\documentclass[reqno,a4paper,oneside,11pt]{amsart}

\usepackage[dvipsnames]{xcolor}
\RequirePackage{doi}
\usepackage{hyperref}
\usepackage{amsmath,amssymb,epsf}
\usepackage{amsfonts,amsbsy,amscd,stmaryrd, wasysym, bbm}
  \SetSymbolFont{stmry}{bold}{U}{stmry}{m}{n}
  \usepackage{mathtools}
  \usepackage{subcaption}
\usepackage{latexsym,amsopn}
\usepackage{amsthm}
\usepackage{esint}
\usepackage{mathrsfs}
\usepackage{slashed} 
\usepackage{enumitem}

\usepackage[margin=1.3in,footskip=0.25in]{geometry}
\allowdisplaybreaks

\usepackage[utf8]{inputenc}
\usepackage[T1]{fontenc}   

\usepackage{graphicx}
\usepackage{multirow}
\usepackage{float}

\usepackage{soul,todonotes}

\newcounter{smallarabics}
\newenvironment{arabicenumerate}
{\begin{list}{{\normalfont\textrm{(\arabic{smallarabics})}}}
  {\usecounter{smallarabics}\setlength{\itemindent}{0cm}
   \setlength{\leftmargin}{5ex}\setlength{\labelwidth}{4ex}
   \setlength{\topsep}{0.75\parsep}\setlength{\partopsep}{0ex}
   \setlength{\itemsep}{0ex}}}
{\end{list}}

\newcounter{smallroman}

\newcommand{\ben}{\begin{arabicenumerate}}  
\newcommand{\een}{\end{arabicenumerate}}

\newtheorem{theorem}{Theorem}[section]

\newtheorem{proposition}[theorem]{Proposition}
\newtheorem{lemma}[theorem]{Lemma}
\newtheorem{corollary}[theorem]{Corollary}
\theoremstyle{definition}
\newtheorem{definition}[theorem]{Definition}
\newtheorem{remark}[theorem]{Remark}
\newtheorem{example}[theorem]{Example}

\newcommand{\beq}{\begin{equation}}
\newcommand{\eeq}{\end{equation}}
\newcommand{\bea}{\begin{aligned}}
\newcommand{\eea}{\end{aligned}}
\newcommand{\bear}{\begin{array}{rl}}
\newcommand{\eear}{\end{array}}
\newcommand{\bex}{\begin{example}}
\newcommand{\eex}{\end{example}}
\def\bel{\begin{lemma}}
\def\eel{\end{lemma}}
\def\bet{\begin{theoreme}}
\def\eet{\end{theoreme}}
\def\bed{\begin{definition}}
\def\eed{\end{definition}}
\def\ber{\begin{remark}}
\def\eer{\end{remark}}
\def\bep{\begin{proposition}}
\def\eep{\end{proposition}}
\newcommand{\qeds}{\qed\medskip}

\let\origmaketitle\maketitle
\def\maketitle{
  \begingroup
  \def\uppercasenonmath##1{} 
  \let\MakeUppercase\relax 
	\origmaketitle
  \endgroup
}

\def\bar{\overline}

\def\cinf{C^\infty}
\def\proof{
\noindent{\bf Proof.}\ \ }

\usepackage{calrsfs}
\DeclareMathAlphabet{\pazocal}{OMS}{zplm}{m}{n}
\def\cA{{\pazocal A}}
\def\cB{{\pazocal B}}

\def\cD{{\pazocal D}}

\def\cF{{\pazocal F}}

\def\cL{{\pazocal L}}
\def\cM{{\pazocal M}}
\def\cN{{\pazocal N}}
\def\cO{{\pazocal O}}
\def\cP{{\pazocal P}}

\def\cR{{\pazocal R}}
\def\cS{{\pazocal S}}
\def\cT{{\pazocal T}}
\def\cU{{\pazocal U}}
\def\cV{{\pazocal V}}
\def\cW{{\pazocal W}}

\def\cZ{{\pazocal Z}}

\def\sH{\mathcal{H}}
\def\sI{\mathcal{I}}
\def\sL{\mathcal{L}}
\def\sR{\mathcal{R}}

\def\fM{{\mathfrak M}}
\def\fN{{\mathfrak N}}
\def\fO{{\mathfrak O}}
\def\fT{{\mathfrak T}}
\def\fV{{\mathfrak V}}
\def\fY{{\mathfrak Y}}
\def\ft{{\mathfrak t}}
\def\fk{{\mathfrak k}}
\def\fh{{\mathfrak h}}
\def\mfi{{\mathfrak i}}

\def\rr{{\mathbb R}}
\def\zz{{\mathbb Z}}
\def\cc{{\mathbb C}}
\def\nn{{\mathbb N}}

\def\ss{{\mathbb S}}

\def\bm{{\mathbbm{m}}}

\def\IRG{{\rm IRG}}
\def\ORG{{\rm ORG}}

\def\i{{\rm i}}

\def\id{{\rm id}}

\let\Im\relax
\let\Re\relax
\let\Ker\relax
\DeclareMathOperator{\Ker}{Ker}
\DeclareMathOperator{\Ran}{Ran}
\DeclareMathOperator{\Im}{Im}
\DeclareMathOperator{\Re}{Re}
\DeclareMathOperator{\Dom}{Dom}

\DeclareMathOperator{\supp}{supp}

\DeclareMathOperator*{\tho}{\text{\th}}
\DeclareMathOperator*{\edt}{\text{\dh}}
\DeclareMathOperator{\Sol}{Sol}

\def\14{\frac{1}{4}}

\def\12{\frac{1}{2}}

\def\e{{\rm e}}
\def\d{{\rm d}}
\newcommand{\bop}{{\mathrm{b}}}

\def\bbbone{{\mathchoice {\rm 1\mskip-4mu l} {\rm 1\mskip-4mu l}
{\rm 1\mskip-4.5mu l} {\rm 1\mskip-5mu l}}}

\def\c{{\rm c}}

\def\id{{\rm id}}

\def\coinf{C_{0}^\infty}

\DeclareSymbolFont{boldoperators}{OT1}{cmr}{bx}{n}
\SetSymbolFont{boldoperators}{bold}{OT1}{cmr}{bx}{n}

\makeatletter
\newcommand*{\defeq}{\mathrel{\rlap{%
                     \raisebox{0.34ex}{$\m@th\cdot$}}%
                     \raisebox{-0.4ex}{$\m@th\cdot$}}%
                     =}
                     
\makeatother
\makeatletter
\newcommand*{\eqdef}{=\mathrel{\rlap{%
                     \raisebox{0.34ex}{$\m@th\cdot$}}%
                     \raisebox{-0.4ex}{$\m@th\cdot$}}%
                     }
\makeatother

\def\Sol{{\rm Sol}}

\def\WF{{\rm WF}}

\def\mo{\mathit{o}}
\def\mi{\imath}

\def\dVol{\mathop{}\!{\rm d}{ vol}}
\def\diff{\mathop{}\!d}

\def\MI{{\rm M}_{\rm I}}
\def\MII{{\rm M}_{\rm II}}

\def\MK{{\rm M}}
\def\MIUII{{\rm M}_{{\rm I}\cup {\rm II}}}
\def\MIII{M_{\rm III}}
\def\2Sol{{\rm Sol}_{{\rm L}^{2}}}

\newcommand{\abs}[1]{{\left\vert #1 \right\vert}}
\newcommand{\norm}[1]{{\left\Vert #1 \right\Vert}}
\newcommand{\IP}[1]{{\left\langle #1 \right\rangle}}
\newcommand{\cv}[1]{{\underline{ #1 }}}

\numberwithin{equation}{section}

\usepackage{tikz}
\usetikzlibrary{cd, patterns,arrows,decorations.pathreplacing, decorations.markings,arrows.meta, shapes.multipart}
\usetikzlibrary{shapes.misc}

\begin{document}
\title[Unruh state for Teukolsky on Kerr]{\Large The Unruh state for bosonic Teukolsky fields on subextreme Kerr spacetimes}
\author{Dietrich \textsc{H\"afner}}
\address{Universit\'e Grenoble Alpes, Institut Fourier, 100 rue des Maths, 38610 Gi\`eres, France}
\email{dietrich.hafner@univ-grenoble-alpes.fr}

\author{Christiane \textsc{Klein}}
\address{University of York, Department of Mathematics, Heslington, York YO10 5DD, United Kingdom\\ Universit\'e Grenoble Alpes, Institut Fourier, 100 rue des Maths, 38610 Gi\`eres, France\\ CY Cergy Paris Universit\'e, 2 avenue Adolphe Chauvin, 95302 Cergy-Pontoise, France}
\email{christiane.klein@york.ac.uk}
\thanks{\emph{Acknowledgments.} We would like to thank Pascal Millet for many helpful discussions at various stages of this work and Peter Hintz for useful comments on a preliminary version of this paper. We acknowledge support from the ANR funding  ANR-20-CE40-0018-01. CK was funded by the Deutsche Forschungsgemeinschaft (DFG, German
Research Foundation) – Projektnummer 531357976.}
\maketitle
\begin{abstract}
We perform the quantization of Teukolsky scalars of spin $0$, $\pm 1$, and $\pm 2$ within the algebraic approach to quantum field theory. We first discuss the classical phase space, from which we subsequently construct the algebra. This sheds light on which fields are conjugates of each other. Further, we construct the Unruh state for this theory on Kerr and show that it is Hadamard on the black hole exterior and the interior up to the inner horizon. This shows not only that Hadamard states exist for this theory, but also extends the existence and Hadamard property of the Unruh state to (bosonic) Teukolsky fields on Kerr, where such a result was previously missing. 
\end{abstract}

\tableofcontents

\section{Introduction}
The quantization of gravity remains, to this day, one of the large open questions of theoretical physics. There are many ideas on what a complete theory of quantized gravity might look like, see e.g. the reviews \cite{Kiefer, deBoer}, but no complete and generally accepted theory. In absence of that, one comparably conservative approach to exploring the phenomenology and physical consequences of the quantization of gravity is to study the quantization of linearized gravity in a non-trivial background spacetime. This can, for example, be done in the framework of algebraic quantum field theory. 

The study of linearized gravity is particularly important on backgrounds describing observable gravitational phenomena, since a good theoretical understanding of these spacetimes is fundamental for the planning and evaluation of observations. An example are rotating black holes, described by Kerr spacetimes. Today, an increasing number of mergers of rotating black holes is observed via gravitational waves \cite{LIGO}, and the shadows of supermassive black holes are made visible by large international collaborations \cite{EHT}. Studying linearized quantum gravity on Kerr spacetimes could help to identify visible effects of quantum gravity in the gravitational wave signature of black hole mergers or the shadow of black holes. It may also give valuable insight into questions such as black hole evaporation. However, rotating black hole spacetimes are mathematically challenging to deal with.

The effects of this can, for example, be seen from the existing work on linearized quantum gravity. The construction of the algebra of observables for free theories such as linearized gravity is quite well understood, see e.g. \cite{FH} and \cite{HS} for a general framework for gauge theories. However, the construction, or even the proof of existence, of physical states for linearized gauge theories is not as straightforward and must be done case by case. In the case of linearized gravity, it is known that the Minkowski vacuum state is a physical state in the sense that it is a Hadamard state \cite{H:PhD}. Moreover, Hadamard states have been constructed based on Wick rotation \cite{GMW, GWgrav}, see also \cite{CMS}, by full gauge fixing on spacetimes with compact Cauchy surfaces \cite{Ggrav}, and for radiative degrees of freedom on asymptotically flat spacetimes without trapping using a bulk-to-boundary construction \cite{BDM}. Noticeably, none of these approaches covers the case of Kerr black hole  spacetimes.  

One of the big difficulties in finding appropriate states for linearized gravity is to reconcile gauge invariance with the requirement of positivity of the state. To circumvent the issue of gauge freedom, one possibility is to quantize a slightly different theory, which does not use the metric perturbation but rather the Teukolsky scalars as its fundamental degrees of freedom. The Teukolsky scalars are derived from the perturbations of the metric, or, for spin one, from the perturbations of the electromagnetic potential, by contracting the Weyl- or field-strength tensor with the elements of a geometrically preferred frame and considering the linearized Bianchi-identities or vacuum Maxwell equations. In the class of algebraically special spacetimes of Petrov type D, the equations for certain contractions of the Weyl- and field-strength tensor perturbations decouple from the rest and take the form of linear, second-order, normally hyperbolic PDEs. 
While this approach is not feasible in generic background spacetimes, in spacetimes of Petrov type D,  which crucially includes Kerr black hole  spacetimes, it allows the degrees of freedom of linearized gravity or electromagnetism to be expressed in a gauge-invariant way \cite{Teukolsky1, Teukolsky2, SteWa}. 

Quantizing the Teukolsky scalars was first proposed in \cite{CCH} as a method of quantizing linearized gravity in Kerr spacetimes. A modification of this theory, introducing the gravitational (or electromagnetic) Hertz potential as a "dual" field, and giving an expression for the Hadamard parametrix, was presented in \cite{Iuliano:2023}. Both works utilize quantized Teukolsky fields to explore the phenomenology of linearized quantum gravity on Kerr spacetimes. However, these works considered neither the structure of the algebra of observables of the Teukolsky scalars, nor the rigorous construction and Hadamard property of the states they employed, in much, if any, detail. In particular, the question of the positivity of the states constructed in \cite{CCH, Iuliano:2023} remained unresolved. Hence, the rigorous construction of the theory and its states is still an open problem. 

 In this work, we therefore consider the quantization of the theory of Teukolsky scalars of spin $0$, $\pm 1$, and $\pm 2$ in the algebraic approach in a rigorous fashion, beginning with the construction of the algebra of observables on subextreme Kerr spacetimes. We consider the Kerr spacetime as consisting of the black hole  exterior and its interior up to, but not including, the inner horizon of the black hole.
 
 Even at this stage, one finds oneself in an unusual situation not commonly considered in the literature: The Teukolsky scalar fields are sections of a non-trivial bundle over the Kerr spacetime. This bundle, the bundle of spin-weighted scalars of spin weight $s$ and boost weight $s$ (or spin $s$ for short if they agree), is an associate bundle to a bundle of geometrically adapted null tetrads, see e.g. \cite{Millet1} for a detailed account. The difficulty in working on this bundle is that the bundle of spin-weighted scalars does not possess a natural or otherwise preferred hermitian fibre metric, which is a key ingredient in the quantization of the theory, unless one globally fixes a scaling for the null vectors in the tetrad to reduce the bundle. Since it will be crucial later on that the scaling of the null vectors can be adapted freely, we need to work on the full bundle.  To provide the missing structure of the spin-weighted scalar bundle, we introduce an extended theory combining Teukolsky scalars of spin $+s$ and $-s$ similar to the ideas in \cite{Iuliano:2023}. The corresponding extended Teukolsky operator can then be shown to be a formally hermitian Green hyperbolic operator, allowing the algebra to be constructed following an established scheme. In some sense, the quantization of the Teukolsky scalar faces issues akin to the quantization of a Green hyperbolic operator with constraints. In this regard, the quantization of Teukolsky is of interest in its own right, independent of its relation to linearized gravity, since it is an example of a theory described by a normally hyperbolic differential operator on a bundle with very little natural structure.

After the construction of the algebra of observables, we construct the Unruh state for this theory on subextreme Kerr spacetimes. This state, which was first considered for scalar fields on Schwarzschild spacetimes in the context of black hole evaporation \cite{Unruh}, displays the behaviour that one would expect at late times in gravitational collapse to a black hole, see e.g. \cite{Hafner} for an argument in the fermionic case. Intuitively, the reason for this is that the Unruh state is constructed such that it resembles the Minkowski vacuum at past null infinity, while containing the expected Hawking radiation at future null infinity. This behaviour indicates that the Unruh state is not an equilibrium state, but rather a non-equilibrium steady state. In particular, in the absence of a global timelike Killing field in the exterior of a Kerr black hole, the Unruh state will rely on two different Killing fields for its definition. In this way, the Uhruh state circumvents the No-Go theorem
by Pinamonti-Sanders-Verch \cite{PSV}.\footnote{The No-Go theorem by Kay-Wald \cite{KW} is circumvented by the fact that we do not consider the state on the whole Kruskal extension of the Kerr spacetime.} This is a large advantage of the Unruh state over the Hartle-Hawking state \cite{HartleHawking, Israel}, which cannot be extended to Kerr according to these No-Go results. Another boon of the Unruh state is that it remains well-defined across the black hole  event horizon, in contrast to the Boulware state \cite{Boulware}, which is only well-defined up to the black hole  horizon. In fact, the extendibility across the black hole  horizon, together with the invariance of the state under the Killing flow generating the horizon, fixes the scaling limit of the two-point function thereon and, by time-translation invariance, on the past horizon \cite{KW}, see also \cite{MoPi}.
Finally, the Unruh state satisfies the Hadamard property, which distinguishes the class of physically reasonable states. This was first proven for the scalar-field Unruh state on Schwarzschild spacetime in \cite{DMP}, and was later extended to a variety of asymptotically de-Sitter black holes \cite{BrJo, HWZ, Klein} and to massless free fermions on Kerr spacetimes \cite{GHW, HaK}. 

As in previous approaches, our construction of the Unruh state is based on a bulk-to-boundary construction which was first introduced independently by Hollands \cite{Hollands:2000} and Moretti \cite{Moretti1}. The crucial step in this construction is the embedding of the algebra of observables assigned to the spacetime into an algebra on its past (conformal) boundary, consisting of the ingoing black hole  horizon and past null infinity. This requires a similar embedding of the (charged) symplectic phase space of the theory on the spacetime into a corresponding (charged) symplectic function space on the aforementioned boundary. In particular, it requires showing the conservation of the (charged) symplectic form. The key ingredients to showing this conservation are sufficiently strong decay results for solutions of the Teukolsky equations. These are provided by the detailed analysis in \cite{Millet2}, see also \cite{Millet:thesis}, based on the Fredholm theory developed in \cite{Vasy2013}. To apply these results, we require a precise understanding of the effect that the discrete spacetime symmetry of simultaneously reversing the time and azimuthal direction has on the non-trivial bundle of spin-weighted scalars and the Teukolsky differential operator.

The biggest challenge in constructing the state and making sure that it is well-defined is to show its positivity. This stems from the fact that the fibre metric of the enlarged vector bundle on which our extended theory lives is not positive. The reason is that, just like the (charged) symplectic form of the Teukolsky phase space, it mixes spin-weighted fields with spin $+s$ and spin $-s$. To solve this problem, we identify a physical subspace of the extended phase space and define the Unruh state on the corresponding physical subalgebra. To identify the physical subspace, we use that the Teukolsky-Starobinsky identities \cite{SC, TP} relate solutions of the Teukolsky equation for spin $+s$ with those for spin $-s$. In addition, the Hertz potential appearing in most of the reconstruction schemes \cite{Chrzanowski, Ori, Pound:2021, TZSHPG, BGL} provides another relation between different vacuum solutions to the Teukolsky equation which was also used in \cite{Iuliano:2023}.
The issue of positivity of the state then depends delicately on the properties, predominantly the positivity and invertibility, of the angular differential operator appearing in the Teukolsky-Starobinsky identities \cite{SC, TP, CTdC}, as well as its commutator with the Teukolsky operator, which we work out. 

Another difficulty faced in the construction of the state that was already encountered in the work \cite{DMP} is that on the ingoing black hole  horizon, on which one part of the state is defined, one would like to work in the Kruskal coordinates which cover the bifurcation sphere. However, in terms of these coordinates, the decay of functions in the required boundary function space is too weak for them to be square integrable. This issue is further amplified in the present case by the need to work with a tetrad which smoothly extends to the bifurcation sphere. The reason is the scaling of this tetrad towards past timelike infinity compared to more commonly used ones. Beyond the methods employed in \cite{DMP}, it requires us to make subtle changes to the definition of the state in order to re-distribute decay and ultimately show the positivity of the state. More concretely, we utilize the connection of the spin $+s$- and spin $-s$-fields in the physical subspaces to identify differential operators mapping them to spin-weighted functions with vanishing boost weight. In addition to alleviating the problem of the tetrad scaling, this bundle has a hermitian fibre metric allowing for the definition of natural function spaces on the boundaries.

Finally, we wish to show that the state is Hadamard. As alluded to earlier, the Hadamard property is required for a state to be considered as physically reasonable. In its original formulation (given for scalar fields), the Hadamard property demands that the two-point function of the state has the singularity structure of a Hadamard parametrix, see e.g. \cite{Wald1977_SET, FSW, KW} for early works. This was motivated on the one hand by the fact that this singularity structure is a natural extension of that of the Minkowski vacuum to curved spacetimes. On the other hand, the Hadamard parametrix is local and covariant, providing the basis for the local and covariant point-split renormalization scheme. A major advancement in the theory was the work of Radzikowski \cite{Radzikowski} which showed that the Hadamard property of states for the free scalar field can be expressed in terms of the {\it microlocal spectrum condition}, a condition on the form of the wavefront set of the state's two-point function. A generalization of this notion to a wide class of Green-hyperbolic operators, which we make use of in this work, has been given recently in \cite{F}. 
The proof of the Hadamard property of the Unruh state for the Teukolsky scalars on Kerr combines ideas from different previous works on the Unruh state. However, to complete the proof, these ideas are not sufficient. We combine them with microlocal estimates, which again depend on the estimates shown in \cite{Millet2} in a crucial way. 

Overall, we show
\begin{theorem}
    Let $s\in \{0,1,2\}$. The theory of a spin-$\pm s$ Teukolsky scalar on any subextreme Kerr spacetime, consisting of the black hole  exterior and its interior up to the inner horizon (i.e. the spacetime $\MI\cup \sH_+\cup\MII$, see Section~\ref{subsec:Kstar and starK}), can be quantized by using an enlarged classical phase space encompassing Teukolsky fields of spin $+s$ and spin $-s$. A physical subalgebra corresponding to the algebra of the Hertz potential can be identified.

    On the physical subalgebra, one can define an analogue of the Unruh state. The state is well-defined and Hadamard on any subextreme Kerr spacetime.
\end{theorem}

\subsection{Outlook}
Our rigorous approach to the quantization of Teukolsky scalars shows that there are some subtleties in the construction of the state, which should be taken into account if this theory is used in numerical exploration of physical effects as was done in \cite{CCH, Iuliano:2023}.
Since we quantize source-free solutions to the Teukolsky equation, the field-strength tensor of electromagnetism or the metric perturbation can be reconstructed from the Teukolsky scalars \cite{WK,  Wald, Chrzanowski, PSW, Ori, Pound:2021, TZSHPG, BGL}. For this reason, we are optimistic that the state constructed here can serve as a foundation for the construction of states for linearized gravity and electromagnetism. This could not only help to show the existence of Hadamard states for these theories on rotating black hole spacetimes, but also give access to a physically motivated, explicitly formulated example of a state on these spacetimes. The latter is an important ingredient for the numerical study of quantum effects on black hole spacetimes, which is an active field of research, see for example \cite{OZ, ZCOO, ABOT, PFDW}.

\subsection{Structure of the paper} 
The paper is structured as follows. In Section~\ref{sec:bundle}, we introduce in some detail the bundles of spin-weighted scalars and discuss some of their properties. Additionally, we introduce the extended bundle used in the quantization of the Teukolsky scalars. Section~\ref{sec:Kerr} introduces the Kerr spacetime including various extensions thereof. This section also contains a discussion on the different tetrads which are used in this work to locally trivialize the bundle of spin-weighted scalars. The Teukolsky operator is introduced in Section~\ref{sec:Teukolsky op}. After discussing its properties pertinent to the present work, we introduce its extension that will be used for the quantization. We use it to construct a charged symplectic phase space of the classical theory. We then introduce the Teukolsky-Starobinsky identities and use them to identify a physical subspace of the classical phase space. Subsequently, the analytic framework for the decay estimates used in this work is set up in Section~\ref{sec:analytic}. The results of \cite{Millet2} required in this work are recalled in Section~\ref{sec:millet}. They are employed in Section~\ref{sec:sympl form} to show that the charged symplectic form of the enlarged Teukolsky phase space, which is determined by the integral of a conserved current over a Cauchy surface of the Kerr spacetime, is conserved even when this surface is taken to be the past boundary of the spacetime, or, more specifically, the union of the ingoing black hole  horizon and past null infinity. This allows us to embed the spacetime phase space into a charged symplectic function space on this boundary. Section~\ref{sec:algebra} recalls the quantization of a (charged) symplectic space in the form of a CCR-algebra, as well as the definition of Hadamard states thereon. Following this definition, the algebra of Teukolsky scalars as well as its physical subalgebra are defined based on the spaces constructed earlier. Subsequently, we define the Unruh state in Section~\ref{sec:Unruh} using a bulk-to-boundary construction, and show that it is indeed a well-defined state, and in particular positive, on the physical subalgebra. Finally, we show that the Unruh state thus defined is a Hadamard state in Section~\ref{sec:Hadamard}.

\subsection{Preliminaries, notations, and conventions}

Our convention of signs and prefactors for the Fourier transform on $\rr^n$ is
\begin{align}
   \hat f (k)&=(2\pi)^{-n/2}\int\limits_{\rr^n} e^{\i k\cdot x} f(x) \d^n x\, , 
\end{align}
where $k\cdot x$ is the usual inner product on $\rr^n$.

We use the physics convention that a sesquilinear form is anti-linear in its first and linear in its second argument.

We use the notation $\cc^*=\cc\setminus\{0\}$, $\rr_+=[0,\infty)$, $\rr_+^*=(0,\infty)$, $\cc_+^{(*)}=\rr+\i\rr_+^{(*)}$, and analogously for $\rr_-^{(*)}$ and $\cc_-^{(*)}$.
For $x\in \cc^n$, $\IP{x}=(1+\abs{x}^2)^{1/2}$ will denote its Japanese bracket.

In this work, a spacetime is considered to be a tuple $(M,g,\fO, \fT)$, where $M$ is a smooth, 4-dimensional, connected, Hausdorff, paracompact manifold equipped with a smooth Lorentzian metric $g$ of signature $(+---)$. We assume that $(M,g)$ is globally hyperbolic, orientable and time-orientable. $\fO$ is an equivalence class of nowhere vanishing 4-forms which expresses the orientation of the spacetime, and $\fT$ is an equivalence class of nowhere vanishing $g$-timelike 1-forms (i.e. $g^{-1}(\fT, \fT)>0$) expressing the time orientation. In both cases, $\fO$ and $\fO'$ or $\fT$ and $\fT'$ are equivalent if there is a strictly positive function $f\in C^\infty(M;\rr)$ so that $\fO'=f\fO$ or $\fT'=f\fT$. The volume form induced by the metric $g$ will be denoted by $\dVol_g$.

If $x\in M$, then the causal future of $x$, $J^+(x)$, denotes the set of all points $y\in M$ that can be reached from $x$ by a smooth, future-directed causal curve (i.e. a curve $\gamma:[a,b]\to M$ with $\gamma(a)=x$, $\gamma(b)=y$, and tangent vector $\gamma'$ satisfying $\fT(\gamma')> 0$ and $g(\gamma', \gamma')\geq 0$). Similarly, the causal past of $x$, $J^-(x)$, denotes the set of points that can be reached by a smooth past-directed causal curve. We denote $J(x)=J^+(x)\cup J^-(x)$, and for $K\subset M$, $J^\pm(K)=\bigcup_{x\in K} J^\pm(x)$ and $J(K)=J^+(K)\cup J^-(K)$. In a similar way, we define the timelike future of $x\in M$, $I^+(x)$ as the set of points $y\in M$ that can be reached from $x$ by a smooth, non-trivial future-directed timelike curve ($\gamma:[a,b]\to M$ with $a\neq b$, $\fT(\gamma')>0$, and $g(\gamma',\gamma')>0$). In particular, note that $x\notin I^+(x)$, but $x\in J^+(x)$. $I^-(x)$ is defined analogously. Again, we set $I^\pm(K)=\bigcup_{x\in K} I^\pm(x)$ and $I(K)=I^+(K)\cup I^-(K)$.

If $E\xrightarrow{\pi}{M}$ is a vector bundle over $M$, we denote its dual bundle by $E^\#$, its complex conjugate bundle by $\bar{E}$, and its anti-dual bundle by $E^*$. The fibre $\pi^{-1}(x)$ for some $x\in M$ will frequently be denoted as $E_x$. We denote a vector bundle $E$ with its zero-section removed as $\dot E$. 

 If $F\xrightarrow{\tilde\pi}M$ is a vector bundle over the same manifold, we denote by ${\rm Hom}(E;F)$ the vector bundle of homomorphisms from $E$ to $F$, i.e., the elements of the fibres ${\rm Hom}(E;F)_x$ are linear maps from $E_x$ to $F_x$. We write ${\rm Hom}(E)$ for ${\rm Hom}(E;E)$.

The space of smooth sections of $E$ is denoted by $\Gamma(E)$. $\Gamma_{c/sc/pc/fc}(E)$ denote the space of sections with compact/space compact/past compact/future compact support. Here, a section $f\in \Gamma(E)$ is space-compact if there exists a compact set  $K\subset M$ so that $\supp(f)\subset J(K)$. $f\in \Gamma_{pc/fc}(E)$ means that $\supp f\cap J^\mp(x)$ is compact for any $x\in M$. 

Similarly, if $S\subset M$, $\Gamma_S(E)$ denotes the space of sections with compact support contained in $S$. For sections defined on $S\subset M$, we write $\Gamma(E\vert_S)$.

\paragraph{\bf Stokes formula}
Suppose that the vector field $V^a$ is divergence free, $\nabla_a V^a=0$. Then the 3-form $*V_a$, where $*$ is the Hodge dual, is closed, i.e., $
\d(*V_a)
=0.$
By the Stokes theorem we have for a bounded domain $U$ with smooth boundary $\partial U$
\begin{equation*}
\int\limits_U d(*V_a)=\int\limits_{\partial U} *V_a. 
\end{equation*}
Let $\Sigma$ be a spacelike or null hypersurface. Let $n^a$ be a future normal to $\Sigma$ ($n_a\vert_{T\Sigma}=0$) and $l^a$ transverse to $\Sigma$ such that $l_ an^a=1.$ Then we have 
\begin{align*}
\int\limits_{\Sigma} *V_c=
\int\limits_{\Sigma} V_an^a l^b\lrcorner \dVol_g.
\end{align*} 
If $\Sigma$ is spacelike we can take $n^a=l^a$.

\section{The space of spin weighted functions}
\label{sec:bundle}

\subsection{The bundle of oriented, normalized principal null frames}

In this work, we will assume that the spacetime is a solution to the vacuum Einstein equations $Ric(g)=0$ of Petrov type D. In other words, at every point $p\in M$, there will be two linearly independent principle null directions of multiplicity at least 2, i.e. two linearly independent null vector fields $V^\pm$ defined on a neighbourhood $\cU$ of $p$ so that for all $x\in\cU$ and $ v\in T_x M $ satisfying $g_x(v,V^\pm(x))=0$, one has $C(V^\pm(x) , v, V^\pm(x) ,w)=0$  for all $w\in T_xM$. Here, $C_{abcd}$ is the Weyl tensor constructed from $g$. See \cite[Chapter 5]{ON} for more detail.

On $(M,g,\fO,\fT)$, let $(l,n,m)\in T_{\cc}M^3$ be chosen such that $l$ and $n$ are real, align with the principal null directions, and are future directed. Moreover, we demand that $g(l,n)=-g(m,\bar{m})=1$, and all other combinations vanish. Finally, we demand that the null frame $(l,n,m,\bar{m})$ is oriented in the sense that 
\begin{align*}
\left(2^{-1/2}(n+l),\Re(m),-\Im(m),2^{-1/2}(l-n)\right)
\end{align*}
is oriented, i.e. $\fO(2^{-1/2}(l+n), \Re(m), -\Im(m), 2^{-1/2}(l-n))>0$.

\begin{proposition}{\cite{Millet1}}
\label{prop:N}
 The set $\fN$ of all $(l,n,m)\in T_\cc M^3$ satisfying the above conditions is a smooth submanifold of $T_\cc M^3$ and a $\cc^*\rtimes_h\zz/2\zz$\footnote{$\cc^*\rtimes_h \zz/2\zz$ is the outer semi-direct product of the (multiplication) group $\cc^*$ and the (addition) group $\zz/2\zz$. Here, $h:\zz/2\zz\to {\rm Aut}(\cc^*)$ is a group homomorphism. The group $\cc^*\rtimes_h\zz/2\zz$ consists of the set $\cc^*\times \zz/2\zz$ endowed with the group operation
 \begin{equation}
     (z_1,[a]_2)\bullet (z_2, [b]_2)=(z_1 h([a]_2)(z_2), [a+b]_2)\, .
 \end{equation}}
 principal bundle for the right actions of $\cc^*$ (as a group under multiplication) given by $(l,n,m)\cdot z=(\abs{z}l,\abs{z}^{-1}n,z\abs{z}^{-1}m)$ and of $\zz/2\zz$ given by $(l,n,m)\cdot[1]_2=(n,l,\bar{m})$, with $h([1]_2)(z)=z^{-1}$ (performing the action of $\cc^*$ first). The projection map is $\pi_\fN=(\pi_{T_\cc M^3})\vert_\fN$.

$\fN$ has exactly two connected components if either $M$ is simply connected, or $M$ is connected and there exist two smooth, null, future oriented vector fields $V^\pm$ that are independent and principal at any $x\in M$. In this case, the connected component 
\begin{align*}
\fN_0=\{(l,n,m)\in \fN:\, l\text{ resp. }n\text{ collinear to } V^+\text{ resp. }V^-\}
\end{align*}
is given the structure of a $\cc^*$ principal bundle by the subgroup $\{(z,[0]_2):\, z\in \cc^*\}$, while the action of $(1,[1]_2)$ induces a diffeomorphism between $\fN_0$ and the other connected component.
\end{proposition}

From now on, we assume that there is a global, smooth choice of the principal null directions $V^\pm$, so that $\fN$ has two connected components. In this case, one may consider the background to consist of the tuple $\fM=(M,g,\fO, \fT, \fV)$, where $(M,g,\fO,\fT)$ is the spacetime, and $\fV$ is an equivalence class of nowhere vanishing 1-forms so that $g^{-1}(\fV)$ is a principal null direction at any $x\in M$. Equivalence of $\fV$ and $\fV'$ holds if there is a non-vanishing function $f\in C^\infty(M;\rr)$ so that $\fV'=f\fV$. The equivalence class $\fV$ expresses the choice of "outgoing" principle null direction. One can then define $\fN_0$ on this background by setting $V^+=g^{-1}(\fV)$ and defining $\fN_0$ as above, or by expressing the condition that $l$ be collinear to $V^+$ as the condition that $\fV(l)=0$.

In agreement with the literature, we will refer to the diffeomorphism induced by the action of $(1,[1]_2)$ as  the "prime" operation, denoted by a prime.

\begin{definition}
\label{def:bck map and bck pred map}
     If $\psi:M\to \widetilde{M}$ is an embedding, and if $\widetilde{M}$ can be equipped with $(\widetilde{g},\widetilde{\fO},\widetilde{ \fT},\widetilde{ \fV})$ so  that $\widetilde{M}$ is globally hyperbolic and of Petrov type D, $(\psi(M)\subset \widetilde{M}, \widetilde{g})$ is causally convex, and $\psi^*(\widetilde{g},\widetilde{\fO}, \widetilde{\fT}, \widetilde{\fV})=(g, \fO,\fT, \fV)$, then $\psi$ induces a {\it background map} $\fM\to \widetilde{\fM}=(\widetilde{M},\widetilde{g},\widetilde{\fO},\widetilde{\fT},\widetilde{\fV})$ which we will also denote by $\psi$ in an abuse of notation. 
   If $\cU_{1/2}$ are open subsets of $M$ and $\psi:\cU_1\to\cU_2$ is a diffeomorphism satisfying $\psi^*(g,\fO,\fT,\fV)\vert_{\cU_2}=(g,\fO,\fT,\fV)\vert_{\cU_1}$, then we call (the background map induced by) $\psi$ a \it{local background preserving diffeomorphism}.
\end{definition}

\begin{definition}
\label{def:N0 cover map}
    Let $\upsilon: \fM\to  \widetilde{\fM}$ be a background map. 
    Then we define the cover map $\Upsilon:\fN_0\to \widetilde{\fN}_0$ by
    \begin{align}
        \fN_{0,p}\ni (l,n,m) \mapsto \upsilon_*(l,n,m) \in  \widetilde{\fN}_{0,\upsilon p}
    \end{align}
    where $p\in M$, and $\upsilon_*$ is the push-forward on $T_\cc M^3$ restricted to $\fN_0$.\footnote{ Here and in the following, the bundle $\widetilde{\fN}_0$ denotes the bundle $\fN_0$ constructed over the background $\widetilde{\fM}$, and similar for other bundles constructed from the background.} It is a linear bijection from $\pi_{\fN_0}^{-1}(p)$ to $\pi_{ \widetilde{\fN}_0}^{-1}(\upsilon p)$ for all $p\in M$. \\
    The push-forward $\Upsilon^*: \Gamma(\fN_0)\to\Gamma(\widetilde{\fN}_0)$ is given by $\Upsilon^*(f)=\Upsilon\circ f \circ \upsilon^{-1}$ for all $f\in \Gamma(\fN_0)$ \cite{SV}. 
\end{definition}

\begin{remark}
\label{rem:bp diffeos}
Let $\cU_1$, $\cU_2\subset M$ and $\upsilon:\cU_1\to \cU_2$ be a background preserving diffeomorphism. In this case, the cover map $\Upsilon$ maps $\fN_0\vert_{\cU_1}$ to $\fN_0\vert_{\cU_2}$.  If there is an  $\upsilon$-invariant section $(l,n,m)\in \Gamma( \fN_{0}\vert_\cU)$ defined on some region $\cU\subset M$ containing $\cU_{1/2}$, i.e., a section which satisfies
\begin{align}
    \upsilon_*((l,n,m)(\upsilon^{-1}q))=(l,n,m)(q) \quad \forall q\in \cU_2\, ,
\end{align}
then for any $p\in \cU_1$ and any $a=(l,n,m)(p)\cdot z\in \fN_{0,p}$, one can write
\begin{align}
    \Upsilon(a)= (l,n,m)(\upsilon p)\cdot z\, .
\end{align}
 As a consequence, if $(\tilde{l},\tilde{n},\tilde{m})\in \Gamma(\fN_{0}\vert_{\cU})$, so that  for every $x\in\cU$
    \begin{align*}
        (\tilde{l},\tilde{n},\tilde{m})(x)=(l,n,m)(x)\cdot z(x)\, ,\quad z\in\cinf(\cU;\cc^*)\, ,
    \end{align*}
 then 
\begin{align*}
    (\Upsilon^*(\tilde{l},\tilde{n},\tilde{m}))( q)=\Upsilon((\tilde{l},\tilde{n},\tilde{m})(\upsilon^{-1} q))=(\tilde{l},\tilde{n},\tilde{m})( q)\cdot z^{-1}(q)\cdot z(\upsilon^{-1}q) \quad \forall q\in \cU_2\, .
\end{align*}   
\end{remark}

\subsection{The bundle of spin-weighted scalars}

\begin{definition} 
Let $s,w\in \tfrac{1}{2}\zz$, and let the representation $\rho_{(s,w)}:\cc^*\to \text{GL}(\cc)$ be given by
\begin{align}
 z\mapsto (a\mapsto \abs{z}^{-(w-s)}z^{-s} a)\, .
\end{align}
$s$ is called the spin weight and $w$ the boost weight of the representation.
We define $\cB(s,w)\xrightarrow{\pi_{(s,w)}}M$ as the complex line bundle associated to $\fN_0$ (as a $\cc^*$ principal bundle with the right action described above) and the representation $\rho_{(s,w)} $.\footnote{Alternatively, one can consider $\cB(s,w)$ as an associated bundle to the $\cc^*$-bundle $\cA_0$ of orthonormal spin frames $(o, \iota)$, such that $l=o\otimes \bar{o}$ and $n=\iota\otimes \bar{\iota}$ are independent and future directed along the principal null directions $V^+$ and $V^-$ (see Proposition \ref{prop:N}), with right action $(o, \iota)\cdot \lambda=(\lambda o, \lambda^{-1}\iota)$. In this picture, the representation is
\begin{align*}
\widetilde\rho_{(s,w)}:\cc^*\to \text{GL}(\cc)\, \quad \lambda\mapsto (a\mapsto \lambda^{-w-s}\overline{\lambda}^{-w+s} a)
\end{align*}
This can be identified with the above by using the double cover map $d$ from $\cA_0$ to $\fN_0$, which for any $a\in \cA_0$ satisfies $d(a\cdot \lambda)=d(a)\cdot \lambda^ 2$, see e.g. \cite{Millet1}.} 
In other words, $\cB(s,w)=\fN_0\times \cc /\sim$, where for $a, a' \in \fN_{0,x}$, $c, c' \in \cc$, $(a, c)\sim(a',c')$ if there is a $z\in \cc^*$ so that $(a', c')=(a\cdot z,\rho_{(s,w)}(z^{-1})c)$. 

\end{definition}

The sections of this bundle will be called the spin weighted functions of spin $s$ and boost weight $w$. For $(s,w)\in \zz\times\zz$, they arise as the components of tensor densities in a null frame determined by $(l,n,m)\in  \Gamma(\fN_0)$. They can be identified with equivariant functions on $\fN_0$, i.e. functions $f:\fN_0\to\cc$ so that for any $z\in \cc^*$ and $a\in \fN_0$, $f(a\cdot z)=\rho_{(s,w)}\left(z^{-1}\right)f(a)$.

\begin{remark}
\label{rem:prime cc}
As remarked in \cite{Millet1}, one has the identification $\cB(s+s', w+w')=\cB(s,w)\otimes\cB(s',w')$. Moreover, for the components of conformal tensor densities, the prime relation on $\fN$ induces a relation $':\cB(s,w)\to \cB(-s,-w)$. In addition, complex conjugation maps $\bar{\cdot}:\cB(s,w)\to \cB(-s,w)$.
\end{remark}

\begin{remark}
Let $(l,n,m)$ be a smooth section of $\fN_0$ defined on some open set $\cU\subset M$. Then this induces a local trivialization of $\fN_0$, $\psi:\cU\times \cc^*\to \pi_{\fN_0}^{-1}(\cU)$, $(x,z)\mapsto (l,n, m)(x)\cdot z$. It also induces a local trivialization of $\cB(s,w)$ over $\cU$ for any $s,w\in\rr$, namely the map $\psi_{(s,w)}:\cU\times \cc\ni (x,c)\mapsto [((l,n,m)(x),c)]$. 

Moreover, given another smooth section $(\tilde{l},\tilde{n},\tilde{m})$ of $\fN_0$ defined over $\tilde\cU\subset M$ so that for $x\in \cU\cap \tilde\cU$, one has $(\tilde{l},\tilde{n}, \tilde{m})(x)=(l,n,m)(x)\cdot z(x)$, let  $\tilde{\psi}_{(s,w)}$ be the local trivializations of $\cB(s,w)$ over $\tilde\cU$ induced by this section. Then the transition function $\tau:=\tilde\psi_{(s,w)}^{-1}\circ \psi_{(s,w)}: ((\cU\cap\tilde\cU)\times \cc)\to ((\cU\cap\tilde\cU)\times \cc)$ 
is given by 
\begin{align*}
\tau (x,c)=\left(x, \rho_{(s,w)}(z^{-1}(x))c\right)=\left(x,\abs{z(x)}^{w-s}  z(x)^s c\right)\, .
\end{align*}
\end{remark}

\begin{definition}
\label{def:cover of diffeos}
    Let $\upsilon:\fM\to \widetilde{\fM}$ be a background map. 
     Then we define the cover map $\Upsilon_{(s,w)}:\cB(s,w)\to \widetilde{\cB}(s,w)$ of $\upsilon$ by
    \begin{align}
       \cB(s,w)_p\ni [(a,c)] \mapsto [(\Upsilon a,c)]=[(\upsilon_*a,c)]\in \widetilde{\cB}(s,w)_{\upsilon p}\, ,
    \end{align}
     for any $p\in M$, $a\in \fN_{0,p}$, and $c\in \cc$.
     $\Upsilon$ is the cover map of $\upsilon$ on $\fN_0$ defined in Definition~\ref{def:N0 cover map}.

    The map $\Upsilon_{(s,w)}$ induces a push-forward map
    \begin{align*}
       \Gamma(\cB(s,w))\ni f\mapsto\Upsilon_{(s,w)}^*f= \Upsilon_{(s,w)}\circ f \circ \upsilon^{-1}\in \Gamma(\widetilde{\cB}(s,w))\, .
    \end{align*}
\end{definition}

\begin{lemma}
\label{lemma:trafo of triv. flow}
    Let $\cU$, $\cU_1$, $\cU_2\subset M$, so that $\cU_1\cup\,\cU_2\subset \cU$. Let $\upsilon:\cU_1\to\cU_2$ be a background preserving diffeomorphism. Let $(l,n,m)(p)\in  \Gamma(\fN_{0}\vert_{\cU})$ be an $\upsilon$-invariant section.
    Let $C:x\mapsto [a(x), c(x)]\in \Gamma(\cB(s,w)\vert_{\cU_1})$ with $a(x)=(l,n,m)(x)\cdot z(x)$ be a section of $\fN_0$ defined on $\cU_1\cup\cU_2$. Then for any $p\in \cU_2$
    \begin{align}
    \label{eq:trafo w inv tetrad}
        (\Upsilon_{(s,w)}^*C)(p)=[a(p), \rho_{(s,w)}(z^{-1}(p))\rho_{(s,w)}(z(\upsilon^{-1}p))c(\upsilon^{-1}p)]\, .
    \end{align}
\end{lemma}

\begin{proof}
    Using the results of Remark~\ref{rem:bp diffeos}, it is straightforward to compute
    \begin{align*}
        (\Upsilon_{(s,w)}^*C)(p)&=\Upsilon_{(s,w)}(C(\upsilon^{-1}p))=[(\Upsilon a(\upsilon^{-1}p), c(\upsilon^{-1}p))]\\
        &=[(a(p)\cdot z^{-1}(p)\cdot z(\upsilon^{-1}p), c(\upsilon^{-1}p))]\\
        &=[(a(p), \rho_{(s,w)}(z^{-1}(p))\rho_{(s,w)}(z(\upsilon^{-1}p))c(\upsilon^{-1}p))]\, .
    \end{align*}
\qeds
\end{proof}

\begin{lemma}
\label{lemma:dual}
The dual bundle of $\cB(s,w)$ is $\cB(s,w)^\#=\cB(-s,-w)$, the anti-dual bundle is $\cB(s,w)^*=\cB(s,-w)$.
\end{lemma}

\begin{proof} The dual representation to $\rho_{(s,w)}(z)$ is $\rho_{(s,w)}(z^{-1})=\rho_{(-s,-w)}(z)$. Let $x\in M$, $[a, c]\in \cB(s,w)_x$ and $[b, d]\in \cB(-s,-w)_x$. There is a $z\in \cc^*$ so that $b=a\cdot z$. Then define 
\begin{align}
\label{eq: dual eval}
[b,d]\left([a,c]\right):=  \rho_{(s,w)}(z^{-1})c d= \rho_{(-s,-w)}(z)c d\, .
\end{align}
A straightforward calculation then shows that this is indeed well-defined, linear and therefore maps elements of $\cB(-s,-w)_x$ to elements of $\cB(s,w)^\#_x=\text{Hom}(\cB(s,w)_x,  \cc)$. Similarly, let $\phi\in \cB(s,w)^\#_x=\text{Hom}(\cB(s,w)_x, \cc)$. 
Then for any $[a,b]$ and $[c,d]$ in $\cB(s,w)_x$, with $c=a\cdot z_c\in \fN_{0,x}$, and any $z\in \cc^*$, one must have
\begin{align*}
    \phi([a,b])&=\phi([a\cdot z, \rho_{(s,w)}(z^{-1})b])\\
    \phi([a,b+\rho_{(s,w)}(z_c)d])&=\phi([a,b])+\phi([c,d])\, .
\end{align*}
 For $a\in \fN_{0,x}$, define $\phi_a:\cc\to\cc$ as the map $\phi_a(b)=\phi([a,b])$. It follows from the above that $\phi_a\in \text{Hom}(\cc,\cc)$, which is isomorphic to $\cc$, so we may write $\phi_a(b)=\phi_a b$ by suppressing this isomorphism. It then follows that $\phi_{a\cdot z}=\rho_{(s,w)}(z)\phi_a=\rho_{(-s,-w)}(z^{-1})\phi_a$. Hence, $\phi$ can be identified with $[a, \phi_a]\in \cB(-s,-w)_x$, and consequently $\cB(s,w)^\#=\cB(-s,-w)$.

Since $s$ and $w$ were arbitrary in this discussion, it follows immediately that in the same way $\cB(s,-w)$ can be identified with the dual of $\cB(-s,w)$, the complex conjugate of $\cB(s,w)$ by Remark \ref{rem:prime cc} 
\qeds
\end{proof}
 As a result of this, it follows immediately
\begin{corollary}
\label{cor:B(s,0) IP}
    There is a hermitian fibre metric on $\cB(s,0)$ given by 
    \begin{equation}
        \IP{\cdot,\cdot}_{0,x}:(f,g)\in \cB(s,0)_x\times\cB(s,0)_x\mapsto \IP{f,g}_{0,x}:=g(\bar{f})\in \cc\, ,
    \end{equation}
    using the notation from \eqref{eq: dual eval}.
\end{corollary}

In the following, the connection on the bundle of spin-weighted scalars plays an important role. We use the common definition, see \cite{GHP}, and \cite{Millet1} for a detailed account.
\begin{definition}
\label{def:connection}
 Let $\fM=(M,g, \fO, \fT, \fV)$ be a background, and let $(l,n,m)\in \Gamma(\fN_0\vert_\cU)$ for some set $\cU\subset M$. Then we define the connection $\Theta_a: \Gamma(\cB(s,w))\to \Gamma(\cB(s,w))\otimes T^*\cM$ to be given by 
\begin{align}
\Theta_a&=\nabla_a-(w+s)w_a-(w-s)\bar{w}_a\, , &  w_a&=\frac{1}{2}\left(n^b\nabla_al_b+m^b\nabla_a\bar{m}_b\right)\, 
\end{align}
in the trivialization induced by $(l,n,m)$ on $\cU$.
\end{definition}

Contracting with the corresponding null frame, one obtains the GHP operators \cite{GHP}
\begin{subequations}
\begin{align}
\label{eq:tho}
\tho&=l^{a}\Theta_a :\,\Gamma\cB(s,w))\to \Gamma(\cB(s,w+1)) \\
\label{eq:thop}
 {\tho}'&=n^{a}\Theta_a\, :\,\Gamma(\cB(s,w))\to \Gamma(\cB(s,w-1))\\
\edt&=m^{a}\Theta_a:\,\Gamma(\cB(s,w))\to \Gamma(\cB(s+1,w)) \\
 {\edt}'&=\bar{m}^{a}\Theta_a:\,\Gamma(\cB(s,w))\to \Gamma(\cB(s-1,w))\,.
\end{align}
\end{subequations}
When referring to a GHP-operator acting on spin-weighted functions of a specific spin $s_0$ and boost weight $w_0$, we will indicate the weights in form of a subscript, for example
${\edt}_{(s_0,w_0)}:\Gamma(\cB(s_0,w_0))\to \Gamma(\cB(s_0+1,w_0))$\, .

The contracted derivatives of the null frame produce the spin coefficients, of which \cite{GHP, Aksteiner}
\begin{subequations}
\begin{align}
\rho&=m^{a}\bar{m}^b\nabla_bl_a\in \Gamma(\cB(0,1)) & \rho'&=\bar{m}^{a}m^b\nabla_bn_a \in \Gamma(\cB(0,-1))\\
\tau&=m^{a}n^b\nabla_bl_a \in \Gamma(\cB(1,0))& \tau'&=\bar{m}^{a}l^b\nabla_bn_a\in \Gamma(\cB(-1,0))
\end{align}
\end{subequations}
and their complex conjugates are the only ones which are spin-weighted functions and do not vanish in a Petrov type D spacetime.

\subsection{The reduced bundle}
\label{sec:red bundle}
Let us note that if the spacetime is of Petrov type D, one can make a global choice of the normalisation of the two null vector fields $l$ and $n$. Doing so allows one to identify \cite{Millet1}
\begin{align*}
\fN_{0,r}:=\fN_0/\rr_+^*\simeq \left\{m\in T^*_\cc M: (l,n,m)\in \fN_0\right\}\, ,
\end{align*}
where the quotient is with respect to the action of $\rr_+^*\subset \cc^*$. With this identification, one can associate a local smooth section of $\fN_0$ to any local smooth section of $\fN_{0,r}$.
The reduced bundle $\fN_{0,r}$ has the structure of a $U(1)$-bundle.

\begin{definition}
 For $s\in \tfrac{1}{2}\zz$, take $\nu_s:\phi\mapsto e^{-i s\phi}$ as a representation of $U(1)$ on $\cc$. Then we define $\cB(s)$ as the associated complex line bundle to $\fN_{0,r}$ and the representation $\nu_s$.
\end{definition}
Then one has \cite{Millet1}
\begin{lemma}
\label{isombundles}
$\cB(s)$ is isomorphic to $\cB(s,w)$ for any $w\in \tfrac{1}{2}\zz$.
\end{lemma}

\begin{proof}
First, let $[m,a]\in \cB(s)$, with $m\in \fN_{0,r}$ and $a\in \cc$.  Fix a pair $(l,n)$ of real null vectors so that $(l,n,m)\in \fN_0$. Then the map $\mu_{(w,l,n)}:[m,a]\mapsto [(l,n,m),a]\in \cB(s,w)$ is well defined:
\begin{align*}
&\mu_{(w,l,n)}([m\cdot\phi, \nu_s(\phi^{-1})a])=[(l,n,e^{i\phi} m),e^{is\phi}a ]\\
&=[(l,n,m)\cdot e^{i\phi}, \rho_{(s,w)}(e^{-i\phi})a]=[(l,n,m),a]=\mu_{(w,l,n)}([m,a])\, .
\end{align*}
From the definition, we see that the kernel of $\mu_{(w,l,n)}$ is the zero section. Moreover, let $[(\widetilde{l},\widetilde{n},\widetilde{m}), c]$ be any element of $\cB(s,w)$. Then there is a $z\in \cc$ such that $(\widetilde{l},\widetilde{n},\widetilde{m})=(l,n,m)\cdot z$, so $[(\widetilde{l},\widetilde{n},\widetilde{m}), c]=[(l,n,m),\rho_{(s,w)}(z)c]$, and this can be identified with $\mu_{(w,l,n)}([m,\rho_{(s,w)}(z)c]$. Hence $\mu_{(w,l,m)}$ is an isomorphism.
\qeds
\end{proof}

Note that the above identification depends on the choice of $l$ and $n$, and that different choices lead to different identifications.
 \begin{remark}
\label{remconnection1}
Let $\mu_{(w,l,n)}:\cB(s)\rightarrow \cB(s,w)$ the ismorphism of Lemma \ref{isombundles}. The connection on $\cB(s,w)$ induces via $\mu_{(w,l,n)}$ a connection on $\cB(s)$ that we denote again by $\Theta$. More specifically, we put for $\phi\in \cB(s)$ and $\nu^a\in T\cM$ 
\begin{align}
\nu^a\Theta_a\phi=\mu_{(w,l,n)}^{-1}(\nu^a\Theta_a\mu_{(w,l,n)}(\phi)). 
\end{align}   
Note that this construction depends on $w$, $l$, and $n$, so it defines a family of connections on  $\cB(s)$ to be more accurate.   
\end{remark}

\subsection{The double bundle}
\label{subsec:double bundle}
\begin{definition}
Let $\cV_s$ be the Whitney sum $\cV_s:=\cB(s,s)\oplus \cB(s,-s)=\cB(s,s)\oplus\cB(s,s)^*$. It is a $\cc^2$-bundle over $M$ with the $\cc^*$-representation $\rho_{s,s}\oplus\rho_{s,-s}$. The sections of this bundle are pairs of spin-weighted functions. 
\end{definition}
\begin{remark}
Let $a$ be a smooth section of $\fN_0$, and let $f\in \Gamma(\cV_s)$. Then we will typically denote $f(x)=[a(x),(f_s(x),\overline{f_{-s}}(x))]$.
\end{remark}

\begin{lemma}
\label{lemma:dual 2}
There is a linear bundle isomorphism $j:\cV_s\to\cV_s^*$, and an antilinear bundle isomorphisms $c:\cV_s\to\overline{\cV_s}$. They induce a sesquilinear form $\IP{\cdot,\cdot}:\cV_{s,x}\times \cV_{s,x}\to \cc$. The sesquilinear form is hermitian and non-degenerate, defining a hermitian structure on the bundle $\cV_s$ in the notation of \cite{GW}  (or a hermitian fibre metric in the notation of \cite{F}). 
\end{lemma}

\begin{proof} By Lemma \ref{lemma:dual}, the dual bundle to $\cV_s$ is $\cV_s^\#=\cB(-s,-s)\oplus \cB(-s,s)$, the anti-dual is $\cV_s^*=\cB(s,-s)\oplus\cB(s,s)$, and the complex conjugate bundle is $\overline{\cV_s}=\cB(-s,s)\oplus\cB(-s,-s)$.

Let $[a, (f_s,\overline{f_{-s}})]\in \cV_{s,x}$ with $f_{\pm s} \in \cc$ and $a\in \fN_{0,x}$. We define the maps 
\begin{align}
j&: \, \cV_s\to \cV_s^*\, , & \left[a,(f_s,\overline{f_{-s}})\right]&\mapsto [a,(\overline{f_{-s}},f_s)]\, ,\\
c&: \, \cV_s\to \overline{\cV_s}\,, &  \left[a,(f_s,\overline{f_{-s}})\right]&\mapsto [a,(\overline{f_{s}},f_{-s})]\, .
\end{align}
One can see immediately that these maps are well-defined and bijective. Moreover, if by an abuse of notation we also  denote 
\begin{align*}
    c:\cV_s^*\ni[a,(\overline{f_{-s}},f_s)]\mapsto [a,(f_{-s},\overline{f_s})]\in \cV_s^\#\, ,\\
    j:\overline{\cV_s}\ni [a,(\overline{f_{s}},f_{-s})]\mapsto [a,(f_{-s},\overline{f_s})]\in \cV_s^\#\, ,
\end{align*}
 then it is straightforward to see that $j\circ c=c\circ j$.

Now, let $x\in M$, $a$, $b\in \fN_{0,x}$, $f=[a, (f_s,\overline{f_{-s}})]$ and $g=[b,(g_s,\overline{g_{-s}})]\in \cV_{s,x}$. We define 
\begin{align}
\label{eq:cV fibre metric}
\IP{g,f}_{ x}:=j(f)(cg)\, ,
\end{align} 
where, by the direct sum structure, and using the definition in \eqref{eq: dual eval},
\begin{align}
j(f)(cg):&=[a,\overline{f_{-s}}]\left([b,\overline{g_s}]\right)+[a,f_s]\left([b,g_{-s}]\right)\\ \nonumber
&=\rho_{(-s,s)}(z)\overline{g_s}\overline{f_{-s}}+\rho_{(-s,-s)}(z)g_{-s}f_s\, ,
\end{align}
with $z\in \cc^*$ such that $b=a\cdot z$. Moreover, the form is hermitian, since $\rho_{(s,w)}(z)=\overline{\rho_{(-s,w)}(z)}=\overline{\rho_{(s,-w)}(z^{-1})}$.
Finally, assume that there is an $f\in \cV_{s,x}$ so that $\IP{f,g}_{ x}=0$ for all $g\in \cV_{s,x}$. By choosing $g$ of the form $[a,(0,1)]$ or $[a,(1,0)]$ one can immediately observe that this requires $f=[b,(0,0)]$, showing that the form is non-degenerate. 
\qeds
\end{proof}

\begin{definition}
\label{def: hermit form}
Let $f,h\in \Gamma_c(\cV_s)$. Then we will denote
\begin{align}
\label{eq:herm form bundle section}
 \IP{f,h}=\int\limits_M \IP{ f(x),h(x)}_x \dVol_g\, .
\end{align}
This  pairing is hermitian, sesquilinear, and non-degenerate.
It can be extended to any pair of sections such that the integral remains finite.
\end{definition}

\begin{remark}
\label{rem:uniqueness fibre metric}
    It was shown in \cite{FewsterKlein} that the hermitian fibre metric constructed in the proof of Lemma~\ref{lemma:dual 2} is uniquely specified, up to a constant $\alpha\in \cc^*$, by demanding that for any linear differential operator $P:\Gamma(\cB(s,s))\to\Gamma(\cB(s,s))$, the extended operator $P\oplus {}^\star P:\Gamma(\cV_s)\to\Gamma(\cV_s)$ is {\it formally hermitian} with respect to this fibre metric, i.e.
    \begin{align}
        \IP{F, (P\oplus {}^\star P)H}=\IP{(P\oplus {}^\star P)F,H}
    \end{align}
    for any $F, H\in \Gamma(\cV_s)$ so that $\supp F\cap\supp H$ is compact. 
     Here, the operator ${}^\star P:\Gamma(\cB(s,-s)\to \Gamma(\cB(s,-s))$ is the {\it formal anti-dual} of $P$ defined by 
    \begin{align}
        \int\limits_M f(\overline{Ph})\dVol_g=\int\limits({}^\star P f)(\overline{h})\dVol_g
    \end{align}
    for all $f\in \Gamma(\cB(s,-s))$ and $h\in \Gamma(\cB(s,s))$ so that $\supp f\cap\supp h$ is compact, using the notation of Lemma~\ref{lemma:dual}. 
\end{remark}

Note that the  hermitian fibre metric on $\cV_s$ constructed in Lemma~\ref{lemma:dual 2} is not positive. This will be a major obstacle in the construction of a positive state later on.

\section{The Kerr spacetime}
\label{sec:Kerr}
In this section, we introduce the Kerr spacetime. Let $M>0$ be the black hole mass and $a$ the black hole angular momentum per unit mass. We also require $\abs{a}<M$. For such a choice of constants 
\begin{align}
\Delta=r^2-2Mr+a^2
\end{align}
has two real, positive roots, $r_{-}=M- \sqrt{M^2-a^2}$ and $r_{+}=M+\sqrt{M^2-a^2}$. We can then define the exterior of a Kerr black hole as the manifold $\MI=\rr_t\times(r_+,\infty)_r\times\ss_{\theta,\varphi}^2$, endowed with the metric
\begin{align}
\label{eq:Kerr metric BL}
g=\frac{\Delta-a^2\sin^2\theta}{\varrho^2}\d t^ 2+\frac{4Mar\sin^2\theta}{\varrho^2}\d t\d \varphi-\frac{\varrho^2}{\Delta}\d r^2 -\varrho^2\d \theta^2-\frac{\sigma^2\sin^2\theta}{\varrho^2}\d \varphi^2\, ,
\end{align}
where $\theta\in [0,\pi]$ and $\varphi\in \rr/2\pi\zz$ are spherical coordinates on $\ss^2$, and
\begin{subequations}
\begin{align}
\varrho^2&=r^2+a^2\cos^2\theta\, ,\\
\sigma^2&=(r^2+a^2)^2-a^2\sin^2\theta\Delta=\varrho^2(r^2+a^2)+2Mra^2\sin^2\theta\, .
\end{align}
\end{subequations}
The global coordinates $(t,r,\theta,\varphi)$ are called the {\it Boyer-Lindquist coordinates}.
This manifold is axisymmetric, manifestly time invariant and globally hyperbolic \cite{GHW}, and we define an orientation and time orientation on $\MI$ by setting $\fO$ to be the volume form $ \dVol_g=\varrho^2 \sin\theta \d t\wedge \d r\wedge \d\theta\wedge \d\varphi$ induced by the metric and by choosing $\d t$ to be future directed.

Similarly, the interior of the same Kerr black hole up to the inner horizon is defined as the manifold $\MII=\rr_t\times(r_-,r_+)_r\times\ss_{\theta,\varphi}$ endowed with the same metric and the time orientation such that $-\d r$ is future directed. One can also define the region beyond the inner horizon as $\MIII=\rr_t\times(-\infty, r_-)_r\times\ss^2_{\theta,\varphi}\backslash \{r=\cos\theta=0\}$ with the same metric, but we will not consider this region here.

\subsection{The Kerr-star and star-Kerr spactimes}
\label{subsec:Kstar and starK}
Let $(t^*, r,\theta, \varphi^*)$ denote the {\it Kerr-star} ($K^*$) and $({}^*t,r,\theta,{}^*\varphi)$ the {\it star-Kerr} (${}^*K$) coordinates, and set $\MIUII^{in}=\rr_{t^*}\times(r_-,\infty)_{r}\times\ss^2_{\theta,\varphi^*}$ endowed with the metric
\begin{align}
\label{eq:g Kstar}
g=g_{tt}\d t^{*2}+2g_{t\varphi}\d t^*\d \varphi^*+g_{\varphi\varphi}\d \varphi^{*2}-\varrho^2\d\theta^2- 2\d t^* \d r+ 2a\sin^2\theta\d\varphi^* \d r\, .
\end{align}
Here, $g_{\mu\nu}$ refers to the corresponding metric components of \eqref{eq:Kerr metric BL}.
Similarly, we define $\MIUII^{out}=\rr_{{}^*t}\times(r_-,\infty)_{r}\times\ss^2_{\theta,{}^*\varphi}$, endowed with the metric \eqref{eq:g Kstar}, with $\d t^*\leftrightarrow -\d {}^*t$ and $\d \varphi^*\leftrightarrow -\d {}^*\varphi$.

To avoid ambiguity, we denote the vector fields $\partial_r$, $\partial_\theta$ of the $K^*$ or ${}^*K$ coordinates as $\partial_{r^*}$ and $\partial_{\theta^*}$, or $\partial_{{}^*r}$ and $\partial_{{}^*\theta}$, respectively. Then the time orientation of $\MIUII^{in}$ or $\MIUII^{out}$ is chosen by demanding $- \partial_{r^*}$ or $\partial_{{}^*r}$ to be future directed, respectively. 

For later use, we note that the inverse  metric in the $K^*$-coordinates is given by
\begin{align*}
\varrho^2g^{-1}=&-a^2\sin^2\theta \partial_{t^*}^2-2(r^2+a^2)\partial_{t^*}\partial_{r^*}-\Delta\partial_{r^*}^2-2a\partial_{t^*}\partial_{\varphi^*}\\
&-2a\partial_{r^*}\partial_{\varphi^*}-\frac{1}{\sin^2\theta}\partial_{\varphi^*}^2-\partial_{\theta^*}^2\, ,
\end{align*}
and in ${}^*K$-coordinates, one obtains the same upon replacing $\partial_{t^*}\to -\partial_{{}^*t}$ and $\partial_{\varphi^*}\to -\partial_{{}^*\varphi}$.

Moreover, the Kerr exterior $\MI$ and interior  $\MII$ can be identified with the submanifolds $\MI^{in/out}=\MIUII^{in/out}\cap\{r>r_+\}$ and $\MII^{in/out}=\MIUII^{in/out}\cap\{r_-<r<r_+\}$ respectively by the maps defined by
\begin{subequations}
\begin{align}
\label{eq:Kstar coord trafo}
 (t^*, r,\theta,\varphi^*)&=(t+ r_*(r),r,\theta, \varphi+ A(r))\text{ and}\\
 ({}^*t, r,\theta,{}^*\varphi)&=(t- r_*(r),r,\theta, \varphi- A(r))\,.
\end{align} 
\end{subequations}
  Here,
\begin{align}
r_*(r)=\int \limits_{r_0}^{r}\frac{r'^{2}+a^2}{\Delta(r')}\diff r'\, ,\quad \quad A(r)=\int\limits_{r_0}^r\frac{a}{\Delta(r')}\diff r' \, ,
\end{align}
for some arbitrary but fixed $r_0\in (r_-,r_+)$ or $r_0\in(r_+,\infty)$ respectively.

This embedding preserves the time orientation for $\MI$ in both cases, but for $\MII$, the embedding into $\MII^{in}$ preserves the time orientation, while the embedding into $\MII^{out}$ reverses it.
We will denote by $\sH_\pm:=\MIUII^{in/out} \cap \{r=r_+\}$ the future/ past event horizon of the black hole. 
We then define $\cM:=\MI\cup\MII\cup\sH_+\sim \MIUII^{in}$ as the physical part of the Kerr spacetime. By \cite{GHW}, it is globally hyperbolic.

We will often want to work with a combination of the two coordinate systems $K^*$ and ${}^*K$. To this end, for $a$ small we introduce new coordinates 
\begin{align*}
\ft:=&\left\{\begin{array}{c}t-r_*,\, r\le 3M,\\
t+r_*,\, r\ge 4M,\end{array}\right.\\
\varphi_*:=&\left\{\begin{array}{c}\varphi-A(r),\, r\le 3M,\\
\varphi+A(r),\, r\ge 4M.\end{array}\right.  
\end{align*}
We interpolate smoothly for $3M \le r\le 4M$. 
Then for $r\leq 3M$, the metric agrees with the one in star-Kerr coordinates, while for $r>4M$, it agrees with the metric in Kerr-star coordinates in \eqref{eq:g Kstar}. 


\subsection{The Kruskal extension}
\label{subsec:Kruskal ext}
Next, let ${\rm M}_{K}=\rr_{U}\times \rr_V\times\ss^2_{\theta,\varphi_+}$, where the global coordinates $(U,V,\theta,\varphi_+)$ are called the {\it KBL coordinates}.
For $i\in \{-,+\}$, let
\begin{align}
\kappa_{i}=\frac{\abs{\partial_r\Delta(r_i)}}{2(r_i^2+a^2)}=\frac{r_+-r_-}{2(r_i^2+a^2)}
\end{align} 
denote the surface gravity of the horizon located at $r=r_i$, and define $r\in (r_-,\infty)$ as a smooth function on $\MK_{K}$ implicitly by 
\begin{align}
\label{eq:G(r)}
-e^{-2\kappa_+ r}(r-r_-)^{\frac{r_-}{r_+}}=:G(r)=\frac{r-r_+}{UV}\, .
\end{align}
We endow ${\rm M}_{K}$ with the metric
\begin{align}
\label{eq:g Kruskal}
g=&-\widetilde{g}_{1}(U^{2}dV^{2}+ V^{2}dU^{2})- \widetilde{g}_2 dUdV-\widetilde{g}_3\left(UdV-VdU\right)^2\\\nonumber
&-\widetilde{g}_4(UdV- VdU)d\varphi_+  - \varrho^{2}d\theta^{2}+ g_{\varphi\varphi} d\varphi^{2}_+.
\end{align}
Here, we have defined
\begin{subequations}
\begin{align}
\widetilde{g}_{1}&=\frac{G^{2}(r)a^{2}\sin^{2}\theta}{4\kappa_{+}^{2}\varrho^{2}}\frac{(r-r_{-})(r+r_{+})}{(r^{2}+ a^{2})(r_{+}^{2}+a^{2})}\Big(\frac{\varrho^{2}}{r^{2}+a^{2}}+\frac{\varrho_{+}^{2}}{r_{+}^{2}+a^{2}}\Big),\\
\widetilde{g}_{2}&=\frac{G(r)(r-r_{-})}{2\kappa_{+}^{2}\varrho^{2}}\Big(\frac{\varrho^{4}}{(r^{2}+a^{2})^{2}}+\frac{\varrho_{+}^{4}}{(r_{+}^{2}+a^{2})^{2}}\Big),\\
\widetilde{g}_3&=\frac{G^2(r)a^2\sin^2\theta}{4\kappa_+^2\varrho^2}\frac{(r+r_+)^2}{(r_+^2+a^2)^2},\\
\widetilde{g}_4&= \frac{G(r)a\sin^{2}\theta}{\kappa_{+}^{2}\varrho^{2}(r_{+}^{2}+a^{2})}\big(\varrho_{+}^{2}(r-r_{-})+ (r^{2}+a^{2})(r+r_{+})\big)\,,
\end{align}
\end{subequations}
and $\varrho_+^2=\varrho^2(r_+, \theta)$. The time orientation on this spacetime is chosen by selecting $\d(U+V)$ to be future directed.

 As for the Kerr-star and star-Kerr spacetimes, one can isometrically embed $\MI$ and $\MII$ into ${\rm M}_K$ by identifying
\begin{align*}
\varphi_+=\varphi-\frac{a}{r_+^2+a^2}t
\end{align*}
and
\begin{align*}
U&=-e^{-\kappa_+{}^*t}\, ,  & V&=e^{\kappa_+ t^*} & &\text{on }\MI\, ,\\
U&=e^{-\kappa_+{}^*t}\, , & V&=e^{\kappa_+ t^*} & &\text{on }\MII\, ,
\end{align*}
thus mapping $\MI$ to $\MK_K\cap\{U<0,\, V>0\}$ and $\MII$ to $\MK_K\cap\{U>0,\, V>0\}$ or, alternatively
\begin{align*}
U&=e^{-\kappa_+{}^*t}\, , & V&=-e^{\kappa_+ t^*} & &\text{on }\MI\, ,\\
U&=-e^{-\kappa_+{}^*t}\, , & V&=-e^{\kappa_+ t^*} & &\text{on }\MII\, ,
\end{align*}
 mapping $\MI$ to ${\rm M}_K\cap\{U>0,\, V<0\}$ and $\MII$ to ${\rm M}_K\cap\{U<0,\, V<0\}$.
The first embedding conserves the time orientation, while the second reverses it.

The complement of these four regions are the two {\it long horizons} $\sH:=\{V=0\}$ and $\sH_R:=\{U=0\}$, which intersect at the {\it bifurcation sphere} $S(r_+):=\{U=V=0\}$.
$\sH\cup\sH_R$ form the bifurcate Killing horizon of the Killing vector field
\begin{align}
    v_\sH=\kappa_+(V\partial_V-U\partial_U)=\partial_t+\tfrac{a}{r_+^2+a^2}\partial_\varphi\, .
\end{align}
The last expression is valid outside of $\sH\cup\sH_R$.

The Kruskal extension and the different embedded submanifolds are illustrated in Figure~\ref{fig:Kerr_fig}.

\begin{figure}
   \includegraphics[width=0.55\textwidth]{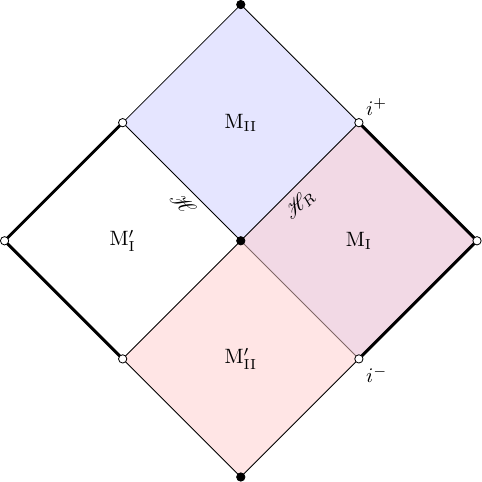}
    \caption{Penrose diagram of the Kruskal extension $\MK_K$ of the Kerr spacetime. The blue and red regions represent the spacetimes $\cM\sim \MIUII^{in}$ and $\MIUII^{out}$ embedded into $\MK_K$, respectively. The prime on $\MI'$ and $\MII'$ denotes that they are equipped with the opposite time orientation.}
    \label{fig:Kerr_fig}
\end{figure}

\subsection{The conformal extension}
\label{subsec:conf ext}
To be able to work at conformal infinity, we require a conformal extension of $\MI$. Let $\breve{g}=x^2g$ be the conformally rescaled metric, where $x=r^{-1}\in[0, \frac{1}{r_+-\epsilon})$  with $0<\epsilon<(r_+-r_-)$.  Consider the manifold $\cM_\epsilon:=\rr_{ \ft}\times (0,\frac{1}{r_+-\epsilon})_x\times\ss^2_{\theta,\varphi_*}\subset \MIUII^{out}$. Then we define its conformal extension $\overline{\cM_\epsilon}$ as the manifold $\rr_{\ft}\times[0, \frac{1}{r_+-\epsilon})_x\times\ss^2_{\theta,\varphi_*}$. We will also use the notation $\breve{\rm M}_{\rm I}:=\overline{\cM_0}$ for the conformal extension of the black hole  exterior. The rescaled metric $\breve g$ in these coordinates for $x<\frac{1}{4M}$, where $(\ft,x,\theta,\varphi_*)=(t^*,x,\theta,\varphi^*)$, reads
\begin{align}
\label{eq:g conf}
\breve g&= \left(x^2-\frac{2Mx^3}{\varrho_x^2}\right) \d t^{*2}+\frac{4Max^3\sin^2\theta}{\varrho_x^2}\d t^* \d \varphi^* + 2\d t^* \d x - 2a\sin^2\theta\d \varphi^* \d x\\\nonumber
&-\sin^2\theta\left(1+a^2x^2+\frac{2Ma^2x^3\sin^2\theta}{\varrho_x^2}\right)\d \varphi^{*2}-\varrho_x^2\d \theta^2  \, .
\end{align}
Here, we have used the notation $\varrho_x^2=1+a^2x^2\cos^2\theta$.  
We observe that, in this regime, $\dVol_{\breve{g}}=\varrho_x^2 \d t_*\wedge \d x\wedge \d^2\omega^{*}$, where $\d^2\omega^*$ is the standard volume element on the unit two-sphere $\ss^2_{\theta,\varphi^*}$.
The inverse of the metric $\breve g$ for $x<\frac{1}{4M}$ is
\begin{align}
\label{eq:inv conf metric}
\breve{g}^{-1}=&-\frac{a^2\sin^2\theta}{\varrho^2_x}\partial_{t^*}^2+2\frac{1+a^2x^2}{\varrho_x^2}\partial_x\partial_{t^*}-2\frac{a}{\varrho_x^2}\partial_{\varphi^*}\partial_{t^*}-\frac{x^2\Delta_x}{\varrho_x^2}\partial_x^2\\\nonumber
&+2\frac{x^2a}{\varrho_x^2}\partial_x\partial_{\varphi^*}-\frac{1}{\varrho_x^2}\partial_\theta^2-\frac{1}{\varrho_x^2\sin^2\theta}\partial_{\varphi^*}^2\, ,
\end{align}
where we have also introduced the notation $\Delta_x=x^2\Delta=1-2Mx+x^2a^2$.

 Similarly, we could have used the manifold $ \MIUII^{out}$ to define the conformal extension in ${}^*K$-coordinates as the manifold $\widetilde{\rm M}_{\rm I}=\rr_{{}^*t}\times[0, r_+^{-1})_x\times\ss^2_{\theta,{}^*\varphi}$ with metric $\breve g$. In these coordinates, one obtains the metric $\breve g$  
by replacing $\d t^*$ and $\d \varphi^*$ by $-\d {}^*t$ and $-\d {}^*\varphi$ in \eqref{eq:g conf}.

Past and future null infinity are then given by the hypersurfaces 
\begin{subequations}
\begin{align}
\sI_-&=\rr_{t^*}\times \{x=0\}\times\ss^2_{\theta,\varphi^*}\subset \overline{\cM_{\epsilon}}\\
\sI_+&=\rr_{{}^*t}\times \{x=0\}\times\ss^2_{\theta,{}^*\varphi}\subset \widetilde{\rm M}_{\rm I}
\end{align}
\end{subequations}
respectively, see Figure~\ref{fig:conf_ext_fig} for an illustration.
The Ricci scalar for $\breve g$ is given by
\begin{align*}
\breve{R}=-\frac{6}{\varrho^2}\left(2\Delta_r-r\partial_r\Delta_r\right)\, .
\end{align*}

\begin{figure}
    \includegraphics[width=0.25\textwidth]{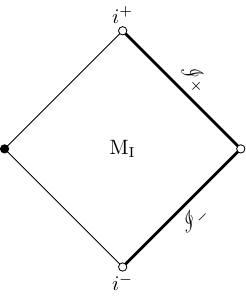}
    \caption{The black hole exterior $\MI$ and its conformal extension. The extension $\breve{\rm M}_{\rm I}$ based on the $K^*$-coordinates contains $\sI_-$, while the extension $\widetilde{\rm M}_{\rm I}$ based on the ${}^*K$-coordiates contains $\sI_+$.} 
    \label{fig:conf_ext_fig}
\end{figure}


\subsection{The Kinnersley tetrad and its renormalizations}
\label{subsec:tetrad}

The Kerr spacetime is of Petrov type D, so the discussion from Section \ref{sec:bundle} applies here.

In particular, one can locally construct a complex null frame whose null vectors are future directed and align with the principal null directions of the spacetime, 
\begin{align}
V^{\pm}=\frac{r^2+a^2}{\Delta}\partial_t\pm\partial_r+\frac{a}{\Delta}\partial_{\varphi}\, .
\end{align}
One choice for $\MI$ that is often employed is the Kinnersley tetrad.
The two real null vectors of the Kinnersely tetrad are then $
l=V^+,\, n=\frac{\Delta}{2\rho^2}V^-$.

In Boyer-Lindquist coordinates, the tetrad reads 
\begin{subequations}
\label{eq:Kinnersley}
\begin{align}
l&=\frac{r^2+a^2}{\Delta}\partial_t+\partial_r+\frac{a}{\Delta}\partial_\varphi\\
n&=\frac{r^2+a^2}{2\varrho^2}\partial_t-\frac{\Delta}{2\varrho^2}\partial_r+\frac{a}{2\varrho^2}\partial_\varphi\\
m&=\frac{ia\sin\theta}{\sqrt{2}p}\partial_t+\frac{1}{\sqrt{2}p}\partial_\theta+\frac{i}{\sqrt{2}p\sin\theta}\partial_\varphi\, ,
\end{align}
\end{subequations}
where we have defined $p=r+ia\cos\theta$.
The main advantage of the Kinnersley tetrad is that in this tetrad the operator $\tho$ (see \eqref{eq:tho}) has the simple form $\tho=l^a\nabla_a$.
However, as noted in \cite{Millet1}, this is not a globally defined smooth section of $\fN_0$, since it cannot be extended to the axis of rotation, $\MI\cap\{\sin\theta=0\}$. 
Moreover, the Kinnersley tetrad does not smoothly extend to $\sH_+$ in $\MIUII^{in}$, but it extends smoothly to $\sH_-$ in $\MIUII^{out}$, and in fact to all of $\MIUII^{out}$. This can be seen by computing the Kinnersley tetrad in $K^*$ and ${}^*K$ coordinates:
\begin{subequations}
\begin{align}
\label{KinnersleystarKstarl}
l&=\frac{2(r^2+a^2)}{\Delta}\partial_{t^*}+\partial_{r^*}+\frac{2a}{\Delta}\partial_{\varphi^*} & &\hspace{-1.2cm}=\partial_{{}^*r}\, ,\\
\label{KinnersleystarKstarn}
n&=-\frac{\Delta}{2\varrho^2}\partial_{r^*} & &\hspace{-1.2cm}=\frac{r^2+a^2}{\varrho^2}\partial_{{}^*t}-\frac{\Delta}{2\varrho^2}\partial_{{}^*r}+\frac{a}{\varrho^2}\partial_{{}^*\varphi}\, .
\end{align} 
\end{subequations}
Note that we formally find at the horizon $\sH_-$
\begin{align}
l=\partial_{{}^*r},\quad n=\frac{r_+^2+a^2}{\varrho_+^2}\left(\partial_{{}^*t}+\frac{a}{r_+^2+a^2}\partial_{{}^*\varphi}\right)\, .
\end{align}

At past null infinity, we will need to work with the conformally rescaled metric $\breve g$, and therefore also with a conformally rescaled tetrad. In order to be able to apply the results by \cite{Araneda}, we set 
\begin{align}
\label{eq:conf Kinnersley}
\breve{l}^a=l^a, \quad \breve{n}^a=x^{-2}n^a. 
\end{align}
Using the $K^*$-coordinates, one thus finds
\begin{align}
\breve{l}^a=-x^{2}\partial_x+\frac{2(1+a^2x^2)}{\Delta_x}\partial_{t^*}+\frac{2ax^2}{\Delta_x}\partial_{\varphi^*},\, \quad \breve{n}^a=\frac{\Delta_x}{2\varrho_x^2}\partial_x\, , 
\end{align} 
giving the formal results $\breve{l}^a=2\partial_{t^*}$ and $\breve{n}^a=\frac{1}{2}\partial_x$ at $\sI_-$.

\begin{remark}
    The flows induced by the Killing fields $\partial_t+c\partial_\varphi$ for any $c\in\rr$ are background preserving diffeomorphisms, see Definition~\ref{def:bck map and bck pred map}. 
    Moreover,  the (conformally rescaled) Kinnersley tetrad is {\it stationary and axisymmetric},  i.e., it is  a $\psi_b^c$-invariant section for any $b$, $c\in \rr$, where $(\psi_b^c)_{b\in\rr}:\MI\to\MI$ is the  flow induced by the Killing vector field $\partial_t+c \partial_\varphi$.
\end{remark}

\begin{remark}
\label{rem:compl triv}
    As discussed in \cite{Millet1}, there exists no global, smooth section of $\fN_0$ on Kerr, or in other words the bundle $\fN_0$ is non-trivial. In particular, the vector field $m$ of the Kinnersley tetrad does not extend to the axis of rotation where $\sin\theta=0$. However, since the axis of rotation has vanishing volume form in all cases considered in this work, it will in general be sufficient to work with one tetrad containing $m$ as its complex vector.

    Nonetheless, for some constructions a complete system of trivializations for $\fN_0$ or $\cB(s,w)$ is desirable.
    Generalizing the result presented in \cite{Millet1}, let $\varphi_c=\varphi-ct$, where $c\in \rr$ is some constant. Then a complete system of trivializations for $\fN_0$ is given by the pair of sections
    \begin{subequations}
    \label{eq:compl triv syst for killing}
    \begin{align}
        x&\mapsto (l,n,m)(x)\cdot e^{i\varphi_c(x)} \,\, \ : \quad  x\in \cM\setminus \{\theta=0\}\\
        x&\mapsto (l,n,m)(x)\cdot e^{-i\varphi_c(x)} \, : \quad  x\in \cM\setminus \{\theta=\pi\}\, ,
    \end{align}
    \end{subequations}
    where $(l,n,m)$ is the Kinnersley tetrad or a (stationary and axisymmetric) rescaling thereof.
    The reason for the generalisation  compared to \cite{Millet1} is that this pair of sections  provides a complete system of trivializations of $\cV_s$ by $\psi_b^c$-invariant sections. In particular, setting $c=0$ (which corresponds to the system of trivializations discussed in \cite{Millet1})  the tetrads in \eqref{eq:compl triv syst for killing} are invariant under the flow of $\partial_t$, while  for $c=a(r_+^2+a^2)^{-1}$  they are invariant under the flow induced by the Killing field $v_\sH$. This allows us to use the form \eqref{eq:trafo w inv tetrad} from Lemma~\ref{lemma:trafo of triv. flow} to express the cover map of this flow on $\cB(s,w)$.
    \end{remark}

None of the above tetrads cover the bifurcation sphere of the Kruskal extension.
To find a tetrad covering the bifurcation sphere, we write the Kinnersley tetrad in KBL coordinates as
\begin{subequations}
\begin{align}
l^a&=2\kappa_+\frac{r^2+a^2}{\Delta}V\partial_V+\frac{a}{\Delta}\left(1-\frac{r^2+a^2}{r_+^2+a^2}\right)\partial_{\varphi_+}\, ,\\
n^a&=-\kappa_+\frac{r^2+a^2}{\varrho^2}U\partial_U+\frac{a}{2\varrho^2}\left(1-\frac{r^2+a^2}{r_+^2+a^2}\right)\partial_{\varphi_+}\, .
\end{align}
\end{subequations}
One can then see that by rescaling as
\begin{align}
\label{eq:Kruskal tetrad}
l^a\to \mathfrak{l}^a:= -U l^a\, ,\quad n^a\to \mathfrak{n}^a:=(-U)^{-1}n^a\, , \quad m^a\to \mathfrak{m}^a:=m^a\, ,
\end{align}
one finds $(\mathfrak{l}, \mathfrak{n}, \mathfrak{m})\in \Gamma(\fN_0)$ with
\begin{subequations}
\begin{align}
\mathfrak{l}^a&=-2\kappa_+\frac{r^2+a^2}{(r-r_-)G(r)}\partial_V+aU\frac{r+r_+}{(r-r_-)(r_+^2+a^2)}\partial_{\varphi_+}\, ,\\
\mathfrak{n}^a&=\kappa_+\frac{r^2+a^2}{\varrho^2}\partial_U+\frac{aV(r+r_+)G(r)}{2\varrho^2(r_+^2+a^2)}\partial_{\varphi_+}\, . 
\end{align}
\end{subequations}
Since $G(r)$ is smooth and non-vanishing on the whole Kruskal domain, this tetrad is smooth on the whole Kruskal extension ${\rm M}_{K}$ as well. In particular, it is smooth at the bifurcation sphere $S(r_+)=\{U=V=0\}$.

\subsection{Identification of spin-weighted functions}
\label{sec:ID to fixed scaling}
In the following, let us work on $\MI$. Assume we have made a global smooth choice of the normalisation of the null vectors $l$ and $n$, so that we can identify $\cB(s,w)$ with $\cB(s)$ as described in Section \ref{sec:red bundle}. Let $\mathfrak{V}_{\ss^2}$ denote the bundle of oriented orthonormal frames on $\ss^2$ with respect to the standard metric on the unit sphere. One has
\begin{lemma}{\cite[Proposition 5.2.1]{Millet1}}
\label{lem:Millet5.2.1}
Let 
\begin{align*}
f:\begin{cases}
\rr_{t}\times (r_+,\infty)\times \mathfrak{V}_{\ss^2}&\to T_\cc \MI\\
(t,r,(\omega,X,Y))&\mapsto \frac{1}{\sqrt{2}p}\left(-ia\IP{X+iY,e_3}_{\rr^3}\partial_t+X+iY\right)\, ,
\end{cases}
\end{align*}
where $\IP{\cdot}_{\rr^3}$ is the canonical scalar product of $\rr^3$, $\cc$-bilinearly extended to $\cc^3$, $(e_1,e_2,e_3)$ is the canonical basis of $\rr^3$, and $p=r+ia\cos\theta$. Then $f$ is an isomorphism of principal bundles between $\rr_{t}\times (r_+,\infty)\times \mathfrak{V}_{\ss^2}$ and $\fN_{0,r}$.
\end{lemma}

\begin{remark}
The same holds for $\rr_{t^*}\times (r_+,\infty)\times \mathfrak{V}_{\ss^2}$, where now $\MI$ is identified with $\rr_{t^*}\times(r_+,\infty)\times \ss^2_{\theta,\varphi^*}$.
\end{remark}

We define (see \cite{Millet:thesis, Beyer:2014})
\begin{definition}
    The bundle $\cB_s^{\ss^2}$ over the sphere $\ss^2$ is defined as the complex line bundle associate to $\mathfrak{V}_{\ss^2}$ and the representation
    \begin{equation}
        U(1) \ni e^{i\phi}\mapsto (c\mapsto e^{-is\phi}c)\, .
    \end{equation}
\end{definition}

\begin{remark}
\label{rem:Triv B_s^S}
    As for the bundle $\cB(s,s)$, a local trivialization of $\cB_s^{\ss^2}$ is induced by choosing a section of $\mathfrak{V}_{\ss^2}$.  The trivialization $\mathcal{T}_m$ on $\ss^2_{\theta,\varphi}\setminus\{\sin\theta=0\}$ is given by 
        \begin{align}
            \pi_3 f^{-1}(m)&=( \partial_\theta,  \csc\theta \partial_\varphi)\, ,
        \end{align}
    where $m$ is as in the Kinnersley tetrad, $\pi_3:\rr_{t}\times (r_+,\infty)\times \mathfrak{V}_{\ss^2}\to \mathfrak{V}_\ss^2$ is the natural projection, and 
    \begin{align}
        f^{-1}: &\fN_{0,r}\to \rr_{t}\times (r_+,\infty)\times \mathfrak{V}_{\ss^2}\, ,\\
        &m \mapsto \left(t,r,\left(\Re\left(\sqrt{2}p {\rm pr}_{T_{(\theta,\varphi)}\ss^2, \partial_t}m\right),\Im\left( \sqrt{2}p {\rm pr}_{T_{(\theta,\varphi)}\ss^2, \partial_t}m\right)\right)\right)
    \end{align}
    is the inverse of the function in Lemma~\ref{lem:Millet5.2.1},\footnote{Note that in contrast to \cite{Millet:thesis}, the multiplication by $\sqrt{2}p$ is completed before taking real and imaginary parts. One can check that this is the correct order.} with ${\rm pr}_{T_{(\theta,\varphi)}\ss^2, \partial_t}$ the linear projection to $T_{(\theta,\varphi)}\ss^2$ parallel to $\partial_t$.
\end{remark}

\begin{remark}
\label{rem3.7}
By the above discussion we obtain an isomorphism from $\cB(s)$ to $\rr_{t}\times(r_+,\infty)\times \cB_s^{\ss^2}$. This then induces a connection, again denoted by $\Theta$, on $\rr_{t}\times(r_+,\infty)\times \cB_s^{\ss^2}$, see Remark \ref{remconnection1}. In turn, the connection on $\rr_{t}\times(r_+,\infty)\times \cB_s^{\ss^2}$ induces a connection on $\cB_s^{\ss^2}$ in the following way. For $\phi\in \Gamma(\cB_s^{\ss^2})$, we define $\tilde{\phi}\in \Gamma(\rr_t\times (r_+,\infty)\times \cB_s^{\ss^2})$ by $\tilde{\phi}(p,\omega)=(p,\phi(\omega))$. 
For $p_0\in \rr_t\times(r_+,\infty)$, we consider the associated projection $\pi_{p_0}:\rr_t\times(r_+,\infty)\times \cB_s^{\ss^2}\rightarrow \cB_s^{\ss^2}$ and define a connection $\Theta_{p_0}^{\ss^2}$ on $\cB_s^{\ss^2}$ by 
\begin{align*}
\nu^a\Theta^{\ss^2}_{p_0,a}\phi=\pi_{p_0}(\nu^a\Theta_a\tilde{\phi}), 
\end{align*}
where $\nu^a\in T\ss^2$ and $T\ss^2$ is understood as a subspace to $T\cM$. 
\end{remark}
Note that the bundle $\cB_s^{\ss^2}$ carries a natural hermitian metric $\bm$. If $[m,c],\, [m,d]\in \cB_s^{\ss^2}$ we define $\bm([m,c],[m,d])=\bar{c}d$. It is easy to see that this is well defined. This and the above connection can now be used to define natural $L^2$ and Sobolev spaces.   
\begin{definition}
We define the space $L^2(\cB_s^{\ss^2})$ as the completion of $\Gamma(\cB_s^{\ss^2})$ for the norm 
\begin{align}
\Vert u\Vert^2_{L^2(\cB_s^{\ss^2})}=\int_{\ss^2}\bm_{\omega}(u,u)d^2\omega. 
\end{align}
We then fix a finite family of smooth vector fields $(\cZ_i)_i$ generating $\Gamma(T\ss^2)$ as a $C^{\infty}(\ss^2)$ module. The Sobolev norms $\Vert.\Vert_{H^m_{[s]}(\ss^2)}$ are defined inductively by
\begin{subequations}
\begin{align}
\Vert u\Vert_{H^0_{[s]}(\cB_s^{\ss^2})}=\Vert u\Vert_{L^2(\cB_s^{\ss^2})},\\
\Vert u\Vert^2_{H^{m+1}_{[s]}}:=\Vert u\Vert^2_{H^{m}_{[s]}}+\sum_{i}\Vert \Theta^{\ss^2}_{p_0,\cZ_i}u\Vert^2_{H^{m}_{[s]}}.
\end{align}
\end{subequations}
By having a closer look at the potentials $\cZ^a_iw_a$, it is easy to see that the so-defined Sobolev norms are equivalent for different choices of $p_0$. We will fix $p_0\in \rr_t\times (r_+,\infty)$ and drop the index $p_0$ in the following. The Sobolev spaces $H^{m}_{[s]}(\cB_s^{\ss^2})$ are defined as completions of $\Gamma(\cB_s^{\ss^2})$ for these norms. 
\end{definition}
Note that the above discussion of course depends on the choice of the tetrad vectors $l,n$, but is not restricted to block $I$. The construction works out in every region of the spacetime on which we can define a tetrad $(l,n,m,\bar{m})$ with $l,n$ globally smooth. 

\section{The Teukolsky operator and phase space}
\label{sec:Teukolsky op}
 In this section, we will introduce the Teukolsky operator of spin $s$. For $s=0$, the Teukolsky operator will simply be the  d'Alembert operator on $\cM$. For spin $\pm 1$, it arises upon contracting the vacuum Maxwell equations on a Petrov type-D spacetime in terms of the electromagnetic field-strength tensor $F_{ab}$ with a tetrad in $\fN$ and decoupling the components. Similarly, for spin $\pm 2$, the Teukolsky operator arises when considering linearized gravity on a Petrov type-D background. In particular, they appear when writing the Bianchi identities for the perturbation of the Weyl tensor, $\delta C_{abcd}$, in terms of its contraction with a tetrad in $\fN$. In this section, we give an abstract introduction to the Teukolsky operator before making the link to these origins.

\subsection{The Teukolsky operator}
\label{subsec:Teukolsky op}
Following \cite{Millet1} and \cite{Aksteiner}, the Teukolsky operator acting on sections of $\cB(s,s)$ is given by
\begin{align}
T_s=2\left[(\tho-2s\rho-\bar{\rho})({\tho}' -\rho')-(\edt-2s\tau-\bar{\tau}')({\edt}'-\tau')\right]-(4s^2-6s+2)\Psi_2\, ,
\end{align}
where
\begin{align}
\Psi_2&=C_{abcd}m^{a}l^b\bar{m}^cn^d
\end{align}
is the only non-vanishing contraction of the Weyl tensor with the null tetrad $(l,n,m,\bar{m})$.
In the following, we will discuss a number of properties of the Teukolsky operator.

 For the first one, we note that in any spacetime of Petrov type D, there is a Killing spinor $\kappa_{AB}$ of valence two, which can be written as $\kappa_{AB}=-2\zeta\mo_{(A}\mi_{B)}$ \cite{Aksteiner:PhD} in a principal dyad. The Killing spinor equation is solved by $\zeta\propto \Psi_2^{-1/3}$. Associated to the Killing spinor is the complex Killing vector
\begin{align}
    \xi=-\zeta(-\rho' l+\rho n+\tau'm-\tau \bar{m})\, .
\end{align}
The Kerr-NUT class, of which the Kerr spacetime is an example, can be characterized by the condition $\Im(\xi)=0$, leading to the relations
\begin{align}
    \frac{\rho}{\bar\rho}=\frac{\rho'}{\bar{\rho}'}=-\frac{\tau'}{\bar\tau}=-\frac{\tau}{\bar{\tau}'}=\frac{\overline{\zeta}}{\zeta}\, .
\end{align}
On Kerr spacetime, for example, one can choose $\kappa_{AB}$ so that $\xi^a=(\partial_t)^a$, resulting in 
   $\zeta=r-ia\cos(\theta)$, 
see \cite{AAB}.

In this class of spacetimes, one can construct a second real Killing vector
\begin{align}
\eta^a=\zeta\left[(\Im\zeta)^2(\rho' l^a-\rho n^a)-(\Re\zeta)^2(\tau \bar{m}^a-\tau' m^a)\right]\, 
\end{align}
from the Killing spinor.
In the case of Kerr, if $\xi=\partial_t$, then $\eta=a\partial_\varphi+a^2\partial_t$.

The Lie derivatives with respect to $\xi$ and $\eta$ acting on a section of $\cB(s,w)$ can be written as \cite[Section 2.4.2]{Aksteiner:PhD}
\begin{subequations}
\label{eq:Lie ders}
\begin{align}
\cL_{\xi,(s,w)}=&\xi^a\Theta_a-\frac{s+w}{2}\zeta\Psi_2-\frac{w-s}{2}\overline{\zeta}\overline{\Psi_2}\, , \\
\cL_{\eta,(s,w)}=&\eta^a\Theta_a+(s+w)h+(w-s)\overline{h}\, ,\\
h=&\frac{1}{16}\zeta(\zeta^2+\overline{\zeta}^2)\Psi_2-\frac{1}{8}\zeta\overline{\zeta}^2\overline{\Psi_2}+\frac{1}{4}\rho\rho'\zeta^2(\overline{\zeta}-\zeta)+\frac{1}{4}\tau\tau'\zeta^2(\overline{\zeta}+\zeta)\, .
\end{align}
\end{subequations}

\begin{remark}
In the Kerr-NUT class, in any tetrad invariant under $\cL_{\xi}$ and $\cL_{\eta}$, such as the Kinnersley tetrad on Kerr, the Lie derivatives simplify to \cite{Aksteiner:PhD}
\begin{align}
\cL_{\xi,(s,w)}=\xi^a\nabla_a\,, \quad \cL_{\eta,(s,w)}=\eta^a\nabla_a\, .
\end{align}
Restricting to a Kerr spacetime, one can use their representation in the Kinnersley tetrad to show that they annihilate all spin coefficients as well as $\Psi_2$, commute with all four GHP differential operators \cite{Pound:2021}, and satisfy \cite[Eq.(2.65)]{Aksteiner:PhD} 
\begin{align}
    [\cL_\xi,\zeta]=[\cL_\xi,\overline{\zeta}]=0\, .
\end{align} 
Under complex conjugation, $\cL_{\xi,(s,w)}$ and $\cL_{\eta,(s,w)}$ are mapped to $\cL_{\xi,(-s,w)}$ and $\cL_{\eta,(-s,w)}$, respectively.
\end{remark}

\begin{proposition}{\cite[Corollary 5.4.2]{Aksteiner:PhD}}
On Kerr spacetimes, $\zeta\overline{\zeta} T_s$ can be written in the form 
 \begin{subequations}   
 \label{eq:T_s symm decomp}
\begin{align}
\zeta\overline{\zeta} T_s &=2(\cR_s-\cS_s)\, ,\\
\cR_s&=\varrho^2 \left(\tho-\rho-\bar{\rho}\right)\left({\rho}'-2s\rho'\right)+\frac{2s-1}{2}(\zeta+\overline{\zeta})\cL_\xi\, \\
\cS_s&=\varrho^2\left(\edt-\tau-\overline{\tau}'\right)\left({\edt}'-2s\tau'\right)+\frac{2s-1}{2}(\zeta-\overline{\zeta})\cL_\xi\, .
\end{align}
\end{subequations}
Moreover, $\left[\cR_s,\cS_s\right]=0$, and $\cR_s$ and $\cS_s$ commute with $\zeta \overline{\zeta}T_s$.
\end{proposition} 

\begin{remark}
  It is conjectured in \cite{Aksteiner:PhD} that this is valid on the whole class of Kerr-NUT spacetimes.  
\end{remark}

Another property of the Teukolsky operator  on Kerr that will be important later is the following
\begin{lemma}
\label{lemma:invar Teukolsky Killing flow}
    Let $X=\partial_t+c\partial_\varphi$, with $\c\in\rr$, so that $X$ is a Killing field of the Kerr spacetime. Let $(\upsilon_b)_{b\in\rr}$ be the flow induced by $X$. Take $(\Upsilon_{(s,s),b}^*)_{b\in \rr}$ to be the group of push-forwards on $\Gamma(\cB(s,s))$  corresponding to the cover maps of $\upsilon_b$ on $\cB(s,s)$ as discussed in Lemma~\ref{lemma:trafo of triv. flow}. Then one has
    \begin{align}
        T_s\circ \Upsilon_{(s,s),b}^*=\Upsilon_{(s,s),b}^*T_s \quad \forall b\in \rr\, .
    \end{align}
\end{lemma}
\begin{proof}
    Let $\phi\in \Gamma(\cB(s,s))$, and let $\phi(x)$ denote $\phi$ in the Kinnersley tetrad (or a rescaling thereof so that the rescaling factor depends only on $r$ and $\theta$).  Since the flows $(\upsilon_b)_{b\in \rr}$ are background preserving and the Kinnersley tetrad is $\upsilon_b$-invariant, one has by Lemma~\ref{lemma:trafo of triv. flow}
    \begin{align*}
        \Upsilon_{(s,s),b}^*\phi(t,r,\theta, \varphi)=\phi(t-b, r,\theta, \varphi- cb)\, .
    \end{align*}
    Using the expression for $T_s$ in the Kinnersley tetrad, it is then straightforward to show that the required result holds on $\cM\setminus\{\sin\theta=0\}$. Similarly, computing $T_s$ in the tetrads in \eqref{eq:compl triv syst for killing}, one finds that the additional terms compared to the expression in the Kinnersley tetrad are independent of $t$ and $\varphi$, and hence the same result holds in these tetrads, showing that it holds on all of $\cM$.
    \qeds
\end{proof}

 For the next property of the Teukolsky operator, we make use of the commutation relations of the GHP differential operators \cite{Aksteiner} and identities for the spin coefficients to recast the Teukolsky operator into the form \cite{Bini}
\begin{align}
T_s=g^{ab}\left(\Theta_a+2sB_a\right)\left(\Theta_b+2sB_b\right)-4s^2\Psi_2\, ,
\end{align} 
where $B_a=\tau\bar{m}_a-\rho n_a$ is a section of $\cB(0,0)\otimes T^*\cM$ \cite{Aksteiner}. 
Combining $w_a$ and $B_a$ to $\Gamma_a=B_a-w_a$, this can be written in the useful form \cite{Bini}
\begin{align}
\label{eq: Teuk Bini form}
T_s=g^{ab}(\nabla_a+2s\Gamma_a)(\nabla_b+2s\Gamma_b)-4s^2\Psi_2\, .
\end{align}
This form makes it apparent that the Teukolsky operator is a normally hyperbolic operator  i.e. has principal symbol $\sigma_{T_s}(x, \xi)=g_x^{-1}(\xi,\xi)\id_{\cB(s,s)}$.
The following properties of $\Gamma$ will be crucial in the subsequent discussions.
\begin{remark}
Specializing to the  exterior $\MI$, working in Boyer-Lindquist coordinates, and selecting the Kinnersley tetrad \eqref{eq:Kinnersley}, $\Gamma$ is given by \cite{Bini} \footnote{Note that there is a relative factor of two between our definition of $\Gamma$ and that in \cite{Bini}.}(recall $p=r+ia\cos\theta$)
\begin{align}
\label{eq:Gamma}
\Gamma=-\frac{1}{2\varrho^2}\left(\left[\frac{M(r^2-a^2)}{\Delta}-p\right]\partial_t+(r-M)\partial_r+\left[\frac{a(r-M)}{\Delta}+i\frac{\cos\theta}{\sin^{2}\theta}\right]\partial_\varphi\right) 
\end{align}
and we find the identities
\begin{subequations}
\begin{align}
l^a\Gamma_a&=\frac{p}{\varrho^2}\, , \quad
n^a\Gamma_a=\frac{1}{2\varrho^4}\left[p\Delta-\varrho^2(r-M)\right]\, ,\\
\nabla_a \Gamma^a&=-\frac{1}{2\varrho^2} \, , \quad \Gamma_a \Gamma^a=\frac{1}{4\varrho^2}\cot^2\theta+\Psi_2\, .
\end{align}
\end{subequations}
\end{remark}

\begin{remark}
Under a change of tetrad, i.e trivialization of $\cB(s,s)$, such that $l\to \lambda\bar\lambda l$, $n\to (\lambda\bar\lambda)^{-1}n$, $\Gamma_a$ changes as
\begin{equation}
\Gamma_a\to \Gamma_a-\frac{\nabla_a \lambda}{\lambda}\, .
\end{equation}

For example, consider the tetrad $(\mathfrak{l},\mathfrak{n},m,\overline{m})$ defined in \eqref{eq:Kruskal tetrad}. Then for this tetrad, we find
\begin{subequations}
\begin{align}
\mathfrak{l}^a\Gamma_a&=(-U)l^a\Gamma_a=-U\frac{p}{\varrho^2}\, ,\\
 \mathfrak{n}^a\Gamma_a&=-\frac{G(r) V}{2\varrho^2}\left[\frac{(r-r_-)p}{\varrho^2}+\frac{rr_+-M(r+r_+)-a^2}{r_+^2+a^2}\right]\, ,
 \label{eq:nGamma Kruskal}
\end{align}
\end{subequations}
Thus,  $\mathfrak{n}^a\Gamma_a$ vanishes on $\sH=\{V=0\}$ while $\mathfrak{l}^a\Gamma_a$ vanishes on $\{U=0\}$.
\end{remark}

\begin{remark}
    If we consider the conformally rescaled tetrad $(\breve l, \breve n, \breve m)$ on $\overline{\cM_\epsilon}$, see Section~\ref{subsec:conf ext},  then we find
\begin{align}
\breve{\Gamma}_a= \breve{\bar{m}}_a \breve{m}^b\breve{n}^c \breve{\nabla}_c \breve{l}_b-\breve{n}_a\breve{m}^b\breve{\bar{m}}^c\breve\nabla_c\breve l_b +\frac{1}{2}\breve l^b \breve\nabla_a\breve n_b-\frac{1}{2}\breve{m}^b\breve\nabla_a \breve{\bar{m}}_b=\Gamma_a\, ,
\end{align}
 where $\breve\nabla_a$ is the Levi-Civita connection of $\breve g$, and $\breve{g}$ was used for lowering indices.
In particular, $\breve{l}^a\breve{\Gamma}_a=l^a\Gamma_a$ and $\breve{n}^a\breve{\Gamma}_a=x^{-2}n^a\Gamma_a$, so
 $\breve{l}^a\breve{\Gamma}_a$ vanishes like $x$, while $\breve{n}^a\breve{\Gamma}_a$ approaches $-M/2+\cO(x)$ as $x\to 0$.
\end{remark}

 Following up on this last result, we have by \cite[Lemma 3.2]{Araneda}
\begin{lemma}
\label{lem:rescaling}
Let $\phi\in \Gamma(\cB(s,s))$ transform like $\phi\to \breve \phi=x^{-1}\phi$ under the conformal transformation. Let us assume that under the conformal transformation, the section $(l,n,m)$ of $\fN_0$ is mapped to $(l,x^{-2}n, x^{-1}m)$. Denote by $\breve{T}_s$ the Teukolsky operator of the conformally transformed spacetime,
\begin{align}
\label{eq:T_s conf rescaled}
\breve{T}_s=\breve{g}^{ab}(\breve{\nabla}_a+2s\breve{\Gamma}_a)(\breve{\nabla}_b+2s\breve{\Gamma}_b)-4s^2\breve{\Psi}_2\, ,
\end{align}
where, as discussed above, $\breve{\Gamma}_a=\Gamma_a$, and $\breve{\nabla}_a$ is the Levi-Civita-connection for $\breve{g}$. Then 
\begin{align}
\breve{T}_s\breve{\phi}=x^{-3}\left[T_s+\frac{1}{6}(R-x^2\breve R)\right]\phi\, ,
\end{align}
where $R$ is the Ricci scalar.
\end{lemma}

\begin{proof} By \cite[Lemma 3.2]{Araneda}, combined with eq. (2.45) and (3.4) of the same paper, we obtain 
\begin{align*}
\breve{T}_s\breve{\phi}=x^{-3}\left[T_s+(4s^2+2)\Psi_2+\frac{R}{6}\right]\phi-\left[(4s^2+2)\breve{\Psi}_2+\frac{\breve{R}}{6}\right]x^{-1}\phi 
\end{align*}
From the conformal invariance of the Weyl tensor, $\breve{C}^{a}_{\hphantom{a}bcd}=C^{a}_{\hphantom{a}bcd}$ and the transformation properties of the tetrad, one obtains that $\breve{\Psi}_2=x^{-2}\Psi_2$. Therefore, the terms containing $\Psi_2$ and $\breve{\Psi}_2$ cancel exactly, and we obtain the desired result.
\qeds
\end{proof}

\subsection{The Teukolsky operator under time reversal}
\label{sec:flip}
The Kerr exterior $\MI$ has a discrete symmetry under the simultaneous flip of $t$ to $-t$ and $\varphi$ to $-\varphi$. In this section, we will discuss the effect of this flip on the Teukolsky operator, starting from the effect on the background and the tetrad. 
Let  $\psi:\MI\to \MI$  be the diffeomorphism that corresponds to the flip $(t,r,\theta,\varphi)\mapsto (-t,r,\theta,-\varphi)$. We note that the background as described in the previous section is given by $\fM=(\MI, g, \dVol_g, {\rm d} t, g(V^+))$. 
It is then straightforward to see that the background map corresponding to the diffeomorphism $\psi$ maps $\fM$ to the modified background $\widetilde{\fM}=(\MI, g, \dVol_g, -{\rm d}t, g(V^-))$. By Definition~\ref{def:N0 cover map}, this implies that the lift of $\psi$ will map (sections of) $\fN_0$ to (those of) $\widetilde{\fN}_0$,  the corresponding bundle over $\widetilde{\fM}$. 

To compare the results with the quantities in the original background, we consider the Carter tetrad
\begin{equation}
(l_C,n_C,m_C)=(l,n,m)\cdot (\Delta/2)^{1/2}\zeta^{-1}\, ,
\end{equation}
where $(l,n,m)$ is the Kinnersley tetrad.
The Carter tetrad is defined on $\MI$ (with the exception of the axis of rotation). 
By a straightforward computation in Boyer-Lindquist coordinates, one can  check that, under  the push-forward induced by the lift of $\psi$ to $\fN_0$ introduced in Definition~\ref{def:N0 cover map}, this tetrad behaves as
\begin{align}
    (\Psi^*(l_C,n_C,m_C))(x)=(-n_C,-l_C,\bar{m}_C)(x)
\end{align}
for any $x\in \MI$. 

Using this result, one can see that a general section $(l,n,m)\in \Gamma(\fN_{0}\vert_{\MI})$ so that $(l,n,m)(x)=(l_C,n_C,m_C)(x)\cdot z(x)$ with $z\in \cinf(\MI;\cc^*)$ satisfies
\begin{align}
    (\Psi^*(l,n,m))(x)=(-n, -l, \bar{m})(x)\cdot z(\psi^{-1}x)\cdot z(x)\, .  
\end{align}
We can use this to derive the behaviour of a section $\Phi(x)=[(l,n,m)(x), \phi(x)]\in \Gamma(\cB(s,w)\vert_{\MI})$. Together with Definition~\ref{def:cover of diffeos}, one obtains
\begin{align}
    (\Psi^*_{(s,w)}\Phi)(x)=[(-n,-l,\bar{m})(x), \rho_{(s,w)}(z(\psi^{-1}x))\rho_{(s,w)}(z(x))\phi(x)]
\end{align}

We also note the following
\begin{lemma}
    There is an isomorphism $\iota$ between $\widetilde{\cB}(s,w)$ and $\cB(-s,-w)$ given by
    \begin{align}
       \widetilde{\cB}(s,w)_x\ni[(-n,-l,\bar{m}), c]\mapsto [(l,n,m),c]\in \cB(-s,-w)_x
    \end{align}
    for any $x\in \MI$, $(l,n,m)\in \fN_{0,x}$ (on background $\fM$), and $c\in\cc$.
\end{lemma}

\begin{proof}
    If the map is well-defined, it is immediately apparent that it is a bijection. We therefore show well-definedness. For any $z\in \cc^*$, and $x$, $(l,n,m)$, and $c$ as in the statement of the lemma, we have
    \begin{align*}
        \iota_x[(-n,-l,\bar{m}), c]&=\iota_x[(-n,-l,\bar{m})\cdot z^{-1}, \rho_{(s,w)}(z)c]\\
        &=\iota_x[(-\abs{z}^{-1}n, -\abs{z}l, \bar{z}\abs{z}^{-1}\bar{m}), \rho_{(-s,-w)}(z^{-1})c]\\
        &=[(\abs{z}l,\abs{z}n,z\abs{z}^{-1} m), \rho_{(-s,-w)}(z^{-1})c]\\
        &=[(l,n,m)\cdot z, \rho_{(-s,-w)}(z^{-1})c]\\
        &=[(l,n,m),c]\, .
    \end{align*}
    \qeds
\end{proof}
From here on, we will assume that the tetrad $(l,n,m)$ is chosen such that $z(\psi^{-1}x)=z(x)$ for simplicity and because this is the case for the Kinnersley tetrad, which we are ultimately interested in. Under this assumption, we have 
\begin{equation}
\label{4.23}
    \iota_x\circ \Psi^*_{(s,w)}([(l,n,m), \phi])(x)=[(l,n,m)\cdot z^{-2}, \phi(\psi^{-1}x)]\, ,
\end{equation}
    which is now a section of $\cB(-s,-w)$ in the original background $\fM$.

In a next step, we note that since $\psi$ is an isometry of $(\MI,g)$, one has the identity
$\psi_{*}(\nabla_X Y)=\nabla_{\psi_{*}X}(\psi_{ *}Y)$ for any pair of vector fields $X$ and $Y$, where $\psi_{ *}$ is the usual push-forward of $\psi$ acting on tangent vector fields.
With the help of this identity, we can compute
    \begin{align}
  \Psi^*({\tho}_{(s,w)}^{(l,n,m)})&=\Psi^*(\nabla_{l}-w g(n, \nabla_{l} l)-s g(m, \nabla_{l}\bar{m}))\\\nonumber
  &=-\left[\nabla_{\abs{z}^{2}n}+wg(\abs{z}^2n, \nabla_{\abs{z}^2n}(\abs{z}^{-2}l))+sg\left(\frac{\bar{z}^2}{\abs{z}^{2}}m, \nabla_{\abs{z}^2n}\left(\frac{z^2}{\abs{z}^{2}}\bar{m}\right)\right)\right]\\\nonumber
  &=-{{\tho}'}_{(-s,-w)}^{(l,n,m)\cdot z^{-2}}   
    \end{align}
    which we can now interpret, again, as the corresponding quantity on $\fM$ in the light of the isomorphism $\iota$.
   In a similar fashion, one computes
   \begin{subequations}
\begin{align}
    &\Psi^*({{\tho}'}_{(s,w)}^{(l,n,m)})=-{\tho}_{(-s,-w)}^{(l,n,m)\cdot z^{-2}}\, , \\
    &\Psi^*({\edt}_{(s,w)}^{(l,n,m)})={{\edt}'}_{(-s,-w)}^{(l,n,m)\cdot z^{-2}}\, ,\\
   &\Psi^*({{\edt}'}_{(s,w)}^{(l,n,m)})={\edt}_{(-s,-w)}^{(l,n,m)\cdot z^{-2}}\, ,\\
    &\Psi^*(\rho_{(l,n,m)})=\Psi^*(g(m_C, \nabla_{\bar{m}_C}l_C))=-g(\bar{m}_C, \nabla_{m_C}n_C)=-\rho'_{(l,n,m)\cdot z^{-2}}\, , \\
    &\Psi^*(\tau_{(l,n,m)})=\tau'_{(l,n,m)\cdot z^{-2}}\, .
\end{align}
\end{subequations}

Combining all these results, we can finally compute 
\begin{align}
    &\Psi^*(T_s^{(l,n,m)})\\\nonumber
    &=\Psi^*\left[2[({\tho}-2s\rho-\bar\rho)({\tho}'-\rho')-(\edt-2s\tau-\bar{\tau}')({\edt}'-\tau')]-4(s-1)(s-\frac{1}{2})\Psi_2\right]\\\nonumber
    &={T_s'}^{(l,n,m)\cdot z^{-2}}\, ,
\end{align}
where $T_s'$ denotes the image of $T_s$ under the prime operation introduced in Proposition \ref{prop:N}, see also Remark \ref{rem:prime cc}. It satisfies the identity \cite{Millet1}
\begin{align}
\label{eq:Ts prime rel}
T_s^\prime=\Psi_2^{\tfrac{2s}{3}}T_{-s}\Psi_2^{-\tfrac{2s}{3}}\, .
\end{align}
Together with the relation $\zeta\propto \Psi_2^{-1/3}$ \cite{AAB}, we finally obtain
\begin{align}
    \label{eq:Ts flip relation}
       \iota\Psi^*(T_s^{(l,n,m)}\phi^{(l,n,m)})(x)=\zeta^{-2s}T_{-s}^{(l,n,m)\cdot z^{-2}}{\zeta}^{2s}\phi^{(l,n,m)\cdot z^{-2}}(\psi^{-1}x)\, .
    \end{align}

In \cite{Millet:thesis}, the Teukolsky operator $T_s$ is considered in the tetrad $(\Delta l, \Delta^{-1}n,m)$,  where $(l,n,m)$ is the Kinnersley tetrad \eqref{eq:Kinnersley}.
If we apply the above to the Teukolsky operator in this tetrad, we obtain $z=\sqrt{2\Delta} \zeta$, and thus  $\iota\cdot \Psi^*$ maps $T_s\phi_s$ in the tetrad of \cite{Millet:thesis} to $T_s'\phi_{-s}$ in the tetrad $(l,n,m)\cdot (2\zeta^2)^{-1}$ or to $T_{-s}\zeta^{2s}\phi_{-s}$ in the tetrad $(l,n,m)\cdot \frac{1}{2}$.

\subsection{The Teukolsky operator on the bundle $\cV_s$}
We define the Teukolsky operator on the $\cc^2$-bundle $\cV_s$ defined in Section~\ref{subsec:double bundle} as $\cP_s=T_s\oplus \overline{T_{-s}}$.

\begin{lemma}
\label{lemma: self dual}
The operator $\cP_s$ is formally hermitian with respect to the hermitian form defined in Definition~\ref{def: hermit form}.  
\end{lemma}

\begin{proof}
 By the results of \cite{FewsterKlein}, see also Remark~\ref{rem:uniqueness fibre metric}, this follows immediately if one has $\overline{T_{-s}}={}^\star T_s$. To show this, let $f\in \Gamma(\cB(s,-s))$, $h\in \Gamma(\cB(s,s))$ so that $\supp f\cap\supp h$ is compact.
We may thus partially integrate without producing boundary terms. The proof proceeds by direct computation in some local trivialization. By construction, the result will be independent of this choice. In such a trivialization we have
\begin{align}
\int\limits_\cM (T_sh(x))(\bar{f}(x))\dVol_g =\int \limits_\cM\bar{f}(x) (T_s h)(x)\dVol_g\, ,
\end{align}
where we use the notation of Lemma~\ref{lemma:dual} on the left-hand side, while the right-hand side indicates multiplication.

 Let us choose the formulation \eqref{eq: Teuk Bini form} for the Teukolsky operator. Since the Levi-Civita connection is metric compatible, using partial integration in the form of Green's second formula leads to
\begin{align*}
\int \limits_\cM \bar{f}(x)(T_{s}h_{s})(x)\dVol_g=&\int\limits_\cM h(x)g^{ab}\left(\nabla_a\nabla_b+4s^2\Gamma_a\Gamma_b-4s^2\Psi_2\right) \bar{f}(x)\dVol_g \\
& -2s\int\limits_\cM g^{ab}h(x)(\nabla_b\Gamma_a+\Gamma_b\nabla_a)\bar{f}(x)\dVol_g\, .
\end{align*}
 Collecting the different terms, one obtains
\begin{align}
\label{eq:PI identitiy}
\int\limits_\cM \bar{f}(x) (T_s h)(x)\dVol_g=\int\limits_\cM (T_{-s}\bar{f})(x) h(x) \dVol_g=\int\limits_{\cM}h(x)(\overline{\overline{T_{-s}}f}(x)) \dVol_g\, ,
\end{align}
 concluding the proof.
\qeds
\end{proof}
As a consequence of this result, both $\cP_s$ and its formal adjoint are normally hyperbolic. 
We note that the action for the Teukolsky scalars $f\in \Gamma(\cV_s)$ can be written in the form 
\begin{align}
    S[f]=-\IP{f, \cP_sf}\, 
\end{align}
using the notation of Definition~\ref{def: hermit form}. This action is symmetric under the global phase shift $f\to e^{i\epsilon} f$, which has the corresponding conserved Noether current
\begin{align*}
j^a[f](x)&=-i(\IP{f(x), \cD_{s,a}f(x)}_x-\IP{\cD_{s,a}f(x),f(x)}_x)\, ,\\
\cD_{s,a}&=(\Theta_a+2sB_a)\oplus \overline{(\Theta_a-2sB_a)}\, .
\end{align*}

Consequently, if we define for $f,h\in \Gamma(\cV_s)$
\begin{align}
J_a[f,h](x)=\IP{f(x), \cD_{s,a}h(x)}_x-\IP{\cD_{s,a}f(x),h(x)}_x\, ,
\end{align}
then we have
\begin{lemma}
\label{lemma:cons J}
Let $f,h\in C^\infty(\cM;\cV_s)$ satisfy $\cP_s f=\cP_s h=0$. Then $J_a[f,h]$ is conserved, i.e. $\nabla_aJ^a[f,h](x)=0$.
\end{lemma}

\begin{proof}
By direct computation  we have (suppressing the dependence on $x$ for brevity)
\begin{align*}
\nabla_aJ^a[f,h]=&g^{ab}\IP{\nabla_a f, \cD_{s,b}h}+\IP{f,\nabla_a\cD_{s,b}h}-\IP{\nabla_a\cD_{s,b}f,h}-\IP{\cD_{s,b}f,\nabla_ah}\\
=&  g^{ab}\IP{f,\cD_{s,a}\cD_{s,b}h}-\IP{\cD_{s,a}\cD_{s,b}f,h}=\IP{f, \cP_s h}-\IP{\cP_s f,h}=0\, .
\end{align*}
\qeds
\end{proof}


\subsection{Construction of the classical phase space}
\label{sec:phase space}
 The construction of the classical phase space will be based on the Green hyperbolicity of the Teukolsky operator. We briefly recall some properties of semi-Green hyperbolic operators and their retarded/advanced Green operators. The definitions and results can be found in \cite{Baer, BaerGinoux}.
 \begin{definition}
    Let $B\xrightarrow{\pi}M$ be a finite-rank vector bundle over a globally hyperbolic spacetime $(M,g,\fO,\fT)$, and let $P:\Gamma(B)\to \Gamma(B)$ be a linear differential operator. Then $P$ is called {\it semi-Green hyperbolic} if there exist {\it retarded and advanced Green operators} $E^\pm_P:\Gamma_c(B)\to \Gamma(B)$ satisfying the conditions
    \begin{subequations}
    \label{eq:Green}
\begin{align}
\label{eq:Green 1}
P\circ E_P^\pm f&=f\, ,\\
E_P^\pm\circ P f &=f\, ,\label{eq:Green 2}\\
\supp (E_P^\pm f)&=J^\pm(\supp f)\,  \label{eq:Green 3}
\end{align} 
\end{subequations}
for all $f\in \Gamma_c(B)$. If both $P$ and ${}^\star P$ are semi-Green hyperbolic, then $P$ is {\it Green hyperbolic}.
\end{definition}

\begin{proposition}
\label{prop:Green hyp props}
 Let $P:\Gamma(B)\to \Gamma(B)$ be a semi-Green hyperbolic operator. Then
 \begin{enumerate}
     \item The retarded/advanced Green operators $E^\pm_P: \Gamma_c(B)\to \Gamma_{pc/fc}(B)$ and the {\it commutator function} $E_P:=E_P^--E_P^+: \Gamma_c(B)\to \Gamma_{sc}(B)$ are unique and continuous.
     \item There are unique, continuous extensions $\widetilde{E}^\pm_P:\Gamma_{pc/fc}(B)\to\Gamma_{pc/fc}(B)$ satisfying \eqref{eq:Green} for all $f\in \Gamma_{pc/fc}(B)$.
     \item The commutator function satisfies 
\begin{align}
   P (E_P f)=0\, , \quad \Ker E_P=P\Gamma_c(B)\, ,\quad  \supp (E_P f)=J(\supp f)\, 
\end{align}
for all $f\in \Gamma_c(B)$. 
\item If $f\in \Gamma_{sc}(B)$ solves $Pf=0$, then there is a $f'\in \Gamma_c(B)$ so that $f=E_Pf'$.
 \end{enumerate}
\end{proposition}
\begin{proof}
    The first two points are \cite[Theorem 3.8]{Baer}, \cite[Corollary 3.11]{Baer} and \cite[Corollary 3.12]{Baer}. The last two are immediate consequences of \cite[Theorem 3.5]{BaerGinoux}.
    Since we will need the construction later on, let us sketch a proof of the last point for the convenience of the reader.
    
    Let $f\in \Gamma_{sc}(B)$ so that $Pf=0$. Let $\Sigma_\pm$ be two Cauchy surfaces of $M$, with $\Sigma_+\subset I^+(\Sigma_-)$, and let $\chi\in C^\infty(M;\rr)$ be a positive function satisfying $\chi=1$ in $J^-(\Sigma_-)$ and $\chi=0$ in $J^+(\Sigma_+)$.
    Then $\chi f\in \Gamma_{fc}(B)$ and $P(\chi f)=[P,\chi]f\in \Gamma_c(B)$. Moreover, one clearly has 
    $(1-\chi)f\in \Gamma_{pc}(B)$ and $P(1-\chi)f=-P\chi f$, since $[P,1]=0$ and $Pf=0$.
Finally, using the extended advanced and retarded Greens operators $\widetilde{E}_P^\pm$ from part 2 of the proposition, we have
\begin{align*}
E_P P(\chi f)&=E_P^-(P(\chi f))-E_P^+(P(\chi f))=\widetilde{E}_P^- P(\chi f)+\widetilde{E}_P^+P((1-\chi)f)\\
&=\chi f+(1-\chi)f=f\,.
\end{align*}
Therefore $f'=P\chi f\in \Gamma_c(B)$ satisfies the requirement, completing the proof.
    \qeds
\end{proof}

After this little digression, let us return to the Teukolsky operator.
 As seen by the formulation \eqref{eq: Teuk Bini form}, the Teukolsky operators $T_s$ and $\overline{T_{-s}}$ are normally hyperbolic. Since moreover $\cM$ is globally hyperbolic, it follows from \cite[Cor. 3.4.3]{BGP}
\begin{proposition}
\label{prop: E_s}
 $T_s$ and $\overline{T_{-s}}$ are semi-Green hyperbolic operators. The retarded and advanced Green operators of $T_s$ will be denoted by $E^\pm_s$, those of $\overline{T_{-s}}$ by ${\overline{E_{-s}}}^\pm$.
\end{proposition}

Since  $\cP_s=T_s\oplus \overline{T_{-s}}$ is a direct sum of Green-hyperbolic operators and formally hermitian, it follows from \cite[Lemma 3.17]{Baer}
\begin{proposition}
\label{prop:Delta_s}
    $\cP_s$ is a Green hyperbolic operator. Its retarded and advanced Green operators are given by $\Delta_s^\pm=E_s^\pm \oplus \overline{E_{-s}}^\pm$.
\end{proposition} 
 Further, since $\cP_s$ is formally hermitian, one has the standard result (see e.g. the discussion in \cite{F})
\begin{lemma}
\label{lem:Green FH}
Let  $f,h\in \Gamma_c(\cV_s)$. Then $(f,\Delta_s^\pm h)=(\Delta_s^\mp f,h)$.
\end{lemma}

\begin{proof}
By the support properties of $\Delta_s^\pm$, $\supp(\Delta^\mp_s f)\cap\supp(\Delta^\pm_s h)$ is compact. We can then use \eqref{eq:Green 1}, combined with Lemma \ref{lemma: self dual} to obtain 
\begin{align*}
\IP{f,\Delta_s^\pm h}=\IP{\cP_s\Delta_s^{\mp}f, \Delta_s^\pm h}=\IP{\Delta_s^\mp f, \cP_s\Delta_s^\pm h}=\IP{\Delta_s^\mp f,h}\, ,
\end{align*}
\qeds
\end{proof}

 Let us define the following spaces:
\begin{definition}
\label{def:phase spaces}
Let 
\begin{equation}
\Sol_{s}(\cM):=\{\phi\in \Gamma_{sc}(\cV_s) : \cP_s\phi=0\}\, .
\end{equation}
For $\phi,\psi\in \Sol_{s}(\cM)$, we define
\begin{align}
\sigma_s(\phi,\psi):=(-1)^{s}\int\limits_\Sigma J[\phi,\psi](n_\Sigma)\dVol_\gamma\, ,
\end{align}
where $\Sigma$ is a spacelike Cauchy surface with future-pointing unit normal vector $n_\Sigma$ and induced metric $\gamma$.
In addition, let $TS_s(\cM):=\Gamma_c(\cV_s)/\cP_s \Gamma_c(\cV_s)$.
Let $[f], [h]\in TS_s(\cM)$.  Then we define the sesquilinear form
\begin{align}
([f],[h])_{\Delta_s}&:=(-1)^s\IP{[f],\Delta_s [h]}.
\end{align}
\end{definition}

The following result is also fairly standard, we give the proof for the convenience of the reader.  
\begin{proposition}
\label{prop:iso phase space}
\begin{enumerate}
\item $(TS_s(\cM),(\cdot,\cdot)_{\Delta_s})$ is a charged symplectic space, in particular $(\cdot,\cdot)_{\Delta_s}:TS_s(\cM)\times TS_s(\cM)\to \cc$ is well-defined, anti-hermitian, and non-degenerate.
\item $(\Sol_{s}(\cM),\sigma_s)$ is a charged symplectic space. In particular, $\sigma_s$ is independent of the choice of $\Sigma$, anti-hermitian and non-degenerate. 
\item $\Delta_s:TS_s(\cM)\to \Sol_s(\cM)$ is a morphism of charged symplectic spaces, i.e. it is bijective and 
\begin{align}
\label{eq: ID of sympl forms}
\sigma_s(\Delta_s([f]),\Delta_s([h]))=([f],[h])_{\Delta_s}\, .
\end{align}
\end{enumerate}
\end{proposition}

\begin{proof}
\begin{enumerate}
    \item The well-definedness and non-degeneracy of $(\cdot,\cdot)_{\Delta_s}$ follows immediately from the fact that $\Ker \Delta_s=\cP_s\Gamma_c(\cV_s)$ in part 3 of Proposition~\ref{prop:Green hyp props}. Moreover, by Lemma~\ref{lem:Green FH}
\begin{align*}
\IP{f,\Delta_s h}=\IP{f,\Delta_s^-h}-\IP{f,\Delta_s^+h}=\IP{\Delta_s^+ f,h}-\IP{\Delta_s^- f, h}=-\IP{\Delta_s f,h}=-\overline{\IP{h,\Delta_s f}}\, ,
\end{align*}  
showing that $([f],[h])_{\Delta_s}$ is anti-hermitian.
\item  The independence of $\sigma_s$ from the choice of $\Sigma$ follows from the fact that $J_a[\phi,\psi]$ is conserved as shown in Lemma \ref{lemma:cons J}, together with the support properties of $\phi$ and $\psi$, and an application of Stokes' theorem. The anti-hermiticity and weak non-degeneracy of $\sigma_s$ will follow from \eqref{eq: ID of sympl forms} once we have shown the third part.
\item As shown in Proposition~\ref{prop:Green hyp props}, for $\Delta_s:\Gamma_c(\cV_s)\to \Gamma_{sc}(\cV_s)$, one has $\Ker \Delta_s=\cP_s\Gamma_c(\cV_s)$ and $\Ran \Delta_s=\Sol_{s}(\cM)$, showing that it descends to a bijective map $\Delta_s:TS_s(\cM)\to \Sol_s(\cM)$. To show \eqref{eq: ID of sympl forms}, let $f, h\in \Gamma_c(\cV_s)$ be representatives of $[f]$ and $[h]$, and let $\Sigma$ be a spacelike Cauchy surface with future-pointing unit normal vector $n$ and induced metric $\gamma$. Then we find
\begin{align*}
(f,h)_{\Delta_s} =&(-1)^s\IP{f,\Delta_s h}= (-1)^{s+1}\IP{\Delta_s f,h}\\
=& (-1)^{s+1}\left[\int\limits_{J^+(\Sigma)} \IP{\Delta_s f, h}_x \dVol_g +\int\limits_{\mathclap{J^-(\Sigma)}} \IP{\Delta_s f, h}_x \dVol_g\right]\\
=&(-1)^{s+1}\left[\int\limits_{J^+(\Sigma)} \IP{\Delta_s f,\cP_s\Delta_s^- h}_x \dVol_g +\int\limits_{\mathclap{J^-(\Sigma)}} \IP{\Delta_s f,\cP_s\Delta^+_s h}_x \dVol_g\right]\\
=&(-1)^{s+1}\left[\hphantom{\partial i}
\int\limits_{\mathclap{\partial J^+(\Sigma)}} J_a[\Delta_s f, \Delta_s^- h] n^{a}_+\dVol_\gamma +\int\limits_{\mathclap{\partial J^-(\Sigma)}} J_a[\Delta_s f,\Delta^+_s h] n^{a}_- \dVol_\gamma\right]\\
=&(-1)^{s}\int\limits_{\Sigma} J_a[\Delta_s f, \Delta_s h] n^{a}\dVol_\gamma
\end{align*}
To get to the fourth line, we used partial integration and Green's second formula to shift the operator $\cP_s$ to $\Delta_sf$. This term drops out, and we are left with the boundary term. The minus-sign used in the last equality appears because, in our sign convention, Stokes' theorem leads to $n_\pm^{a}$ outward-pointing from the perspective of $J^{\pm}(\Sigma)$, and therefore $n_-^{a}=-n_+^{a}$ is future pointing.
\end{enumerate}
\qeds
\end{proof}

Thus, we have two equivalent representations of the classical phase space.

\subsection{The Teukolsky-Starobinsky identities}
The Teukolsky equations with spins in $\{0, \pm 1,\pm 2\}$ occur in the description of Klein-Gordon scalars, electromagnetism, and metric perturbation in Petrov type D spacetimes.

In particular, the operator $T_0$ reduces to the d'Alembert operator on $\cinf(\cM;\cc)$.

For $s=1$, the Teukolsky equation arises when one decouples different components of the Maxwell equation. The spin-weighted function $\phi_1\in\Gamma(\cB(1,1))$ is identified with $F_{lm}$, and $\phi_{-1}\in\Gamma(\cB(-1,-1))$ with $F_{\bar{m}n}$, where $F_{ab}$ is the electromagnetic field-strength tensor, and $F_{xy}$, $x,y\in\{l,n,m,\bar{m}\}$, denotes a contraction of $F_{ab}$ with elements of the null tetrad. They satisfy $T_1\phi_1=0$ and $T_1'\phi_{-1}=0$ if $F_{ab}$ satisfies the vacuum Maxwell equations $\nabla_{[a}F_{bc]}=\nabla^aF_{ab}=0$.

For $s=2$, one can in a similar way identify $\phi_2\in \Gamma(\cB(2,2))$ with $ \delta C_{lmlm}$ and $\phi_{-2}\in\Gamma(\cB(-2,-2))$ with $\delta C_{\bar{m}n\bar{m}n}$, where $ \delta C_{abcd}$ is the perturbation of the Weyl tensor. If this Weyl tensor perturbation belongs to a solution of the linearised vacuum Einstein equations, then $T_2\phi_2=0$ and $T_2'\phi_{-2}=0$.

Recall that $T_s'$ denotes the image of $T_s$ under the prime operation introduced in Proposition \ref{prop:N}, and that by \eqref{eq:Ts prime rel} and the relation $\zeta\propto \Psi_2^{-1/3}$ \cite{AAB}, the above implies that $\zeta^{ 2s}\phi_{-s}$ solves $T_{-s}(\zeta^{2s}\phi_{-s})=0$ for $s\in\{0,1,2\}$.

Then $\phi_s$ and $\phi_{-s}$, $s\in \{0,1,2\}$, satisfy the Teukolsky-Starobinsky identities
\cite{WK, Pound:2021}
\begin{subequations}
\begin{align}
    {\tho}^{2s}\zeta^{2s}\phi_{-s}&={\edt}'^{2s}\zeta^{2s}\phi_s-3M\delta_{s,2}\cL_\xi \overline{\phi_{s}}\\
    {\tho}'^{2s}\zeta^{2s}\phi_s&={\edt}^{2s}\zeta^{2s}\phi_{-s}+3M\delta_{s,2}\cL_\xi\overline{\phi_{-s}}\, ,
\end{align}
\end{subequations}
where $\delta_{s,2}$ is the Kronecker Delta, i.e. $\delta_{s,2}=1$ if $s=2$, and zero otherwise.
Here, by ${\tho}'^{2s}\phi_s$, we mean ${\tho}'_{(s,-s+1)}{\tho}'_{(s,-s+2)}\dots {\tho}'_{(s,s)}\phi_s$,
and similar for other powers of GHP differential operators.

$\phi_s$ and $\phi_{-s}$ also satisfy the decoupled, higher-order equations
\cite{WK, Pound:2021}
\begin{subequations}
\begin{align}
\label{eq:TS plus}
{\tho}^{2s}\overline{\zeta}^{2s}{\tho}'^{2s}(\zeta^{2s}\phi_s)&={\edt}^{2s}\overline{\zeta}^{2s}{\edt}'^{2s}(\zeta^{2s}\phi_s)-9\delta_{s,2}M^2\cL_\xi^s\phi_s\\
{\tho}'^{2s}\overline{\zeta}^{2s}{\tho}^{2s} (\zeta^{2s}\phi_{-s})&={\edt}'^{2s}\overline{\zeta}^{2s}{\edt}^{2s}(\zeta^{2s}\phi_{-s})-9\delta_{s,2}M^2\cL_\xi^s\phi_{-s}\, .
\label{eq:TS minus}
\end{align}
\end{subequations}
Note that the identities \eqref{eq:TS plus} and \eqref{eq:TS minus} are trivial in the case $s=0$.

Later, we will also need the complex conjugate of \eqref{eq:TS minus} for $\Phi=\zeta^{2s}\phi_{-s}$. It can be obtained by replacing
 ${\tho}_{(s,w)}\to {\tho}_{(-s,w)}$, ${\tho}'_{(s,w)}\to {\tho}'_{(-s,w)}$, ${\edt}_{(s,w)}\to{\edt}'_{(-s,w)}$, and ${\edt}'_{(s,w)}\to{\edt}_{(-s,w)}$, and reads
 \begin{align}
\label{eq:TS bar minus}
{\tho}'^{2s}\zeta^{2s}{\tho}^{2s}{\overline{\Phi}}=\left[{\edt}^{2s}\zeta^{2s}{\edt}'^{2s}-9\delta_{s,2}M^2\cL_\xi^s\overline{\zeta}^{-2s}\right]\overline{\Phi}\, .
\end{align} 


\subsection{The physical phase space}
In this section, let us consider the connection between the Teukolsky scalars and the physical theories that are described by them when considered on a spacetime of Kerr-NUT type, namely the linear scalar field, Maxwell and linearised gravity. 

For $s=0$, the Teukolsky equation reduces to the wave equation, so that solutions of the one can be directly identified with solutions of the other.

While for higher spins the identification is not that direct, it is still possible to recover the physical degrees of freedom from either one of the Teukolsky scalars $\phi_s$ or $\phi_{-s}$. 
For both $s=1$ and $s=2$, the electromagnetic vector potential or metric perturbation respectively can be reconstructed in a certain gauge, up to non-propagating and gauge degrees of freedom, from a Hertz potential, see e.g. \cite{Wald, Ori, WK, Pound:2021}. 

In this paper, we use the reconstruction scheme in the ingoing radiation gauge. In this setup, one reconstructs the field-strength tensor or the metric perturbation by applying a differential operator to a Hertz potential $\Phi_{-s,\IRG}\in \Gamma(\cB(-s,-s))$.
 If $\Phi_{-s,\IRG}$ satisfies 
\begin{align}
\label{eq:T-eq IRG Potential}
    T_{-s}\Phi_{-s,\IRG}=0\, ,
\end{align}
then the resulting field-strength tensor or metric perturbation are solutions to the vacuum Maxwell or linearised Einstein equations respectively.
One can then compute the Teukolsky scalars arising as contractions of the null tetrad with $F_{ab}$ and $\delta C_{abcd}$ respectively from the resulting vector potential or metric perturbation. 
These Teukolsky scalars are related to the Hertz potential $\Phi_{-s,\IRG}$ by \cite{WK, Pound:2021}
\begin{subequations}
\begin{align}
    \label{eq:Hertz-rel 1}
    \phi_s&=\frac{1}{2^s}{\tho}^{2s}\overline{\Phi_{-s,\IRG}}\, ,\\
    \phi_{-s}&=\frac{1}{2^s}\left({\edt}'^{2s}\overline{\Phi_{-s,\IRG}}-3M\delta_{s,2}\zeta^{-2s}\cL_\xi \Phi_{-s,\IRG}\right)
\end{align}
\end{subequations}
They satisfy $T_s \phi_s=0$ and $T_s'\phi_{-s}=0$ if \eqref{eq:T-eq IRG Potential} is valid. Moreover, inserting \eqref{eq:Hertz-rel 1} into the Teukolsky-Starobinsky-identity \eqref{eq:TS bar minus} and multiplying by $\overline{\zeta}^{2s}$ for convenience, one obtains the additional relation
\begin{align}
    \label{eq:Hertz-rel 2}
    2^s\overline{\zeta}^{2s}{\tho}'^{2s}\zeta^{2s}\phi_s=\left[\overline{\zeta}^{2s}{\edt}^{2s}\zeta^{2s}{\edt}'^{2s}-9\delta_{s,2}M^2\overline{\zeta}^{2s}\cL_\xi^2\overline{\zeta}^{-2s}\right]\overline{\Phi_{-s,\IRG}}\, .
\end{align}

\begin{definition}
   For $s\in \{0,1,2\}$, we define the operators
    \begin{align}
       \label{eq:def A_s}
        &A_s:\cinf(\cM;\cB(s,-s))\to \cinf(\cM;\cB(s,-s))\, , \\\nonumber
        &A_s f:= \left[\overline{\zeta}^{2s}{\edt}^{2s}\zeta^{2s}{\edt}'^{2s}-9\delta_{s,2}M^2\cL_\xi^2\right] f\, , \quad f\in \cinf(\cM;\cB(s,-s))\, ,\\
        &\widetilde{A}_s:\cinf(\cM;\cB(s,s))\to \cinf(\cM;\cB(s,s))\, , \\\nonumber
        &\widetilde{A}_s f:= \left[{\edt}^{2s}\overline{\zeta}^{2s}{\edt}'^{2s}\zeta^{2s}-9\delta_{s,2}M^2\cL_\xi^2\right] f\, , \quad f\in \cinf(\cM;\cB(s,s))\, .
    \end{align}
    Note that for $s=0$, this corresponds to setting $A_0=\id=\widetilde{A}_0$.
\end{definition}

\begin{proposition}
\label{prop:A_s properties}
    The operators $A_s: \Gamma(\cB(s,-s))\to \Gamma(\cB(s,-s))$ and $\widetilde{A}_s: \Gamma(\cB(s,s))\to \Gamma(\cB(s,s))$ have the following properties:
    \begin{enumerate}
        \item On $\MI$, fix $l$ and $n$ as in the Kinnersley tetrad to identify $A_s$ and $\widetilde{A}_s$ with operators  $A_s$, $\widetilde{A}_s:\Gamma(\cB(s)\vert_{\MI})\to \Gamma(\cB(s)\vert_{\MI})$. Then these operators are identical and can be expressed as
        \begin{align}
        \label{eq:A_s in Lpm}
            \widetilde{A}_s=A_s=\frac{1}{2^{2s}}\left[\sL^+_{-s+1}\sL^+_{-s+2}\dots \sL^+_s\sL^-_{-s+1}\cdots\sL^-_s-144\delta_{s,2}M^2\cL_\xi^2\right]\, ,
        \end{align}
        where 
        \begin{align}
        \label{eq:def sL}
            \sL_n^\pm=\partial_\theta\pm i\left(a\sin\theta\partial_t+\frac{1}{\sin\theta}\partial_\varphi\right)+n\cot\theta\, 
        \end{align}
        in the trivialization induced by the Kinnersley tetrad.
        \item $A_s$ commutes with $\varrho^2\overline{T_{-s}}$, while $\widetilde{A}_s$ commutes with $\varrho^2T_s$. Consequently,
        \begin{align*}
            A_s\overline{E^\pm_{-s}}f=\overline{E^\pm_{-s}}\varrho^{-2}A_s\varrho^2 f\, , \quad f\in \Gamma_c( \cB(s,-s))\, .
        \end{align*}
        \item  On $\MI$, fix $l$ and $n$ as in the Kinnersley tetrad \eqref{eq:Kinnersley} to identify $\Gamma(\cB(s,\pm s)\vert_{\MI})$ with $\Gamma(\cB(s)\vert_{\MI})$, and subsequently with $\cinf(\rr_{t^*}\times(r_+,\infty)_r;\Gamma(\cB_s^{\ss^2}))$ as in Section \ref{sec:ID to fixed scaling}. In this way, one can view $A_s=\widetilde{A}_s$ as an operator on $H^{4s}(\rr_{t^*};H^{4s}_{[s]}(\ss^2)) $ for any fixed but arbitrary $r_0\in (r_+,\infty)$. Then for any $\omega\in \rr$, the Fourier-transformed operator 
        \begin{align}
            a_s(\omega):=e^{i\omega t^*}A_se^{-i\omega t^*}: H^{m+4s}_{[s]}(\ss^2)\to H^{m}_{[s]}(\ss^2)
        \end{align}
        is positive and elliptic. 
        \item If $\phi_s\in \Gamma_{sc}(\cB(s,s))$ is a space-compact solution to $T_s\phi_s=0$, and if $\widetilde{A}_s\phi_s=0$, then $\phi_s=0$. Similarly, if $\overline{\phi_{-s}}\in \Gamma_{sc}( \cB(s,-s))$ is a space-compact solution to $\overline{T_{-s}}\overline{\phi_{-s}}=0$, and if $A_s\overline{\phi_{-s}}=0$, then $\overline{\phi_{-s}}=0$.
        \item  On $\MI$, fix $l$ and $n$ as in the Kinnersley tetrad to identify $\Gamma(\cB(s,\pm s))$ with the space $\cinf(\rr_{t^*}\times(0,r_+^{-1})_x; \cB_s^{\ss^2})$, and consider $A_s$ and $\widetilde{A}_s$ as operators thereon. Let $\phi \in \cinf(\rr_{t^*}\times [0,r_+^{-1})_x; \cB_s^{\ss^2})$. Then
        \begin{align}
            \lim\limits_{x\to 0} (A_s\phi)(t^*, x, \theta,\varphi^*)=A_s(\lim\limits_{x\to 0}\phi(t^*,x,\theta,\varphi^*))
        \end{align}
    \end{enumerate}
\end{proposition}

\begin{proof}
  {\bf{1)}} We show the calculation for $A_s$, the one for $\widetilde{A}_s$ works in the same way. We begin by recalling that upon adding subscripts to indicate on which spin and boost weight the derivative operators act, \eqref{eq:def A_s} reads
  \begin{align*}
      \left[\overline{\zeta}^{2s}{\edt}_{(s-1,-s)}\dots{\edt}_{(-s,-s)}\zeta^{2s}{\edt}'_{(-s+1,-s)}\dots{\edt}'_{(s,-s)}-9\delta_{s,2}M^2\cL_{\xi,(s,-s)}^2\right]\, .
  \end{align*}
  Moreover, employing the form of $m^a$ and $\Theta_a$ in the Kinnersley tetrad \eqref{eq:Kinnersley}, a straightforward computation shows
  \begin{subequations}
 \begin{align}
\label{eq:sL edt id}
{\edt}_{(s,w)}&=\frac{1}{\sqrt{2}}\left[\frac{1}{\overline{\zeta}}\sL^+_{-s}+(w-s)\partial_\theta \overline{\zeta}^{-1}\right]\, ,\\
{\edt}'_{(s,w)}&=\frac{1}{\sqrt{2}}\left[\frac{1}{\zeta}\sL^-_{s}+(w+s)\partial_\theta \zeta^{-1}\right]\, .
\label{eq:sL edt' id}
\end{align}
\end{subequations}
Plugging these identities into the definition of $A_s$ and carefully commuting factors of $\zeta$ and $\overline{\zeta}$ with the differential operators, one obtains \eqref{eq:A_s in Lpm}.

{\bf{2)}} Since $A_0=\id$, this case is trivial.
To show that $A_s$ and $\widetilde{A}_s$ commute with $\zeta\bar\zeta \overline{T_{-s}}$ and $\zeta\bar\zeta T_s$ respectively for $s\in\{1,2\}$, we want to expand $A_s$ and $\widetilde{A}_s$ in terms of symmetry operators. As noted earlier, this includes the operators $\overline{\cS_{-s}}$ and $\overline{\cR_{-s}}$ ($\cS_s$ and $\cR_s$ for $\widetilde{A}_s$) defined in \eqref{eq:T_s symm decomp}. Beyond that, it follows from Lemma~\ref{lemma:invar Teukolsky Killing flow} that $\cL_\xi$ and $\cL_{\eta}$ are symmetry operators for $T_s$ and $\overline{T_{-s}}$; in fact they commute with all GHP spin coefficients, as well as with $\Psi_2$ \cite{Pound:2021}.

Consequently, we wish to expand $A_s$ and $\widetilde{A}_s$ in powers of $\overline{\cS_{-s}}$, $\overline{\cR_{-s}}$,  (or $\cS_s$ and $\cR_s$), $\cL_\xi$, and $\cL_\eta$.

For $s=1$, the result for $\widetilde{A}_s$ is simply \cite[eq. (45b)]{WK}. For $A_s$, we start from the operator identity \cite{WK} 
\begin{align*}
    {\edt}'{\edt}'\overline{\zeta}^2{\edt}{\edt}\zeta^2\phi_{-1}=\left[\cS_1'^2+\cL_\eta\cL_\xi\right]\phi_{-1}\, .
\end{align*}
Renaming $\psi_{-1}=\zeta^2\phi_{-1}$, and multiply by $\zeta^2$, we find
\begin{align*}
   \zeta^2{\edt}'{\edt}'\overline{\zeta}^2{\edt}{\edt}\psi_{-1}=\left[\zeta^2\cS_1'^2\zeta^{-2}+\cL_\eta\cL_\xi\right]\psi_{-1}\, .
\end{align*}
Recalling the relations \eqref{eq:Ts prime rel}, \eqref{eq:T_s symm decomp}, and $\zeta\propto \Psi_2^{-1/3}$, we find that we can write upon complex conjugation
\begin{align*}
    A_1\overline{\psi_{-1}}=\left[\overline{\cS_{-1}}^2+\cL_\eta\cL_\xi\right]\overline{\psi_{-1}}\, .
\end{align*}

For $s=2$, we work in the Kinnersley tetrad and in Boyer-Lindquist coordinates, where $A_s$ and $\widetilde{A}_s$ agree.\footnote{We believe that this proof could be done using GHP operators, spin coefficients, and their algebraic relations. However, since $A_2$ is of 8th, and $\overline{T_{-2}}$ of second order, this would likely require computer assistance. We didn't attempt such a proof here.} These choices let us expand $\overline{\cS_{-2}}$ explicitly,\begin{align}
\label{eq:S_2 compl conj Kinnersley}
\overline{\cS_{-2}}&=\frac{1}{2}\left(\frac{1}{\sin\theta}\partial_\theta \sin \theta\partial_\theta+\frac{1}{\sin^2\theta} \partial_\varphi^2+a^2\sin^2\theta\partial_t^2+2a\partial_t\partial_\varphi\right.\\\nonumber
&\left.-4\i a\cos\theta\partial_t+4\i\frac{\cos\theta}{\sin^2\theta}\partial_\varphi-4\cot^2\theta+2\right)\, ,
\end{align}
and we recall that in this tetrad \footnote{This identification is valid in any tetrad that is invariant under the Killing flows of $\partial_t$ and $\partial_\varphi$.} $\partial_t=\cL_\xi$ and $\partial_\varphi=\frac{1}{a}\cL_\eta-a\cL_\xi$,
see \cite{Aksteiner:PhD}, Sections 2.4 and 2.5. Moreover, $\overline{\cS}_{-s}$ is related to $\cS_s$ by 
\begin{align}
\label{eq:Symm op relation}
 \overline{\cS_{-s}}=\cS_s+2s\, . 
\end{align}

Since the second term in $A_2$ is already a power of $\cL_\xi$, it remains to expand $\overline{\zeta}^4{\edt}^4\zeta^4{\edt}'^4$ in terms of $\overline{\cS_{-2}}$, $\partial_t$ and $\partial_\varphi$. We start from \eqref{eq:A_s in Lpm}, where the coordinate form of $\sL^\pm_n$ is given by \eqref{eq:def sL}. Subsequently, we follow the same procedure as in \cite{CTdC} for the proof of the existence of the Teukolsky-Starobinsky constants. We start with the coordinate expression of $\sL^{-}_1\sL^-_2$ and rewrite it by replacing $\partial_\theta^2 = 2 \overline{\cS_{-2}} + G$, where $G$ is a second-order differential operator which is at most first order in $\theta$ and which can be read off from \eqref{eq:S_2 compl conj Kinnersley}.
We continue this procedure by multiplying by $\sL_0^-$ from the left, and again replacing all second-order derivatives in $\theta$  as above. Proceeding with the remaining operators $\sL^\pm_n$ in \eqref{eq:A_s in Lpm} in the same manner, we rewrite $A_2$ in terms of $\overline{\cS_{-2}}$, $\partial_t$, $\partial_\varphi$, and, a priori, $\theta$ and a differential operator of first order in $\partial_\theta$.

 Making use of a modification of the Mathematica notebook provided by the authors of \cite{CTdC}, we find
\begin{align*}
A_2=&18 a^3 \partial_\varphi \partial_t^3 + 9 a^4 \partial_t^4 + 
 2 a \partial_\varphi \partial_t (\overline{\cS_{-2}}-2) (5\overline{\cS_{-2}}-13) + \left(6 - 5 \overline{\cS_{-2}} + \overline{\cS_{-2}}^2\right)^2 \\
 &+  a^2 \partial_t^2 (9 \partial_\varphi^2 + 2 (\overline{\cS_{-2}}-2) (5\overline{\cS_{-s}}-7))-9M^2\partial_t^2\, ,
\end{align*}
which is the desired expansion in terms of $\overline{\cS_{-s}}$, $\partial_t$ and $\partial_\varphi$.\footnote{This expression agrees with that found in \cite[Lemma 2.7]{CTdC} upon setting $\overline{S_{-2}}=-\frac{1}{2} L+2$, $\partial_t=-\i \frac{\nu}{a}$, and $\partial_\varphi=\i m$. Similarly, the expression agrees with \cite[Theorem 81.5]{chandrasekharBook} upon replacing $\overline{S_{-2}}=-\frac{1}{2} \lambda+2$, $\partial_t=\i \sigma^+$, and $\partial_\varphi=\i m$, recalling that $\alpha^2=a^2+\frac{am}{\sigma^+}$.}
Applying the relation \eqref{eq:Symm op relation} to the result above gives the corresponding result for $\widetilde{A}_2$. Since this is independent of $\theta$ and contains no first-order term in $\partial_\theta$, this shows that $A_s$ and $\widetilde{A}_s$ commute with $\varrho^2 \overline{T_{-s}}$ and $\varrho^2 T_s$, respectively.

We can then use the fact that the retarded and advanced propagators $\overline{E^\pm_{-s}}$ for $\overline{T_{-s}}$ and their properties listed in Proposition~\ref{prop:Green hyp props} extend to past compact/future compact sections of $\cB(s,-s)$ \cite[Theorem 3.8]{Baer} to compute
\begin{align*}
    A_s\overline{E_{-s}}^\pm f=\overline{E_{-s}}^\pm\overline{T_{-s}}A_s\overline{E_{-s}}^\pm f=\overline{E_{-s}}^\pm\varrho^{-2}A_s\varrho^2\overline{T_{-s}}\overline{E_{-s}}^\pm f=\overline{E_{-s}}^\pm\varrho^{-2}A_s\varrho^2 f\,,
\end{align*}
concluding the proof of the second part.

{\bf{3)}} We remain in a tetrad where $l$ and $n$ are fixed to the scaling of the Kinnersley tetrad \eqref{eq:Kinnersley}. Therefore, $A_s$ and $\widetilde{A}_s$ agree and can be treated simultaneously.

Using the result of part 1, it is then easy to see that $a_s(\omega)$ is a differential operator of order $4s$ with principal symbol 
\begin{equation*}
    \sigma_{4s}(a_s(\omega))=2^{-2s} \left(\xi_\theta^2+\frac{\xi_\varphi^2}{\sin^2\theta}\right)^{2s}\,.
\end{equation*}

Next, we note that for any fixed $\omega\in\rr$, the Fourier-transformed symmetry operator $e^{i\omega t^*}\cS_se^{-i\omega t^*}:\Gamma(\cB_s^{\ss^2})\to\Gamma(\cB_s^{\ss^2})$ has countably many real eigenvalues with eigenfunctions of the form $e^{im\varphi^*}S^{[s],(a\omega)}_{m\ell}(\theta)$ with $m\in \zz$ and $\ell\geq \max\{\abs{m},s\}$ called  spin-weighted spheroidal harmonics \cite{Teukolsky2}. Moreover, $\left\{e^{im\varphi}S^{[s],(a\omega)}_{m\ell}(\theta)\right\}_{m,\ell}$ form an orthonormal basis of $H^n_{[s]}(\ss^2)$ \cite[Proposition 2.1]{TdC}. This result follows from Sturm-Liouville theory.

In fact, for any $\omega\in \rr$, the elements of $\left\{e^{im\varphi}S^{[s],(a\omega)}_{m\ell}(\theta)\right\}_{m,\ell}$ are also eigenfunctions of $a_s(\omega)$, see for example the discussion in \cite[Lemma 3.1]{GH} for $s=2$ or in \cite[Def. 2.1, Lemma 2.7]{CTdC} for more general spin. 
Moreover, by \cite[Lemma 3.1]{GH} and \cite[Lemma 2.11]{CTdC}, the eigenvalues
 \begin{align*}
 a_s(\omega)e^{im\varphi}S^{[s],(a\omega)}_{m\ell}(\theta)=N(s,\omega,m,\ell)e^{im\varphi}S^{[s],(a\omega)}_{m\ell}(\theta)
\end{align*}  
satisfy $N(1,\omega,m,\ell)>0$ and $N(2,\omega,m,\ell)>9M^2\omega^2$, respectively. In particular, the eigenvalues are positive and non-zero. The invertibility therefore follows from the fact that the eigenfunctions form an orthonormal basis of $H^n_{[s]}(\ss^2)$. 

{\bf{4)}} We show the case of $\phi_s$ and $\widetilde{A}_s$, the case for $\overline{\phi_{-s}}$ and $A_s$ works in the same way.

Let us fix the scaling to that of the Kinnersley tetrad. While this does not work in all of $\cM$, we can apply it on $\cM\setminus \sH_+$.
Let $\phi_s\in \Gamma_{sc}(\cB(s,s))$ be a space-compact, non-trivial solution to $T_s\phi_s=0$, and assume $\widetilde{A}_s\phi_s=0$. For some  $\delta<r_+-r_-$, let $r_0\in(r_+-\delta,r_+)\cup(r_+,\infty)$ be chosen such that $\phi_s\vert_{r_0}$ is not vanishing identically. Since $\phi_s$ is a non-trivial solution to $T_s\phi_s=0$, such $r_0$ can always be found. 
We can combine the results of \cite{Millet:thesis} with their equivalent under the symmetry $(t,\varphi)\to(-t,-\varphi)$ 
to argue that for any such  $r_0$, $\phi_s\vert_{r_0}$  and $\widetilde{A}_s\phi_s\vert_{r_0}$ (which vanishes by assumption) are elements of $L^2(\rr_t; \Gamma(\cB_s^{\ss^2}))$.\footnote{See also the more detailed account of the results of \cite{Millet2} in Section \ref{sec:millet}.} Denote $\omega=(\theta, \varphi)$, and consider 
\begin{align*}
    0&=\int\limits_{\rr_t}\int\limits_{\ss^2} \overline{\phi_s}(r_0,t,\omega)\widetilde{A}_s\phi_s(r_0,t,\omega)\d^2\omega \d t = 
    \int\limits_{\rr_k}\int\limits_{\ss^2}\overline{\hat{\phi}_s}(r_0,k,\omega)a_s(k)\hat{\phi}_s(r_0,k,\omega)\d^2\omega \d k \\
    &=\int\limits_{\rr_k}\sum\limits_{m,\ell}\abs{\hat{\phi}_{s,m,\ell}(r_0,k)}^2N(s,k,m,\ell) \d k 
\end{align*}
where we have used Plancherel theorem in $t$ in the second step, and the completeness and orthonormality of the $\left(e^{im\varphi}S^{(\nu)}_{m\ell}(\theta)\right)$-basis in the third step. As discussed in part 3) above, $N(s,k,m,l)>0$. Moreover, since by our assumptions $\phi_s\vert_{r_0}$ is smooth and does not vanish identically, there must be $m$ and $\ell$ so that $\hat{\phi}_{s,m,\ell}(r_0,k)\neq 0$, so that we conclude
\begin{align*}
    \int\limits_{\rr_k}\sum\limits_{m,\ell}\abs{\hat{\phi}_{s,m,\ell}(r_0,k)}^2N(s,k,m,\ell) dk >0\, ,
\end{align*}
 leading to a contradiction. This concludes the proof of part 4.

 {\bf 5)}  We wish to take the limit along the flow induced by $\partial_x$ (or alternatively $\partial_{r^*}$ in the $K^*$-coordinate system). We note that, using $l$ and $n$ as in the Kinnersley tetrad as described in the assumptions, one has
 \begin{equation}
     [A_s,\partial_x]=[\widetilde{A}_s,\partial_x]=0\, ,
 \end{equation}
 or in other words $A_s$ is independent of $x$. Since one also has $[A_s,x]=0$, which can be checked by a straightforward computation in the (conformally rescaled) Kinnersley tetrad, one immediately obtains the desired result.
\qeds
\end{proof}

\begin{definition}
Let $s\in\{0,1,2\}$, and define the operator 
\begin{align}
\label{eq: Bs def}
    B_s: \Gamma(\cB(s,s))\to \Gamma(\cB(s,-s))\, ,\quad \phi_s\mapsto \overline{\zeta}^{2s}{\tho}'^{2s}\zeta^{2s} \phi_{s}\, .
\end{align}
We define the physical subspace of $\Sol_{s}(\cM)$ as
\begin{align}
    \Sol_{s,p}(\cM):=\{(\phi_s,\overline{\phi_{-s}})\in \Sol_{s}(\cM):  \overline{\phi_{-s}}\in\Ran(A_s) \text{ and }\overline{\phi_{-s}}&=B_s\phi_s\}\, .
\end{align}
Here, $A_s$ is considered as an operator on $\Ker(\overline{T_{-s}})\subset \Gamma_{sc}(\cB(s,-s))$.
\end{definition}

An interpretation based on the ingoing radiation gauge \eqref{eq:Hertz-rel 2} is that the elements of $\cS_{1,p}(\cM)$ are of the form $(F_{lm}, 2^{-1}A_1\overline{\Phi_{1,\IRG}})$, while those of $\cS_{2,p}(\cM)$ are of the form $(C_{lmlm}, 2^{-2}A_2\overline{\Phi_{2,\IRG}})$.
Assuming that $\Phi_{-s,\IRG}\in \Ker_{sc} (T_{-s})$, one has indeed $\phi_{\pm s}\in \Ker_{sc} (T_{\pm s})$ by \eqref{eq:Hertz-rel 1} and Proposition~\ref{prop:A_s properties}.

For later use, we also note the following property:
\begin{lemma}
\label{lemma:A_s exchange}
   Let $\phi=(\phi_s,\overline{\phi_{-s}})\in\Sol_{s,p}(\cM)$.  Then 
   \begin{subequations}
    \begin{align}
        {\tho}^{2s}A_s\overline{\phi_{-s}}=\widetilde{A}_s{\tho}^{2s}\overline{\phi_{-s}}\, ,\\
        A_sB_s\phi_s=B_s\widetilde{A}_s\phi_s\, ,\\
        \widetilde{A}_s^{-1}{\tho}^{2s}\overline{\phi_{-s}}={\tho}^{2s} A_s^{-1}\overline{\phi_{-s}}\, .
    \end{align}
    \end{subequations}
\end{lemma} 

\begin{proof} 
 Let $(\phi_s,\overline{\phi_{-s}})$ as above. Then by the Teukolsky-Starobinsky identities \eqref{eq:TS plus} and
 \eqref{eq:TS minus} and the fact that $\overline{\phi_{-s}}=B_s\phi_s$, one has
 \begin{align*}
     B_s\widetilde{A}_s\phi_s= \overline{\zeta}^{2s} {\tho}'^{2s}\zeta^{2s}{\tho}^{2s} \overline{\zeta}^{2s} {\tho}'^{2s} \zeta^{2s}\phi_s= \overline{\zeta}^{2s} {\tho}'^{2s}\zeta^{2s}{\tho}^{2s} B_s\phi_s=A_s B_s\phi_s\, .
 \end{align*}
Moreover, by \eqref{eq:TS bar minus}, one has
\begin{align*}
    {\tho}^{2s}A_s\overline{\phi_{-s}}={\tho}^{2s}\overline{\zeta}^{2s}{\tho}'^{2s}\zeta^{2s}{\tho}^{2s}\overline{\phi_{-s}}\, .
\end{align*}
 For $\phi\in \Sol_{s,p}(\cM)$, one has ${\tho}^{2s}\overline{\phi_{-s}}=\widetilde{A}_s\phi_s$, and by part (2) of Proposition~\ref{prop:A_s properties}, this entails ${\tho}^{2s}\overline{\phi_{-s}}\in\Ker_{sc}(T_s)$. One can thus apply the Teukolsky-Starobinsky-identity \eqref{eq:TS plus} to obtain
 \begin{align*}
     {\tho}^{2s}A_s\overline{\phi_{-s}}={\tho}^{2s}\overline{\zeta}^{2s}{\tho}'^{2s}\zeta^{2s}{\tho}^{2s}\overline{\phi_{-s}}=\widetilde{A}_s{\tho}^{2s}\overline{\phi_{-s}}\, .
 \end{align*}
 Finally, combining ${\tho}^{2s}\overline{\phi_{-s}}=\widetilde{A}_s\phi_s$, and $\overline{\phi_{-s}}=B_s\phi_s$ with \eqref{eq:Hertz-rel 1}, it is straightforward to see 
 \begin{align*}
     \widetilde{A}_s^{-1}{\tho}^{2s}\overline{\phi_{-s}}=\phi_s=2^{-s}{\tho}^{2s}\overline{\Phi_{-s,IRG}}={\tho}^{2s} A_s^{-1}\overline{\phi_{-s}}\, .
 \end{align*}
 \qeds
\end{proof}
\begin{definition}
 We set
 \begin{subequations}
\begin{align}
\cT_{s,sc}(\cM)&:=\{\phi\in \Gamma_{sc}( \cB(s,s)): {T_{s}}\phi=0\}\, ,\\
\cT_{-s,sc}(\cM)&:=\{\phi\in \Gamma_{sc}( \cB(s,-s)): \overline{T_{-s}}\phi=0\}\, .
\end{align}   
\end{subequations}
\end{definition}
We can show
 \begin{lemma}
 \label{lemma:phys subspace bij}
     The map $\overline{\psi_{-s}}\mapsto({\tho}^{2s}\overline{\psi_{-s}}, A_s\overline{\psi_{-s}})$ is a bijection between $\cT_{-s,sc}(\cM)$ and $\Sol_{s,p}(\cM)$
 \end{lemma}

 \begin{proof}
We show that the map is surjective and injective. 
First of all, note that by the definition of $\Sol_{s,p}(\cM)$, if $\phi=(\phi_s, \overline{\phi_{-s}})\in \Sol_{s,p}(\cM)$, then there must be a $\overline{\psi_{-s}}\in \cT_{-s,sc}(\cM)$ so that $\overline{\phi_{-s}}=A_s\overline{\psi_{-s}}$. As a consequence, by Lemma~\ref{lemma:A_s exchange}, and \eqref{eq:TS bar minus}, the identity $\widetilde{A}_s\phi_s={\tho}^{2s}\overline{\phi_{-s}}={\tho}^{2s}A_s\overline{\psi_{-s}}=\widetilde{A}_s{\tho}^{2s}\overline{\psi_{-s}}$
is satisfied by all $\phi_s\in \cT_{s,sc}$ so that $\overline{\phi_{-s}}=B_s\phi_s$ is contained in $\Ran(A_s)$. By part 4 of Proposition~\ref{prop:A_s properties}, this implies that $\phi_s$ must be of the form ${\tho}^{2s}\overline{\psi_{-s}}$, and that the map is surjective. 

To show that the map is also injective, we note that by part 4 of Proposition~\ref{prop:A_s properties}, $\overline{\phi_{-s}}=0$ implies $\overline{\psi_{-s}}=0$.
  \qeds   
 \end{proof}

\begin{remark}
    Instead of choosing the physical subspace in the way done above, we could have instead chosen the space
    \begin{align*}
        \{(\phi_s,\overline{\phi_{-s}})\in \Sol_{s}(\cM):  \phi_s\in \Ran(\widetilde{A}_s)\text{ and }\phi_s={\tho}^{2s}\overline{\phi_{-s}}\}\, .
    \end{align*}
    In this case, elements $(\phi_s,\overline{\phi_{-s}})$ could be naturally identified with
    $(2^{-1}\widetilde{A}_1\widetilde{\Phi}_{s,\ORG}, \overline{\zeta^{2}F_{\bar{m}n}})$ and $(2^{-2}\widetilde{A}_2\widetilde{\Phi}_{s,\ORG}, \overline{\zeta^{4}C_{\bar{m}n\bar{m}n}})$. Here $\widetilde{\Phi}_{s,\ORG}$ is the Hertz potential in outgoing radiation gauge, which satisfies $ T_s\widetilde{\Phi}_{s,\ORG}=0$.
It is related to the Teukolsky scalars of the corresponding metric perturbations by
\begin{subequations}
\begin{align}
    \label{eq:Hertz-rel 3}
    \phi_s&=\frac{1}{2^s}\left({\edt}^{2s}\overline{\zeta}^{2s}\overline{\widetilde{\Phi}_{s,\ORG}}+3M\delta_{s,2}\cL_\xi \widetilde{\Phi}_{s,\ORG}\right)\, ,\\
    \phi_{-s}&=\frac{1}{2^s}{\tho}'^{2s}\overline{\zeta}^{2s}\overline{\widetilde{\Phi}_{s,\ORG}}\, ,
\end{align}
\end{subequations}
    and these Teukolsky scalars satisfy $T_s\phi_s=T_s'\phi_{-s}=0$.
    For this choice of physical subspace there is a bijection between the physical subspace and $\cT_{s,sc}(\cM)$.
\end{remark}

\section{Analytic framework}
\label{sec:analytic}
In this section, we recall some basic definitions of $b-$ and scattering Sobolev spaces and the algebra of scattering pseudodifferential operators. We follow closely the presentation in \cite{Millet2}. 
Let $\cN$ be an $n$-dimensional manifold with boundary and $E$ a vector bundle over $\cN$ with connection $\Theta$. We suppose that there are a hermitian metric $m$ on $E$ and volume forms $\dVol^{b/sc}(x)$ on ${\cN}$. In local adapted coordinates $x=(x_0,x'),\, x_0\ge 0,\, x'\in \rr^{n-1}$, with $x_0=0$ locally defining the boundary of $\cN$, we choose $\dVol^b=\left\vert\frac{\d x_0\d x'}{x_0}\right\vert$ and $\dVol^{sc}=\left\vert\frac{\d x_0\d x'}{x_0^{n+1}}\right\vert$. We will often write $m_x(u,u)=:\vert u\vert_x^2$. 
\subsection{Sobolev spaces}
\label{sec:func setting}
We define the natural space $L_{b/sc}^2(E)$ as the completion of $\Gamma_c(E)$ for the inner product 
\begin{align}
\langle u,v\rangle =\int\limits_{\cN} m_x( u,v)\dVol^{b/sc}(x). 
\end{align}
\begin{definition}
 We define ${}^bT\cN$, the bundle of $b-$vectors on $\cN$, as the bundle whose smooth sections are smooth vector fields tangent to the boundary.  In local coordinates $({x_i})_{i=0}^{n-1}$ near a point of the boundary, if $x_0$ is a defining function of the boundary, such vector fields can be written as
\begin{align}
a_0(x)x_0\partial_{x_0}+\sum_{i=1}^{n-1}a_i(x)\partial_{x_i}
\end{align}
with $(a_i)_{i=0}^{n-1}$ a family of functions which are smooth up to the boundary. The dual bundle is denoted by ${}^bT^*\cN$.
We define ${}^{sc}T\cN$, the bundle of $sc$-vectors or scattering vectors, as the bundle whose sections are of the form $x_0Z$ for $Z\in \Gamma({}^{b}T\cN)$. The dual bundle is denoted by ${}^{sc}T^*\cN$. If $m\in \nn$, we denote by ${\rm Diff}^m_{b/sc}(E)$ the algebra of differential operators generated by the set $\{\id_E\}\cup\{\Theta_{X_1}...\Theta_{X_j},\, j\le m,\, X_i\in \Gamma({}^{b/sc}T\cN)\}$.        
\end{definition}
Starting with the algebra ${\rm Diff}^m_{b/sc}(E)$ we want to define natural function spaces.  
\begin{definition}
\label{def:function spaces and norms}
Fix a family of smooth vector fields $(Z_i)_{i=1}^n$ generating $\Gamma({}^{b/sc}T\cN)$ as a $C^{\infty}(\cN)$-module and fix a $b/sc$- volume form $\dVol^{b/sc}$. We define $H^0_{b/sc}(E)$ as the spaces $L^2_{b/sc}(E)$.  For $r\in \nn,\, r\ge 1$, we recursively define the $b/sc$-Sobolev spaces $H^{r}_{b/sc}(E)$ by completion of $\Gamma_c( E\vert_{\cN^0})$ in the norm
\begin{align}
\Vert u\Vert_{H^{r+1}_{b/sc}}^2:=\Vert u\Vert^2_{H^r_{b/sc}}+\sum_{i=1}^n\Vert \Theta_{Z_i}u\Vert^2_{H^r_{b/sc}}. 
\end{align}
Here, $\cN^0$ is the interior of $\cN$.
Similarly, for $m\in\nn$, we define the $C^m$-norm by 
\begin{align}
\Vert u\Vert_{C^m_{b/sc}}=\sup_{x\in \cN,\vert \alpha\vert\le m}\abs{ \Theta_{Z_1}^{\alpha_1}...\Theta_{Z_n}^{\alpha_n}u}_{x}. 
 \end{align}  
\end{definition}
Suppose now that, locally, we have $\cN=[0,1)_{x_0}\times Y$. We then define the weighted Sobolev spaces 
\begin{align}
\bar{H}^{m,\mu}_{b/sc}=x_0^{\mu}\bar{H}^m_{ b/sc}, 
\end{align}
where the bar indicates that the construction is based on extendible distributions.  We will also need the corresponding semiclassical spaces, where the norm is defined recursively by 
 \begin{align}
\Vert u\Vert_{H^{r+1}_{b/sc,h}}^2:=\Vert u\Vert^2_{H^r_{b/sc,h}}+\sum_{i=1}^n\Vert h\Theta_{Z_i}u\Vert^2_{H^r_{b/sc,h}}. 
\end{align}
\begin{remark}
It can sometimes be useful to define the $b-$ Sobolev spaces with respect to the $sc-$ volume form. We will denote these spaces by $\bar{H}^{m,\mu}_{(b)}$ etc.   
\end{remark}

We would like to work with $\cB(s,s)$, but the bundle does not carry any  natural or preferred hermitian metric. Rather than working with $\cB(s,s)$, we fix the scaling of the tetrad and identify $\cB(s,s)$ first with $\cB(s)$ and subsequently with $\rr_{\ft}\times[0,\frac{1}{r_+-\epsilon})_x\times \cB^{\ss^2}_s.$

Note that we have a natural identification between elements of $\Gamma(\cB(s))$ and $C^{\infty}(\rr_{\ft}\times(r_+-\epsilon,\infty);\Gamma(\cB^{\ss^2}_s))$, and between $\cD'(\cB(s))$ and $\cD'(\rr_{\ft}\times(r_+-\epsilon,\infty),\cD'(\cB^{\ss^2}_s))$.  Recall that for a vector bundle $E$, a distribution $u\in \cD'(E)$ is a linear continuous map $u: \Gamma_c( E^\#)\to \cc$, where $E^\#$ is the dual bundle.\footnote{More accurately, a distribution $u\in \cD'(E)$ is a linear continuous map $u: \Gamma_c( E^\#\otimes \Omega)\to \cc$, where $\Omega$ is the bundle of 1-densities over the base manifold $\cN$. However, since the metric-induced volume form provides a distinguished 1-density, we will use this to identify functions and 1-densities, suppressing $\Omega$.}

On $X=(r_+-\epsilon,\infty)_r\times \ss_{\theta,\varphi}^2$ and $\overline{X}:=\left[0,\frac{1}{r_+-\epsilon}\right)_x\times \ss_{\theta,\varphi}^2$, we consider the $b$-volume form $\dVol^b=\frac{\d x\d^2\omega}{x}$ and the scattering volume form $\dVol^{sc}=\frac{\d x\d^2\omega}{x^4}$. We can now define the following scalar product on $\Gamma_c( (0,\frac{1}{r_+-\epsilon})_x\times \cB^{\ss^2}_s)$:
\begin{align}
 \langle u,v \rangle_{b/sc}=\int\limits_X \bm_{\omega}(u(x,\omega),v(x,\omega))\dVol^{b/sc}(x,\omega)\, .  
\end{align}    
We define  $L^2_{b/sc}( (0,\frac{1}{r_+-\epsilon})_x\times \cB_s^{\ss^2})$ as the completion of ${\Gamma}_c( (0,\frac{1}{r_+-\epsilon})_x\times \cB^{\ss^2}_s)$ for the associated norm. Starting with these $L^2$ spaces, the Sobolev spaces are constructed as in the general setting. To do this, we remark that the connection $\Theta$ on $\cB(s,s)$ induces a family of connections $\Theta_{\ft_0}$ on $X\times \cB_s^{\ss^2}$, as in the construction in Remark \ref{rem3.7}. Because of the stationarity of the spacetime, these connections are in fact independent of $\ft_0$ if a stationary tetrad scaling is selected for the identification of $\cB(s,s)$ and $\cB(s)$. Assuming this to be the case, we drop the index $\ft_0$.
\begin{remark}
The drawback of the above definition is that the Sobolev spaces require working in a fixed scaling of the tetrad's null vectors to identify $\cB(s)$ and $\cB(s,s)$. Changing the  scaling of the tetrad will change the Sobolev norms of the field. 
\end{remark}
 \begin{remark}
Note that the connection we use is slightly different from the connection used in \cite{Millet2}. Millet requires for the connection that 
\begin{align*}
\Theta_{\partial_{\varphi}}=\partial_{\varphi}+is\cos\theta,\, \Theta_{\partial_{\theta}}=\partial_{\theta},\, \Theta_{\partial_x}=\partial_x 
\end{align*}
 in a suitable local trivialization. Nonetheless, the two connections give equivalent Sobolev norms. This follows from 
\begin{align*}
\Theta_a=\nabla_a-2sw_a,\quad w_a=\frac{1}{2}\left(n^b\nabla_al_b+m^b\nabla_a\bar{m}_b\right)
\end{align*}
on $\cB(s,s)$ and the fact that the contractions of $w_a$ with smooth vector fields give smooth bounded functions on $\cM$ in the scaling provided by the Kinnersley tetrad. Indeed, $w_a$ is smooth because it is built from smooth vector  fields, therefore the contraction with smooth vector fields is also smooth. As for possible growth issues at infinity we have to consider star-Kerr coordinates. All Christoffel symbols are bounded because the metric coefficients as well as their derivatives and inverses are bounded. A similar remark holds for the Kinnersley tetrad, see \eqref{KinnersleystarKstarl}, \eqref{KinnersleystarKstarn}. Therefore, for $\nu=\partial_r$, $\nu=\partial_{{}^*t}$, $\nu=\partial_{{}*\varphi}$, and $\nu=\partial_{\theta}$, the contractions $\nu^aw_a$ are smooth and bounded. This remains valid for other choices of the tetrad which display a similar behaviour at infinity.
\end{remark}

\subsection{Scattering pseudodifferential operators}
In this section, we introduce scattering pseudodifferential operators. Let us recall that symbols $S^{m,l}(\rr_z^n\times\rr_{\xi}^n;\cc)$ are smooth functions $a\in C^{\infty}(\rr_z^n\times\rr_{\xi}^n;\cc)$ such that 
\begin{align}
\label{estsymb}
\forall \alpha,\beta\in \nn^n\, \exists C_{\alpha\beta}>0:\quad \vert \partial_z^{\alpha}\partial_{\xi}^{\beta}a(y,\xi)\vert\le C_{\alpha\beta} \langle z\rangle^{l-\abs{\alpha}}\langle \xi\rangle^{m-\abs{\beta}}. 
\end{align}
We similarly define the space of semiclassical symbols of order $(m,l)$, $S_h^{m,l}(\rr_z^n\times\rr_{\xi}^n;\cc)$, as the set of $h$-indexed families $(u_h)_{h\in [0,1)}\in C^{\infty}([0,1)_h;S^{ m,l}(\rr_z^n\times\rr_{\xi}^n;\cc))$ such that for all $h\in [0,1)$ one has $u_h\in S^{m,l}(\rr_z^n\times\rr_{\xi}^n;\cc)$ and the estimate \eqref{estsymb} holds uniformly in $h\in[0,1)$.  We also put $S^{-\infty,-\infty}=\cap_{m,l} S^{m,l}$ and $S_h^{-\infty,-\infty}=\cap_{m,l} S_h^{m,l}$.    
 For a symbol $p\in S^{m,l}(\rr_z^n\times\rr_{\xi}^n ;\cc)$, we define the operator $Op(p)$ by its action on $\phi\in C_0^\infty(\rr^n;\cc)$,
\begin{align}
\label{quant1}
Op(p)\phi(z)=(2\pi)^{-n}\int e^{\i\xi(z-z')}p(z,\xi)\phi(z')\d^n z'\d^n\xi.
\end{align}  
We similarly define $Op_h(p_h)$ for a symbol $p_h\in S^{m,l}_h(\rr_z^n\times \rr_{\xi}^n;\cc)$ by 
\begin{align}
\label{quant2}
Op_h(p_h)u(y)=(2\pi h)^{-n}\int e^{\i h^{-1}\xi(z-z')}p_h(z,\xi)u(z')\d^n z'\d^n\xi.
\end{align}  
The classes of pseudodifferential operators $\Psi^{m,l}(\rr^n),\, \Psi_{h}^{m,l}(\rr^n)$ are then the classes of operators obtained by the above procedures. It is useful to look at the above procedure in the compactified picture to make the link to the scattering calculus of Melrose. We introduce spherical coordinates $(r,\omega)$ and put $x=\frac{1}{r}$. $\rr^n$ is then compactified to $\bar{\rr}^n$ by adding $x=0$.  A general covector $\zeta dz$ becomes 
\begin{align}
\zeta dz=\tau \frac{dx}{x^2}+\mu \frac{d\omega}{x}
\end{align}
with 
\begin{align}
\tau=-\omega\cdot \zeta,\quad \mu=\zeta-(\omega\cdot \zeta)\omega. 
\end{align}
Conversely, 
\begin{align}
\zeta=-\tau \omega+\mu,\quad \mu\cdot\omega=0. 
\end{align}
This entails $\abs{\zeta}^2=\tau^2+\abs{\mu}^2$. After this change of coordinates, the symbol $a(x,\omega,\tau,\mu)$ satisfies
\begin{align}
\label{symbest}
\vert V_1...V_k\partial^{\beta}_{(\tau,\mu)}a\vert\lesssim x^{-l+k}\langle (\tau,\mu)\rangle^{m-\abs{\beta}}
\end{align}
for all scattering vector fields $V_1,...,V_k$. The class of symbols satisfying \eqref{symbest} is denoted $S^{m,l}({}^{sc}T^*\bar{\rr^n})$ in the following. We similarly define the classes $S^{m,l}_{h}({}^{sc}T^*\bar{\rr^n})$. 
The symbol classes $S^{m,l}({}^{sc}T^*\bar{X};\cc)$, $S_{h}^{m,l}({}^{sc}T^*\bar{X};\cc)$ are defined in the same way. We then define the corresponding pseudodifferential operators by going back to euclidean coordinates and using the quantizations \eqref{quant1} and \eqref{quant2}. This defines the operator classes $\Psi_{sc}^{m,l}(\bar{X};\cc)$ and $\Psi_{sc,h}^{m,l}(\bar{X};\cc)$. 
We can then define the algebra of scattering pseudodifferential operators on the bundle $E:=[0,\frac{1}{r_+-\epsilon})_x\times \cB_s^{\ss^2}$.
\begin{definition}
We define $\Psi^{m,l}_{sc}(E)$ as the set of properly supported linear operators $A:\Gamma_c(E)\rightarrow \cD'(E)$ such that for all $\chi_1,\, \chi_2$ smooth and compactly supported on some open  trivializing set for $E,\, \chi_1A\chi_2\in \Psi^{m,l}_{sc}(\bar{X};\cc)$.
\end{definition} 
The invariant definition of the principal symbol is as follows.  Let 
\begin{align*}
\pi:{}^{sc}T^*\bar{X}\rightarrow \bar{X}
\end{align*}
be the canonical projection and $\pi^*E$ the pullback of the bundle $E$. Then the invariantly defined principal symbol $\sigma^m(A)$ is an element of 
\begin{align*}
\sigma^m(A)\in S^{m,l}/S^{m-1,l-1}\left({}^{sc}T^*\bar{X};\pi^*{\rm Hom(E)}\right). 
\end{align*} 
That means that a representation of $\sigma^m(A)$ is a map assigning to $(y,\xi)\in {}^{sc}T^*\bar{X}$ an element of ${\rm Hom}(E_y)$. We refer to \cite{HiLN} for details. Note that the hermitian metric $\bm$ on $\cB_s^{\ss^2}$ induces a hermitian metric on $E$ and thus a hermitian metric on $\pi^*{\rm Hom}(E)$. The symbol classes $S^{m,l}\left({}^{sc}T^*\bar{X};\pi^*{\rm Hom(E)}\right)$ are then defined by replacing the absolute value in \eqref{symbest} by this hermitian metric.  In order to not overload the notations we will write $S_{sc}^{m,l}$ instead of $S^{m,l}\left({}^{sc}T^*\bar{X};\pi^*{\rm Hom(E)}\right)$. 

The definition of the set $\Psi_{sc,h}^{m,l}(E)$ is obtained by replacing the operator $A$ by a family $(A_h)_{h\in [0,1)}$ of continuous linear operators from $C_c^{\infty}(E)$ to $\cD'(E)$ and $\Psi^{m,l}_{sc}(\bar{X};\cc)$ by $\Psi_{sc,h}^{m,l}(\bar{X};\cc)$ in the previous definition.  
 
Eventually, we denote by $\Psi^{m,l}_{sc,c}(E)$ the set of operators with Schwartz kernel supported in $((r_+-\epsilon+\eta,\infty)\times \ss^2)^2$ for some $0<\eta<\epsilon$. $\Psi^{m,l}_{sc,h,c}$ denotes the class of corresponding semiclassical operators.  

We say that an operator $P\in \Psi_{sc}^{m,l} (E)$ is elliptic if its principal symbol $p\in S_{sc}^{m,l}/S_{sc}^{m-1,l-1}$ has an inverse, in other words, if there exists $q\in S_{sc}^{-m,-l}/S_{sc}^{-m-1,-l-1}$ such that $pq=qp=[1]$, where $[1]$ is the equivalence class of the symbol $1$ in the space $S_{sc}^{0,0}/S_{sc}^{-1,-1}$. 

We can localize this definition near a point $(x,k)\in {}^{sc}T^*\bar{X}\setminus 0$ as follows: We say that the pseudodifferential operator $P\in \Psi^{m,l}_{sc} (E)$ with principal symbol $p\in S_{sc}^{m,l}/S_{sc}^{m-1,l-1}$ is elliptic at $(x,k)\in {}^{sc}\dot{T}^*\bar{X}$ if there exists $q\in S_{sc}^{-m,-l}/S_{sc}^{-m-1,-l-1}$ such that $pq=qp=[g]$ with $g=1$ on a conic neighbourhood of $(x,k)$. The subset of ${}^{sc}\dot{T}^*\bar{X}$ at which $P$ is elliptic is denoted ${\rm Ell}(P)$. The complementary subset is the characteristic set denoted by ${\rm Char}(P)$.

A point $(x,k)\in {}^{sc}\dot{T}^*\bar{X}$ \underline{is not} in the wavefront set of $P$ if there exists $A\in \Psi^{0,0}_{sc,c} (E)$ with $(x,k)\in {\rm Ell}(A)$ such that $AP\in \Psi^{-\infty,-\infty}_{sc,c} (E)$. We denote the wavefront set by $\WF(P)$. 

All these definitions have an analogue in the semiclassical setting. We say that a semiclassical operator $P\in\Psi_{sc,h}^{m,l} (E)$ is elliptic if its principal symbol $p\in S_{sc}^{m,l}/hS_{sc}^{m-1,l-1}$ has an inverse, in other words there exists $q\in S_{sc}^{-m,-l}/hS_{sc}^{-m-1,-l-1}$ such that $pq=qp=[1]$ where $[1]$ is the equivalence class of the symbol $1$ in the space $S_{sc}^{0,0}/hS_{sc}^{-1,-1}$.

Again, we can localize this definition near a point $(x,k)\in {}^{sc}T^*\bar{X}$. We say that the semiclassical operator $P\in \Psi^{m,l}_{sc,h} (E)$ with principal symbol $p\in S_{sc}^{m,l}/hS_{sc}^{m-1,l-1}$ is elliptic at $(x,k)$ if there exists $q\in S_{sc}^{-m,-l}/hS_{sc}^{-m-1,-l-1}$ such that $pq=qp=[g]$ with $g=1$ on a neighbourhood of $(x,k)$. The subset of ${}^{sc}T^*\bar{X}$ at which $P$ is elliptic is denoted ${\rm Ell}_h(P)$. The complementary subset is the characteristic set denoted by ${\rm Char}_h(P)$. 

The point $(x,k)\in {}^{sc}T^*\bar{X}$ \underline{is not} in the semiclassical wavefront set of a semiclassical operator $P\in \Psi^{m,l}_{sc,h} (E)$, denoted by $\WF_h(P)$, if there exists a semiclassical operator $A\in \Psi^{0,0}_{sc,h} (E)$ with $(x,k)\in {\rm Ell}_h(A)$ such that $AP\in h^{\infty}\Psi^{-\infty,-\infty}_{sc,h} (E)$.

\section{Classical field estimates}
\label{sec:millet}
In this section we recall some results on the classical field obtained by Millet in \cite{Millet2}. 
\subsection{Influence of the time reversal}
\label{sec:trev}
We want to consider backward solutions of 
\begin{align}
\label{eq:4.1}
T_{s}\phi_{s}=f_s,
\end{align}
whereas Millet considers forward solutions of the same equation. As we have seen in the previous section, the identification between $\Gamma(\cB(s,s))$  and $C^{\infty}(\rr_{\ft}\times (r_+-\epsilon,\infty);\Gamma(\cB_s^{\ss^2}))$, and therefore the size of the Sobolev norms, depends on the tetrad. Two different tetrads are of interest here. The Kinnersley tetrad $(l^K, n^K,m^K)$ and the tetrad used by Millet, 
\begin{align}
(l^M,n^M,m^M)=(\Delta l^K,\Delta^{-1}n^K, m^K). 
\end{align}
For a smooth section $\phi_s\in \Gamma(\cB(s,s))$, we write in the following $\phi^K_s$ and $\phi_s^M$ for the corresponding elements in $C^{\infty}(\rr_{\ft}\times [0,\frac{1}{r_+-\epsilon});\Gamma(\cB_s^{\ss^2}))$ obtained by using the Kinnersley tetrad, or the tetrad used in \cite{Millet2}, respectively. Equation \eqref{eq:4.1} can then be written as   
\begin{align}
\label{eq:4.1a}
T^K_{s}\phi^K_{s}=f_s^K,
\end{align}
where $f_s^K$ has to be understood as an element of $C_0^{\infty}(\rr_{\ft}\times (r_+,\infty);\Gamma(\cB_s^{\ss^2}))$. 
We apply the flip  $\psi: \, (t,\varphi)\mapsto (-t,-\varphi)$. By the discussion in Section \ref{sec:flip}, we know that this flip sends $T^K_s\phi_s^K$ to $T^M_{-s}2^s\zeta^{2s}\phi_{-s}^M$, where $\phi_{-s}$ is defined as $\phi_{-s}=\iota\circ\Psi_{(s,s)}^*\phi_s$. 
This entails that the function $\widetilde{\phi}^M_{-s}=2^s\zeta^{2s}\phi_{-s}^M$ is a solution of 
\begin{align}
\label{6.3}
T_{-s}^M\widetilde{\phi}_{-s}^M=f_{-s}^M.  
\end{align}
We want to compute the relationship between $\phi_{-s}^M$ and $\phi_s^K$. Let $z=\sqrt{\frac{2}{\Delta}}\zeta$. We then have by \eqref{4.23}
\begin{align}
\phi_{-s}(x)&=\iota\circ \Psi^*[(l^K,n^K,m^K),\phi_s(x)]=[(l^K,n^K,m^K)z^{-2},\phi_s(\psi^{-1}x)]\\\nonumber
&=[(l^M,n^M,m^M),\zeta^{-2s}2^{-s}\phi_s(\psi^{-1}x)].
\end{align}
Therefore, $\phi^M_{-s}(x)=\zeta^{-2s}2^{-s}\phi_s^K(\psi^{-1}x)$ and thus $\widetilde{\phi}^M_{-s}(x)=\phi_s^K(\psi^{-1}x)$. Note that $f_s$ and $f_{-s}$ have compact supports, and that $f^M_{-s}(x)=\zeta^{-2s}2^{-s}f_s^K(\psi^{-1}x)$. Using a local trivialization, we can then see that  there is a constant $C>0$ so that
\begin{subequations}
\begin{align}
&\Vert \phi^K_s\Vert_{\bar{H}^{m,l}_b([0,\frac{1}{r_+-\epsilon})\times \cB^{\ss^2}_s)}=\Vert \widetilde{\phi}^M_{-s}\Vert_{\bar{H}^{m,l}_b([0,\frac{1}{r_+-\epsilon})\times \cB^{\ss^2}_s)},\\
&C^{-1} \Vert f_{-s}^M\Vert_{\bar{H}_b^{m,l}([0,\frac{1}{r_+-\epsilon})\times \cB^{\ss^2}_{-s})}\le \Vert f_s^K\Vert_{\bar{H}^{m,l}_b([0,\frac{1}{r_+-\epsilon}))\times \cB^{\ss^2}_s)}\\\nonumber
&\le C \Vert f_{-s}^M\Vert_{\bar{H}_b^{m,l}([0,\frac{1}{r_+-\epsilon})\times \cB^{\ss^2}_{-s})}.
\end{align}
\end{subequations}
Note that in the above equations the spaces $\bar{H}_b^{m,l}([0,\frac{1}{r_+-\epsilon})\times \cB^{\ss^2}_{\mp s})$ are constructed using different tetrads and spacelike slices. To clarify the difference between the time function $\ft$  defined in Section~\ref{subsec:Kstar and starK} and the time function $\ft$ in \cite{Millet2}, we will also introduce indices on the time functions ($K$ for our time function, $M$ for the one used in \cite{Millet2}). We then define 
\begin{align*}
\Sigma^{K/M}_{\ft_0}=\{\ft^{K/M}=\ft_0\}.
\end{align*}
The Sobolev spaces in the previous subsection are then defined using the time slices $\Sigma^{K}_{-\ft_0}$ and $\Sigma^{M}_{\ft_0}$ for some $\ft_0\ge 0$. We then note that the restriction $\psi\vert_{\Sigma^{K}_{-\ft_0}}$  of the flip is a diffeomorphism onto $\Sigma^{M}_{\ft_0}$. 

Let us now consider the same flip $\psi$ at the level of the Fourier transform. We start with
\begin{align*}
\hat{T}_s(\sigma)\hat{\phi}^K_s(\sigma)=\hat{f}_s^K(\sigma)\Leftrightarrow T_s^Ke^{i\sigma\ft^K}\hat{\phi}^K_s(\sigma)=e^{i\sigma\ft^K}\hat{f}_s^K(\sigma)\,, .
\end{align*}
We now apply the flip and express everything in the tetrad used in \cite{Millet2}. This gives 
\begin{align*}
e^{i\sigma\ft^M}T_{-s}^M2^s\zeta^{2s}e^{-i\sigma\ft^M}\hat{\phi}_{-s}^M(\sigma)=f^M_{-s}(\sigma)\Leftrightarrow \hat{T}^M_{-s}(-\sigma)\hat{\tilde{\phi}}^M_{-s}(\sigma)=\hat{f}^M_{-s}(\sigma). 
\end{align*} 
The same remark on the norms holds also for the Fourier transformed sections and we therefore have to study $\hat{T}^M_{-s}(-\sigma)$.  
We will also use Kruskal coordinates and the tetrad  $(\mathfrak{l}, \mathfrak{n},\, \mathfrak{m})$. For a section $\phi_s\in \Gamma(\cB(s,s))$, we denote the corresponding element of $C^{\infty}(\rr_U\times\rr_V;\Gamma(\cB_s^{\ss^2}))$ by $\phi_s^{U}$. 

By the above discussion the results of Millet \cite{Millet2} easily translate to our situation. 
\subsection{Semiclassical flow}
\label{subsec:semiflow}
Here, we work with the tetrad used in \cite{Millet2}, the operator $\hat{T}^M_{s,h}(z)$, and the coordinate $\ft^M$.  

We drop the index $M$ in the following. We will consider the semiclassical operator
\begin{align*}
\hat{T}_{s,h}(z)=h^2\hat{T}_s(h^{-1}z)\, , \quad  z=\frac{\sigma}{\vert \sigma\vert}\, ,\quad h=\vert \sigma\vert^{-1}. 
\end{align*}
 Here, we suppose $0\le {\Im z}\le Ch$ for some $C>0$. Then, for $z_0\in\{\pm 1\}$, one has $z-z_0=\cO(h)$.
The semiclassical principal symbol is $p_h(\xi)=-\rho^2g^{-1}(\xi-z_0\d\ft)$. Let $H=H_{p_h}$ be the associated hamiltonian vector field.

We define 
\begin{align*}
\Sigma_{\pm}=p_h^{-1}\{0\}\cap\{\pm (\varrho^2 z_0+Mr\xi_r)>0\}.
\end{align*}
$\Sigma_{\pm}$ are connected components of the characteristic set $\Sigma_h=p_h^{-1}\{0\}$. Different time functions will be adapted in different regions, see \cite[Lemma 5.12]{Millet2} and \cite[Lemma 5.13]{Millet2} for a discussion on these coordinate changes. For an interval $I$ we define $U_I=I\times \ss^2$. 
\subsubsection{Semiclassical flow near the horizon}
We use the coordinate $t^*$ on the radial interval $I=(r_+-\epsilon,3M)$ and coordinates $\xi=\xi_r\d r+\zeta \d\varphi^*+\eta \d\theta$ for cotangent vectors. 
We will work on the fibre radial compactification ${}^{sc}\bar{T^*}X$ of ${}^{sc}T^*X$ and introduce
\begin{align*}
\widetilde{\rho}=\frac{1}{\sqrt{\xi_r^2+\frac{\zeta^2}{\sin^2\theta}+\eta^2}}. 
\end{align*} 
The rescaled hamiltonian vector field $\widetilde{H}=\widetilde{\rho}H_{p_h}$ extends to $\{\widetilde{\rho}=0\}$.
Let 
\begin{align*}
\Lambda_+&=\{\Delta_r=0,\, \eta=0,\, \zeta=0,\, \xi_r>0\},\, L_+=\Lambda_+\cap \{\widetilde{\rho}=0\},\\
\Lambda_-&=\{\Delta_r=0,\, \eta=0,\, \zeta=0,\, \xi_r<0\},\, L_-=\Lambda_-\cap \{\widetilde{\rho}=0\}.
\end{align*}
\subsubsection{Semiclassical flow near $x=0$}
We use the time function $^*t$ and the interval $I=(6M,\infty)$. We define the coordinates on the cotangent bundle by $\xi=\xi_r\d r+\zeta \d\varphi+\eta \d\theta$. Then $p_h(\xi)=-\rho^2g^{-1}(\xi-z_0\d{}^*t)$.  We will consider the renormalized Hamiltonian vector field $\widetilde{H}_{p_h}=r^{-1}H_{p_h}$ which extends to a continuous vector field on ${}^{sc}T^*((r_+-\epsilon,\infty]\times \ss^2)$.

We now put $x=\frac{1}{r}$,  $\xi_{sc}=\xi_r$,  $\eta_{sc}=\frac{\eta}{r}$, and $\zeta_{sc}=\frac{\zeta}{r}$. Let us further define 
\begin{align*}
{\mathcal R}_{out}=\{\xi_{sc}=\zeta_{sc}=\eta_{sc}=x=0\},\quad {\mathcal R}_{in}=\{\xi_{sc}=-2z_0,\, \zeta_{sc}=\eta_{sc}=x=0\}. 
\end{align*}
\subsubsection{Semiclassical flow near the trapped set}
We now use the time coordinate $t$ on the radial interval $I=(r_{min},r_{max})$. Here, $(r_{min},r_{max})$ is such that bicharacteristics leaving $(r_{min},r_{max})\times \ss^2$ leave any compact set of $(r_+,\infty)\times \ss^2$ in either the past or the future, see \cite{Millet2} for details. 
For given $z_0\in \{-1,1\}$, let $\Gamma_{\pm}$ be the set of points in $T^*U_I$ such that the bicharacteristics with corresponding initial data are trapped in the past\footnote{Here, past and future is with respect to the flow of the bicharacteristic.} ($+$) resp. future $(-)$ and $K_{z_0}=\Gamma_+\cap \Gamma_-$. 
\subsubsection{Properties of the semiclassical flow}
We will need the following (see \cite[Proposition 5.38]{Millet2}):
\begin{proposition}
\label{pp0}
We define the surface $B_{\epsilon}=\{r=r_+-\epsilon\}\subset {}^{sc}\bar{T}^*((r_+-2\epsilon,\infty]\times \ss^2)$. In this proposition, we say that a curve $\gamma$ defined 
on some interval $J$ is of type $(A,B)$ where $A$ and $B$ are two sets in $\{L_-,L_+,B_{\epsilon},{\mathcal R}_{in},{\mathcal R}_{out}\}$ if $\gamma$ tends to $A$ at $\inf J$ and to $B$ at $\sup J$. Moreover, if $A=B_{\epsilon}$ (resp. $B=B_{\epsilon}$), then $\inf J>-\infty$ (resp. $\sup J<\infty$), in all other cases the corresponding bound of $J$ is infinite.

Let $\gamma$ be a bicharacteristic for the renormalized Hamiltonian flow on ${}^{sc}\bar{T^*}((r_+-\epsilon,\infty]\times \ss^2)$ maximally defined on some interval $J$. 

Let $ z_0=-1$. 
\begin{itemize}
\item If $\gamma\subset \Sigma_+$, then either $\gamma\subset L_+$ or $\gamma$ is of type $(L_+,B_{\epsilon})$. 
\item If $\gamma\subset \Sigma_-$, then either $\gamma\subset L_-\cup K_{z_0}\cup {\mathcal R}_{in}\cup {\mathcal R}_{out}$, or $\gamma$ is of type $({\mathcal R}_{out},L_-)$, $({\mathcal R}_{out},K_{z_0})$, $({\mathcal R}_{out},{\mathcal R}_{in})$,  $(B_{\epsilon},L_-)$, $(B_{\epsilon},K_{z_0})$, $(B_{\epsilon},{\mathcal R}_{in})$, $(K_{z_0},L_-)$, or $ (K_{z_0},{\mathcal R}_{in})$. 
\end{itemize}
Let $z_0=1$. 
\begin{itemize}
\item If $\gamma\subset \Sigma_-$, then either $\gamma\subset L_-$ or $\gamma$ is of type $(B_{\epsilon},L_-)$. 
\item If $\gamma\subset \Sigma_+$, then either $\gamma\subset L_+\cup K_{z_0}\cup {\mathcal R}_{in}\cup {\mathcal R}_{out}$, or $\gamma$ is of type $(L_+,B_{\epsilon})$, $(L_+, K_{z_0})$, $(L_+,{\mathcal R}_{out})$, $({\mathcal R}_{in},K_{z_0})$, $({\mathcal R}_{in}, B_{\epsilon})$, $({\mathcal R}_{in},{\mathcal R}_{out})$, $(K_{z_0}, B_{\epsilon})$, or $(K_{z_0},{\mathcal R}_{out})$. 
\end{itemize}
\end{proposition}
Figure \ref{fig:semflow} summarizes the situation for the semiclassical flow. 
\begin{figure}
\centering
\begin{tikzpicture}
\draw (-3, 2) -- (-3,-2) node[yshift = -0.2cm] (B) {$B_{\epsilon}$};
\draw (0,0) node[draw] (K) {$K_{ 1}$};
\draw (2,0) node[draw] (Rin) {$\mathcal{R}_{\text{out}}$};
\draw (2,-1) node[draw] (Rout) {$\mathcal{R}_{\text{in}}$};
\draw (-1, 2) node[draw] (Lp) {$L_+$};
\draw (-1, -2) node[draw] (Lm) {$L_-$};
\draw[decoration={markings, mark=at position 0.5 with {\arrow{>}}}, postaction={decorate}] (Lp)--(K);
\draw[decoration={markings, mark=at position 0.5 with {\arrow{>}}}, postaction={decorate}] (Lp)--(Rin);
\draw[decoration={markings, mark=at position 0.5 with {\arrow{>}}}, postaction={decorate}] (Lp)--(-3,0.5);
\draw[decoration={markings, mark=at position 0.5 with {\arrow{>}}}, postaction={decorate}] (K)--(Rin);
\draw[decoration={markings, mark=at position 0.5 with {\arrow{>}}}, postaction={decorate}] (K)--(-3,0);
\draw[decoration={markings, mark=at position 0.5 with {\arrow{>}}}, postaction={decorate}] (Rout)--(Rin);
\draw[decoration={markings, mark=at position 0.5 with {\arrow{>}}}, postaction={decorate}] (Rout)--(K);
\draw[decoration={markings, mark=at position 0.5 with {\arrow{>}}}, postaction={decorate}] (Rout)--(-3,-0.5);
\draw[decoration={markings, mark=at position 0.5 with {\arrow{>}}}, postaction={decorate}] (-3, -1.5)--(Lm);
\end{tikzpicture}\hspace{1cm}
\begin{tikzpicture}
\draw (-3, 2) -- (-3,-2) node[yshift = -0.2cm] (B) {$B_{\epsilon}$};
\draw (0,0) node[draw] (K) {$K_{ -1}$};
\draw (2,0) node[draw] (Rin) {$\mathcal{R}_{\text{in}}$};
\draw (2,1) node[draw] (Rout) {$\mathcal{R}_{\text{out}}$};
\draw (-1, 2) node[draw] (Lp) {$L_+$};
\draw (-1, -2) node[draw] (Lm) {$L_-$};
\draw[decoration={markings, mark=at position 0.5 with {\arrow{>}}}, postaction={decorate}] (K)--(Lm);
\draw[decoration={markings, mark=at position 0.5 with {\arrow{>}}}, postaction={decorate}] (Rin)--(Lm);
\draw[decoration={markings, mark=at position 0.5 with {\arrow{>}}}, postaction={decorate}] (-3,-0.5)--(Lm);
\draw[decoration={markings, mark=at position 0.5 with {\arrow{>}}}, postaction={decorate}] (Rin)--(K);
\draw[decoration={markings, mark=at position 0.5 with {\arrow{>}}}, postaction={decorate}] (-3,0)--(K);
\draw[decoration={markings, mark=at position 0.5 with {\arrow{>}}}, postaction={decorate}] (Rin)--(Rout);
\draw[decoration={markings, mark=at position 0.5 with {\arrow{>}}}, postaction={decorate}] (K)--(Rout);
\draw[decoration={markings, mark=at position 0.5 with {\arrow{>}}}, postaction={decorate}] (-3,0.5)--(Rout);
\draw[decoration={markings, mark=at position 0.5 with {\arrow{>}}}, postaction={decorate}] (Lp)--(-3, 1.5);
\end{tikzpicture}
\caption{Structure of the semiclassical Hamiltonian flow for the original operator $\hat{T}_{s,h}^M(z_0)$ (for $ z_0=1$ on the left and $z_0 = -1$ on the right), see Section \ref{subsec:semiflow}. Source : \cite{Millet2}.}
\label{fig:semflow}
\end{figure}
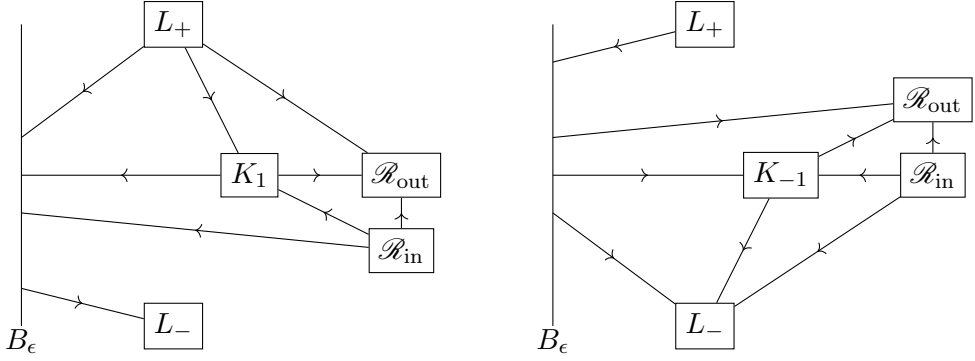

\subsection{Fredholm-type estimates}
 Fredholm estimates for non-elliptic problems in the spirit of this section were first introduced in \cite{Va13}. We will closely follow the exposition for the Teukolsky operator in \cite{Millet2}. We will consider the operator $\hat{T}_{s}^K(\sigma)$. Recall from the discussion in Section \ref{sec:trev}
 that estimates for this operator follow from the corresponding estimates for the operator $\hat{T}^M_{-s}(-\sigma)$.  This means that we have to consider the flow of Section \ref{subsec:semiflow} for $\widetilde{z}_0=-z_0$. In the following, we will drop the index $K$ and write $\Psi^{m,l}_{sc,h(,c)}$ for $\Psi^{m,l}_{sc,h(,c)}([0,\frac{1}{r_+-\epsilon})_x\times \cB^{\ss^2}_s)$ to avoid clutter of notation.  
 
 In \cite{Millet2}, Fredholm-type estimates were obtained from the precise understanding of the semiclassical Hamiltonian flow. We will recall some of them that will be useful for us.  
One obtains regularity near $ L_{\pm}$ and $\sR_{in}$ "for free". This regularity will then be propagated by the Propagation of Singularities Theorem.  
\subsubsection{Propagation of Singularities}
The following Proposition is the standard propagation of singularities, see e.g. \cite[Theorem 8.2.7]{HiLN}. 
 \begin{proposition}
\label{pp3}
Let $B_0,\, B_1,\, G\in \Psi^{0,0}_{sc,h,c}$  with compactly supported Schwartz kernels, and $r,\ell\in \rr$. Assume that for every $x\in \WF_h(B_0)\cap\Sigma_h$, there exists a $t>0$ (resp. $t<0$) such that $e^{ t H_{p_h}}x\in {\rm Ell}_h(B_1)$ and $(e^{sH_{p_h}})_{s\in [0,t]} x$  (resp. $(e^{sH_{p_h}})_{s\in [t,0]}x)$ remains in the elliptic set of $G$. Then for every $N>0$, there exists a constant $C_N>0$ such that we have the estimate
\begin{align*}
\Vert B_0u\Vert_{\bar{H}^{r,\ell}_{(b),h}}\le C_N\left(h^{-1}\Vert G\hat{T}_{s,h}(z)u\Vert_{\bar{H}^{r-1,\ell}_{(b),h}}+\Vert B_1u\Vert_{\bar{H}^{r,\ell}_{(b),h}}+h^N\Vert u\Vert_{\bar{H}_{(b),h}^{r-N,\ell}}\right) 
\end{align*}
\end{proposition} 
 for any $u\in \cD' \left(\left[0, \frac{1}{r_+-\epsilon}\right)_x\times \cB_s^{\ss^2}\right)$ in the strong sense, i.e., if the right-hand side is finite, then the left-hand side is finite and the inequality holds.
\subsubsection{Radial point estimates}
We first recall the estimates at $L_{\pm}$. Suitable estimates at radial points were first used in this setting in \cite{Va13}. The following is \cite[Proposition 6.1]{Millet2}:
\begin{proposition}
\label{pp1}
Let $ r\ge r'> \frac{1}{2}-( s-\frac{1}{\kappa_+}{\Im}(h^{-1}z)),\, \ell\in \rr$. If $A,B,G\in \Psi^{0,0}_{sc,h}$ with compactly supported Schwartz kernels are such that $A$ and $G$ are elliptic on $L_{\pm}$ and every forward (or backward) bicharacteritic curve from $\WF_h(B)$ tends to $L_{\pm}$ with closure in the elliptic set of $G$, then for all $N\in \nn$ there exists a constant $C_N>0$ such that $Au\in \bar{H}^{r',\ell}_{(b),h}$ implies\footnote{Note that the error $h\Vert u\Vert_{\bar{H}^{-N,\ell}_{(b),h}}$ in \cite{Millet2}  is improved here to $h^N\Vert u\Vert_{\bar{H}^{-N,\ell}_{(b),h}}$ by iterating the estimates, see \cite{Vasy2013}, footnote 33 on page 418.}
\begin{align*}
\Vert Bu\Vert_{\bar{H}^{r,\ell}_{(b),h}}\le C_N\left(h^{-1}\Vert G\hat{T}_{s,h}(z)u\Vert_{\bar{H}^{r-1,\ell-1}_{(b),h}}+h^N\Vert u\Vert_{\bar{H}^{-N,\ell}_{(b),h}}\right). 
\end{align*} 
\end{proposition}
Estimates at ${\mathcal R}_{in}$ will also be important. An adapted estimate was first shown in \cite{Va21}. The following is \cite[Proposition 6.13]{Millet2}:
\begin{proposition}
\label{pp2}
Let $\cU$ be a neighbourhood of $\sR_{in}$ separated from the zero section. Then there exist $B_0,G\in \Psi^{0,0}_{sc,h}$ with $\WF_h(B_0)\cup \WF_h(G)\subset \cU$, $\WF_h(B_0)\subset {\rm Ell}_h(G)$ and $B_0$ elliptic at $\sR_{in}$. If $\eta>0$ is such that $-\eta h\le {\rm Im}(z)\le 0$, and in addition $r+\ell+\frac{1}{2}-2s>0$ and $u\in H^{r',\ell'}_{(b)}$ with $r'+\ell'+\frac{1}{2}-2s>0$, then for all $N\in \nn,$ there exists $C_N>0$ independent of $u$,  $h$, and $z$ so that 
\begin{align*}
\Vert B_0 u\Vert_{H^{r,\ell}_{(b),h}}\le C_N\left(h^{-1}\Vert G\hat{T}_{s,h}(z)u\Vert_{H^{r,\ell-1}_{(b),h}}+h^N\Vert u\Vert_{H^{-N,\ell}_{(b),h}}\right). 
\end{align*}   
\end{proposition}
\subsubsection{Estimates at the trapped set}
 We will also need estimates near the trapped set. Suitable estimates were first shown by Wunsch-Zworski \cite{WuZw11} and Dyatlov \cite{Dy16}. The following version follows from \cite[Theorem 4.7]{HiVa16} and \cite[Proposition 6.22]{Millet2}:
\begin{proposition}
\label{pp4}
 Let $z=z_0+\cO(h)$ with $z_0\in \{-1,1\}$, and let $\ell\in \rr$ and $r>0$.  Then there exist a bounded neighbourhood $\cU$ of the trapped set $K_{-z_0}$ and operators $B_K, B_0$ in $\Psi^{0,0}_{sc,h}$ with $B_K$ elliptic on a neighbourhood of $K_{- z_0}$, $WF_h(B_0)\subset\cU, WF_h(B_K)\subset \cU$, and $\WF_h(B_0)\cap \Sigma_h\subset \Sigma_{- z_0}$ with either ${\rm Ell}_h(B_0)\cap \Gamma_+=\emptyset\, (z_0=-1)$ or ${\rm Ell}_h(B_0)\cap \Gamma_-=\emptyset\, (z_0=1)$.  Moreover, for all $N\in \nn$, there exists $C_N>0$ such that for all  $u\in \bar{H}^{r,\ell}_{(b),h}$ we have :
 \begin{align*}
 \Vert B_K u\Vert_{\bar{H}^{r,\ell}_{(b),h}}\le C_N \left(h^{-1}\Vert B_0u\Vert_{\bar{H}^{r,\ell}_{(b),h}}+h^{-2}\Vert G\hat{T}_{s,h}(z)u\Vert_{\bar{H}^{r,\ell-1}_{(b),h}}+h^N\Vert u\Vert_{\bar{H}^{-N,\ell}_{(b),h}}\right).
 \end{align*}
 \end{proposition}

\begin{remark}
    Note that the weights in $x$ in the definition of the spaces $H^{r,\ell}_{(b),h}$ only play a role in Proposition~\ref{pp2}. However, since all of these estimates are combined later on, the weights are included in all of them.
\end{remark}


\subsection{Decay estimates}
\label{SecDecayest}
We need decay estimates for backward solutions of \eqref{eq:4.1a}.
By the discussion at the beginning of this section, these estimates are equivalent to the corresponding estimates for the forward solution of \eqref{6.3}. 

\subsubsection{Energy decay}
We have the following local energy decay estimate : 
\begin{theorem}
\label{theorem:7.1}
Let $\alpha\in (0,1)$. Let $f^K\in C_0^{\infty}\left(\rr_t\times (0,\frac{1}{r_+-\epsilon})_x; \Gamma\left(\cB_s^{\ss^2}\right)\right)$. Let $\phi^K_s$ be the unique backward solution to the forcing problem \eqref{eq:4.1a}. Then for all $r,j\in \nn^{\ge 0}$, there exists $C>0$ so that 
\begin{align}
\label{th:7.1}
\Vert (\ft\partial_{\ft})^j \phi^K_s\Vert_{\bar{H}_{b}^{r,1-}}\le C\langle \ft\rangle^{-\alpha}. 
\end{align}
\end{theorem}
\begin{proof}
We apply the flip $\psi$ sending $t\mapsto -t$ and $\varphi\mapsto -\varphi$. By the discussion in Section~\ref{sec:flip}, we know that this flip sends $T^K_s$ to $T^M_{-s}$ and $\phi^M_{s}$ to $\psi_{-s}^M=2^{s}\zeta^{2s}\phi^M_{-s}$, where the $M$ stands for the tetrad used in \cite{Millet2}. Now, $\psi_{-s}^M$ is a forward solution of \eqref{6.3}. The result then follows from \cite[Theorem 8.1]{Millet2}. 
\qeds
\end{proof}
\begin{theorem}
\label{th6.7}
Let $f^K\in C_0^{\infty}\left(\rr_{\ft}\times \left(0,\frac{1}{r_+-\epsilon}\right)_x;\Gamma(\cB_s^{\ss^2})\right)$. Let $\phi^K_s$ be the unique backward solution to the forcing problem \eqref{eq:4.1a}. Then for all $j\in \nn$, there exists $C_j>0$ such that uniformly in $\ft\le -1$ we have 
\begin{align*}
\vert(\ft\partial_{\ft})^j\phi^K_s(\ft,x,\omega^*)\vert\le C_j\vert\ft\vert^{-3-\vert s\vert-s}. 
\end{align*}
\end{theorem}
\proof
We again use the flip $\psi$. Let $\Sigma_{\ft_0}=\ft^{-1}\{\ft_0\}$. By finite speed of propagation, the data on $\Sigma_{\ft_0}$ is compactly supported for suitable chosen $\ft_0$. Then \cite[Theorem 1.2]{Millet2} and \cite[Corollary 8.5]{Millet2} give the result. 
\qed
\subsubsection{Radiation field and decay at null infinity}

Recall from the discussions in Section~\ref{sec:Kerr} the coordinates $x=\frac{1}{r}$ and $\ft$, and define the new coordinates $v=-x\ft$ and $\tau=-\ft^{-1}$. Let $\cN=[0,1)_{v}\times[0,1)_{\tau}$. Let $\bar{H}_b^{m,\mu,\nu}(\cN\times \cB_s^{\ss^2})=v^{\mu}\tau^{\nu}\bar{H}_b^{m}(\cN\times \cB_s^{\ss^2})$, where $\bar{H}_b^{m}(\cN\times\cB_s^{\ss^2})$ is the $b-$ Sobolev space defined in Section~\ref{sec:func setting} with extendible conditions at $v=1$ and $\tau=1$. Let $f^K\in C_0^{\infty}(\rr_{\ft}\times (r_+-\epsilon); \Gamma(\cB_s^{\ss^2}))$, and let $\phi^K_{s}$ be the backward solution of \eqref{eq:4.1a}.
\begin{lemma}
\label{lem6.3}
There exists $\phi_{s}^{K,rad}\in \bar{H}_b^{\infty,(3+s+\vert s\vert)-}([0,1)_{\tau}\times\cB_s^{\ss^2})$ such that
\begin{align}
\label{equ:4.2}
\phi^K_{s}-v\phi_{s}^{K,rad}\in \bar{H}_b^{\infty,2-,(3+s+\vert s\vert)-}(\cN\times \cB_s^{\ss^2}).
\end{align}
\end{lemma}

\begin{proof}
We apply the flip  $t\to -t$ and $\varphi\to -\varphi$. By the discussion in Section~\ref{sec:flip}, we know that this flip sends $T^K_s$ to $T^M_{-s}$ and $\phi^{K}_{s}$ to $\psi_{-s}^M=2^{s}\zeta^{2s}\phi^M_{-s}$, where the $M$ stands for the tetrad used in \cite{Millet2}. Additionally, $v^K$ is sent to $v^M$ and $\tau^K$ to $\tau^M$, where the index $K$ indicates the coordinates used here and the index $M$ indicates the coordinates used in \cite{Millet2}. Now we know that $\psi^{M}_{-s}$ is a forward solution of  $T^M_{-s}\psi^M_{-s}=0$,
and we can apply \cite[Proposition 7.45]{Millet2} to $\psi^M_{-s}$. Translating this to $\phi_{s}^K$ gives the result. 
\qeds
\end{proof}

\begin{corollary}
\label{cor:4.1}
Let $f^K\in C_0^{\infty}(\rr_{\ft}\times (r_+,\infty); \Gamma(\cB_s^{\ss^2}))$ and $\phi_s^K$ a backward solution of \eqref{eq:4.1a}. Then there exists a constant $c_f$ such that $rE_s^-(f^K)$ extends smoothly to $\sI_-$ and $\supp rE_s^-(f^K)\cap \sI_-\subset\{\ft\le c_f\}$. 
Furthermore, for all $n\in \nn$ there exists $C>0,\, m>0$ such that 
\begin{align}
\sup_{\omega\in \ss^2}\left\vert \partial_\ft^n(rE_s^-(f^K))\vert_{\sI^-}(\ft,\omega))\right\vert_{\omega}\le C \Vert f^K\Vert_{C^m} \langle \ft\rangle^{-2-s-\vert s\vert-n+}.
\end{align}
Here $\vert.\vert_{\omega}$ is the fibre norm on $\cB_s^{\ss^2}$ at the point $\omega\in \ss^2$. 
\end{corollary}
\begin{proof}
This is a direct consequence of Lemma \ref{lem6.3}.
\qeds
\end{proof}

\begin{definition}
We define 
\begin{enumerate}
\item $\cS_s(\cM):=\left\{(\phi_s,\overline{\phi_{-s}})\in \Gamma(\cV_s)\,{\rm s.t.}\, \phi^K_{\pm s}\in C^{\infty}\left(\rr_{\ft}\times\left[0,\frac{1}{r_+-\epsilon}\right), \Gamma\left(\cB_{\pm s}^{\ss^2}\right)\right)\right.$ fulfils\\
$\hphantom{\cS_s(\cM):=\{}$\eqref{th:7.1} and  \eqref{equ:4.2} for $\left. \vphantom{ \left(\frac{1}{r_+}\right)}\ft\le 0 \right\},$
\item $\cS_s(\sI_-):=\left\{\breve{\phi}=(\breve \phi_s,\overline{\breve \phi}_{-s})\in \Gamma(\breve{\cV}_{s}\vert_{\sI_-}): \exists \ft_0(\breve{\phi})\in \rr\, \text{s.t.}\, 1_{\ft>\ft_0}\breve\phi=0\right.$ and\\ $\hphantom{\cS_s(\sI_-):=\{}\left. \forall n\in \nn,\, \exists C_n(\breve\phi) : \sup\limits_{\omega\in \ss^2}\abs{\partial_\ft^n \breve\phi^K_s}_{\omega}\leq C_n(\breve{\phi}) \langle\ft\rangle^{-2-s-\vert s\vert-n+},\right.$\\
$\hphantom{\cS_s(\sI_-):=\{}\left. \sup\limits_{\omega\in \ss^2}\abs{\partial_\ft^n\overline{\breve\phi^K_{-s}}}_{\omega}\leq C_n(\breve{\phi}) \langle\ft\rangle^{-2+s-\vert s\vert-n+} \right\}$,\\
where $1_{\ft>\ft_0}$ is the characteristic function of $(\ft_0,\infty)$.
\item The trace operator $T_{\sI}:\Sol_{s}(\cM)\to  \Gamma(\breve{\cV}_{s}\vert_{\sI_-})$ is defined by $\phi\mapsto \breve{\phi}\vert_{x=0}$.\footnote{ Recall $\breve\phi= x^{-1}\phi$.}
\end{enumerate}
\end{definition}


\subsubsection{Decay along the long horizon}
\label{subsec:decay hor}
We again consider a backward solution of \eqref{eq:4.1} for $f\in \Gamma_c(\cB(s,s))$. 
\begin{theorem}
\label{th:4.2}
There exists $d>1 $ and $U(\phi_s)\in\rr$ so that $1_{U>U(\phi_s)}\phi_s\vert_\sH=0$ and so that for all $n\in \nn$ there exists $C>0,\, m>0$ satisfying
\begin{align*}
\sup_{\omega\in \ss^2}\abs{\partial_U^n\phi_s^{U}\vert_{\sH}(U,\omega)}_{\omega}\leq C\Vert f^K\Vert_{C^m}\langle U\rangle^{s-n}(\log\langle U\rangle)^{-d-n}.
\end{align*} 
\end{theorem}
\begin{proof}
The existence of $U(\phi_s)$ follows from the finite speed of propagation. $\phi_s^{U}$ is smooth by the general theory of hyperbolic equations. The only issue is to show decay for $U\rightarrow -\infty$ and we can work in Kinnersley's tetrad and use $(\ft,r,\varphi_+,\theta)$ coordinates. We have $\phi_{s}^{U}=\abs{U}^{s}\phi_{s}^K$. We then have to show decay estimates on $\phi_s^K$. Again, we use the flip $\psi$, which sends $\phi_s^K$ to $2^s\zeta^{2s}\phi_{-s}^M$, forward solution of  \eqref{6.3}. The result then follows directly from \cite[Proposition 8.4]{Millet2}. 
\qeds
\end{proof}
\begin{definition}
    For $s\in\{0,1,2\}$, we define
    \begin{enumerate}
    \item $\cS_s(\sH):=\left\{ \phi=(\phi_s,\overline{\phi_{-s}}) \in \Gamma(\cV_{s}\vert_{\sH}): \exists d>1, U(\phi)\in\rr \text{ s.t. } 1_{U>U(\phi)}\phi=0\, \right.$ and \\
    $\hphantom{\cS_s}\left. \forall n\in\nn\:\exists C_n(\phi)>0: \sup\limits_{\omega\in \ss^2}\abs{\partial_U^n\phi^{U}_{\pm s}\vert_{\sH}(U,\omega)}_{\omega}\leq C_n(\phi)\langle U\rangle^{\pm s-n}(\log\langle U\rangle)^{-d-n}\right\}.$\\
Here $1_{x> x_0}$ is the characteristic function of  $(x_0,\infty)$.
\item The trace operator $T_\sH:\Sol_{s}(\cM)\to  \Gamma(\cV_{s}\vert_\sH)$ is defined by $\phi\mapsto \phi\vert_{\sH}$. 
\end{enumerate}
\end{definition}

It follows from part 4 of Lemma~\ref{prop:Green hyp props} applied to $\cP_s$ that every space-compact solution $\phi\in\Sol_s(\cM)$ is a difference of forward and backward solutions of an appropriate forcing problem and therefore that it and its restrictions lie in the above spaces. 

\section{Conservation of the symplectic form}
\label{sec:sympl form}
We have seen in Lemma \ref{lemma:cons J} that $\nabla_aJ^a[f,h]=0$ for $f,h\in C^\infty(\cM;\cV_s)$ satisfying $\cP_s f=\cP_s h=0$. By Stokes' theorem, this gives the conservation of an associated charged symplectic form $\sigma$ on a family of Cauchy surfaces, see Definition~\ref{def:phase spaces}.
We would like to use the conservation of $J$ and Stokes' theorem to express $\sigma$ in terms of integrals over past null infinity and the past horizon. Because of the singularity at $i^-$, this argument cannot be applied directly. The decay estimates of Section \ref{SecDecayest} are crucial in this context.    

\subsection{Conformal rescaling}
Let $\phi_s\in \Gamma(\cB(s,s))$ satisfy $T_s\phi_s=0$.
 The conformal rescaling induces the maps $\breve\cdot: \fN_0\to  \breve{\fN}_0$, $(l,n,m)\mapsto (\breve l, \breve n, \breve m)= (l, x^{-2}n, x^{-1}m)$ that we have introduced before, as well as the map $\breve\cdot:  \Gamma( \cV_{s}\vert_{ \MI})\to \Gamma(\breve{\cV}_{s}\vert_{\breve{\mathrm{M}}_{\mathrm{I}}})$, $[(l,n,m), \phi]\mapsto[(\breve l,\breve n, \breve m), \breve\phi]=[(\breve l,\breve n, \breve m), x^{-1}\phi]$.
Then, by Lemma~\ref{lem:rescaling} and the fact that $R=0$ on Kerr, we have
\begin{align}
\left[\breve{T}_s+\frac{1}{6}\breve{R}\right]\breve{\phi}_s=0\, ,
\end{align}
where we recall that $\breve{T}_s$ is the Teukolsky operator of the conformally rescaled metric, see \eqref{eq:T_s conf rescaled}.

Let now $\phi$, $\psi\in \Sol_{s}(\cM)$, and let $\breve{J}^a$ be the current given by 
\begin{align*}
\breve{J}_a[\breve{\phi},\breve{\psi}]&=\IP{\breve{\phi}, \breve{\cD}_{s,a}\breve{\psi}}-\IP{\breve{\cD}_{s,a}\breve{\phi},\breve \psi}\, \\
\breve{\cD}_{s,a}&=(\breve{\Theta}_a+2s\breve{B}_a)\oplus \overline{(\breve{\Theta}_a-2s\breve{B}_a)}=(\breve{\nabla}_a+2s\Gamma_a)\oplus(\breve{\nabla}_a-2s\overline{\Gamma}_a)\, .
\end{align*}
Then one can easily see by direct computation that 
\begin{align}
\label{eq:sigma conf covar}
\breve{J}_a[\breve\phi,\breve\psi]=x^{-2}J_a[\phi,\psi]\, .
\end{align}

Moreover, using exactly the same calculation as in the proof of Lemma~\ref{lemma:cons J}, but putting breves on all quantities, one obtains 
\begin{lemma}
\label{lemma:sigma conf inv}
$\breve{J}_a[\breve\phi,\breve\psi]$ is conserved in the conformally rescaled spacetime, i.e. 
\begin{align}
\breve{\nabla}_a\breve{g}^{ab}\breve{J}_b[\breve{\phi},\breve{\psi}]=0\,.
\end{align}
\end{lemma}
In addition, as a direct consequence of the above and Stokes' theorem, we obtain
\begin{lemma}
Let $\Sigma_t$ be a Cauchy surface of $\breve{\rm M}_{\rm I}$. In an abuse of notation, we will use the same name for the Cauchy surface of $\MI$ that can be identified with (the interior of) $\Sigma_t$ via the embedding of $\MI$ into $\breve{\rm M}_{\rm I}$. Let $\phi,\psi\in \Sol_{s}(\MI)$. Let $\breve{v}^a$ be a future directed normal vector to $\Sigma_t$, and $\breve{w}^a$ transverse to $\Sigma_t$ such that $\breve{g}(\breve{v},\breve{w})=1$. Assume that $\breve{v}^a=x^{c_1}v^a$, $\breve{w}^a=x^{c_2}w^a$ so that $c_1+c_2=-2$, or in other words $g(v,w)=1$. Then 
\begin{align}
\int\limits_{\Sigma_t}\breve{J}_a[\breve\phi,\breve\psi]\breve{v}^a(\breve{w}\lrcorner \dVol_{\breve g})=\int\limits_{\Sigma_t}J_a[\phi,\psi]v^a(w\lrcorner\dVol_g)
\end{align}
 and it is independent of the choice of Cauchy surface $\Sigma_t$. 
\end{lemma}



\subsection{Symplectic form on $\sH$}
We have 
  \begin{proposition}
    For $\phi,\, \psi\in \cS_s(\sH)$,
    \begin{align}
    \label{eq:def sigma H}
\sigma_{\sH,s}(\phi, \psi)&:=(-1)^s\int\limits_{\sH} *J_a[\phi,\psi]=(-1)^{s}r_+^{-2}\int\limits_\sH*\breve{J}_a[\breve\phi,\breve\psi]\nonumber\\
&=2(-1)^{s}\int\limits_{\mathclap{\rr_U\times \ss^2}}\IP{\phi,\Theta_U \psi}(r_+^2+a^2)\d U\d^2\omega_+\, ,
\end{align}
is a well-defined charged symplectic form. Here, $\d^2\omega_+$ is the standard volume element of $\ss^2_{\theta,\varphi_+}$, and $\IP{\cdot, \cdot}$ is the fibre metric of $\cV_s$ defined in \eqref{eq:cV fibre metric}. $\cS_s(\sH)$ therefore has the structure of a charged symplectic space. If the integral on the right-hand side of \eqref{eq:def sigma H} is taken over some $\cO\subset \sH$, we refer to that expression as $\sigma_{\cO,s}$.
\end{proposition}

\begin{proof}
 As we consider the horizon $\sH$, we will work in the Kruskal tetrad \eqref{eq:Kruskal tetrad} which smoothly extends through the bifurcation sphere.  Let us work in the KBL-coordinates introduced in Section~\ref{subsec:Kruskal ext}. We can choose $\partial_U$ as the normal vector to $\sH$. Considering the metric in \eqref{eq:g Kruskal}, we then find that for the appropriately chosen vector $w$ such that $w$ is future directed, null, transversal to $\sH$ and $g(\partial_U,w)=1$, one obtains (with $\mathfrak{i}$ the injection $\rr_U\times\ss^2_{\theta,\varphi_+}\to M_K$)
\begin{align}
\mathfrak{i}^*(w^a\lrcorner \dVol_g)=(r_+^2+a^2)\d U\wedge \d^2\omega_+\, .
\end{align}
Taking into account that $\partial_U$ is proportional to $\mathfrak{n}^a$ on the horizon, with a proportionality factor that is bounded with bounded inverse, and that $\mathfrak{n}^a\Gamma_a$ in the Kruskal tetrad vanishes on $\sH$, we obtain for $\phi$, $\psi\in \cS_{s}(\sH)$
\begin{align}
\label{eq:current horizon}
&\int\limits_{\sH} *J_a[\phi,\psi]\\\nonumber
&=\int\limits_{\sH}\left[\overline{\phi_s^{U}} \partial_{U}\overline{\psi_{-s}^{U}}+\phi_{-s}^{U}\partial_{U}\psi_s^{U}-\overline{\partial_U\phi_s^{U}}\overline{\psi_{-s}^{U}}-\partial_U\phi_{-s}^{U}\psi_s^{U}\right](r_+^2+a^2)\d U\d^2 \omega_+\, .
\end{align}
Here, $\sH$ is identified with $\rr_U\times\ss^2_{\theta,\varphi_+}$. From this form, it is apparent by partial integration that the form is indeed anti-hermitian. Let now $U_0=\max\{U(\phi_s),U(\psi_s)\}$. Then
\begin{align*}
\abs{\,\int\limits_{\sH}*J_a[\phi,\psi]}=\abs{\hphantom{-\infty}\int\limits_{\mathclap{(-\infty,U_0]\times \ss^2}}*J_a[\phi,\psi]}\, .
\end{align*}
Let $U_1<\min(-2,U_0)$. We then estimate
\begin{align*}
\abs{\hphantom{-\infty}\int\limits_{\mathclap{(-\infty,U_0]\times \ss^2}}*J_a[\phi,\psi]}\leq\abs{\hphantom{-\infty}\int\limits_{\mathclap{(-\infty,U_1)\times \ss^2}}*J_a[\phi,\psi]}+\abs{\hphantom{123}\int\limits_{\mathclap{[U_1,U_0]\times \ss^2}}*J_a[\phi,\psi]}\, .
\end{align*}
The second integral is finite because $\phi$ and $\psi$ are smooth. By the definition of $\cS_s(\sH)$, we have 
\begin{align*}
    \abs{\overline{\phi_s^{U}} \partial_{U}\overline{\psi_{-s}^{U}}(U,\omega_+)}\leq c_1 \abs{U}^s\abs{\log\abs{U}}^{-d}\abs{U}^{-s-1}\abs{\log\abs{U}}^{-\widetilde d}=c_1\abs{U}^{-1}\abs{\log\abs{U}}^{-(d+\widetilde d)}
\end{align*}
on $\{U<U_1\}\cap\sH$ for some constants $c_1$ depending on $\phi$ and $\psi$, and $d>1, \widetilde{d}>1$.
Similar estimates can be obtained for the other terms.
We then find
\begin{align*}
&\abs{\hphantom{-\infty}\int\limits_{\mathclap{(-\infty,U_1)\times \ss^2}}*J_a[\phi,\psi]}
\leq C \int\limits_{\mathclap{(-\infty,U_1)}}\abs{U}^{-1}\abs{\log\abs{U}}^{-d-\widetilde d}\d U<\infty\, ,
\end{align*}
where $C$ is a positive constant that depends on $\phi$ and $\psi$. The first equality in \eqref{eq:def sigma H} follows from \eqref{eq:sigma conf covar}
, the second is obtained by integration by parts. 
\qeds
\end{proof}

\subsection{Symplectic form on $\Sigma_{\ft_0}$}
In the coordinates $(\ft,x,\theta,\varphi_*)$ on $\breve{\mathrm{M}}_\mathrm{I}$, let $\Sigma_{\ft_0}$ be the hypersurface $\{\ft=\ft_0\}$.
\begin{proposition}
\label{prop5.6}
For $\phi,\psi\in \cS_s(\cM)$ 
\begin{align}
\sigma_{\Sigma_{\ft_0}}(\breve\phi,\breve\psi):=(-1)^s\int\limits_{\Sigma_{\ft_0}}*\breve{J}_a[\breve{\phi},\breve{\psi}]
\end{align}
is  well-defined. In addition, we have 
\begin{align}
\sigma_{\Sigma_{\ft_0}}(\breve\phi,\breve\psi)\rightarrow 0,\quad \ft_0\rightarrow-\infty. 
\end{align}
\end{proposition}
\begin{proof}
We start with the vector 
field $\breve{w}=\partial_{\ft}$ which is always transverse to $\Sigma_{\ft_0}$. Let $\breve{v}^a$ be a normal vector to $\Sigma_{\ft_0}$ such that $\breve{v}_a\breve{w}^a=1$. It will be useful to have an explicit expression in the region $r>4M$.   

For $r>4M$, the coordinates $(\ft,x,\theta,\varphi_*)$ agree with $(t^*,x,\theta,\varphi^*)$, and hence one must have $\breve{v}^a=\breve{g}^{-1}(\d t^*)$. From \eqref{eq:inv conf metric}, one finds explicitly 
\begin{equation*}
\breve{v}^a=:\varrho_x^{-2}\breve{v}_0^a= \varrho_x^{-2}(-a^2\sin^2\theta \partial_{t^*}+(1+a^2x^2)\partial_x-a\partial_{\varphi^*})\, .
\end{equation*}
Denoting $\omega^*=(\theta,\varphi^*)$, and $\d^2\omega^*=\sin\theta \d\theta \d\varphi^*$, we have $
\dVol_{\breve{g}}=\varrho_x^2 \d\ft\wedge \d x\wedge \d^2\omega^*$, and thus $\dVol_{\Sigma_{\ft_0}}=\breve{w}^a\lrcorner \dVol_{\breve{g}}=\varrho_x^2 \d x\wedge \d^2\omega^*$.
We then have 
\begin{align*}
\sigma_{\Sigma_{\ft_0}}(\breve\phi,\breve\psi)=&(-1)^s\int\limits_{\mathclap{\Sigma_{\ft_0}}}\breve{J}_a[\breve\phi,\breve\psi]\breve{v}^a(\breve{w}\lrcorner \dVol_{\breve g})=(-1)^s\int\limits_{\mathclap{\Sigma_{\ft_0}}}\breve{J}_a[\breve\phi,\breve\psi]\breve{v_0}^a \d x \d^2\omega^*\\
=&(-1)^s\int\limits_{\Sigma_{\ft_0}}\breve{v}_0^a\left[\overline{\breve{\phi}_s}(\breve{\nabla}_a-2s\overline{\Gamma}_a)\overline{\breve{\psi}_{-s}}+\breve{\phi}_{-s}(\breve{\nabla}_a+2s\Gamma_a)\breve{\psi}_s\right.\\
&\left.\hphantom{\int\limits_{\Sigma_{\ft}}\breve{v}_0^a}-\overline{\breve{\psi}_{-s}}(\breve{\nabla}_a+2s\overline{\Gamma}_a)\overline{\breve{\phi}_s}-\breve{\psi}_s(\breve{\nabla}_a-2s\Gamma_a)\breve{\phi}_{-s}\right]\d x \d^2\omega^*\, .
\end{align*}

A straightforward calculation using \eqref{eq:Kstar coord trafo}, \eqref{eq:g conf}, and \eqref{eq:Gamma} shows that
\begin{align}
\breve{v}_0^a\Gamma_a=x^2\left[\frac{\i a\cos\theta}{2}-\frac{M(1-a^2x^2)}{\Delta_x}\right]
\end{align}
in the Kinnersley tetrad.

Therefore, for $\breve{\phi},\, \breve{\psi}$ supported in $x\in[0,\frac{1}{4M}]$, the symplectic form $\sigma_{\Sigma_{\ft_0}}$ reduces to 
\begin{align}
\label{eq:sympl sigma t}
\sigma_{\Sigma_{\ft_0}}(\breve{\phi},\breve{\psi})=&(-1)^s\int\limits_{\Sigma_{\ft_0}}\left[(1+a^2x^2)\left(\overline{\breve{\phi}_s}\partial_x \overline{\breve{\psi}_{-s}}+ \breve{\phi}_{-s} \partial_x \breve{\psi}_s - \overline{\breve{\psi}_{-s}}\partial_x \overline{\breve{\phi}_s} -\breve{\psi}_s \partial_x \breve{\phi}_{-s} \right)\right.\nonumber\\
&-a^2\sin^2\theta \left(\overline{\breve{\phi}_s}\partial_{\ft}\overline{\breve{\psi}_{-s}}+\breve{\phi}_{-s}\partial_{\ft}\breve{\psi}_s-\overline{\breve{\psi}_{-s}}\partial_{\ft}\overline{\breve{\phi}_s}-\breve{\psi}_s\partial_{\ft}\breve{\phi}_{-s}\right)\\\nonumber
&-a\left(\overline{\breve{\phi}_s}\partial_{\varphi^*}\overline{\breve{\psi}_{-s}}+\breve{\phi}_{-s}\partial_{\varphi^*}\breve{\psi}_s-\overline{\breve{\psi}_{-s}}\partial_{\varphi^*}\overline{\breve{\phi}_s}-\breve{\psi}_s\partial_{\varphi^*}\breve{\phi}_{-s}\right)\\\nonumber
&+4sx^2\frac{M(1+a^2x^2)}{\Delta_x}\left(\overline{\breve{\phi}_{s}}\overline{\breve{\psi}_{-s}}-\breve{\phi}_{-s}\breve{\psi}_s\right)\\\nonumber
&\left.+4sx^2\frac{\i a\cos\theta}{2}\left(\overline{\breve{\phi}_{s}}\overline{\breve{\psi}_{-s}}+\breve{\phi}_{-s}\breve{\psi}_s\right)\right]\d x \d^2\omega^*
\end{align}
upon identifying $\Sigma_{\ft_0}=\left[0,r_+^{-1}\right]_x\times\ss^2_{\theta,\varphi^*}$. 

 To show that $\sigma_{\Sigma_{\ft_0}}$ is well-defined, we can cut the integral over $\Sigma_{\ft_0}$ into two parts, one part for $x\in [\frac{1}{4M}, r_+^{-1}]$ and the other part for $x\in[0,\frac{1}{4M}]$. The first part is clearly well-defined and goes to zero as $\ft_0\to -\infty$ by the properties of $\cS_s(\cM)$. For the second part, we can suppose that $\breve{\phi}$ and $\breve{\psi}$ are supported in $x\le \frac{1}{4M}$ and use the explicit formula \eqref{eq:sympl sigma t}. We start with the term in the first line of \eqref{eq:sympl sigma t}. We only consider the term 
\begin{align*}
\int\limits_{\Sigma_{\ft_0}}(\breve{\phi}_{-s}x\partial_x\breve{\psi}_s)\frac{\d x \d^2\omega^*}{x},
\end{align*}
the others can be treated in the same way. We note that by \eqref{equ:4.2}, $x\partial_x\breve{\psi}_{s}=v\partial_v\breve{\psi}_{s}=\widetilde{\psi}_s$
with 
\begin{align}
\label{6.9c}
\widetilde{\psi}_s\in \bar{H}^{\infty,1-,(2+s+\vert s\vert)-}_{b}. 
\end{align}
This allows us to write for $\ft_0\le -1$ 
\begin{align}
\widetilde{\psi}_{s}(\ft_0,x,\omega^*)=\int\limits^{\ft_0}_{-\infty}\partial_{\ft}\widetilde{\psi}_s(\ft,x,\omega^*)\d\ft\, ,
\end{align}
since for $x$ and $\omega^*$ fixed, $\abs{\widetilde{\psi}_s(\ft,x,\omega^*)}\to 0$ as $\ft\to -\infty$ by Theorem~\ref{th6.7}.
Thus, we can estimate
\begin{align}
\vert \widetilde{\psi}_{s}(\ft_0,x,\omega)\vert^2\lesssim \vert \ft_0\vert^{-\epsilon}\int\limits^{\ft_0}_{-\infty}\abs{\ft}^{-1+\epsilon}\vert \ft\partial_{\ft}\widetilde{\psi}_{s}(\ft,x,\omega)\vert^2\d\ft\, . 
\end{align}
We can then estimate for $\delta>0$
\begin{align*}
&\int\limits_{\ss^2}\int\limits_{0}^{\frac{1}{4M}} x^{-\delta}\vert \widetilde{\psi}_{s}(\ft_0,x,\omega^*)\vert^2\frac{\d x \d^2\omega^*}{x}\le C \vert\ft_0\vert^{-\epsilon}\int\limits_{\ss^2}\int\limits_{-\infty}^{\ft_0}\int\limits_0^{\frac{1}{4M}}\abs{\ft}^{-1+\epsilon}\abs{ \ft\partial_{\ft}\widetilde{\psi}_{s}(\ft,x,\omega^*)}^2\frac{\d x \d\ft\d^2\omega^*}{x^{1+\delta}}\\
&=C \vert\ft_0\vert^{-\epsilon}\int\limits_{\ss^2}\int\limits_{-\infty}^{\ft_0}\int\limits_0^{-\frac{1}{\ft}}\vert\ft\vert^{-1+\epsilon}\abs{ \ft\partial_{\ft}\widetilde{\psi}_{s}(\ft,x,\omega^*)}^2\frac{\d x \d\ft\d^2\omega^*}{x^{1+\delta}}\\
&+C \vert\ft_0\vert^{-\epsilon}\int\limits_{\ss^2}\int\limits_{-\infty}^{\ft_0}\int\limits_{-\frac{1}{\ft}}^{\frac{1}{4M}}\vert\ft\vert^{-1+\epsilon}\abs{ \ft\partial_{\ft}\widetilde{\psi}_{s}(\ft,x,\omega^*)}^2\frac{\d x \d\ft\d^2\omega^*}{x^{1+\delta}}\\
&=C(I_1+I_2). 
\end{align*}
Let us first estimate $I_2$. We have by Theorem~\ref{th6.7} 
\begin{align*}
I_2\lesssim \vert\ft_0\vert^{-\epsilon}\int\limits_{\ss^2}\int\limits_{-\infty}^{\ft_0}\int\limits_{-\frac{1}{\ft}}^{\frac{1}{4M}}\vert\ft\vert^{-1+\epsilon}\abs{\ft}^{-6}\frac{\d x \d\ft\d^2\omega^*}{x^{1+\delta}}\lesssim \vert\ft_0\vert^{-\epsilon}\int\limits_{\ss^2}\int\limits_{-\infty}^{\ft_0}\int\limits_{-\frac{1}{\ft}}^{\frac{1}{4M}}\abs{\ft}^{-6+\epsilon+\delta}\d x \d\ft\d^2\omega^*\lesssim \abs{\ft_0}^{-5+\delta}. 
\end{align*}
We now estimate $I_1$. We change coordinates by going from $x$ and $\ft$ to $v=-x\ft$ and $\tau=-\ft^{-1}$. Then
\begin{align*}
I_1= \vert\ft_0\vert^{-\epsilon}\int\limits_{\ss^2}\int\limits_0^1\int\limits_0^{-\frac{1}{\ft_0}}\tau^{-\epsilon-\delta}v^{-\delta}\vert (-\tau\partial_{\tau}+v\partial_v) \widetilde{\psi}_{s}\vert^2\frac{\d\tau \d v\d^2\omega^*}{v\tau}.
\end{align*}
By \eqref{6.9c}, we have $(-\tau\partial_{\tau}+v\partial_v)\widetilde{\psi}_{s}=\tau^{2+s+\vert s\vert-\delta}v^{1-\delta}\check{\psi}_s$ 
with $\check{\psi}_s\in \bar{H}_b^{\infty,0,0}$. As a result, we have
\begin{align*}
I_1\lesssim \vert{\ft_0}\vert^{-\epsilon}\int\limits_{\ss^2}\int\limits_0^1\int\limits_0^{-\frac{1}{\ft_0}}\tau^{4+2s+2\vert s\vert-3\delta-\epsilon}v^{2-3\delta}\vert \check{\psi}_s\vert^2\frac{\d\tau \d v\d^2\omega^*}{\tau v}\lesssim \abs{\ft_0}^{-3}
\end{align*}
for $\epsilon, \delta$ small enough. At the same time, we know by \eqref{th:7.1} 
that 
\begin{align*}
\int\limits_{\Sigma_{\ft_0}} x^{-2+\delta}\vert \phi_{-s}\vert^2\frac{\d x\d^2\omega^*}{x}\lesssim \langle \ft_0\rangle^{-2\alpha}.
\end{align*}
We then use the Cauchy-Schwarz inequality to obtain the desired estimate.

Next, we consider the terms in the third line of \eqref{eq:sympl sigma t}, which can be treated using \eqref{th:7.1}
 and the Cauchy Schwartz inequality, e.g.  
\begin{align*}
\left\vert\,\int\limits_{\Sigma_{\ft_0}}\overline{\breve{\phi}_s}\partial_{\varphi^*}\overline{\breve{\psi}_{-s}}\d x\d^2\omega^*\right\vert\le  \Vert\phi_s\Vert_{\bar{H}_b^{0,1/2}}\Vert \partial_{\varphi^*}\psi_{-s}\Vert_{\bar{H}_b^{0,1/2}}\xrightarrow{\ft_0\to -\infty} 0\, . 
\end{align*}
The same argument applies to the terms in line 2 and 4. The terms in line 2 have even better decay in $\ft$, the terms in line 4 have even better decay in $x$. 
\qeds
\end{proof}

\subsection{Symplectic form at null infinity}

We now consider $\sI_-$, in other words the surface $\{x=0\}\subset\breve{\mathrm{M}}_{\mathrm{I}}$ in the coordinates $(\ft=t^*,x,\theta,\varphi^*)$. 
\begin{proposition}
\label{prop:sigma I-}
    For $\breve{\phi},\breve{\psi}\in \cS_s(\sI_-)$, 
    \begin{align}
    \label{eq:def sigma I}
\sigma_{\sI,s}(\breve{\phi},\breve{\psi}):=(-1)^s\int\limits_{\sI_-}*\breve{J}[\breve{\phi},\breve{\psi}]=
2(-1)^{s}\int\limits_{\mathclap{\rr_{\ft}\times \ss^2}}\IP{\breve{\phi}, \breve{\Theta}_{\ft}\breve{\psi}}\d\ft \d^2\omega^*\, 
\end{align}
is a well-defined charged symplectic form. Here $\omega^*=(\theta,\varphi^*)$, and $\d^2\omega^*$ is the standard volume element of $\ss^2_{\omega^*}$. $\cS_s(\sI_-)$ therefore has the structure of a charged symplectic space. If the integral  on the right-hand side of \eqref{eq:def sigma I} is taken over some $\cO\subset \sI_-$, we refer to that as $\sigma_{\cO,s}$.
\end{proposition}
\begin{proof}
Recall that, in this limit, the conformally transformed Kinnersley tetrad behaves as $\breve{l}^a=2\partial_{\ft}$ and $\breve{n}^a=\tfrac{1}{2}\partial_x$, and that the metric determinant of the conformally transformed metric $\breve g$, restricted to $x=0$, is $-\sin^2\theta$. Therefore, a short calculation shows that $\breve{l}^a$ is a future directed normal vector of $\sI_-$, while $\breve{n}^a$ is transversal to it, and we have $\dVol_{\sI_-}=\breve{n}^a\lrcorner \dVol_{\breve{g}}=\frac{1}{2} \d\ft\wedge \d^2\omega^*$.
Moreover, we recall that $\breve{l}^a\Gamma_a$ vanishes linearly in $x$ as $x\to 0$.
We then find for the current 
\begin{align}
\label{eq:sympl scri}
\int\limits_{\sI_-}*\breve{J}_a[\breve\phi,\breve\psi]=&\int\limits_{\sI_-}\breve{l}^a\breve{J}_a[\breve{\phi},\breve{\psi}](\breve{n}^a\lrcorner\dVol_{\breve{g}})\\\nonumber
=&\int\limits_{\mathclap{\rr_{\ft}\times \ss^2}}(\overline{\breve{\phi}_{s}}\ft\partial_{\ft}\overline{\breve{\psi}_{-s}}+\breve{\phi}_{-s}\ft\partial_{\ft}\breve{\psi}_s-\overline{\breve{\psi}_{-s}}\ft\partial_{\ft}\overline{\breve{\phi}_{s}}-\breve{\psi}_s\ft\partial_{\ft}\breve{\phi}_{-s})\frac{\d\ft \d^2\omega^*}{\ft}\,.
\end{align}
It is apparent using partial integration that this form is indeed anti-hermitian.
We use \eqref{equ:4.2} to see that $\breve{\psi}_s\in \bar H_{\bop}^{\infty,(2-s+\vert s\vert)-}$, and similar for the other terms. 
To show the convergence of the integral let us consider, for example, the term 
\begin{align*}
\int\limits_{\mathclap{\rr_{\ft}\times \ss^2}}\breve{\phi}_{-s}\ft\partial_{\ft}\breve{\psi}_s\frac{\d\ft \d^2\omega^*}{\ft}=\int\limits_{\mathclap{(-\infty,0]_\tau\times \ss^2}}\breve{\phi}_{-s}\tau\partial_{\tau}\breve{\psi}_s\frac{\d\tau \d^2\omega^*}{\tau}. 
\end{align*}
Now, the fact that $\breve{\phi}_{-s}\in \bar H_{\bop}^{\infty,(2-s+\vert s\vert)-}$ and $\tau\partial_{\tau}\breve{\psi}_s\in \bar H_{\bop}^{\infty,(2+s+\vert s\vert)-}$ assures the convergence of the integral.

The decay also allows us to rewrite $\sigma_{\sI,s}(\breve{\phi},\breve{\psi})$ using partial integration as on $\sH$:
\begin{align*}
&\int\limits_{\sI_-}*\breve{J}_a[\breve\phi,\breve\psi]=2\int\limits_{\sI_-}\left[\overline{\breve\psi_s} \partial_{\ft} \overline{\breve\psi_{-s}} + \breve\phi_{-s} \partial_{\ft} \breve\psi_s\right] \d\ft \d^2\omega^*\\
=&2\int\limits_{\sI_-} \IP{\breve{\phi}, \breve{l}^a\breve{\cD}_{s,a}\breve{\psi}}(\breve{n}^a\lrcorner\dVol_{\breve g})=2\int\limits_{\sI_-}\IP{\breve{\phi}, \breve{l}^a\breve{\Theta}_{a}\breve{\psi}}(\breve{n}^a\lrcorner\dVol_{\breve g})\, .
\end{align*}
The last equality follows from the fact that $\breve{l}^a \breve{B}_a$ vanishes on $\sI_-$.
\qeds
\end{proof}

\subsection{Conservation of the symplectic form and symplectic boundary spaces}
\begin{theorem}
\label{theorem:cons sympl form}
 Let $\phi, \psi\in \Sol_{s}(\cM)$. We have 
\begin{align*}
\sigma_{\Sigma_{\ft_0}}(\breve{\phi},\breve{\psi})=\sigma_{\sH\cap\{\ft\le \ft_0\},s}(T_\sH\phi,T_\sH\psi)+\sigma_{\sI\cap\{\ft\le \ft_0\},s}(T_{\sI}{\phi},T_{\sI}{\psi})\, .
\end{align*}
\end{theorem}
\begin{proof}
By the divergence theorem we have for $\ft_1<\ft_0$
\begin{align*}
\sigma_{\Sigma_{\ft_0}}(\breve{\phi},\breve{\psi})=\sigma_{\sH\cap\{\ft_1\le \ft\le \ft_0\},s}(T_\sH\phi,T_\sH\psi)+\sigma_{\sI\cap\{\ft_1\le \ft\le \ft_0\},s}(T_\sI\phi,T_\sI\psi)+\sigma_{\Sigma_{\ft_1}}(\breve{\phi},\breve{\psi}).
\end{align*}
Noting that $\phi,\psi\in \cS_s(\cM)$, it is then sufficient to take the limit $\ft_1\rightarrow -\infty$ and to apply Proposition \ref{prop5.6}.
\qeds
\end{proof}
As a corollary of the above, we obtain
\begin{corollary}
\label{cor:sympl form}
   Let $\phi, \psi\in \Sol_{s}(\cM)$. Then  
   \begin{align*}
(\sigma_{\sH,s}\oplus \sigma_{\sI,s})((T_\sH\oplus T_\sI)\phi,(T_\sH\oplus T_\sI)\psi)=\sigma_s(\phi,\psi)\, ,
\end{align*}
where for $f,f'\in \cS_s(\sH)\, ,\,\,\breve{h},\breve{h}'\in \cS_s(\sI_-)$ the charged symplectic form $(\sigma_{\sH,s}\oplus \sigma_{\sI,s})$ is defined as
$(\sigma_{\sH,s}\oplus \sigma_{\sI,s})((f,\breve{h}),(f',\breve{h}')):=\sigma_{\sH,s}(f, f')+\sigma_{\sI,s}(\breve{h},\breve{h}')$.
Consequently, the map $(T_\sH\oplus T_\sI): \Sol_{s}(\cM)\to \cS_s(\sH)\oplus \cS_s(\sI_-)$ is an injective morphism of charged symplectic spaces that conserves the charged symplectic form.
\end{corollary}
\begin{proof}
Since $\phi$ and $\psi$ are of space-compact support, the conservation of the current together with an application of Stokes' theorem allows us to compute $\sigma_s(\phi,\psi)$ on the Cauchy surface 
\begin{align*}
(\sH\cap\{U\ge U(\ft_0)\})\cup(\Sigma_{\ft_0}\cap\{x\ge \epsilon\})\cup \{t(x,\ft)=t(\epsilon,\ft_0),\, 0\leq x\le\epsilon\}\, 
\end{align*}  for some $0<\epsilon<(4M)^{-1}$,  see Figure~\ref{fig:Cauchy surf} for an illustration. Here, we assume $\ft_0<0$, $t(x,\ft)=\ft-r_*(x)$ is the Boyer-Lindquist coordinate, and we use the identification $U=-e^{-\kappa_+\ft}$, which is valid on $\sH_-$.
\begin{figure}
    \includegraphics[width=0.4\textwidth]{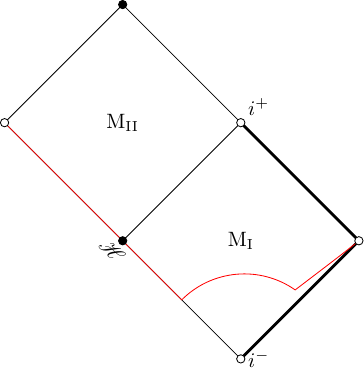}
    \caption{Illustration of the Cauchy surface used in the proof of Corollary~\ref{cor:sympl form} (red). The arc corresponds to the piece $\Sigma_{\ft_0}\cap\{x\ge \epsilon\}$. }
    \label{fig:Cauchy surf}
\end{figure}

Using Stokes' theorem and the space-compactness, as well as working in the conformally extended spacetime, we can then apply Proposition~\ref{prop5.6} and Proposition~\ref{prop:sigma I-} to show that this is the same as computing $\sigma_s$ on the Cauchy surface
\begin{align*}
 (\sH\cap\{U\ge U(\ft_0)\})\cup\Sigma_{\ft_0}\cup(\sI_-\cap\{\ft\ge \ft_0\}) \, .   
\end{align*}
Finally, it follows from Theorem~\ref{theorem:cons sympl form} that 
\begin{align*}
    \sigma_s(\phi, \psi)=&\sigma_{\Sigma_{\ft_0}}(\breve{\phi},\breve{\psi})+\sigma_{\sH\cap\{U\ge U(\ft_0)\},s}(T_\sH\phi,T_\sH\psi)+\sigma_{\sI\cap\{\ft\ge \ft_0\},s}(T_{\sI}{\phi},T_{\sI}{\psi})\\
    =&\sigma_{\sH\cap\{U\ge U(\ft_0)\},s}(T_\sH\phi,T_\sH\psi)+\sigma_{\sI\cap\{\ft\ge \ft_0\},s}(T_{\sI}\phi,T_{\sI}\psi)\\
    &+\sigma_{\sH\cap\{U\le U(\ft_0)\},s}(T_\sH\phi,T_\sH\psi)+\sigma_{\sI\cap\{\ft\le \ft_0\},s}(T_{\sI}{\phi},T_{\sI}{\psi})\\
    =& \sigma_{\sH,s}(T_\sH\phi,T_\sH\psi)+\sigma_{\sI,s}(T_{\sI}\phi,T_{\sI}\psi).
\end{align*}
\qeds
\end{proof}

\begin{definition}
\label{def:boundary spaces}
  We define $(\cS_s(\sH)\oplus \cS_s(\sI), \sigma_{\sH,s}\oplus \sigma_{\sI,s})$ to be the charged symplectic boundary space.  We define the physical boundary space as the subspace $\cS_{s,p}(\sH)\oplus\cS_{s,p}(\sI_-)$, where
\begin{subequations}
\begin{align}
    \cS_{s,p}(\sH):=&\left\{(\phi_s,\overline{\phi_{-s}})\in \cS_s(\sH):\overline{\phi_{-s}}\in\Ran(A_s)\text{ and }\overline{\phi_{-s}}=B_s\phi_s\right\}\, ,\\
    \cS_{s,p}(\sI_-):=&\left\{(\breve{\phi}_s,\overline{\breve{\phi}_{-s}})\in \cS_s(\sI_-): \overline{\breve{\phi}_{-s}}\in \Ran(A_s)\text{ and }\breve{\phi}_{s}={\tho}^{2s}A_s^{-1}\overline{\breve{\phi}_{-s}},\right\}\, .
\end{align}
\end{subequations}
\end{definition}

\begin{remark}
    At the horizon, we use the scaling of the Kruskal tetrad \eqref{eq:Kruskal tetrad} to identify $\cS_s(\sH)$ with a subset of $\cinf(\rr_U;\cB_s^{\ss^2}\otimes\cB_s^{\ss^2})$. Then the same computation used to derive \eqref{eq:nGamma Kruskal}, together with the identity $\kappa_+=\tfrac{r_+-M}{r_+^2+a^2}$, shows that ${\tho}'$ takes the form 
    \begin{align}
        {\tho}'\vert_{\sH}=\frac{r_+-M}{\varrho_+^2}\nabla_U
    \end{align}
    in the scaling given by the Kruskal tetrad when restricted to $\sH$, where $V=0$. It is therefore well-defined on $\cinf(\rr_U;\cB_s^{\ss^2})$. Moreover, in the same tetrad, a direct calculation shows that $[{\tho}', \zeta]=\cO(V)$, vanishing on $\sH$. Recalling that $\varrho^2=\abs{\zeta}^2$, this implies
    \begin{align}
    \overline{\phi_{-s}^U}=(r_+-M)^{2s}\partial_U^{2s}\phi_s^U\, 
\end{align}
for $(\phi_s,\overline{\phi_{-s}})\in \cS_{s,p}(\sH)$ expressed in the Kruskal tetrad.
    
    Similarly, in the Kinnersley tetrad, it is simple to see that $[A_s,x]=0$ and that, as discussed in Proposition~\ref{prop:A_s properties}, one can restrict $A_s$ to fixed $x$. Moreover, since $[\breve m,x]=[m,x]=0$ for $m$ in the Kinnersley tetrad, $A_s$ in the Kinnersley tetrad agrees with the version in the conformally rescaled tetrad. Hence, it can be considered as an operator on $\sI_-=\breve{\mathrm{M}}_{\mathrm{I}}\cap\{x=0\}$,  and it follows from Proposition~\ref{prop:A_s properties} that $[A_s, T_\sI]=0$. Additionally, working in the conformally rescaled Kinnersley tetrad, one finds
    \begin{align}
        \breve{\tho}_{(s,w)}=-x^2 \partial_x+\frac{2(1+a^2x^2)}{\Delta_x}\partial_{t^*}+\frac{2ax^2}{\Delta_x}\partial_{\varphi^*}+2wx\, ,
    \end{align}
    so that upon restriction to $\sI_-$, where $x=0$, one obtains a well-defined operator on $\cinf(\rr_{t^*};\cB_s^{\ss^2})$.
\end{remark}

\begin{remark}
   By restricting the trace operators to the physical phase space $\cS_{s,p}(\cM)$, and employing the relation \eqref{eq:Hertz-rel 1},  Lemma~\ref{lemma:A_s exchange}, and the fact that $[\breve{\tho},x]=-x^2$, one can embed the trace of the physical bulk phase space into the physical boundary space, i.e., $ T_\sH(\cS_{s,p}(\cM))\oplus T_{\sI}(\cS_{s,p}(\cM))\subset (\cS_{s,p}(\sH)\oplus \cS_{s,p}(\sI_-))$.
\end{remark}

\begin{lemma}
\label{lemma:injectivity proj bdy}
    The projections $\pi_\pm:\cS_{s,p}(\sH)\to \Gamma(\cB(s,\pm s)\vert_\sH)$, and $\pi_\pm:\cS_{s,p}(\sI_-)\to \Gamma(\breve{\cB}(s,\pm s)\vert_{\sI_-})$ onto the first (second) component are injective.
\end{lemma}

\begin{proof}
    Let us begin with the projections of $\cS_{s,p}(\sH)$. Injectivity of $\pi_+$ then follows directly from the definition. For the injectivity of $\pi_-$, consider the form of $B_s$ on the horizon in the  Kruskal tetrad. Considering it as an operator on $\Gamma( \cB(s,s)\vert_\sH)$, we find that a general solution to $B_s\phi_s=0$ takes the form 
    \begin{align*}
        \phi_s=\sum\limits_{k=0}^{2s-1}f_j(\theta,\varphi_+)U^k
    \end{align*}
    in the Kruskal tetrad. Comparing to the definition of $\cS_s(\sH)$, it is apparent that such $\phi_s$ can only belong to a $\phi\in \cS_{s,p}(\sH)$ if $\phi=0$. This implies the injectivity of $\pi_-$.

    Regarding the projections on $\sI_-$, it is clear from the definition that $\pi_-$ is injective. To show the injectivity of $\pi_+$, we note that if $\overline{\breve{\phi}_{-s}}$ is not the zero section and of the form $A_s\psi$, then $\psi=A_s^{-1}\overline{\breve{\phi}_{-s}}$ cannot be the zero section either. We then consider the operator $\breve{\tho}^{2s}$ on $ \Gamma(\breve{\cB}(s,-s)\vert_{\sI_-})$, in the conformally rescaled Kinnersley tetrad. In this tetrad, a general solution to ${\breve{\tho}}^{2s}\overline{\breve{\psi}_{-s}}=0$ takes the form
    \begin{align*}
        \overline{\breve{\psi}_{-s}}=\sum\limits_{k=0}^{2s-1}h_j(\theta,\varphi^*)\ft^k\, ,
    \end{align*}
so the only solution satisfying $(0,\overline{\breve{\psi}_{-s}})\in \cS_{s}(\sI_-)$ is $\overline{\breve{\psi}_{-s}}=0$, implying $\overline{\breve{\phi}_{-s}}=0$ , and hence the injectivity of $\pi_+$.
     \qeds
\end{proof}


\section{The CCR-algebra}
\label{sec:algebra}
\subsection{The CCR-algebra of a (charged) symplectic space}
In this part, we wish to construct the $*$-algebra of observables. In the present work, we choose to use the CCR-algebra of smeared fields instead of, e.g., the Weyl algebra.

\begin{definition}
Let $(V,\nu)$ be a (charged) symplectic space. Then we define the CCR-algebra $\cA(V,\nu)$ as $\cA_0(V,\nu)/\cR(V,\nu)$, where $\cA_0(V,\nu)$ is the unital $*$-algebra (over the field $\cc$) freely generated by abstract elements of the form $\Phi(v)$, $\Phi^*(v)$ with labels $v\in (V,\nu)$, and a unit element $\bbbone$. $\cR(V,\nu)$ is the ideal of $\cA_0(V,\nu)$ generated by the relations
\begin{itemize}
\item Linearity: $\Phi(\lambda v+w)-\overline{\lambda} \Phi(v)-\Phi(w)=0$, and  $\Phi^*(\lambda v+w)-\lambda \Phi^*(v)-\Phi^*(w)=0$ for all $v,w\in (V,\nu)$ and $\lambda\in \cc$
\item Commutation relation: $\left[\Phi(v),\Phi^*(w)\right]-\i\nu(v,w)\, \bbbone=0$,\\ $[\Phi(v),\Phi(w)]=[\Phi^*(v),\Phi^*(w)]=0$ for all $v,w\in (V,\nu)$.
\end{itemize}
The action of the $*$-involution is defined on the generators by $\bbbone^*=\bbbone$, $(\Phi(v))^*=\Phi^*(v)$, and $(\Phi^*(v))^*=\Phi(v)$ for all $v\in (V,\nu)$, and then continued such that it is anti-linear and satisfies $(ab)^*=b^*a^*$ for all $a, b\in \cA$.
\end{definition}

\begin{remark}
\label{rem:algbra induced map}
Let $U:(V,\nu)\to (W,\mu)$ be a linear map that conserves the (charged) symplectic form, i.e $\mu(Uv_1,Uv_2)=\nu(v_1,v_2)$  for all $v_1,\,v_2\in V$.
Then, by the weak non-degeneracy of $\nu$ and $\mu$, $U$ must be injective. Moreover, it induces an injective homomorphism of $*$-algebras, $\alpha_U:\cA(V,\nu)\to \cA(W,\mu)$ by setting
\begin{align}
\alpha_U(\Phi_V(v))=\Phi_W(Uv)\, ,\quad \alpha_U(\Phi^*_V(v))=\Phi^*_W(Uv)\, .
\end{align}
If $U$ is an isomorphism, the induced $*$-homomorphism is an isomorphism as well.
\end{remark}

\subsection{The CCR-algebra for the Teukolsky field}
In the present work, there are different choices for the symplectic space $(V,\nu)$. First of all, there are the two equivalent, enlarged phase spaces, $(TS_s(\cM),(\cdot,\cdot)_{\Delta_s})$ and $(\Sol_{s}(\cM),\sigma_s)$. As noted in Remark \ref{rem:algbra induced map}, the CCR-algebras constructed from these spaces will be equivalent. They will be denoted $\cA_s(\cM)$.

Additionally, consider the charged symplectic space $ (\cS_s(\sH)\oplus \cS_s(\sI_-),\sigma_{\sH,s}\oplus \sigma_{\sI,s})$ on the boundary. The algebra constructed from this space will be called the boundary algebra $\cA_{s,B}$.

\begin{remark}
By the results of Section \ref{sec:sympl form}, the trace map $T_{\sH}\oplus T_{\sI}: \Sol_{s}(\cM)\to \cS_s(\sH)\oplus \cS_s(\sI_-)$ preserves the charged symplectic form and is injective. As described in Remark \ref{rem:algbra induced map}, it induces the injective $*$-homomorphism $\alpha_{tr}:\cA_s(\cM)\to \cA_{s,B}$. 
\end{remark}

However, as we have discussed, these phase spaces are artificially doubled. Moreover, while they do have hermitian, non-degenerate forms, these are not positive.
\begin{definition}
 The physical subalgebra of $\cA_s(\cM)$, which will be called $\cA_{s,p}(\cM)$, is the one generated from the linear subspace $(\Sol_{s,p}(\cM),\sigma_s)\subset(\Sol_{s}(\cM),\sigma_s)$. 
 
Further, we will consider the subalgebra $\cA_{s,B,p}$ of $\cA_{s,B}$ generated from the physical subspace $(\cS_{s,p}(\sH)\oplus \cS_{s,p}(\sI_-), \sigma_{\sH}\oplus \sigma_{\sI})$ as the physical boundary algebra. 
\end{definition}

By a restriction of $T_{\sH}\oplus T_{\sI}$ to $\Sol_{s,p}(\cM)$, one can then embed the physical bulk algebra $\cA_{s,p}(\cM)$ into the physical boundary algebra $\cA_{s,B,p}$.

\begin{remark}
    Note that for $s=0$, the physical algebra $\cA_{0,p}(\cM)$ is isomorphic to the usual CCR-algebra of the free, massless scalar field.
\end{remark}

\begin{remark}
\label{rem:alg autom of killing flow}
Let $(\upsilon_b)_{b\in\rr}$ be a continuous one-parameter group of isomorphisms of $\cM$ induced by a Killing field of the form $\partial_t+c\partial_\varphi$ with $c\in\rr$. For every $b\in\rr$, set $\Upsilon_{(s,w),b}$ to be the cover of $\upsilon_b$ on $\cB(s,w)$, as described in Definition~\ref{def:cover of diffeos}. Let $\Upsilon_{(s,w),b}^*$ be the corresponding push-forward maps acting on $\Gamma(\cB(s,w))$. Finally, set for all $b\in \rr$
\begin{align}
  \mathfrak{Y}_b^*=\Upsilon_{(s,s),b}^*\oplus \Upsilon_{(s,-s),b}^*  \, .
\end{align}
 Then $(\mathfrak{Y}_b^*)_{b\in\rr}:\Gamma(\cV_s)\to \Gamma(\cV_s)$ is a group of push-forward maps satisfying
\begin{align}
\label{eq: commute with P_s}
    \mathfrak{Y}_b^*\circ \cP_s=\cP_s\circ \mathfrak{Y}_b^*\, ,
\end{align}
by Lemma~\ref{lemma:invar Teukolsky Killing flow} and the direct sum structure.

As a direct consequence, one has
\begin{align}
    \mathfrak{Y}_b^*\circ\Delta_s=\Delta_s\circ \mathfrak{Y}_b^*\vert_{\Gamma_c(\cV_s)}\, ,
\end{align}
and the maps $\mathfrak{Y}_b^*$ descend to maps 
\begin{align}
    \mathfrak{Y}_b^*:TS_s(\cM)\ni [f]\mapsto [\mathfrak{Y}_b^*f]\in TS_s(\cM)\, .
\end{align}
Hence, they induce a 1-parameter group of $*$-algebra homomorphisms $\alpha_b:\cA_s(\cM)\to \cA_s(\cM)$ by 
\begin{align}
    \alpha_b(\Phi([f]))=\Phi([\mathfrak{Y}_b^*f])\, , \quad \alpha_b(\Phi^*([f]))=\Phi^*([\mathfrak{Y}_b^* f])\, ,
\end{align}
and extended to the whole algebra by linearity. Combining Lemma~\ref{lemma:trafo of triv. flow} with Remark~\ref{rem:compl triv}, it is apparent that the $\alpha_b$ satisfy $\alpha_b\circ\alpha_c=\alpha_{b+c}$ for all $b$, $c\in\rr$.

In the same way, any (family of) endomorphisms of $\Gamma_c(\cV_s)$ satisfying \eqref{eq: commute with P_s} induce (a family of) $*$-algebra homomorphisms of $\cA_s(\cM)$. 
\end{remark}

\subsection{States on the CCR-algebra}
In the algebraic approach to quantum field theory, states are linear, positive, and normalized maps from the algebra of observables to $\cc$. In other words, a state is a linear map $\omega:\cA(V,\nu)\to \cc$ satisfying $\omega(\bbbone)=1$ and $\omega( a^*a)\geq 0$ for any $a\in \cA(V,\nu)$. Due to the structure of the algebra, any state is determined by its $n$-point functions
\begin{align*}
&w^{(n;k)}_\omega:V^{\otimes n}\to \cc\, ,\\
&w^{(n;k)}_\omega(v_1,\dots,v_k;v_{k+1},\dots,v_n):=\omega(\Phi(v_1)\dots\Phi(v_k)\Phi^*(v_{k+1})\dots\Phi^*(v_n))\, .
\end{align*} 
Here, we have taken into account the commutation relations of the algebra to order the arguments. $n$-point functions with other orders of the $\Phi$ and $\Phi^*$ can be obtained from these using the commutation relation. 
Note that the $n$-point function $w^{(n;k)}_\omega$ is anti-linear in the first $k$ arguments and linear in the last $n-k$ arguments.

A particularly simple but interesting class of states on the CCR-algebra $\cA(V,\nu)$ are the quasi-free states.
\begin{definition}
Let $\omega:\cA(V,\nu)\to \cc$ be a state on the CCR-algebra of $(V,\nu)$. Then it is {\it quasi-free} if $w^{(n;k)}_\omega=0$ whenever $n\neq 2k$ and 
\begin{align*}
w^{(2k;k)}_\omega(v_1,\dots,v_k;v_{k+1},\dots,v_{2k})=\sum\limits_{\pi\in S_k}\prod\limits_{i=1}^k w^{(2;1)}_\omega(v_i, v_{\pi(i)+k})\, ,
\end{align*}
where $S_n$ is the group of permutations of $n$ elements. 
\end{definition}
In other words, quasi-free states are completely determined by their two-point functions, see e.g. \cite{GW} or \cite[Proposition 17.14]{DG}.\footnote{Note that in the notation of \cite{DG}, the above defines a gauge invariant, quasi-free state.}

In turn, given a sesquilinear form $w^+:V\times V\to \cc$ with $\omega^+(v,v)\geq 0$ and $w^-(v,v):=\omega^+(v,v)-\i\nu(v,v)\geq 0$ for any $v\in V$, this map is the two-point function of a quasi-free state on $\cA(V;\nu)$. It completely defines the state.

\subsection{Hadamard states}
\label{sec:Hadamard intro}
Not all algebraic states correspond to physically acceptable states of the theory. A widely accepted criterion for identifying physical states is to demand that they satisfy the Hadamard condition. For the class of quasi-free states, the Hadamard condition can be formulated in terms of the wavefront set of the two-point function.

Let $B$ be a finite-dimensional (complex) vector bundle over a manifold $(M,g)$.  Recall that a distribution $u\in \cD'(M;B)$ is a linear continuous map $u: \Gamma_c( B^\#)\to \cc$, where $B^\#$ is the dual bundle. 
We then define the wavefront set as one would expect from the literature, see e.g. \cite{H:vol1}.
\begin{definition}
Let $u\in \cD'(M)$, and let $(x,k)\in \dot{T}^*M$. Let $X^\alpha$ be a coordinate system covering a neighbourhood of $x$. Then $(x,k)$ is a direction of rapid decrease if there is a smooth, compactly supported function $\chi\in \cinf_0(M)$ with $\chi(x)\neq 0$, and an open cone $V\subset T^*_xM$ containing $k$ so that 
\begin{align}
    \sup\limits_{l\in V}\left(1+\norm{l}^N\right)\abs{u\left(\chi e^{iX\cdot l}\right)}<\infty
\end{align}
for any $N\in \nn$. Here, $\norm{\cdot}$ is any norm on $T^*_xM$, and $X\cdot l=X^\alpha(x) l(\partial_{X^\alpha})(x)$.

The wavefront set $\WF(u)\subset \dot{T}^*M$ is the set of all $(x,k)\in \dot{T}^*M$ which are not directions of rapid decrease. 

Let now $u\in \cD'(M;B)$ for some finite-dimensional vector bundle $B$. Then one can define 
\begin{align}
    \WF(u):=\bigcup\limits_{\mathclap{S\in \Gamma(B^\#)}}\WF\left(f\mapsto u(fS)\right)\, .
\end{align}
\end{definition}
Above, $f\in \cinf_0(M)$, and it is sufficient to take the union over $S\in \Gamma(B^\#)$ forming a local basis of $B^\#$ \cite{F}.



Let us now consider the two-point function of a quasi-free state on the physical subalgebra $\cA_{s,p}(\cM)$. Then the two-point function is a sesquilinear form $w^+:\Sol_{s,p}(\cM)\times\Sol_{s,p}(\cM)\to\cc$. By Lemma~\ref{lemma:phys subspace bij}, this induces the distributions
\begin{subequations}
\label{eq:spacetime 2pt fct}
\begin{align}
    &W^+\in \cD'(\cM\times\cM;\cB(s,s)\boxtimes\cB(-s,s))\, ,\\\nonumber
    &W^+(f,h):=w^+\left(({\tho}^{2s}\overline{E_{-s}}(\overline{f}), A_s\overline{E_{-s}}(\overline{f})), ({\tho}^{2s}\overline{E_{-s}}(h), A_s\overline{E_{-s}}(h))\right)\, ,\\
   &W^-\in \cD'(\cM\times\cM;\cB(s,s)\boxtimes\cB(-s,s))\, ,\\\nonumber
    &W^-(f,h):=(w^+-i\sigma)\left(({\tho}^{2s}\overline{E_{-s}}(\overline{f}), A_s\overline{E_{-s}}(\overline{f})), ({\tho}^{2s}\overline{E_{-s}}(h), A_s\overline{E_{-s}}(h))\right)\, ,
\end{align}
\end{subequations}
where $\boxtimes$ denotes the exterior tensor product of vector bundles, $f\in \Gamma_c( \cB(-s,-s))$, and $h\in\Gamma_c( \cB(s,-s))$. We will refer to these distributions as the {\it spacetime two-point functions}.
It is straightforward to see that 
\begin{align}
\label{eq:2pt fct weak bisol}
    W^\pm(T_{-s}f, h)=W^\pm(f, \overline{T_{-s}}h)=0\, .
\end{align}

Following \cite{F}, we define

\begin{definition}
\label{def:Hadamard}
 The state $\omega$ on the algebra $\cA_{s,p}(\cM)$ satisfies the Hadamard condition if 
\begin{align}
\label{eq:Had cond}
    \WF'(W^\pm)\subset \cN^\pm\times\cN^\pm\, ,
\end{align}
where the primed wavefront set is defined as
\begin{align}
    \WF'(W^\pm)=\{(x,k;y,l)\in  \dot{T}^*\cM^2:(x,k;y,-l)\in \WF(W^\pm)\}\, ,
\end{align}
and 
\begin{align}
    \cN^\pm:=\{(x,k)\in  \dot{T}^*\cM: g^{-1}_x(k,k)=0, \pm g^{-1}(k,\fT)>0\}\, 
\end{align}
denote the future and past lightcones, where $\fT$ is the one-form representing the time orientation.   
\end{definition}

\section{The Unruh state}
\label{sec:Unruh}

\subsection{Definition of the Unruh state}

In this section, we will define the Unruh state on the physical subalgebra  $\cA_{s,p}(\cM)$. We do so by defining a state on $\cA_{s,B,p}$ and then using the pullback via the trace map $(T_{\sH}\oplus T_{\sI})\vert_{\Sol_{s,p}(\cM)}$. 

We begin on the horizon.
As a first step, recall that $\pi_\pm:\cS_{s,p}(\sH)\to \Gamma(\cB(s,\pm s)\vert_\sH)$ is the projection onto the first/second component.
\begin{definition}
    Let $\cU_s:\pi_+(\cS_{s,p}(\sH))\to \Gamma( \cB(s,0)\vert_\sH)$ be given by 
\begin{align}
    \cU_s(\phi_s)=\tfrac{\varrho_+^{2s}}{(r_+-M)^s} {\tho}'^s\phi_s\, .
\end{align} 
Further, let $\widetilde{\cU}_s$ be the operator defined on $\pi_{-}(\cS_{s,p}(\sH))$ by 
\begin{align}
   \widetilde{\cU}_s(\overline{\phi_{-s}})= (r_+-M)^{2s}\cU_s\pi_+\circ\pi_-^{-1}\overline{\phi_{-s}}\, ,
\end{align}
For $\phi,\phi'\in\cS_{s,p}(\sH)$, we define
\begin{align}
    w^+_\sH(\phi,\phi'):&=2\int\limits_{\mathclap{\rr\times\ss^2}}\left[\left\langle\cU_s\phi_s, 1_{\rr_+}(\i\Theta_U)(\i\Theta_U) \widetilde{\cU}_s\overline{\phi_{-s}'}\right\rangle_0\right.\\
    &\left.+\left\langle\widetilde{\cU}_s\overline{\phi_{-s}},\vphantom{\overline{\phi_{-2}'}} 1_{\rr_+}(\i\Theta_U)(\i\Theta_U) \cU_s\phi_s'\right\rangle_0\right] (r_+^2+a^2)\d U \d^2\omega_+\nonumber \\
&=4(r_+-M)^{2s}\int\limits_{\mathclap{\rr\times\ss^2}} \overline{D_U^s\phi_s^U} 1_{\rr_+}(D_U)D_U D_U^s\phi_{s}'^U (r_+^2+a^2) \d U \d^2\omega_+\nonumber
\end{align}
where in the last line we have used the definitions of $\cU_s$ and $\widetilde{\cU}_s$, as well as the trivialization given by the Kruskal tetrad \eqref{eq:Kruskal tetrad}, identifying $\phi_s, \phi_s'\in \pi_{+}\cS_{s,p}(\sH)$ with elements of $C^\infty(\rr_U; \Gamma(\cB_s^{\ss^2}))$. We have also introduced the notation $D_U=\i\partial_U$ and used the hermitian fibre metric $\IP{\cdot,\cdot}_0$ on $\cB(s,0)$ introduced in Corollary~\ref{cor:B(s,0) IP}.
\end{definition}

Next, we consider $\sI_-$.
\begin{definition}
     Let $\psi=(\breve{\psi}_s,\overline{\breve{\psi}_{-s}})\in \cS_{s,p}(\sI_-)$. Then we define the operator $\cW_s:\pi_{-}(\cS_{s,p}(\sI_-))\to \Gamma( \breve{\cB}(s,0)\vert_{\sI_-})$ by
\begin{align}
    \cW_s(\overline{\breve{\psi}_{-s}})=\breve{\tho}^s\overline{\breve{\psi}_{-s}}
\end{align}
and analogously we define $\widetilde{\cW}_s:\pi_+(\cS_{s,p}(\sI_-))\to \Gamma(\breve{\cB}(s,0)\vert_{\sI_-})$ by
\begin{align}
    \widetilde{\cW}_s(\breve{\psi}_s)= \breve{\tho}^sA_s^{-1}\pi_-\circ\pi_+^{-1}(\breve{\psi}_s)\,,
\end{align}
where $\pi_\pm:\cS_{s,p}(\sI_-)\to\Gamma( \breve{\cB}(s,\pm s)\vert_{\sI_-})$ is the projection onto the first/second component.
For $\psi,\psi'\in\cS_{s,p}(\sI_-)$, we set
\begin{align}
\label{eq:def wI}
    w^+_\sI(\psi,\psi'):=&2\int\limits_{\mathclap{\rr\times\ss^2}}\left[\left\langle\widetilde{\cW}_s\breve{\psi}_s, 1_{\rr_+}(\i\Theta_\ft)(\i\Theta_\ft) \cW_s\overline{\breve{\psi}_{-s}'}\right\rangle_0\right.\\
    &\left.+\left\langle\cW_s\overline{\breve{\psi}_{-s}},\vphantom{\overline{\phi_{-2}'}} 1_{\rr_+}(\i\Theta_\ft)(\i\Theta_\ft) \widetilde{\cW}_s\breve{\psi}_s'\right\rangle_0\right] \d\ft \d^2\omega^*\nonumber \\\nonumber
=&2^{2s+1}\int\limits_{\mathclap{\rr\times\ss^2}}\left[ \overline{D_\ft^sA_s^{-1}}\breve{\psi}_{-s} 1_{\rr_+}(D_\ft)D_\ft^{s+1} \overline{\breve{\psi}_{-s}}'\right.\\ \nonumber
&\left.+ \overline{D_\ft^s}\breve{\psi}_{s-}1_{\rr_+}(D_\ft)D_\ft^{s+1} A_s^{-1}\overline{\breve{\psi}_{-s}'}\right]\d\ft \d^2\omega^*\nonumber\, ,
\end{align}
where in the last line we have used the scaling of the conformally rescaled Kinnersley tetrad to identify $\breve{\psi}_{-s}$, $\breve{\psi}_{-s}'$ with elements of $\cinf(\rr_\ft;\Gamma(\cB_{-s}(\ss^2)))$.
\end{definition}

\begin{remark}
    Note that $\cU_0$ and $\widetilde{\cU}_0$ both simplify to the identity operator on $\cinf(\sH)$, while $\cW_0$ and $\widetilde{\cW}_0$ simplify to the identity operator on $\cinf(\sI_-)$.
\end{remark}

\begin{remark}
\label{remark:func calc}
    To define $1_{\rr_+}(\i\Theta_U)$,  we note that the hermitian fibre metric $\IP{\cdot,\cdot}_0$ in fact equips $\cB(s,0)$ with a fibre-wise inner product. We can thus define the positive-definite hermitian form 
    \begin{align*}
    \IP{f,h}_{\sH}=\int\limits_{\mathclap{\rr\times\ss^2}}\IP{f,h}_0(r_+^2+a^2)\d U\d^2\omega_+\, , \quad f,h\in \Gamma_{\sH}(\cB(s,0))\, ,
    \end{align*}
    and define $L^2_s(\sH)$ as the completion of $\Gamma_\sH(\cB(s,0))$ with respect to the norm induced by this pairing.

    It then follows along the lines of standard results 
    that $\i\Theta_U$ is self-adjoint on the domain $\{f\in L^2_s(\sH):\Theta_U f\in L^2_s(\sH)\}$.
    Since $1_{\rr_+}(x)\in L^\infty(\rr)$, the operator $1_{\rr_+}(\i\Theta_U)$ can then be defined as a bounded operator by functional calculus.
    
    The same argument applies to $1_{\rr_+}(\i\Theta_\ft)$: Define $L^2_s(\sI_-)$ as the completion of the space $\Gamma_{\sI_-}(\cB(s,0))$ with respect to the norm induced by 
    \begin{align*}
    \IP{f,h}_{\sI}=\int\limits_{\mathclap{\rr\times\ss^2}}\IP{f,h}_0\d\ft \d^2\omega^*\, ,
    \end{align*}
    and realize that $\i\Theta_\ft$ with domain $\{f\in L^2_s(\sI_-):\Theta_\ft f\in L^2_s(\sI_-)\}$
    is self-adjoint, so that $1_{\rr_+}(\i\Theta_\ft)$ can be defined by functional calculus.
\end{remark}

\begin{lemma}
\label{lemma:well def wH}
$w_{\sH}^+$ is well defined. 
\end{lemma}
\begin{proof}
To begin, note that only derivatives along $\sH$ are involved in $\cU_s$, making the operator well-defined on $\pi_+(\cS_{s,p}(\sH))\subset \Gamma( \cB(s,s)\vert_\sH)$.
In addition, by Lemma~\ref{lemma:injectivity proj bdy}, $\pi_-$ is injective and therefore bijective onto its rank, so that $\pi_-^{-1}$ is well-defined thereon. Moreover, all derivatives involved in $\widetilde{\cU}_s$ are tangential to $\sH$, making it well-defined as a map from $\pi_{-}(\cS_{s,p}(\sH))$ to $\Gamma(\cB(s,0)\vert_\sH)$. 

Let $\chi_{\pm}\ge 0,\, \chi_++\chi_-=1$ be a partition of unity with $\supp \chi_-\subset (-\infty,-1),$ and $\supp \chi_+\subset (-2,\infty)$. Let us also set $C_s=4(r_+-M)^{2s}(r_+^2+a^2)$. For $\phi, \phi'\in \cS_{s,p}(\sH)$, we write using the Kruskal tetrad
\begin{align*}
w_{\sH}^+(\phi,\phi')=C_s&\int\limits_{\mathclap{\rr\times \ss^2}}\left[1_{\rr_+}(D_U)D_U\widetilde{\phi}_s^{U}\chi_-\overline{\widetilde{\phi'}_s^{U}}
+1_{\rr_+}(D_U)D_U\widetilde{\phi}_s^{U}\chi_+\overline{\widetilde{\phi'}_s^{U}}\right]\d U\d^2\omega_+.
\end{align*} 
Here, we have introduced the notation $\widetilde{\phi}_s^{U}=D_U^s\phi_s^{U}$. 
Let us first consider the second integral. By  the definition of $\cS_{s,p}(\sH)$ in Definition~\ref{def:boundary spaces}, $\chi_+\widetilde{\phi'}_s^{U}$ is a smooth compactly supported function and thus in $L^2$. We then estimate
\begin{align*}
\norm{ 1_{\rr_+}(D_U)D_U\widetilde{\phi}_s^{U}}_{L^2}\le \norm{ D_U\widetilde{\phi}_s^{U}}_{L^2}\le \norm{ D_U \chi_+\widetilde{\phi}_s^{U}}_{L^2}+\norm{ D_U \chi_-\widetilde{\phi}_s^{U}}_{L^2}. 
\end{align*}
$\chi_+\widetilde{\phi}_s^{U}$ being compactly supported, the first term is finite. We then have 
\begin{align*}
D_U\chi_-\widetilde{\phi}_s^{U}=\i\chi_-'\widetilde{\phi}_s^{U}+\chi_-D_UD_U^sU^s\phi_s^{K}\, ,
\end{align*}
where the prime denotes a derivative with respect to $U$, and the $K$-superscript refers to the fact that we have switched to the Kinnersley tetrad \eqref{eq:Kinnersley}. 
The first term is smooth and compactly supported and thus in $L^2$. For the second term we use that 
\begin{align*}
D_U^2U^2=(UD_U)^2+3\i UD_U-2\, , \quad D_UU=UD_U+\i \, .
\end{align*}
Therefore, by the  definition of $\cS_{s,p}(\sH)$ in Definition~\ref{def:boundary spaces}, we see that 
\begin{align*}
\vert \chi_-D_U^{s+1}U^s\phi_s^{K}\vert\lesssim\langle U\rangle ^{-1}(\ln \langle U\rangle)^{-\delta}, \quad \delta>1\, . 
\end{align*}
As a consequence, we only have to treat 
\begin{align*}
&\int\limits_{\mathclap{\rr\times \ss^2}} (1_{\rr_+}(D_U)D_U\widetilde{\phi}_s^{U})\overline{\chi_-{\widetilde{\phi}}_s'^{U}}\d U\d^2\omega_+\\
=&\int\limits_{\mathclap{\rr\times \ss^2}} \left[(\chi_+\widetilde{\phi}_s^{U})\overline{(1_{\rr_+}(D_U)D_U\chi_-{\widetilde{\phi}}_s'^{U})} +(D_U1_{\rr_+}(D_U)\chi_- \widetilde{\phi}_s^{U})\overline{\chi_-{\widetilde{\phi}}_s'^{U}}\right]\d U\d^2\omega_+\, .
\end{align*}
By the same argument as above, the first term is finite (using also that $1_{\rr_+}(D_U)$ is a bounded operator). For the second term, we now use \cite[Lemma D.1]{GHW} and the change of variables $U=-e^{-\kappa_+ {}^*t}$ to obtain 
\begin{align*}
\int\limits_{\mathclap{\rr\times \ss^2}}\chi_{\beta}^+(D_{^*t})\chi_-D_U^sU^s\phi_s^{K}\overline{D_{{}^*t}\chi_-D_U^sU^s{\phi'}_s^{K}}\d{}^*t\d^2\omega_+, 
\end{align*}
where $\chi_{\beta}^+(x)=(1+e^{-\beta x})^{-1}$ with $\beta=\frac{2\pi}{\kappa_+}$.
It is now sufficient to use that $\chi_{\beta}^+(D_{^*t})$ is a bounded operator on $L^2(\rr\times \ss^2;\d{}^*t\d^2\omega_+)$ and that, by the definition of $\cS_{s,p}(\sH)$ in Definition~\ref{def:boundary spaces}, we have 
\begin{align*}
\chi_-D_U^sU^s\phi_s^{K},\,  D_{{}^*t}\chi_-D_U^sU^s\overline{\phi}_s'^{K}\in L^2(\rr\times \ss^2;\d{}^*t\d^2\omega_+). 
\end{align*}
\qeds
\end{proof}

\begin{lemma}
\label{lemma:well def wI}
    $w_\sI^+$ is well-defined.
\end{lemma}
\begin{proof}
Again, we begin by remarking that, by Lemma~\ref{lemma:injectivity proj bdy}, $\pi_+^{-1}$ is well-defined on $\pi_+\cS_{s,p}(\sI_-)$, and all differential operators involved are tangential to $\sI_-$, making the operators $\cW_s$ and $\widetilde{\cW}_s$ well-defined.

We will  fix the normalisation of the null vectors $\breve l$ and $\breve n$ to that of the conformally rescaled Kinnersley tetrad. 

By the definition of $\cS_{s,p}(\sI_-)$, for $\breve\psi\in \cS_{s,p}(\sI_-)$, both $\cW_s\overline{\breve\psi_{-s}}$ and $\widetilde{\cW}_s\breve\psi_s$ decay sufficiently fast for $\ft\to \pm \infty$ so that they are in $L^2_s(\sI_-)$, and the same is true for $\Theta_\ft \cW_s\overline{\breve\psi_{-s}}$ and  $\Theta_\ft\widetilde{\cW}_s\breve\psi_s$. As mentioned before, $1_{\rr_+}(\i\Theta_\ft)$ can be defined on $L^2_s(\sI_-)$ as a bounded operator, so that we obtain
\begin{align*}
    \abs{\omega_\sI^+(\breve\psi,\breve\psi')}=&2\left\vert\,\int\limits_{\rr\times\ss^2}\left[\left\langle\cW_s\breve{\psi}_s, 1_{\rr_+}(\i\Theta_\ft)(\i\Theta_\ft) \cW'_s\overline{\breve{\psi}_{-s}'}\right\rangle_0\right.\right.\\
    &\left.\left.\hphantom{abs}+\left\langle\cW'_s\overline{\breve{\psi}_{-s}},\vphantom{\overline{\phi_{-2}'}} 1_{\rr_+}(\i\Theta_\ft)(\i\Theta_\ft) \cW_s\breve{\psi}_s'\right\rangle_0\right] \d\ft \d^2\omega^*\right\vert.
   \end{align*}
This is bounded by the Cauchy-Schwartz inequality. 
\qeds
\end{proof}

We can now define the state $\omega_U$:
\begin{definition}
\label{def:U state}
Let $\phi,\, \psi\in \Sol_{s,p}(\cM)$. Then the Unruh state $\omega_U$ on $\cA_{s,p}(\cM)$ is defined as the quasi-free state determined by the two-point function
\begin{align*}
w^+_U(\phi,\psi):=w^+_\sH(T_\sH\phi,T_\sH\psi)+w^+_\sI(T_\sI\phi,T_\sI\psi)\, .
\end{align*}
\end{definition}


\begin{lemma}
    $w^+_U$ is well-defined.
\end{lemma}
\begin{proof}
For $\phi \in \Sol_{s,p}(\cM)$,  one has $T_\sH\phi\in \cS_{s,p}(\sH)$ and $T_\sI\phi\in \cS_{s,p}(\sI_-)$  as discussed in Section~\ref{sec:sympl form}. Therefore, the statement follows directly from Lemma~\ref{lemma:well def wH} and Lemma~\ref{lemma:well def wI}.
\qeds
\end{proof}

 \begin{remark}
\label{rem:2pt distr estimates}
 We note that the spacetime two-point function 
 \begin{align}
   W^+_U(f,h)=w^+_U\left(({\tho}^{2s}\overline{E_{-s}}(\overline{f}), A_s\overline{E_{-s}}(\overline{f})), ({\tho}^{2s}\overline{E_{-s}}(h), A_s\overline{E_{-s}}(h))\right)\, ,
 \end{align}
 for $f\in \Gamma_c(\cB(s,s))$, $h\in \Gamma_c(\cB(s,-s))$, is indeed a bi-distributions, i.e., it satisfies estimates
 \begin{align}
     \abs{W^+_U(f,h)}\leq C \norm{f}_{C^m}\norm{h}_{C^m}
 \end{align}
 for some $m\in \nn$ and some constant $C>0$ that depends on the supports of $f$ and $h$. This follows from the estimates in the proofs of Lemma~\ref{lemma:well def wH} and Lemma~\ref{lemma:well def wI}, together with the estimates in  Corollary~\ref{cor:4.1} and Theorem~\ref{th:4.2}.
\end{remark}

\subsection{Positivity}
It remains to show the positivity of the two-point function $w^+_U$.
\begin{theorem}
Let $\phi=(\phi_s, \overline{\phi_{-s}})\in \Sol_{s,p}(\cM)$. Then 
\begin{subequations}
\begin{align}
w_U^+(\phi,\phi)\geq 0\, ,\\
w_U^-(\phi,\phi):=w_U^+(\phi,\phi)-\i\sigma(\phi,\phi)\geq 0\, ,
\end{align}
\end{subequations}
and hence $w_U^+$ induces a well-defined quasi-free state $\omega_U$ on $\cA_{s,p}(\cM)$.
\end{theorem}

\begin{proof}
The proof proceeds in multiple steps:\\
(1) Recalling the definition of $w_U^+$ and Theorem~\ref{theorem:cons sympl form}, we can write for $\phi$, $\psi\in  \Sol_{s,p}( \cM)$
\begin{align}
    w_U^-(\phi, \psi)=&w_\sH^+(T_\sH\phi,T_\sH\psi)+w_\sI^+(T_\sI\phi,T_\sI\psi)\\\nonumber 
    &-\i\left[ \sigma_\sH (T_\sH \phi,T_\sH \psi) +\sigma_\sI (T_\sI \phi , T_\sI \psi ) \right]\\\nonumber
    =&w_\sH^-(T_\sH\phi,T_\sH\psi)+w_\sI^-(T_\sI\phi,T_\sI\psi)\, ,
\end{align}
where we use the natural definition
\begin{subequations}
\begin{align}
   w_\sH^-(T_\sH\phi,T_\sH\psi)&=w_\sH^+(T_\sH\phi,T_\sH\psi)-\i\sigma_\sH(T_\sH\phi,T_\sH\psi)\\
    w_\sH^-(T_\sI\phi,T_\sI\psi)&=w_\sI^+(T_\sI\phi,T_\sI\psi)-\i\sigma_\sI(T_\sI\phi,T_\sI\psi)\, .
\end{align}
\end{subequations}
\\
(2) Let us start with $w^+_{\sH}(T_\sH\phi, T_\sH\phi)$. Then writing $(T_\sH\phi)=(\phi_s,\overline{\phi_{-s}})$, and using the definition of $\widetilde{\cU}_s$, we can write
\begin{align*}
w^+_\sH(T_\sH\phi,T_\sH\phi):&=4(r_+-M)^{2s}\left\langle\cU_s\phi_s, 1_{\rr_+}(\i\Theta_U)(\i\Theta_U)\cU_s\phi_{s}\right\rangle_\sH\, .
\end{align*}
Here, we have employed the notation from Remark~\ref{remark:func calc}. From this form, it is directly visible that this expression is positive, since the operator $1_{\rr_+}(\i\Theta_U)(\i\Theta_U)$ acting on $\Dom(\Theta_U)\subset L^2_s(\sH)$ is a positive operator. 
\\
(3)
In the next step, we use the identity $\id=1_{\rr_+}(\i\Theta_U)+1_{\rr_-}(\i\Theta_U)$ following from functional calculus on $L^2_s(\sH)$ and the definition for $\widetilde{\cU}_s$ to express
\begin{align}
    w^-_\sH(T_\sH\phi,T_\sH\phi')=&2^{2}(r_+-M)^{2s}\left\langle\cU_s\phi_s, -1_{\rr_-}(\i\Theta_U)(\i\Theta_U) \cU_s\phi_{s}'\right\rangle_\sH\\\nonumber
    &+2\int\limits_{\mathclap{\rr\times\ss^2}} \left[ \overline{\cU_s\phi_s} (\i\Theta_U) \widetilde{\cU}_s\overline{\phi_{-s}'} +\overline{\widetilde{\cU}_s}\phi_{-s} (\i\Theta_U)\cU_s\phi_s'\right.\\\nonumber
    &\left.-(-1)^s(\overline{\phi_s}(\i\Theta_U)\overline{\phi_{-s}'}+\phi_{-s}(\i\Theta_U)\phi_s')\right](r_+^2+a^2)\d U\d^2\omega_+\, .
\end{align}
The well-definedness of this expression follows from the well-definedness of its two parts.

Comparing also to the previous part of the proof, it is apparent that the first part is positive if $\phi_s=\phi_s'$, since $-1_{\rr_-}(\i\Theta_U)(\i\Theta_U)$ is a positive operator on the domain of $\Theta_U$ in $L^2_s(\sH)$.

We thus focus on the second part and show that it vanishes.
Expressing the second part in the Kruskal tetrad, and utilizing Definition~\ref{def:boundary spaces} of $\cS_{s,p}(\sH)$, it reads
\begin{align*}
    c_s  \int\limits_{\mathclap{\rr\times\ss^2}} \left[ 2 \partial_U^s \overline{\phi_s^U} \partial_U^{s+1} \phi_s'^U -(-1)^s \left( \overline{\phi_s^U} \partial_U^{2s+1} \phi_s'^U +\partial_U^{2s} \overline{\phi_s^U} \partial_U \phi_s'^U \right) \right] \d U \d^2\omega_+\, ,
\end{align*}
where $c_s=2\i(r_+-M)^{2s}(r_+^2+a^2)$.
To show that this vanishes, we will use partial integration. After one partial integration, we obtain
\begin{align*}
   c_s  \int\limits_{\mathclap{\rr\times\ss^2}} \left[ 2 \partial_U^s \overline{\phi_s^U} \partial_U^{s+1} \phi_s'^U -(-1)^{s-1} \left( \partial_U\overline{\phi_s^U} \partial_U^{2s} \phi_s'^U +\partial_U^{2s-1} \overline{\phi_s^U} \partial_U^2 \phi_s'^U \right) \right] \d U \d^2\omega_+\\
   -(-1)^sc_s\int\limits_{\ss^2}\left[\overline{\phi_s^U} \partial_U^{2s}\phi_s'^U +\partial_U^{2s-1} \overline{\phi_s^U} \partial_U \phi_s'^U\right]_{-\infty}^{\infty}\d^2\omega_+\, .
\end{align*}
By the definition of $\cS_{s,p}(\sH)$, we can estimate
\begin{align*}
    \abs{\overline{\phi_s^U}\partial_U^{2s}\phi_s'^U}, \abs{\partial_U^{2s-1}\overline{\phi_s^U}\partial_U\phi_s'^U}\leq \abs{\log\IP{U}}^{-d}\, 
\end{align*}
for some $d>1$,  which vanishes for $U\to -\infty$. For $U\to \infty$, both terms actually vanish to arbitrary order due to the support properties of functions in $\cS_{s,p}(\sH)$. Therefore, the boundary term vanishes. Repeating the same process another $s-1$ times, we arrive at
\begin{align*}
    c_s  \int\limits_{\mathclap{\rr\times\ss^2}} \left[ 2 \partial_U^s \overline{\phi_s^U} \partial_U^{s+1} \phi_s'^U - \left( \partial_U^s\overline{\phi_s^U} \partial_U^{s+1} \phi_s'^U +\partial_U^{s} \overline{\phi_s^U} \partial_U^{s+1} \phi_s'^U \right) \right] \d U \d^2\omega_+=0\, ,
\end{align*}
as desired. 
\\
(4)
Next, consider $w_\sI^+(T_\sI \phi,T_\sI\phi)$. We utilize the notation $T_\sI\phi=(\breve \phi_s,\overline{\breve \phi_{-s}})$. Taking into account the definition of $\widetilde{\cW}_s$, it can be written as
\begin{align*}
w^+_\sI(T_\sI\phi,T_\sI\phi) =&2  \left[ \left\langle \cW_s A_s^{-1}\overline{\breve{\phi}_{-s}}, 1_{\rr_+}(\i\Theta_\ft)(\i\Theta_\ft) \cW_s\overline{\breve{\phi}_{-s}}\right\rangle_\sI\right.\\
    &\left.+\left\langle\cW_s\overline{\breve{\phi}_{-s}}, 1_{\rr_+}(\i\Theta_\ft) (\i\Theta_\ft) \cW_sA_s^{-1}\overline{\breve{\phi}_{-s}}\right\rangle_\sI\right]\, .
\end{align*}
In the following, we work in the conformally rescaled Kinnersley tetrad. This allows us to consider $\overline{\breve{\phi}_{-s}}$ as elements of $\cinf(\rr_\ft;\Gamma(\cB_s^{\ss^2}))$. In that tetrad, we have $\cW_s=2^s\partial_\ft^s$, and $A_s^{-1}$ commutes with $\cW_s$, so that one can write 
\begin{align*}
   \left\langle \cW_s A_s^{-1}\overline{\breve{\phi}_{-s}}, 1_{\rr_+}(\i\Theta_\ft)(\i\Theta_\ft) \cW_s\overline{\breve{\phi}_{-s}}\right\rangle_\sI=2^{2s}\int\limits_{\mathclap{\rr\times\ss^2}} \overline{ A_s^{-1}} \partial_\ft^s\breve{\phi}_{-s}^K 1_{\rr_+}(D_\ft)D_\ft \partial_\ft^s\overline{\breve{\phi}_{-s}^K} \d\ft \d^2\omega^*\, ,\\
   \left\langle\cW_s\overline{\breve{\phi}_{-s}}, 1_{\rr_+}(\i\Theta_\ft) (\i\Theta_\ft) \cW_sA_s^{-1}\overline{\breve{\phi}_{-s}}\right\rangle_\sI=2^{2s}\int\limits_{\mathclap{\rr\times\ss^2}}\partial_\ft^s\breve{\phi}_{-s}^K 1_{\rr_+}(D_\ft)D_\ft A_s^{-1}\partial_\ft^s\overline{\breve{\phi}_{-s}^K} \d\ft \d^2\omega^*\, .
\end{align*}
The positivity of the second term then follows directly from the fact that both $A_s^{-1}$ and  $1_{\rr_+}(D_\ft)D_\ft$ are positive operators on $L^2(\rr_\ft\times \ss^2)$, the intersection of their domains including $\partial_\ft^s \overline{\breve{\phi}_{-s}}$. In addition, the two operators commute, so that their product is again a positive operator.
For the positivity of the first term, we use that the operator $D_\ft$ is  self-adjoint on $\{\phi\in L^2_s(\sI_-):\partial_\ft\in L^2_s(\sI_-)\}$, and so is the real function $1_{\rr_+}(D_\ft)D_\ft$ thereof. We can thus transfer this operator to the other side to arrive at the same situation as for the first term. The positivity of the product follows again from the commutativity of $1_{\rr_+}(D_\ft)D_\ft$ with $A_s^{-1}$.
\\
(5) Similar to step (3), we use functional calculus to obtain 
\begin{align}
\label{eq:wI_minus}
    w^-_{\sI}(\psi,\psi')=&2\IP{ \cW_s A_s^{-1} \overline{\breve{\psi}_{-s}} ,(-1_{\rr_-}(\i\Theta_\ft)) (\i \Theta_\ft) \cW_s \overline{\breve{\psi}_{-s}'} }_\sI\\\nonumber
    &+2\IP{\cW_s \overline{\breve{\psi}_{-s}} ,(-1_{\rr_-}(\i\Theta_\ft)) (\i\Theta_\ft) \cW_s  A_s^{-1} \overline{\breve{\psi}_{-s}'} }_\sI\\\nonumber
    &+2 \int\limits_{\mathclap{\rr\times\ss^2}} \left[\overline{{\tho}^s A_s^{-1}}\breve{\psi}_{-s}(\i\Theta_\ft){\tho}^s\overline{\breve{\psi}_{-s}'}+\overline{{\tho}^s}\breve{\psi}_{-s}(\i\Theta_\ft){\tho}^sA_s^{-1}\overline{\breve{\psi}_{-s}'}\right.\\\nonumber
   &\left.-(-1)^s\left(\overline{{\tho}^{2s}A_s^{-1}}\breve{\psi}_{-s} (\i\Theta_\ft)\overline{\breve{\psi}_{-s}'}+\breve{\psi}_{-s}(\i\Theta_\ft){\tho}^{2s}A_s^{-1}\overline{\breve{\psi}_{-s}'}\right)\right]\d\ft \d^2\omega^*\, .
\end{align}
The positivity of the first two terms follows from essentially the same reasoning as in step (4), where the positivity of $\omega_\sI^+(T_\sI\phi, T_\sI\phi)$ was shown. Similar to step (3), it remains to show that the third and last row in \eqref{eq:wI_minus} cancel.

To do so, we fix the conformally rescaled Kinnersley tetrad, in which the last two terms in \eqref{eq:wI_minus} can be rewritten as
\begin{align*}
    &2^{1+2s}\i \int\limits_{\mathclap{\rr\times\ss^2}} \left[\partial_\ft^s\overline{A_s^{-1}}\breve{\psi}_{-s}^K\partial_\ft^{s+1}\overline{\breve{\psi}_{-s}'^K}+\partial_\ft^s\breve{\psi}_{-s}^K\partial_\ft^{s+1}A_s^{-1}\overline{\breve{\psi}_{-s}'^K}\right.\\\nonumber
   &\left.-(-1)^s\left(\partial_\ft^{2s}\overline{A_s^{-1}}\breve{\psi}_{-s}^K\partial_\ft\overline{\breve{\psi}_{-s}'^K}+\breve{\psi}_{-s}^K\partial_\ft^{2s}A_s^{-1}\overline{\breve{\psi}_{-s}'^K}\right)\right]\d\ft \d^2\omega^*\, .
\end{align*}
Using the decay of elements in $\cS_{s,p}(\sI_-)$ and partial integration as in part (3) of the proof, one finds that the two parts cancel.
\qeds
\end{proof}

\section{The Hadamard property of the Unruh state}
\label{sec:Hadamard}
In this section, we will show
\begin{theorem}
\label{thm:Had}
    The Unruh state $\omega_U$ on $\cA_{s,p}(\cM)$ defined by the two-point function $w^+_U$ is Hadamard in the sense of Definition~\ref{def:Hadamard}.
\end{theorem}

\begin{remark}
    The results of our previous section indicate that the Unruh state could in principle be defined on the whole Kruskal extension $\MK_K$. This would require replacing all arguments based on finite speed of propagation with arguments based on the decay results in Corollary~\ref{cor:4.1} and Theorem~\ref{th:4.2}, as well as adding an additional term  to the two-point function which is analogous to $w^+_{\sI}$, but defined on $\sI_+'$, the conformal future null boundary of $\MI'=\{U>0,V<0\}\subset \MK_K$. The reason we do not discuss this extension is that the Hadamard property of the Unruh state is expected to break down at $\sH$, rendering the extension unphysical. In our proof of the Hadamard property, this manifests in the invalidity of Lemma~\ref{lemma:time ori horizon}, which guarantees that for a point $(y,\xi)\in (\cN^+\cup\cN^-)\vert_\sH$ belonging to any bicharacteristic transversal to the horizon, the condition $k(\partial_U)>0$ implies that $(y,\xi)\in \cN^+$, see also \cite[Remark 8.4]{GHW}. Further indication is given by the uniqueness result of Kay and Wald \cite{KW}, though there are some subtleties related to the algebra on which the state should be considered. Moreover, the work by Kay and Wald does not state where exactly the Hadamard property or the definition of the state would break down. Finally, there is some numerical evidence for the case $s=0$ in the form of the stress-energy tensor computed in the Unruh state for scalar fields on Kerr \cite{ZCLOO}. The results found there, together with the stationarity of the state, imply upon coordinate transformation that the expectation value of the stress-energy tensor of the scalar field in the Unruh state diverges on $\sH_-$. This indicates a breakdown of the Hadamard property.
\end{remark}

To begin the proof of Theorem~\ref{thm:Had}, recall that the two-point functions $w^\pm_U$ defined in Definition~\ref{def:U state} induce spacetime two-point functions $W^\pm_U$ via \eqref{eq:spacetime 2pt fct}. They can be written as
\begin{align}
    W^\pm_{ U}=&W^\pm_\sH+W^\pm_\sI\, ,
\end{align}
where $W^\pm_{\sH}$ and $W^\pm_{\sI}$ are the spacetime two-point functions obtained via \eqref{eq:spacetime 2pt fct} from the pullbacks of $w^\pm_{\sH}$ and $w^\pm_\sI$, respectively.

\subsection{Restriction to the diagonal}
 In this section, we show that it is sufficient to study the primed wavefront set of $W^\pm_U$ over the diagonal in $T^*\cM\times T^*\cM$. The idea is based on \cite[Proposition 6.2]{SVW}, see also \cite{GHW, Klein}.

We recall that $W^\pm_{ U}$ is by construction a weak solution to a Teukolsky equation in both entries, see \eqref{eq:2pt fct weak bisol}. Since $T_{-s}$ is a normally hyperbolic operator, the wavefront set of $W^\pm_{ U}$ must be contained in the lightcone $(\cN\cup\{0\})\times(\cN\cup\{0\})$ by Propagation of Singularities Theorem as in \cite[Lemma 6.5.5]{DH}. Here, $\cN=\cN^+\cup\cN^-$, and $\{0\}$ denotes the zero section of $T^*\cM$.\footnote{Note that the zero section is not contained in $\cN$ according to our definition of the lightcone.}

Moreover, by the same result, if $(x,k;y,l)\in \cN\times\cN$ is in $\WF(W^\pm_{ U})$, then the same must be true for all $(x',k';y',l')\in (\cN\cup\{0\})\times(\cN\cup\{0\})$ with $(x,k)\sim (x',k')$ and $(y,l)\sim (y',l')$, where two points in $(\cN\cup\{0\})$ are related by $\sim$ if they belong to the same bicharacteristic strip of the Teukolsky operator, or in other words the same null geodesic lifted to $T^*\cM$. \footnote{Note that here, we refer to the bicharacteristics of $T_s$ in $T^*\cM$ rather than those of $\hat{T}_s(\sigma)$ in $T^*X$ discussed in Section~\ref{subsec:semiflow}.}

Let $h\in\Gamma_c(\cB(s,-s))$. Then we define the notation
\begin{align}
\label{eq:notation soln of test fct}
    \phi_h=({\tho}^{2s}\overline{E_{-s}} h, A_s\overline{E_{-s}}h)
\end{align}
for the corresponding element of $\Sol_{s,p}(\cM)$, using the bijection from Lemma~\ref{lemma:phys subspace bij}.
Next, we note that 
\begin{lemma}
   Let $f\in \Gamma_c(\cB(-s,-s))$,  and $h\in\Gamma_c(\cB(s,-s))$, and let $\phi_{\bar f}$ and $\phi_h$  be the corresponding elements of $\Sol_{s,p}(\cM)$. Then 
   \begin{align}
   \label{eq:sympl on red space}
       \sigma (\phi_{\bar f}, \phi_h) =(-1)^s \int\limits_\cM f \left( \varrho^2(\overline{A_s})^t\varrho^{-2}{\tho}^{2s}\overline{E_{-s}}-E_s({\tho}^{2s})^t \varrho^{-2} A_s\varrho^2\right) h\dVol_g\, ,
   \end{align}
   where $({\tho}^{2s})^t: \Gamma(\cB(s,-s))\to \Gamma( \cB(s,s))$ and $(\overline{A_s})^t: \Gamma(\cB(s,s))\to\Gamma(\cB(s,s))$ are the differential operators that are formally dual to ${\tho}^{2s}$ and $\overline{A_s}$, i.e., 
   \begin{align*}
       \int\limits_\cM {\tho}^{2s}fh\dVol_g=\int\limits_\cM f({\tho}^{2s})^th\dVol_g\, , \quad
        \int\limits_\cM \overline{A_s}fk\dVol_g=\int\limits_\cM f(\overline{A_s})^tk\dVol_g
   \end{align*}
   for all $f\in \Gamma(\cB(-s,-s))$, $h\in \Gamma(\cB(s,-s))$, and $k\in \Gamma(\cB(s,s))$ so that $\supp f\cap\supp h$ or respectively $\supp f \cap \supp k$ is compact.
   \end{lemma}

\begin{proof}
    First, we note that for $f\in \Gamma_c(\cB(-s,-s))$ and $h\in \Gamma_c(\cB(s,s))$, it follows from the computations in the proof of Lemma~\ref{lemma: self dual} that 
    \begin{align}
        \int\limits_\cM f(x) T_s h(x)\dVol_g=\int\limits_\cM T_{-s}f(x) h(x) \dVol_g\, ,
    \end{align}
   in any trivialization, and consequently
    \begin{align}
    \label{eq:E_s duality}
        \int\limits_\cM f(x) E_s h(x)\dVol_g=-\int\limits_\cM E_{-s}f(x) h(x) \dVol_g\, .
    \end{align}

    Let us now consider $\sigma(\phi_{\bar f},\phi_h)$, with $f$ and $h$ as in the statement of the lemma.
    We note that making use of the symplectomorphism in Proposition~\ref{prop:iso phase space}, and in particular the inverse of $\Delta_s$ constructed in the proof of Proposition~\ref{prop:Green hyp props}, we can write
    \begin{align*}
        \sigma(\phi_{\bar f}, \phi_h)=([(T_s\chi {\tho}^{2s}\overline{E_{-s}}\bar f, \overline{T_{-s}}\chi A_s\overline{E_{-s}}\bar f)], [(T_s\chi {\tho}^{2s}\overline{E_{-s}}h, \overline{T_{-s}}\chi A_s\overline{E_{-s}}h)])_{\Delta_s}\, .
    \end{align*}
    Here, $\chi\in \cinf(\cM;\rr)$ is as in the proof of Proposition~\ref{prop:Green hyp props}, i.e., for two Cauchy surfaces $\Sigma_\pm$, so that $\Sigma_+\subset I^+(\Sigma_-)$, $\chi$ satisfies $\chi\vert_{J^-(\Sigma_-)}=1$ and $\chi\vert_{J^+(\Sigma_+)}=0$.
    Using some fixed trivialization, we can then rewrite this further as
    \begin{align*}
    &([(T_s\chi {\tho}^{2s}\overline{E_{-s}}\bar f, \overline{T_{-s}}\chi A_s\overline{E_{-s}}\bar f)], [(T_s\chi {\tho}^{2s}\overline{E_{-s}}h, \overline{T_{-s}}\chi A_s\overline{E_{-s}}h)])_{\Delta_s}\\
    &= (-1)^s\int\limits_\cM \left( \overline{T_s}\chi{\tho}^{2s}E_{-s}f \overline{E_{-s}}\overline{T_{-s}}\chi A_s\overline{E_{-s}}h + T_{-s}\chi \overline{A_s}E_{-s}f E_sT_s\chi{\tho}^{2s}\overline{E_{s}}h \right) \dVol_g  \\
      &=(-1)^s\int\limits_\cM \left( \overline{T_s}\chi{\tho}^{2s}E_{-s}f A_s\overline{E_{-s}}h-\overline{A_s}E_{-s}f T_s\chi {\tho}^{2s}\overline{E_{-s}}h\right) \dVol_g\\
      &=(-1)^s\int\limits_\cM \left( -{\tho}^{2s}E_{-s}f \varrho^{-2}A_s\varrho^2 h+\varrho^{-2}\overline{A_s}\varrho^2f {\tho}^{2s}\overline{E_{-s}}h\right) \dVol_g\, .
    \end{align*}
    In the first step, we have made use of \eqref{eq:E_s duality} and the fact that $A_s\overline{E_{-s}}:\Gamma_c(\cB(s,-s))\to \cT_{-s,sc}(\cM)$, so that $\overline{E_{-s}}\overline{T_{-s}}\chi A_s\overline{E_{-s}}h= A_s\overline{E_{-s}}h$ for all $h\in \Gamma_c(\cB(s,-s))$,
     see the proof of Proposition~\ref{prop:Green hyp props}. In the second step, we have used part 2 of Proposition~\ref{prop:A_s properties} to exchange $\overline{E_{-s}}$ and $A_s$, combined with \eqref{eq:E_s duality} and the fact that ${\tho}^{2s}\overline{E_{-s}}:\Gamma_c(\cB(s,-s))\to \cT_{s,sc}(\cM)$, so that we obtain a similar reduction as before.     
     
     As a last step, by the definition of the formal dual, and using again \eqref{eq:E_s duality}, we obtain \eqref{eq:sympl on red space}.
    \qeds
\end{proof}


From the foregoing lemma, it is apparent that for $f\in\Gamma_c(\cB(-s,-s))$ and $h\in\Gamma_c(\cB(s,-s))$ with spacelike separated supports, one has $\sigma(\phi_{\bar f},\phi_h)=0$.
Moreover, it follows from Propagation of Singularities (see \cite[Theorem 6.5.3]{DH}) and the fact that differential operators with smooth coefficients do not increase the wavefront set that 
\begin{align}
    \WF(\varrho^2(\overline{A_s})^t{\tho}^{2s}\overline{E_{-s}}-E_s({\tho}^{2s})^t \varrho^{-2} A_s\varrho^2)\subset \cN\times\cN\, .
\end{align}
Thus, following the same argument as in the proof of \cite[Lemma V.1]{Klein},  see also \cite{SVW}, it follows from this property together with the positivity of the state 

\begin{corollary}
\label{cor:Had reduction}
    The spacetime two-point functions satisfy the Hadamard condition \eqref{eq:Had cond} if and only if \eqref{eq:Had cond} holds on the diagonal in $T^*\cM\times T^*\cM$.
    \end{corollary}

\subsection{The backwards trapped set} 
After this preparation, the next step is to check the validity of the Hadamard property for each bicharacteristic of the Teukolsky operator $T_s$. Projecting the bicharacteristics to $\cM$, one obtains the null geodesics $\gamma:(\tau_-,\tau_+)\to \cM$. We can separate this into three different cases, making use of the analysis in \cite{GHW}: past-trapped null geodesics, null geodesics approaching $\sH$ towards the past, and null geodesic approaching $\sI_-$ towards the past.
In this section, we begin with the first case, the past-trapped null geodesics. 
\begin{definition}
    Let $\gamma:(\tau_-,\tau_+)\to \cM$ be a maximally extended null geodesic on $\cM$. Then we say that $\gamma$ is past-trapped  if there exist a $t_0\in \rr$ and $r_+<r_0<R_0<\infty$, so that $r_0\leq r(\gamma(\tau))\leq R_0$ for all $\tau$ for which $\gamma(\tau)\in \MI$ and $t(\gamma(\tau))\leq t_0$. The set of all past-trapped null geodesics will be called $\Gamma^-$.
    
    Similarly, $\gamma$ will be called future-trapped if there are $t_0$ and $r_+<r_0<R_0<\infty$, so that $r_0\leq r(\gamma(\tau))\leq R_0$ for all $\tau$ for which $\gamma(\tau)\in \MI$ and $t(\gamma(\tau))\geq t_0$. The set of all future-trapped null geodesics will be called $\Gamma^+$. We set $\Gamma=\Gamma^+\cup\Gamma^-$. The intersection $K=\Gamma^+\cap\Gamma^-$ is the trapped set.
\end{definition}

\begin{remark}
    We note that, in contrast to the discussion in Section~\ref{subsec:semiflow}, future and past refer to the time orientation of the background rather than the direction of the bicharacteristic flow in this section. We recall that for $(y,\xi)\in K$, one has $\xi_t\neq 0$, see \cite{Dyatlov}. As a consequence, with regard to the semi-classical flow as depicted in Figure~\ref{fig:semflow}, the geodesics in $\Gamma$ correspond to  semiclassical bicharacteristics flowing to or from $K_{\pm 1}$.
\end{remark}

Since the trapped set is located in region $\MI$, we focus on that part of the spacetime, and consider test sections in $\Gamma_{\MI}(\cB(s,s))$.
We will use the notation 
\begin{align}
    X_\pm(x):=\pm 1_{\rr_\pm}(x)x\, .
\end{align}

Let us briefly summarize our strategy for this case.
It will be similar to \cite{DMP} and \cite{Klein}: In Lemma~\ref{prop:KMS and ground}, we will first show that $w^\pm_\sH$ satisfies the analyticity properties of a KMS-state with respect to $v_\sH$, and $w^\pm_\sI$ the analyticity properties of a ground state with respect to $\partial_t$. Following the work of \cite{SV}, we use these properties in Lemma~\ref{lemma:SV2.1} to obtain estimates which restrict the so-called "global asymptotic pair correlation spectrum", which was introduced in \cite{SV} as a modification of wavefront sets in certain settings. In fact, we translate these estimates into restrictions on the wavefront sets of $w^\pm_\sH$ and $w^\pm_\sI$ in Lemma~\ref{lemma:passivity Had proof}. The proof relies on choosing adapted coordinate systems and makes use of an alternative formulation of the condition of a direction of rapid decrease given in \cite[Proposition 2.1c]{V}.
We begin this section by analysing the restriction of the two-point functions to $\Sol_s,p(\MI)$.
\begin{remark}
\label{rem:wh restriction to H-}
Let $h\in C_0^\infty(\rr_-)$, and let $U$ be the KBL coordinate. Then one can check that for $U<0$, one has
\begin{subequations}
    \begin{align}
    (X_{\pm}(D_U)h)(U)=\frac{e^{\kappa_+ {}^*t}}{\kappa_+} \chi_\pm(D_ {{}^*t})h({}^*t)\, ,\\
    \chi_\pm(x)=\frac{xe^{\pm \pi x/\kappa_+}}{e^{ \pi x/\kappa_+}-e^{- \pi x/\kappa_+}}=\frac{\pm x}{1-e^{\mp2\pi x/\kappa_+}}\, ,
\end{align}
\end{subequations}
see for example \cite{DMP}, and compare also \cite[App. D]{GHW}.
Consequently, for $\phi,\phi'\in \cS_{s,p}(\sH)$ with $\supp\phi\cup\supp\phi'\subset\{U<U_0<0\}$, and working in the Kruskal tetrad,\footnote{The Kruskal tetrad is not a global trivialization, since the vector field $m$ does not extend to the axis of roatation $\{\sin\theta=0\}$. However, as discussed in Remark~\ref{rem:compl triv}, there is a complete set of trivializations, each of which can be reached from the Kruskal tetrad by a ${}^*t$-independent phase factor, leaving the results below intact.} we can write
\begin{align}
\label{eq:w_hplus on H-}
    w_{\sH}^\pm(\phi,\phi')=c_s\int\limits_{\mathclap{\rr\times\ss^2}} \overline{\partial_U^s\phi_s}({}^*t,\omega_+)\chi_\pm(D_{{}^*t})\partial_U^s\phi'_s({}^*t,\omega_+)\d{}^*t \d^2\omega_+\, ,
\end{align}
where we have used the notation $c_s=4(r_+-M)^{2s}(r_+^2+a^2)$ to collect the constant appearing in the state.

Note also that by the definition of $\cS_{s,p}(\sH)$, $\cU_s\phi_s$ and $\widetilde{\cU}_s\overline{\phi_{-s}}$ are elements of the space $L^2(\sH_-; \cB(s,0))=L^2(\rr_{{}^*t}\times \ss^2; \cB(s,0);(r_+^2+a^2)\d{}^*t\d^2\omega_+)$ when $\phi=(\phi_s,\overline{\phi_{-s}})\in \cS_{s}(\sH)$ is supported in $\{U<U_0\}$ for some $U_0<0$. Here, the inner product on $\cB(s,0)$ is the same as defined in Corollary~\ref{cor:B(s,0) IP} and used in Remark~\ref{remark:func calc}.

Working in the Kruskal tetrad and ${}^*K$-coordinates, and denoting the Fourier transformation in ${}^*t$ by $\cF_{{}^*t}$, one has on the domain of $\Theta_{{}^*t}$ in $L^2(\sH_-; \cB(s,0))$ 
\begin{align}
\label{eq:Fourier def chi}
   \chi_\pm(\i\Theta_{{}^*t})=\cF^{-1}_{{}^*t} \chi_\pm(k)\cF_{{}^*t}
\end{align}

In the same way,  for $\phi\in \cS_{s,p}(\sI_-)$, $\cW_s\overline{\phi_{-s}}$ and $\widetilde{\cW}_s\phi_s$ are elements of $L^2(\rr_\ft\times \ss^2;\cB(s,0))$ with measure $\d\ft \d^2\omega^{*}$ and $\cB(s,0)$-inner product as in Remark~\ref{remark:func calc}. Thus, by working in the conformally rescaled Kinnersley tetrad and coordinates $(\ft,x,\theta,\varphi^*)$, one obtains on the domain of $\Theta_{t^*}$ in $L^2(\rr_\ft\times \ss^2;\cB(s,0))$
\begin{align}
\label{eq:Fourier def xpm}
     X_\pm(\i\Theta_\ft)=\cF^{-1}_\ft X_\pm(k)\cF_\ft\, .
\end{align}
\end{remark}

\begin{remark}
    Let $(\mathfrak{Y}_{\sH,b}^*)_{b\in \rr}$ be the group of push-forwards on $\Gamma(\cV_s)$ belonging to the Killing field $v_\sH$  as introduced in Remark~\ref{rem:alg autom of killing flow} (i.e., set $c=a(r_+^2+a^2)^{-1}$ in Remark~\ref{rem:alg autom of killing flow}).
    Since the flow generated by $v_\sH$ leaves $\sH_-:= \sH \cap \{ U<0 \} $ invariant, $(\mathfrak{Y}_{\sH,b}^*)_{b\in \rr}$ maps the space
    \begin{align*}
        \cS_{s}(\sH_-)=\{\phi\in \cS_s(\sH):\exists U_0<0: \supp(\phi)\subset\{U<U_0\}\}
    \end{align*}
    into itself. By finite speed of propagation, the traces of elements in $\Sol_{s}(\MI)$ are elements of this space.

    Similarly, if $(\fY_{\sI,b}^*)_{b\in\rr}$ denotes the group of push-forwards on $\Gamma(\cV_s)$ belonging to $\partial_t$ as introduced in Remark~\ref{rem:alg autom of killing flow} (by setting $c=0$ therein), then $\fY_{\sI,b}^*$ maps $\cS_{s}(\sI_-)$ into itself.
\end{remark}

Next,  consider $\phi\in \cS_{s,p}(\sH)\cap\cS_{s}(\sH_-)$, and let $\phi_s(x)$  and $\overline{\phi_{-s}}(x)$ be the first and second component of $\phi$ expressed in the Kruskal tetrad \eqref{eq:Kruskal tetrad}. Then, following Lemma~\ref{lemma:trafo of triv. flow}, one has
\begin{subequations}
\begin{align}
    \Upsilon_{\sH,(s,s),b}^*\phi_s({}^*t, \theta,\varphi_+)&=e^{s\kappa_+ b}\phi_s({}^*t-b,\theta,\varphi_+)\, ,\\
    \Upsilon_{\sH,(s,-s),b}^*\overline{\phi_{-s}}({}^*t, \theta,\varphi_+)&=e^{-s\kappa_+ b}\overline{\phi_{-s}}({}^*t-b,\theta,\varphi_+)\, .
\end{align}
\end{subequations}
Moreover, a direct computation in this tetrad confirms that
\begin{align}
    B_s\Upsilon_{ \sH,(s,s),b}^*\phi_s({}^*t,\theta,\varphi_+)&=\Upsilon_{ \sH,(s,-s),b}^*B_s\phi_s({}^*t,\theta,\varphi_+)\\\nonumber
    &=\Upsilon_{ \sH, (s,-s),b}^*\overline{\phi_{-s}}({}^*t,\theta,\varphi_+)\, ,
\end{align}
indicating that $(\mathfrak{Y}_{ \sH,b}^*)_{b\in \rr}$ maps $\cS_{s}(\sH_-)\cap\cS_{s,p}(\sH)$ to itself. Moreover, one can check in a similar way that
\begin{subequations}
\begin{align}
    \cU_s\Upsilon_{\sH,(s,s),b}^*\phi_s({}^*t,\theta,\varphi_+)&=\Upsilon_{\sH,(s,0),b}^*\cU_s\phi_s({}^*t,\theta,\varphi_+)\, ,\\
    \widetilde{\cU}_s \Upsilon_{\sH,(s,-s),b}^* \overline{\phi_{-s}}({}^*t,\theta,\varphi_+) &=  (r_+-M)^{2s}\Upsilon_{ \sH,(s,0),b}^*\cU_s\phi_s({}^*t,\theta,\varphi_+)\, .
\end{align}
\end{subequations}
Similarly, for $\phi\in \cS_{s,p}(\sI_-)$ with components $(\phi_s,\overline{\phi_{-s}})$ in the conformally rescaled Kinnersley tetrad \eqref{eq:conf Kinnersley} one computes 
\begin{subequations}
\label{eq:flow commutator A tho}
\begin{align}
    \left[\Upsilon_{\sI,(s,-s),b}^*, A_s\right]&=0\, ,\\
   \Upsilon_{\sI,(s,w),b}^*\tho&=\tho\Upsilon_{\sI,(s,w-1),b}^*\, ,
\end{align}
\end{subequations}
so that $ (\mathfrak{Y}_{\sI,b}^*)_{b\in \rr}$ maps $\cS_{s,p}(\sI_-)$ to itself, and moreover
\begin{subequations}
\begin{align}
 \cW_s\Upsilon^*_{\sI,(s,-s),b}\overline{\phi_{-s}}(\ft, \theta, \phi^*)&=\Upsilon^*_{\sI,(s,0),b}\cW_s\overline{\phi_{-s}}(\ft, \theta, \varphi^*)\, ,\\
    \widetilde{\cW}_s\Upsilon^*_{\sI,(s,s),b}\phi_{s}(\ft, \theta, \phi^*)&=\Upsilon^*_{\sI,(s,0),b}\widetilde{\cW}_s\phi_{s}(\ft, \theta, \varphi^*)\, .
\end{align}
\end{subequations}

This allows us to show that $w^\pm_\sH$ has the same analyticity properties as a KMS state, while $w^\pm_\sI$ has the analyticity properties of a ground state:
\begin{proposition}
\label{prop:KMS and ground}
Consider the state defined by the two-point functions $w^\pm_\sH+w^\pm_\sI$ restricted to the algebra $\cA_{s,B,p}(\MI)$. Then
\begin{enumerate}
    \item $w^\pm_\sH$ satisfies the KMS condition with respect to $v_\sH$ for inverse temperature $\beta=\frac{2\pi}{\kappa_+}$, i.e., for all $\phi,\phi'\in \cS_{s,p}(\sH)\cap\cS_{s}(\sH_-)$
    \begin{enumerate}
        \item for all $b\in \rr$, one has
        \begin{align}  
    \label{eq:wH-invariance}
    w_\sH^\pm(\mathfrak{Y}_{\sH,b}^*\phi, \mathfrak{Y}_{ \sH,b}^* \phi')=w_\sH^\pm(\phi,\phi')\, .
    \end{align}
        \item the functions $\rr\ni b\mapsto w_\sH^\pm(\phi, \mathfrak{Y}_{\sH,b}^*\phi')\in \cc$ are bounded.
    \item the identities  
    \begin{align}
       \label{eq:KMS identities}
        \int\limits_\rr \hat f(b)w_\sH^\pm(\phi, \mathfrak{Y}_{\sH,b}^* \phi')\d b=\int\limits_\rr \hat f (b\pm \i\beta)\omega_\sH^\mp(\phi, \mathfrak{Y}_{\sH,b}^* \phi')\d b
    \end{align}
   hold for all $f\in \coinf(\rr;\rr)$, where $\hat f$ denotes its Fourier-transform.
    \end{enumerate}    
    \item $w_\sI^\pm$ satisfies the ground state conditions with respect to $\partial_t$, i.e., for all $\phi,\phi'\in \cS_{s,p}(\sI_-)$
    \begin{enumerate}
        \item for all $b\in \rr$, one has
        \begin{align}  
    \label{eq:wI-invariance}
    w_\sI^\pm(\fY_{\sI,b}^*\phi, \fY_{\sI,b}^* \phi')=w_\sI^\pm(\phi,\phi')\, .
    \end{align}
        \item the functions $\rr\ni b\mapsto w_\sI^\pm(\phi, \fY_{\sI,b}^*\phi')\in \cc $ are bounded.
    \item the identities  
    \begin{align}
       \label{eq:GS identities}
        \int\limits_\rr \hat f(b)w_\sI^\pm(\phi, \mathfrak{Y}_{\sI,b}^* \phi')\d b=0
    \end{align}
   hold for all $f\in \coinf(\rr_\pm^*;\rr)$, where $\hat f $ denotes its Fourier-transform.
    \end{enumerate}
    \end{enumerate}
\end{proposition}
\begin{proof}
Let us begin with statement (1a), the invariance of $w_\sH^\pm$. We use the formulation \eqref{eq:w_hplus on H-} for $w_\sH^\pm$ restricted to functions in $\cS_{s,p}(\sH)\cap \cS_{s}(\sH_-)$. Then we obtain together with the above results
\begin{align}
\label{eq:KMS lemma 1a}
    w_{\sH}^\pm&(\mathfrak{Y}_{ \sH,b}^*\phi,\mathfrak{Y}_{\sH,b}^*\phi')\\\nonumber
    &=c_s\int\limits_{\mathclap{\rr\times\ss^2}}  \overline{\partial_U^s\Upsilon_{\sH,(s,s),b}^*\phi_s} ({}^*t,\omega_+) \left(\chi_\pm(D_{{}^*t})\partial_U^s\Upsilon_{\sH,(s,s),b}^* \phi'_s\right)({}^*t,\omega_+)\d{}^*t \d^2\omega_+\\\nonumber
    &=c_s\int\limits_{\mathclap{\rr\times\ss^2}} \left(\overline{\partial_U^s\phi_s}\right)({}^*t-b,\omega_+)\left(\chi_\pm(D_{{}^*t})\Upsilon_{\sH,(s,0),b}^*\partial_U^s\phi'_s\right)({}^*t,\omega_+)\d{}^*t \d^2\omega_+\\\nonumber
    &=c_s\int\limits_{\mathclap{\rr\times\ss^2}} \left(\overline{\partial_U^s\phi_s}\right)({}^*t-b,\omega_+)\left(\chi_\pm(D_{{}^*t})\partial_U^s\phi'_s\right)({}^*t-b,\omega_+)\d{}^*t \d^2\omega_+\\\nonumber
    &=w_{\sH}^\pm(\phi,\phi')\, .
\end{align}
The first line follows directly from \eqref{eq:w_hplus on H-}, the second from the above results, the  final  by a redefinition of the integration variable. The third line is a consequence of 
\begin{align}
    \chi_\pm(D_{{}^*t})\Upsilon_{\sH, (s,0),b}^* \phi = \Upsilon_{\sH,(s,0),b}^* \chi_\pm(D_{{}^*t}) \phi\, ,
\end{align}
which can be checked by a direct computation, using the representation of $\chi_\pm(D_{{}^*t})$ given in \eqref{eq:Fourier def chi}.

Analogously, by working in the conformally rescaled Kinnersley tetrad \eqref{eq:conf Kinnersley}, and employing the representation \eqref{eq:Fourier def xpm} of $X_\pm(D_\ft)$ via Fourier transform in $\ft$, one verifies 
\begin{align}
    X_\pm(D_\ft) \Upsilon_{\sI,(s,0),b}^*\phi=\Upsilon_{\sI,(s,0),b}^*X_\pm(D_\ft)\phi
\end{align}
for all $\phi$ in the domain of $X_\pm(D_\ft)$.
A calculation that is essentially the same as in \eqref{eq:KMS lemma 1a} then shows the invariance of $w_\sI^\pm$, statement (2a).

 For the boundedness property, statement (1b), consider $\phi\in \cS_{s,p}(\sH)\cap\cS_{s}(\sH_-)$. Then by the arguments presented in the proof of Lemma~\ref{lemma:well def wH}, $\cU_s\phi_s$ and $\chi_\pm(\i\Theta_{{}^*t})\cU_s\phi_{s}$ are elements of $L^2(\rr_{{}^*t}\times \ss^2; \cB(s,0))$. An application of the Cauchy-Schwarz inequality working in the Kruskal tetrad then results in 
 \begin{align*}
     \abs{\omega_\sH^\pm(\phi,\mathfrak{Y}_{\sH,b}^*\phi')}=&\abs{c_s\int\limits_{\mathclap{\rr\times\ss^2}} \left(\overline{\partial_U^s\phi_s}\right)({}^*t,\omega_+)\left(\Upsilon_{\sH,(s,0),b}^*\chi_\pm(D_{{}^*t})\partial_U^s\phi'_s\right)({}^*t,\omega_+)\d{}^*t \d^2\omega_+}\\
     \leq&\abs{c_s}\norm{\overline{\partial_U^s\phi_s}}_{L^2(\sH_-)}\norm{\Upsilon_{\sH,(s,0),b}^*\chi_\pm(D_{{}^*t})\partial_U^s\phi'_s}_{L^2(\sH_-)}\\
     =&\abs{c_s}\norm{\overline{\partial_U^s\phi_s}}_{L^2(\sH_-)}\norm{\chi_\pm(D_{{}^*t})\partial_U^s\phi'_s}_{L^2(\sH_-)}\, ,
 \end{align*}
 where  we have abbreviated $L^2(\sH_-)=L^2(\rr_{{}^*t}\times \ss^2; \cB(s,0))$. The last step follows again by a redefinition of the integration variable. Since the bound is finite and uniform in $b$, this shows the desired result.

 The boundedness in statement (2b) can be shown in the same way by taking into account that for $\phi\in\cS_{s,p}(\sI_-)$ with components $(\phi_s, \overline{\phi_{-s}})$ in the conformally rescaled Kinnersley tetrad, $\partial_\ft^{s(+1)}\overline{\phi_{-s}}$ and  $\partial_\ft^{s(+1)}A_s^{-1}\overline{\phi_{-s}}$ are elements of $L^2(\rr\times\ss^2; \cB(s,0))$.

 It remains to show the identities \eqref{eq:KMS identities} and \eqref{eq:GS identities}.
 To this end, let us first note the identity
 \begin{align}
 \label{eq: w+ to w-}
     \chi_\pm(x)=e^{\pm\beta x}\chi_\mp(x)\, ,
 \end{align}
 from which we can infer that for $f\in L^2(\rr_{{}^*t}\times \ss^2; \cB(s,0))$ we have
 \begin{align}
     \left(\chi_\pm(D_{{}^*t})f\right)({}^*t,\omega_+)=\left(\chi_\mp(D_{{}^*t})f\right)({}^*t\mp \i\beta, \omega_+)\, 
 \end{align}
 in the Kruskal tetrad, using the form \eqref{eq:Fourier def chi} of $\chi_\pm(D_{{}^*t})$. Using the Fourier representation of $\Upsilon^*_{\sH,(s,0),b}$, one has 
 \begin{align}
     \chi_\pm(D_{{}^*t})\Upsilon_{\sH,(s,0),b}^* f=\left(e^{ \i b\cdot}\chi_\pm \right)(D_{{}^*t})f\, , \, .
 \end{align}
Analogously, using the formulation \eqref{eq:Fourier def xpm} for $X_\pm(\i\Theta_\ft)$ in the conformally rescaled Kinnersley tetrad,  and the Fourier representation of and $\Upsilon^*_{\sI,(s,0),b}$, one can infer that
\begin{align}
    X_\pm(D_\ft)\Upsilon_{\sI,(s,0),b}^*f=\left(e^{\i b\cdot} X_\pm\right)(D_\ft) f
\end{align}
in the conformally rescaled Kinnersley tetrad for any $f$ in the domain of $X_\pm(D_\ft)$.

 We note that the function
 \begin{align}
 \label{eq:def chipmb}
     x\mapsto \chi_\pm^b(x):=\chi_\pm(x)e^{ibx}=e^{\i\Re b x}\frac{\pm x}{e^{\pm \Im b x}-e^{\mp (\beta-\Im b)x}}
 \end{align}
 is bounded  for $b$ in the strip $\rr\pm \i(0,\beta)\subset \cc$.\footnote{ To see this, we first note that for $x\to 0$, one finds $\chi^b_\pm(x)\to \beta$. One can therefore extend the function continuously to $x=0$. Moreover, in this strip one has both $\Im b>0$ and $\beta-\Im b>0$. Therefore, for large $\abs{x}$, one of the terms in the denominator becomes small, while the other increases exponentially, so that $\chi_\pm^b(x)\to 0$ as $x\to \pm \infty$.}
 Moreover, for any $x\in \rr$, the function $b\mapsto \chi_\pm^b(x)$ is holomorphic on $\rr\pm \i(0,\beta)$, and its complex derivative $b\mapsto (\i x)\chi_\pm^b(x)$ is bounded  uniformly in $x\in \rr$.

Similarly, the function
\begin{align*}
    b\mapsto X_\pm^b(x):= e^{\i bx}X_\pm(x)
\end{align*}
is a holomorphic function of $b$ and it is bounded uniformly in $x$ for $b$ in the upper/lower half-plane $\cc^*_\pm$. The same is true for its complex derivative with respect to $b$, $\i xX_\pm(x)e^{\i bx}$. 
 
 We can then use this to define the functions 
 \begin{align*}
     F^\pm_\sH[\phi,\phi'](b):=& c_s\int\limits_{\mathclap{\rr\times\ss^2}} \overline{\partial_U^s\phi_s}({}^*t,\omega_+)\chi_\pm^b(D_{{}^*t})\partial_U^s\phi'_s({}^*t,\omega_+)\d{}^*t \d^2\omega_+\, ,\quad b\in\rr\pm \i(0, \beta)\\
      F^\pm_\sI[\psi,\psi'](b):=& 2^{1+2s}\int\limits_{\mathclap{\rr\times\ss^2}} \left[\overline{\partial_\ft^sA_s^{-1}}\psi_{-s}(\ft,\omega^*)X_\pm^b(D_{\ft})\partial_\ft^s\overline{\psi'_{-s}}(\ft,\omega^*)\right.\\
      &\left.+\overline{\partial_\ft^s}\psi_{-s}(\ft,\omega^*)X_\pm^b(D_{\ft})\partial_\ft^sA_s^{-1}\overline{\psi'_{-s}}(\ft,\omega^*)\right]\d\ft \d^2\omega^{*}\, ,\quad b\in \cc^*_\pm
 \end{align*}
 for any $\phi,\phi'\in \cS_{s,p}(\sH)\cap \cS_{s}(\sH_-)$ (using their components in the Kruskal tetrad) and $\psi, \psi'\in \cS_{s,p}(\sI_-)$ (using the components in the conformally rescaled Kinnersley tetrad).
 
On the same domains,  the functions
 \begin{align*}
   F^\pm_\sH[\phi,\phi']'(b):=& c_s\int\limits_{\mathclap{\rr\times\ss^2}} \overline{\partial_U^s\phi_s}({}^*t,\omega_+)(\i D_{{}^*t})\chi_\pm^b(D_{{}^*t})\partial_U^s\phi'_s({}^*t,\omega_+)\d{}^*t \d^2\omega_+\, ,\\
   F^\pm_\sI[\psi,\psi']'(b):=& 2^{1+2s}\int\limits_{\mathclap{\rr\times\ss^2}} \left[\overline{\partial_\ft^sA_s^{-1}}\psi_{-s}(\ft,\omega^*)\i D_\ft X_\pm^b(D_{\ft}) \partial_\ft^s \overline{\psi'_{-s}}(\ft,\omega^*)\right.\\
      &\left.+\overline{\partial_\ft^s}\psi_{-s}(\ft,\omega^*)\i D_\ft X_\pm^b(D_{\ft}) \partial_\ft^s A_s^{-1} \overline{\psi'_{-s}}(\ft,\omega^*)\right]\d\ft \d^2\omega^{*}
 \end{align*}
 are also well-defined by the boundedness of $(\i x)\chi_\pm^b(x)$ and $(\i x)X_\pm^b$ on the respective domains. Moreover, by functional calculus, it follows that $F^\pm_\sH[\phi,\phi']'(b)$ and $F^\pm_\sI[\psi,\psi']'(b)$ are the respective complex derivatives of $F^\pm_\sH[\phi,\phi'](b)$ and $F^\pm_\sI[\psi,\psi'](b)$ for any fixed but arbitrary pairs $\phi,\phi'\in \cS_{s,p}(\sH)\cap \cS_{s}(\sH_-)$ or $\psi, \psi'\in \cS_{s,p}(\sI_-)$, and for $b$ in the respective domain, $b\in \rr \pm \i (0, \beta)$ or $b\in \cc^*_\pm$. Hence, $F^\pm_\sH[\phi,\phi'](b)$ is holomorphic in the strip $b\in \rr \pm \i (0, \beta)$, and $F^\pm_\sI[\psi,\psi'](b)$ is holomorphic in $\cc_\pm^*$. 

 From previous results, we can also infer that for $b\in\rr$
 \begin{subequations}
 \begin{align}
  \lim\limits_{\epsilon\to 0}F^\pm_\sH[\phi,\phi'](b\pm \i\epsilon)= w^\pm_\sH(\phi, \mathfrak{Y}^*_{\sH,b}\phi')\, ,\\
  \lim\limits_{\epsilon\to 0}F^\pm_\sI[\psi,\psi'](b\pm \i\epsilon)= w^\pm_\sI(\psi, \fY^*_{\sI,b}\psi')\, ,
 \end{align}
 \end{subequations}
 and hence both functions remain bounded in the limit $\Im b\to 0$. Moreover, combining the property \eqref{eq: w+ to w-} of $\chi_\pm(x)$ and the definition \eqref{eq:def chipmb} of $\chi_\pm^b(x)$, we obtain for $b\in \rr$ and $\epsilon \in (0,\beta)$
 \begin{align}
 \label{eq:wH+-wH- relation}
     F^\pm_\sH[\phi,\phi'](b\pm \i(\beta-\epsilon))=F^\mp_\sH[\phi,\phi'](b\mp \i\epsilon)\, .
 \end{align}
Taking the limit of $\epsilon\to 0$ on both sides and employing the results above then shows that the limit of $F^\pm_\sH[\phi,\phi'](b)$ as $\Im b\to \pm i2\pi\kappa_+^{-1}$ is bounded as well.
From this, it finally follows that for any $f\in \coinf(\rr;\rr)$
\begin{align*}
    \int\limits_\rr \hat f(b)w_\sH^\pm(\phi, \mathfrak{Y}^*_{\sH, b}\phi')db&=\int\limits_\rr \hat f(b)F_\sH^\pm[\phi,\phi'](b)db\\
    &=\int\limits_\rr\hat f( b)F_\sH^\mp[\phi,\phi'](b\mp \i\beta)\d b\\
    &=\int\limits_{\rr\pm i\beta} \hat f(\widetilde{b}\pm \i\beta)F_\sH^\mp[\phi,\phi'](\widetilde{b})\d\widetilde{b}\\
    &=\int\limits_{\rr}\hat f (\widetilde{b}\pm \i\beta ) F_\sH^\mp[\phi,\phi'](\widetilde{b})\d\widetilde{b}\\
    &=\int\limits_{\rr}\hat f(\widetilde{b}\pm \i\beta)w_\sH^\mp(\phi, \mathfrak{Y}^*_{\sH,\widetilde{b}} \phi') \d\widetilde{b}\, .
\end{align*}
To reach the second line, we use \eqref{eq:wH+-wH- relation}.
The third line is reached by a change of coordinates to $b\mp \i\beta$. The fourth follows from the residue theorem, taking into account that $F^\pm[\phi,\phi'](b)$ is bounded on the strip $\rr\pm \i[0, \beta]$  and holomorphic on $\rr\pm \i(0, \beta)$, and $\hat f(z)$ vanishes for $\Re z\to \pm\infty$ as long as $\Im z$ remains bounded. Renaming the integration variable then concludes the proof of statement (1c).

For statement (2c), we use that for $f\in \coinf(\rr^*_\pm)$, $\hat f(b)$ is holomorphic, and we can find some $\epsilon>0$ depending on the support of $f$ so that 
\begin{align*}
    \abs{\hat f(b)}\leq C_N e^{-\epsilon\Im b} (1+\abs{b})^{-N}
    \end{align*}
for $b\in \cc_\pm$ and for any $N\in \nn$, where $C_N>0$ depends on $f$ and $N$.
Consider the closed path $C^\pm$ in $\cc$ obtained by following the edges of the box $\{-R\leq \Re z<R, 0\leq \pm \Im z \leq R\}$  counter-clockwise, and taking the limit $R\to \infty$. Then using the residue theorem together with the decay estimates for $\hat f$ and the boundedness of $F^\pm_\sI[\psi,\psi']$, one can show that \eqref{eq:GS identities} holds.
\qeds
\end{proof}

 Using these results, it is now possible to show the following result, following the relevant parts of the proof of \cite[Proposition 2.1]{SV} step by step:
 \begin{lemma}
 \label{lemma:SV2.1}
     Let $(g_{j}^\lambda)_{\lambda>0}\subset \Gamma_c(\cB(s,-s))$,  $j\in\{1,2\}$, be one-parameter families of sections, and assume that they satisfy
     \begin{subequations}
     \begin{align}
     \label{eq:estimate EV gjlambda sH}
        w^\pm_\sH(T_\sH\phi_{g^\lambda_j}, T_\sH\phi_{g^\lambda_j})&\leq c(1+\lambda)^{l}\\
        \label{eq:estimate EV gjlambda sI}
        w^\pm_\sI(T_\sI\phi_{g^\lambda_j}, T_\sI\phi_{g^\lambda_j})&\leq c'(1+\lambda)^{l'}
     \end{align}
     \end{subequations}
     for some constants  $c, c'>0$  and $l,l'\in\rr$, where we have used the notation \eqref{eq:notation soln of test fct}. Let $\xi=(\xi_1,\, \xi_2)\in \rr^2 \setminus \{0\}$ so that $ \xi_2>0$.
     
     Then for any  such $(g^\lambda_j)_{\lambda>0}$ and $\xi$, one can find a test function $h\in \coinf(\rr^2)$ satisfying $\hat h(0,0)=1$, and an open neighbourhood $V^\pm_\xi \subset \rr^2 \setminus \{0\}$ of $\pm \xi$, so that (for $k\cdot t=k_1t_1+k_2t_2$)
     \begin{subequations}
     \begin{align}
         \sup\limits_{k\in V^\pm_\xi}\abs{\,\int\limits_{\rr^2}e^{\i\lambda k\cdot t} \hat h(t) \omega^\pm_\sH(\fY^*_{ \sH, t_1}T_\sH\phi_{g^\lambda_1}, \fY^*_{\sH,t_2}T_\sH\phi_{g^\lambda_2}) \d^2t}&=\cO(\lambda^{-\infty})\, ,\\
          \sup\limits_{k\in V^\pm_\xi}\abs{\,\int\limits_{\rr^2}e^{\i\lambda k\cdot t}\hat h(t) \omega^\pm_\sI(\fY^*_{\sI,t_1}T_\sI\phi_{g^\lambda_1}, \fY^*_{\sI,t_2}T_\sI\phi_{g^\lambda_2}) \d^2t}&=0\, ,
          \label{eq:decay estimate wsI}
     \end{align}
     \end{subequations}
      where the second equation holds for sufficiently large $\lambda$. Here, $\abs{f(\lambda)}=\cO(\lambda^{-\infty})$ means that for any $N\in\nn$, one can find $C_N, \lambda_N>0$, so that $\abs{f(\lambda)}\leq C_N\lambda^{-N}$ for all $0<\lambda_N<\lambda$.
 \end{lemma}

 \begin{remark}
 \label{rem:SV2.2}
     By the results of \cite[Lemma 2.2]{V} this remains valid if $\hat h$ is replaced by $f\cdot \hat h$ for some $f\in \coinf(\rr^2)$, after potentially shrinking $V^\pm_\xi$. It also remains valid if the families $(g_j^\lambda)_{\lambda>0}$ depend on additional parameters in such a way that the estimates \eqref{eq:estimate EV gjlambda sH}  and \eqref{eq:estimate EV gjlambda sI} remain valid and are uniform in the additional parameters, see \cite[Remark 2.2]{SV}.
 \end{remark}

 \begin{proof}
 For the convenience of the reader, we will give a brief account of the proof, which follows the relevant steps in the proof of \cite[Lemma 2.1]{SV}.
Throughout, we set $\beta=2\pi\kappa_+^{-1}$.

Let $h_j\in\coinf(\rr;\rr)$, $j\in\{1,2\}$ be any two test functions that satisfy $\widehat{h_j}(0)=1$, and set $h(t_1,t_2)=h_1(t_1)h_2(t_2)$. Then clearly $\hat h(0, 0)=1$, as required. Further, choose some $\epsilon>0$ and set $V^\pm_\xi$ to be a sufficiently small open neighbourhood of $\pm\xi$ so that $\pm k_2>\epsilon$ for all $(k_1,k_2)\in V^\pm_\xi$.

For $\lambda>0$ and $k\in V^\pm_\xi\subset \rr^2$, define the functions
\begin{align*}
    g_{\lambda,k}(p)&:=\sqrt{2\pi}h_1(-p- \lambda(k_1+k_2))h_2(p)\, , \quad p\in \rr\\
    f_{\lambda,k}(p)&:=\sqrt{2\pi}g_{\lambda,k}(p- \lambda k_2)=h_1(-p-\lambda k_1)h_2(p-\lambda k_2)\, ,\quad p\in\rr\, .
\end{align*}
By the definition of $V^\pm_\xi$, there is a $\lambda_0>0$ so that $f_{\lambda,k}\in \coinf(\rr_\pm^*)$ for all $ 0<\lambda_0<\lambda$ and all $k\in V^\pm_\xi$.

Using the convolution theorem, one can show
  \begin{align}
    \widehat{f_{\lambda,k}}(s)=\int\limits_\rr e^{ \i s'\lambda k_1} e^{\i(s+s')\lambda k_2}\widehat{h_1}(s') \widehat{h_2}(s+s')\d s'\, .
    \end{align}
    Combining this with \eqref{eq:wH-invariance} and a change of coordinates, this entails
   \begin{align}   
      \label{eq:convol FT id 1}
        \int\limits_{\rr^2} e^{ \i\lambda k\cdot t}\hat{h}(t)w^\pm_\sH(\fY^*_{\#,t_1}T_\#\phi_{g^\lambda_1}, \fY^*_{\#,t_2}T_\#\phi_{g^\lambda_2})\d^2t\\\nonumber
        =\int\limits_\rr \widehat{f_{\lambda,k}}(s)w^\pm_\#(T_\#\phi_{g^\lambda_1}, \fY^*_{\#,s}T_\#\phi_{g^\lambda_2})\d s\, ,
    \end{align}
  where $\#\in\{\sH, \sI\}$.  
Further, using the convolution and Paley-Wiener theorems, it follows that
    \begin{align}
    \label{eq:convol FT id 2}
   \sup\limits_{\lambda>0, k\in \rr^2} \int\limits_\rr \abs{\widehat{g_{\lambda,k}}(s\pm \i\beta)} \d s\leq c<\infty\, .
\end{align}  
From the Fourier transform, one obtains also
\begin{align}
\label{eq:convol FT id 3}
\widehat{f_{\lambda,k}}(s\pm \i\beta)=e^{ \mp\lambda k_2\beta} e^{ \i s\lambda k_2}\widehat{g_{\lambda,k}}(s\pm \i\beta)\, .
\end{align}

It then follows from \eqref{eq:GS identities} combined with \eqref{eq:convol FT id 1} that
\begin{align*}
    \sup\limits_{k\in V^\pm_\xi}\abs{\,\int\limits_{\rr^2} e^{\i\lambda k\cdot t}\hat{h}(t)w^\pm_\sI(\fY^*_{\sI,t_1}T_\sI\phi_{g^\lambda_1}, \fY^*_{\sI,t_2}T_\sI\phi_{g^\lambda_2})\d^2t}=0
\end{align*}
for all $ 0<\lambda_0<\lambda$, showing \eqref{eq:decay estimate wsI}.

Furthermore, the combination of the above results with \eqref{eq:KMS identities} allows to estimate
\begin{align*}
    &\sup\limits_{k\in V^\pm_\xi}\abs{\,\int\limits_{\rr^2}e^{\i\lambda k\cdot t} \hat h(t) w^\pm_\sH(\fY^*_{\sH,t_1}{ T_\sH}\phi_{g^\lambda_1}, \fY^*_{\sH,t_2}{ T_{\sH}}\phi_{g^\lambda_2}) \d^2t}\\
    =&\sup\limits_{k\in V^\pm_\xi}\abs{\int\limits_{\rr} \widehat{f_{\lambda,k}}(s) w^\pm_\sH({ T_\sH}\phi_{g^\lambda_1}, \fY^*_{\sH,s}{ T_\sH}\phi_{g^\lambda_2}) \d s}\\
    =&\sup\limits_{k\in V^\pm_\xi}\abs{\int\limits_{\rr} \widehat{f_{\lambda,k}}(s\pm \i\beta) w^\mp_\sH({ T_\sH}\phi_{g^\lambda_1}, \fY^*_{\sH,s}{ T_\sH}\phi_{g^\lambda_2}) \d s}\\
    \leq &\sup\limits_{k\in V^\pm_\xi} \e^{\mp \lambda k_2\beta}\int\limits_\rr \abs{\widehat{g_{\lambda,k}}(s\pm \i\beta) }
    \abs{ w^\mp_\sH ( T_\sH\phi_{g^\lambda_1}, \fY^*_{\sH,s}{ T_\sH}\phi_{g^\lambda_2})}\d s\\
    \leq &C e^{-\lambda\epsilon\beta}  (1+\lambda)^{l''}=\cO(\lambda^{-\infty})
\end{align*}
for some positive constant $C$,  and some $l''\in\rr$. The first step uses \eqref{eq:convol FT id 1}, the second \eqref{eq:KMS identities}, the third \eqref{eq:convol FT id 3}. The final step combines \eqref{eq:convol FT id 2} with an application of the Cauchy-Schwarz inequality to $\abs{ w^\mp_\sH ({ T_\sH}\phi_{g^\lambda_1}, \fY^*_{\sH, s}{ T_\sH}\phi_{g^\lambda_2})}$ and \eqref{eq:estimate EV gjlambda sH}.
     \qeds
 \end{proof}

The above finally allows us to show the following result
\begin{lemma}
\label{lemma:passivity Had proof}
    Let $(x,\xi)\in T^*(\MI)$, with $\xi$ non-vanishing and null.
    \begin{enumerate}
        \item If $\xi(v_\sH)<0$, then $(x,\pm\xi; x,\mp\xi)$ is a direction of rapid decrease for $W^\pm_\sH$. 
        \item If $\xi(\partial_t)<0$, then $(x,\pm\xi; x,\mp\xi)$ is a direction of rapid decrease for $W^\pm_\sI$.
    \end{enumerate}
\end{lemma}

\begin{proof}
     This proof follows largely parts 1 and 2 of the proof of \cite{SV}, see also \cite{Klein}.
      Let $\#$ replace either $\sH$ or $\sI$. We can then cover a neighbourhood $\cU$ of $x$ by the coordinate systems $\fk_{\#}: \cU\to \fk_{\#}(\cU)\subset \rr^4$, $y\mapsto (t^\#, \cv{y}^\#)$.
   
    We choose these coordinates such that $x$ is mapped to $0\subset \rr^4$, and $v_\sH=\partial_{t^\sH}$ or $\partial_t=\partial_{t^\sI}$, respectively. This can for example be achieved by setting
     \begin{align*}
      (t^{\sH}, \cv{y}^\sH)(y)&=(t, r, \theta, \varphi_+)(y)-(t(x), r(x), \theta(x), \varphi_+(x))\, ,\\
      (t^{\sI}, \cv{y}^\sI)(y)&=(t^*, r^{-1}, \theta, \varphi^*)(y)-(t^*(x), r^{-1}(x), \theta(x), \varphi^*(x))
    \end{align*}
    
    Moreover, we can fix the tetrad over $\cU$ to be the Kinnersley \eqref{eq:Kinnersley} tetrad, which also fixes a local trivialization of $\cV_s$ over $\cU$.\footnote{In doing so, we assume that $\cU$ does not intersect the axis of rotation, and hence that $x$ does not lie on it. If it does, then we can use the corresponding stereographic coordinates instead $\theta$ and $\varphi_+$ or $\varphi^*$, and replace the Kinnersley tetrad by the one from \eqref{eq:compl triv syst for killing} that extends smoothly to $x$, and where $c=\Omega_+$ ($c=0$) in \eqref{eq:compl triv syst for killing}. We note that we can always assume $\cU$ to be sufficiently small so that it does not intersect both poles of the sphere $\{t^\#=r=\text{const.}\}$.}

   Our coordinate systems are build in such a way that we can find some $c>0$ so that for $\abs{b}<c$, the isomorphisms $\upsilon_{\sH,b}$ and $\upsilon_{\sI,b}$ induced by $v_\sH$ and $\partial_t$ respectively can be written as
   \begin{align}
       \upsilon_{\#,b}(y)= \fk_\#^{-1}\left(\fk_\#(y)+(b,\cv{0})\right)
   \end{align}
   for $y$ in some smaller neighbourhood $K\subset \cU$ of $x$. In particular, for $y\in K$ and $\abs{b}<c$, $\upsilon_{\#,b}(y)\in \cU$. We also assume that $K\subset \cU$ is chosen small enough, so that we can define spatial translations by
   \begin{align}
    \varsigma_{\#, \cv{p}}(y)=\fk_{\#}^{-1}\left(\fk_\#(y)+(0,\cv{p})\right) 
   \end{align}
   for $\cv{p}\in \rr^3$ contained in a small neighbourhood $\cv{C}$ of the origin. Furthermore, for $0<\lambda$ and $\abs{y}$ sufficiently small, we can define 
   \begin{align}
       \delta^\#_\lambda(y)=\fk_\#^{-1}\left(\lambda\cdot\fk_\#(y)\right)\, .
   \end{align}

   The push-forward belonging to the cover map $\Upsilon_{\#,(s,w),b}$ of $(\upsilon_{\#,b})_{\abs{b}<c}$ on $\cB(s,w)$ is the same as before. In the local trivialization induced by  the Kinnersley tetrad, its action on $\phi\in \Gamma_K( \cB(s,w))$ takes the simple form 
   \begin{align*}
       \Upsilon_{\#,(s,w),b}^*\phi(y)=\phi(\upsilon_{\#,b}^{-1}y)\, .
   \end{align*}

    For $s\in \{0,1,2\}$, let $\iota_{s}\in \Gamma(\cB(s,-s))$ be a smooth section so that in the Kinnersley tetrad, one has 
   \begin{align*}
       \iota_{ s}(y)=1 \quad \forall y\in \cU\, .
   \end{align*}
    Then $\iota_{s}$ and $\overline{\iota_{s}}$ provide local bases for $\cB(s,-s)$ and  $\cB(-s,-s)$, respectively. We then consider the scalar bi-distributions
    \begin{align*}
        \widetilde{W}^\pm_\sH(f,h):=W^\pm_\sH(\overline{\iota_{s}}f, \iota_{s}h)\, ,\\
         \widetilde{W}^\pm_\sI(f,h):=W^\pm_\sI(\overline{\iota_{s}}f, \iota_{s}h)\, .
    \end{align*}
    
In the following, we identify $\cU$ and $\fk_{\#}(\cU)$ to simplify notation. Let $g_{\#,i}\in \coinf(\cU)$, $i\in \{1,2\}$, so that $\supp g_{\#,i}\subset K$ and $\widehat{(g_{\#,1}\otimes g_{\#,2})}(0,0)=1$, where the Fourier transform is taken in $\rr^4\supset \fk_{\#}(\cU)$. For an arbitrary but fixed $p\geq 1$ and $ 1\leq \lambda$, set $g^\lambda_{\#,i}(y)= g_{\#,i}(\delta^{\#}_{\lambda^{ p}}y)$ whenever $ \delta^\#_{\lambda^{ p}}y\in \cU$, and $g^\lambda_{\#,i}(y)=0$ otherwise.
For $0<\lambda<1$, we set $g^\lambda_{\#,i}(y)=g^1_{\#,i}(y)$.

Then for all $\lambda>0$, $g_{\#,i}^\lambda\in \coinf(\cM)$ with support in $K$. As a result, $\upsilon_{\#,b}^*\varsigma_{\#, \cv{p}}g^\lambda_{\#,i}\in \coinf(\cM)$ is well-defined for $\abs{b}<c$ and $\cv{p}\in\cv{C}$, and is supported in $\cU$.

 As detailed in Remark~\ref{rem:2pt distr estimates}, we have  
\begin{align*}
\abs{\widetilde{W}^\pm_\#(g_{\#,i}^\lambda, g_{\#,i}^\lambda)}=\abs{w^\pm_\#(T_\#\phi_{\iota_{s}\overline{g^\lambda_{\#,i}}}, T_\#\phi_{\iota_{s}g^\lambda_{\#,i}})}\leq C\norm{\iota_s g_{\#,i}^\lambda}_{C^k(\cU)}^2\leq C'\norm{g_{\#,i}^\lambda}_{C^k(\cU)}^2\, ,
\end{align*}
where $C, c'>0$, $\norm{\iota_s g_{\#,i}^\lambda}_{C^k(\cU)}$ is any $C^k$-norm on $\Gamma_\cU(\cB(s,-s))$ as defined in Definition~\ref{def:function spaces and norms} and 
\begin{equation}
    \norm{f}_{C^k(\cU)}=\sup\limits_{x\in \cU}\sup\limits_{\alpha\leq k}\abs{Z_1^{\alpha_1}\dots Z_4^{\alpha_4}f(x)}
\end{equation} 
is any $C^k$-norm on $\coinf(\cU)$ determined by a family of smooth vector fields $(Z_i)_{i=1}^4$ as in Definition~\ref{def:function spaces and norms}.

We can then use the coordinate derivatives of our coordinate system $\fk_\#$ as the  vector fields in the scalar $C^k$-norm, allowing us to estimate 
\begin{align*}
\norm{g_{\#,i}^\lambda}_{C^k(\cU)}&\leq \lambda^{ kp}\norm{g_{\#,i}^1}_{C^k(\cU)}\leq \norm{g_{\#,i}^1}_{C^k(\cU)} (1+\lambda)^{kp}\, ,
\end{align*}
and consequently
\begin{align*}
\abs{\widetilde{W}^\pm_\#(\varsigma_{\#,\cv{p}}^*g_{\#,i}^\lambda, \varsigma_{\#,\cv{p}}^*g_{\#,i}^\lambda)}&\leq 
\norm{g_{\#,i}^1}_{C^k(\cU)}^2 (1+\lambda)^{2kp}\, 
\end{align*}
for any $\cv{p}\in \cv{C}$.
Combined with Remark~\ref{rem:SV2.2}, this implies that $(\iota_{s }\varsigma_{\#,\cv{p}}^*g^\lambda_{\#,i})_{\lambda>0, \cv{p}\in \cv{C}}$ satisfy the conditions of Lemma~\ref{lemma:SV2.1}.

We note also that by the choice of $\iota_{s}$, the results in Remark~\ref{rem:alg autom of killing flow}, and \eqref{eq:flow commutator A tho}, we have
\begin{align*}
    T_\#\phi_{\iota_{s}\upsilon_{\#,b}^*g^\lambda_{\#,i}}=T_\#\phi_{\Upsilon_{\#,(s,s),b}^*\iota_{s}g^\lambda_{\#,i}} =T_\#\fY^*_{\#,b}\phi_{\iota_{s}g^\lambda_{\#,i}}=\fY^*_{\#,b}T_\#\phi_{\iota_{s}g^\lambda_{\#,i}}\, .
\end{align*}

We can now choose $H_\#\in \coinf(\cU\times \cU)$ as 
\begin{align*}
H_\#(t^{\#}, \cv{y}^\#, \widetilde{t}^\#, \widetilde{\cv{y}}^\#)=\chi(t^\#, \widetilde{t}^\#)\hat h(t^\#, \widetilde{t}^\#)\eta(\cv{y}^\#,\widetilde{\cv{y}}^\#)\,,
\end{align*}
where $\hat h$ is as in Lemma~\ref{lemma:SV2.1}, and $\chi\in \coinf([-c,c]^2;\rr)$ and $\eta\in \coinf(\cv{C}\times\cv{C};\rr)$ are chosen such that $H_\#(0,0)=1$. 

By  the assumptions of the Lemma, $\xi_{t^\#}=\xi(\partial_{t^\#})<0$. Let $V^\pm_{\xi_{t^\#}}\subset \rr^2$ be an open neighbourhood of $(\pm\xi_{t^\#}, \mp\xi_{t^\#})$ as in Lemma~\ref{lemma:SV2.1}, and let $V^\pm_\#\subset \rr^4\times\rr^4\cong T^*(\cM\times\cM)$ (using the trivialization induced by the corresponding coordinate system) be an open neighbourhood of $(\pm \xi, \mp\xi)$ of the form $V^\pm_{\xi_{t^\#}} \times \widetilde{V}_\#^\pm$, where  $\widetilde{V}_\#^\pm$ is an open neighbourhood in $\rr^3\times\rr^3$.

One can then employ the estimates of Lemma~\ref{lemma:SV2.1}, combined with Remark~\ref{rem:SV2.2}, to show that
\begin{align}
    \sup\limits_{(k,k')\in V'^{\pm}_\#}\abs{\,\int\limits_{\rr^8}e^{\i\lambda(k,k')\cdot(y,y')}H_\#(y,y')\widetilde{W}_\#^\pm(\upsilon^*_{\widetilde{t}}\upsilon^*_{\cv{p}}g^\lambda_1, \upsilon^*_{\widetilde{t}}\upsilon^*_{\cv{p}}g^\lambda_2)\d^4y\d^4y'}=\cO(\lambda^{-\infty})\, ,
\end{align}
where $V'^\pm_\#\subset V^\pm_\#$ is a neighbourhood of $(\pm\xi, \mp \xi)$ whose closure is a compact subset of $V^\pm_\#$, see also the proof of \cite[Prop. V.2]{Klein}.

By  the characterisation of directions of rapid decrease in \cite[Prop. 2.1(c)]{V}, this entails that $(x,\pm\xi; x,\mp\xi)$ is a direction of rapid decrease for $\widetilde{W}^\pm_\#$, and thereby for $W^\pm_\#$. 
\qeds
\end{proof}

We can then conclude

\begin{corollary}
\label{cor:Had bw trapped}
    Let $(x,\xi)\in T^*(\cM)$ lie on a bicharacteristic strip that belongs to a past-trapped null geodesic. If $\xi$ is past-directed, then $(x,\pm\xi; x,\mp\xi)$ is a direction of rapid decrease for $W^\pm$. 
\end{corollary}

\begin{proof}
  By the results of \cite[Prop. 6.6]{HaK}, if the geodesic specified by $(x,\xi)$ is past trapped, then it must eventually enter region $\MI$. Moreover, if it satisfies $\xi(v_\sH)\geq 0$ or $\xi(\partial_t)\geq 0$, then $\xi$ must be future-directed. The result then follows immediately from Lemma~\ref{lemma:passivity Had proof}.
  \qeds
\end{proof}
 
\subsection{Geodesics ending at the horizon}
To control the wavefront set for geodesics that approach the horizon, we first proceed with some general setup.

\begin{definition}
    Let $\iota_s\in \Gamma( \cB(s,-s))$ be a smooth section that extends smoothly to $\breve{\mathrm{M}}_\mathrm{I}$ and $\sH$, and that is chosen such that
\begin{align*}
    \iota_s(x)=\begin{cases}
        \left[(\mathfrak{l}, \mathfrak{n}, m)(x), 1\right]\, : & r\leq 3M\\
        \left[(l, n, m)(x), 1\right]\, : & r\geq 6M\, ,
    \end{cases}
\end{align*}
with a smooth transition in between. Then $\iota_s$ and $\overline{\iota_s}$ provide local bases of $\cB(s,-s)$ and $\cB(-s,-s)$, respectively. We then define the scalar  bi-distribution
\begin{align}
    \widetilde{W}^\pm_{ U}: \coinf(\cM)\times\coinf(\cM)\to \cc\, ,\quad
    (f,h)\mapsto W^\pm_{ U}(\overline{\iota_s} f, \iota_s h)=w^\pm_U(\phi_{\iota_s \overline{f}}, \phi_{\iota_s h})\, .
\end{align}
\end{definition}

By construction, the wavefront set of $\widetilde{W}^\pm_{ U}$ agrees with that of $W^\pm_{ U}$, and it is sufficient to consider the former. As for $W^\pm_{ U}$, we may further split $\widetilde{W}^\pm_{ U}$ into $\widetilde{W}^\pm_\sH$ and $\widetilde{W}^\pm_\sI$.

Next, we define the maps 
\begin{subequations}
\begin{align}
    \coinf(\cM)\ni f&\mapsto \fh_s f= \cU_sT_\sH {\tho}^{2s}\overline{E_{-s}}\iota_s f\in  \Gamma(\cB(s,0)\vert_{\sH})\, ,\\
    \coinf(\cM)\ni f&\mapsto \mfi_s f=\cW_sT_\sI A_s\overline{E_{-s}} \iota_s f\in \Gamma( \breve{\cB}(s,0)\vert_{\sI_-})\, ,\\
    \coinf(\cM)\ni f&\mapsto \mfi'_s f=\cW_s T_\sI \overline{E_{-s}}\iota_s f\in \Gamma( \breve{\cB}(s,0)\vert_{\sI_-})\, .
\end{align}
\end{subequations}
We can then write 
\begin{subequations}
\begin{align}
    \widetilde{W}^\pm_\sH(f,h)=&4(r_+-M)^{2s}\IP{\fh_s\overline{f},X_\pm(\i\Theta_U) \fh_s h}_\sH,\\
    \widetilde{W}^\pm_\sI(f,h)=&2 \left[ \IP{\mfi_s\overline{f}, X_\pm(\i\Theta_\ft) \mfi'_s h}_\sI+\IP{\mfi'_s\overline{f}, X_\pm(\i\Theta_\ft) \mfi_s h}_\sI\right]\, ,
\end{align}
\end{subequations}
where we use the notation of Remark~\ref{remark:func calc}  and where we employ part 5 of Proposition~\ref{prop:A_s properties} to eliminate the inverse of $A_s$.

Let $q=(y, \xi)\in \cN$ be such that the corresponding null geodesic  maximally extended on $\MK_K$ intersects $\sH$. In a neighbourhood $\cU\subset \cM$ of $y$ with compact closure, use the Kruskal coordinate system to identify $T^*\cU$ with $\cU\times \rr^4$.
 Then
\begin{lemma}
    There is an $\epsilon>0$ so that for any $q\in \dot{T}^*\cM$ satisfying $0<V(q)<\epsilon$, where $V$is the KBL coordinate, and whose corresponding null geodesic intersects $\sH$, one can find a neighbourhood $K\times V$ of $q$ so that all null geodesics corresponding to some $q'\in K\times V$ intersect $\sH$ inside a compact set $K_q\subset\sH$. 
\end{lemma}
\begin{proof}
By \cite[Proposition 4.4.6]{ON}, any null geodesic intersecting the horizon must cut through it transversally. The result then follows  from the continuous dependence of the geodesic flow on the initial data.
    \qeds
\end{proof}
Let now $K\subset \cU$ and $V\subset \rr^4\setminus\{0\}$ be chosen as in the lemma. Let us assume additionally that $K$ is compact and that we can find an open neighbourhood $\cU'$ of $K$ whose closure is contained in $\cU$. Moreover, let $K_q\subset \sH$ be as in the lemma. By its compactness, there is a $U_0\in \rr$ so that $U(x)\geq U_0$ for all $x\in K_q$.

Subsequently, we pick a cutoff function $\chi_q\in \coinf(\sH; \rr)$ so that  $\supp(1-\chi_q)\cap J^-(\cU)\subset\{ x\in \sH: U(x)\leq U_1:=\min(-3, U_0-3)\}$. Let $f$, $h\in \coinf(\cM)$ be a pair of test functions so that $\supp f\cup\supp h\subset K$.
We can then split $\widetilde{W}^\pm_\sH$ into the pieces
\begin{align}
\label{eq:splitting sH}
    2^{-2}(r_+-M)^{-2s}\widetilde{W}^\pm(f,h)=&\IP{\chi_q\fh_s\overline{f}, X_\pm(\i\Theta_U) \chi_q\fh_s h}_\sH\\\nonumber
    &+\IP{(1-\chi_q)\fh_s \overline{f}, X_\pm(\i\Theta_U)  \chi_q \fh_s h}_\sH\\\nonumber
    &+\IP{\chi_q \fh_s\overline{f},  X_\pm(\i\Theta_U) (1-\chi_q) \fh_s h}_\sH\\\nonumber
    &+ \IP{(1-\chi_q) \fh_s \overline{f},  X_\pm(\i\Theta_U) (1-\chi_q) \fh_s h}_\sH\, .
\end{align}
 We note that this is the same kind of splitting as was used in the proof of Lemma~\ref{lemma:well def wH}. 
Let us for the moment ignore the first term on the right-hand side of \eqref{eq:splitting sH}, and focus on the remaining terms. Then we wish to show that these pieces, as well as $w^\pm_\sI$, do not contain the point $q$ in their wavefront set.

To do so, we want to show that $q$ is a direction of rapid decrease for these pieces.
To this end, we introduce the notation

\begin{definition}
  Let $q'=(y',\xi')\in \dot{T}^*_\cU\cM$, and identify $T^*_\cU\cM$ with $\cU\times \rr^4$ in the Kruskal coordinates. If $f\in \coinf(\cU)$, then we set
 \begin{align*}
     f^\lambda_{\xi'}(x)= f(x)e^{ \i\lambda\xi'\cdot x} \in \coinf(\cM)\, ,
 \end{align*}
 where the product $\xi'\cdot x$ is the usual product in $\rr^4$, and $\xi'$, $x$ are identified with elements of $\rr^4$ using the Kruskal coordinates.
\end{definition}
We then show the estimates
\begin{lemma}
    Let $q$, $\cU$, $\cU'$, $K$, $V$, $\chi_q$ as above. Then we can find a function $f\in \coinf(K)$ with $f(y)=1$ and a $V'\subset V$ so that for all  $N\in \nn$, there are $C_N>0$ and $\lambda_N>0$ so that
     \begin{align*}
       \sup\limits_{k\in V'}\abs{(1-\chi_q)\fh_s f_k^\lambda}\leq \abs{\log \abs{U}}^{-d} C_N\lambda^{ -N}\, ,\\
       \sup\limits_{k\in V'}\abs{\mfi_s' f_k^\lambda}\leq \ft^{-\tilde{d}}C_N\lambda^{ -N}\\
       \sup\limits_{k\in V'}\abs{\mfi_s f_k^\lambda}\leq \ft^{-\tilde{d}}C_N\lambda^{ -N}\, 
   \end{align*}
   for all $\lambda>\lambda_N$ and for some $d$, $\tilde{d}>1$.
\end{lemma}

\begin{remark}
\label{remark:cutoff large U}
    By finite speed of propagation and the resulting support properties of $\fh_s f$,  we may always multiply $\fh_s f$ by a function $\chi''\in \cinf(\sH)$ that is equal to $1$ on $\sH\cap J^-(\cU')$ but vanishes for large $U$.
\end{remark}

\begin{proof}
    The proof follows largely that of \cite[Lemma V.6]{Klein}. It makes critical use of the fact that the classical fields can be considered on the larger spacetime $\MK_K$.  Employing the classical phase space of this larger spacetime, it is then possible to estimate the traces of space-compact solutions of the Teukolsky equations. The estimates are obtained by splitting the solutions into different pieces. Propagation of Singularities and the support properties of the advanced and retarded Green operators are then used to obtain rapid-decay estimates for the various pieces.

    Let $\Sigma\subset \MK_K$ be a Cauchy surface of the Kruskal extension $\MK_K$ of $\cM$ so that $\Sigma$ agrees with $\sH$ on a neighbourhood of $\supp\chi_q$ and $\cU\subset I^+(\Sigma)$.
    Let $\Sigma_\pm\subset \MK_K$ be two other Cauchy surfaces satisfying $\Sigma^\pm\subset I^\pm (\Sigma)$ and $\cU\subset I^+(\Sigma_+)$.    
   Choose a cutoff function $ \chi\in \cinf(\MK_K;[0,1])$ so that $\chi=1$ on $J^-(\Sigma_-)$ and $\chi=0$ on $J^+(\Sigma_+)$, as in the proof of Proposition~\ref{prop:Green hyp props}.

    Let $ [K,V]\subset \MK_K$ denote the set of all points $x\in\MK_K$ that lie on a null geodesic specified by a covector $(y',\xi')\in K\times V$ (using the identification of $T^*_{y'}\MK_K$ and $\rr^4$ in Kruskal coordinates).
    Let $ \eta\in \cinf(\MK_K;[0,1])$ be a smooth cutoff function such that $\eta=1$ on an open neighbourhood $\cU'_q$ of $[K,V]$ that satisfies 
    \begin{align*}
        K_q\subset (\cU'_q\cap \sH) \subset \sH\cap J^-(\cU')\cap \{x\in \sH: U(x)> U_1+2\}\, .
    \end{align*}
    Further, we assume that 
    \begin{align*}
        \supp \eta\cap \sH \subset \sH\cap J^-(\cU')\cap \{x\in \sH: U(x)> U_1+1\}\, .
    \end{align*}

    Let $h_q$, and $h_q'\in \coinf(\MK_K)$ be chosen such that 
    \begin{enumerate}
        \item $\supp h_q\cup \supp h_q'\subset J^-(\cU)$ 
        \item $h_q+h_q'=1$ on $J^-(\cU')\cap J^-(\Sigma_+)\cap J^+(\Sigma_-)$
        \item $\supp h_q'\cap \supp \eta=\emptyset$
        \item $\supp h_q\cap \sH \cap \supp(1-\chi_q)=\emptyset$
    \end{enumerate}  

\begin{figure}
    \includegraphics[width=0.7\linewidth]{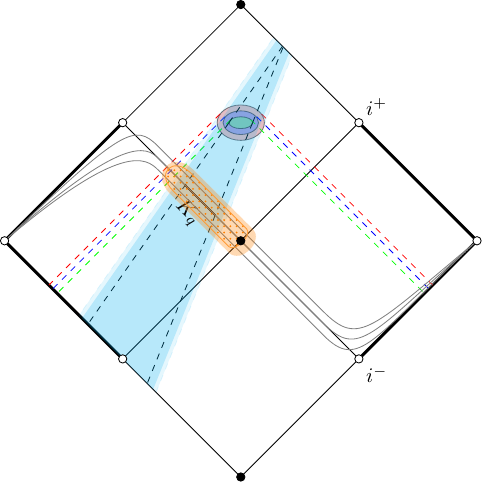}
    \caption{The Kruskal extension $\MK_K$ of $\cM$. The green, blue and red ellipsoids represent the sets $K$, $\cU'$, and $\cU$, respectively. The dashed lines in the corresponding colour delineate their causal pasts. The black dashed lines mark the region $ [K,V]$. The three gray lines are the Cauchy surfaces $\Sigma$ and $\Sigma_\pm$. The cyan region marks the support of $\eta$, the darker shade is the region in which $\eta=1$. The orange region marks the support of $h_q$, the dotted region is where $h_q=1$ .}
    \label{fig:cutoffs}
\end{figure}
\begin{figure}
    \includegraphics[width=\linewidth]{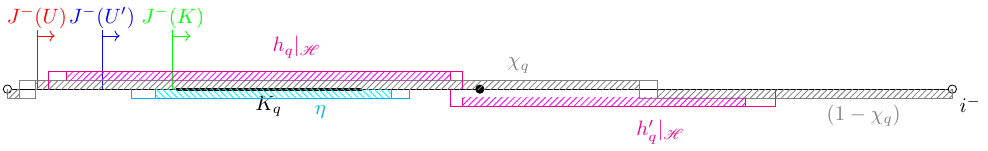}
    \caption{The horizon $\sH$ and the various supports thereon. The thick part of the horizon marks $K_q$, and the intersections of the causal pasts of $\cU$, $\cU'$, and $K$ are to the right of the corresponding marker. The supports of the various functions are marked by the coloured boxes, while the hatched part of the box marks where the function is equal to one.}
    \label{fig:cutoffs2}
\end{figure}

Then the choice of cutoff functions and the fact that $A_s$ commutes with multiplication by $r$ and taking the limit onto $\sI_-$ allows us to write
\begin{align*}
    (1-\chi_q)\fh_s f&=\chi''(1-\chi_q)\cU_s T_\sH {\tho}^{2s}(1-h_q) \overline{E_{-s}}(\iota_s f)\, ,\\
    \mfi_s f&= \cW_s T_\sI A_s(1-h_q) \overline{E_{-s}}(\iota_s f)\, ,\\
    \mfi_s' f&= \cW_s  T_\sI (1-h_q)  \overline{E_{-s}}(\iota_s f)\, .
\end{align*}
 where $\chi''$ is as in Remark~\ref{remark:cutoff large U}. 

To find bounds for these terms, we note that for any $f\in \coinf(\cM)$ supported in $K$, the section $\overline{T_{-s}}\chi\overline{E_{-s}}(\iota_s f)\in \Gamma_c(\cB(s,-s))$ is supported in the set $J^-(\cU')\cap J^-(\Sigma_+)\cap J^+(\Sigma_-)$ where $h_q+h_q'=1$. Moreover, as discussed in the proof of Proposition~\ref{prop:Green hyp props}, one has $\overline{T_{-s}}\chi\overline{E_{-s}}(\iota_s f)=-\overline{T_{-s}}(1-\chi)\overline{E_{-s}}(\iota_s f)$ and $\overline{E_{-s}}(\overline{T_{-s}}\chi\overline{E_{-s}}(\iota_s f))=\overline{E_{-s}}(\iota_s f)$.
Together with the linearity of the Green-hyperbolic operators, this allows us to decompose
    \begin{align}
        \overline{E_{-s}}(\iota_s f)=&h_q \overline{E_{-s}}(\iota_s f)+\overline{E_{-s}}(h_q'\overline{T_{-s}} \chi \overline{E_{-s}}\iota_sf)-\overline{E_{-s}}^-(\eta[\overline{T_{-s}}, h_q]\chi \overline{E_{-s}} \iota_sf)\\\nonumber
        &- \overline{E_{-s}}^-((1-\eta)[\overline{T_{-s}},h_q]\chi \overline{E_{-s}}\iota_sf)- \overline{E_{-s}}^+(\eta [\overline{T_{-s}},h_q](1-\chi) \overline{E_{-s}}\iota_s f)\\\nonumber
        &-\overline{E_{-s}}^+((1-\eta)[\overline{T_{-s}},h_q](1-\chi)\overline{E_{-s}}\iota_s f)\, 
    \end{align}
    for any $f\in \coinf(K)$.
 Note that by the choice of cutoff functions, we have
 \begin{align*}
     & \supp (\eta \chi(\d h_q))\subset \MK_K\setminus(\cM\cup\sH)\, , \quad
     \supp (\eta (1-\chi) (\d h_q))\subset \cM\subset \MK_K \, .    
 \end{align*}
 Consequently,
\begin{align*}
    T_\sH {\tho}^{2s}(1-h_q)\overline{E_{-s}}(\iota_s f)=& T_\sH {\tho}^{2s}\overline{E_{-s}}(h_q' \overline{T_{-s}}\chi \overline{E_{-s}}(\iota_sf))\\
    &-T_\sH{\tho}^{2s}\overline{E_{-s}}^-((1-\eta)[\overline{T_{-s}}, h_q]\chi \overline{E_{-s}}(\iota_s f))\\
    &-T_\sH {\tho}^{2s}\overline{E_{-s}}^+((1-\eta)[\overline{T_{-s}},h_q](1-\chi)\overline{E_{-s}}(\iota_s f))\, ,\\
    T_\sI  (1-h_q) \overline{E_{-s}}(\iota_s f)=& T_\sI \overline{E_{-s}}(h_q' \overline{T_{-s}}\chi \overline{E_{-s}}(\iota_sf))\\
    &-T_\sI  \overline{E_{-s}}^-((1-\eta)[\overline{T_{-s}}, h_q]\chi \overline{E_{-s}}(\iota_s f))\\
    =& T_\sI \overline{E_{-s}}(h_q' \overline{T_{-s}}\chi \overline{E_{-s}}(\iota_sf))\\
    &-T_\sI \overline{E_{-s}}((1-\eta)[\overline{T_{-s}}, h_q]\chi \overline{E_{-s}}(\iota_s f))
    \end{align*}
As argued in \cite{Klein}, the wavefront set of $\overline{E_{-s}}^\pm$ together with the assumption on the supports of the various cutoff functions implies that we can find a neighbourhood $V'\subset V$ and a function $f\in \coinf(K)$ with $f(y)=1$, so that 
\begin{align*}
    \sup\limits_{k\in V'}\norm{ h_q'\overline{T_{-s}}\chi \overline{E_{-s}}(\iota_s f_k^\lambda)}_{C^m(\supp h_q')}=\cO(\lambda^{ -\infty})\, , \\
    \sup\limits_{k\in V'}\norm{(1-\eta)[\overline{T_{-s}}, h_q]\chi \overline{E_{-s}}(\iota_s f_k^\lambda)}_{C^m(\supp h_q)}=\cO(\lambda^{-\infty})\, ,\\
    \sup\limits_{k\in V'}\norm{(1-\eta)[\overline{T_{-s}},h_q](1-\chi)\overline{E_{-s}}(\iota_s f_k^\lambda)}_{C^m(\supp h_q)}=\cO(\lambda^{ -\infty})\, .
\end{align*}
   Plugging this into the estimates  in Corollary~\ref{cor:4.1} or Theorem~\ref{th:4.2}, we find that for any $N\in \nn$, there are $d>1$, $\lambda_N>0$ and $C_N>0$, so that
   \begin{align*}
       \sup\limits_{k\in V'}\abs{(1-\chi_q)\fh_s f_k^\lambda}\leq \abs{\log \abs{U}}^{-d} C_N\lambda^{ -N}\, ,\\
       \sup\limits_{k\in V'}\abs{\mfi_s' f_k^\lambda}\leq \ft^{-d}C_N\lambda^{ -N}\\
       \sup\limits_{k\in V'}\abs{\mfi_s f_k^\lambda}\leq \ft^{-d}C_N\lambda^{ -N}\, ,
   \end{align*}
   for all $\lambda>\lambda_N$, completing the proof.
 \qeds   
\end{proof}
We can now replace $f$ and $h$ in \eqref{eq:splitting sH} by $f^\lambda_k$ for any $k\in V'$ and follow the same estimates as the ones used in the proofs of Lemma~\ref{lemma:well def wH} and Lemma~\ref{lemma:well def wI}. Whenever the term of \eqref{eq:splitting sH} in question contains $\chi_q\fh_sf^\lambda_{k}$, this can produce at most polynomial growth in $\lambda$ due to the results in Remark~\ref{rem:2pt distr estimates}. By potentially restricting $V'$ further, one can ensure that this polynomial growth in $\lambda$ is always dominated by the decay  in $\lambda$ of the other term which follows from the above lemma. Therefore, $q$ is a direction of rapid decrease for the second, third, and fourth term on the right-hand side of \eqref{eq:splitting sH}. 

It remains to treat the first term on the right-hand side of \eqref{eq:splitting sH}. We will do so by an explicit derivation of the wavefront set.

As a first step, we note that due to the structure of the Kerr metric in Kruskal coordinates  given in \eqref{eq:g Kruskal} restricted to $ \sH=\{V=0\}$, the proof of a geometric result on null geodesics intersecting $\sH$ given in \cite[Lemma II.1]{Klein} carries over to the present case without changes. Moreover, the structure of the conformally rescaled metric \eqref{eq:g conf} at $\sI_-$ allows to use the same proof for an analogous result at $\sI_-$, see also \cite[Lemma 5.1]{GHbdy} for a version of this result for null cones. We thus have

\begin{lemma}
\label{lemma:time ori horizon}
    Let $(x,k)\in \dot{T}^*\sH$. Then there is a unique $\eta\in N^*_x\sH$ so that $k+\eta\in T^*_x\MK_K$ is null and not contained in $N^*_x\sH$ if and only if $k(\partial_U)\neq 0$. Moreover, if $k(\partial_U)\neq 0$, then $k+\eta$ is future-directed if and only if $k(\partial_U)>0$. 
    
    Similarly, let $(x,k)\in T^*\sI_-$. Then there is a unique $\eta\in N^*_x\sI_-$ so that $k+\eta\in T^*_x\breve{\mathrm{M}}_\mathrm{I}$ is null and not contained in $N^*\sI_-$ if and only if $k(\partial_{{}^*t})\neq 0$; in this case, $k+\eta$ is future-directed if and only if $k(\partial_{{}^*t})>0$.
\end{lemma}
Here, we have denoted the conormal spaces of $\sH$ (as a submanifold of $\MK_K$) and $\sI_-$ (as a submanifold of $\breve{\mathrm{M}}_\mathrm{I}$) by $N^*\sH$ and $N^*\sI_-$, and used the decomposition $T^*_\sH\MK=T^*\sH\oplus N^*\sH$ and $T^*_{\sI_-}\breve{\mathrm{M}}_\mathrm{I}=N^*\sI_-\oplus T^*\sI_-$.

Next, we note that by combining Propagation of Singluarities \cite{DH} with the calculus of wavefront sets, in particular \cite[Thm. 8.2.4]{H:vol1} and \cite[Thm. 8.2.14]{H:vol1}, we can conclude that the kernel of $\fh_s$, which we will also denote by $\fh_s$, has the wavefront set (see for example \cite{Klein} and also \cite{KW, H:PhD, DMP, GHW, H:vol1})
\begin{align}
    \WF'(\fh_s)=\left\{(x,k; y,\xi)\in T^*(\sH\times \cM)\vert \exists \eta\in N^*_x\sH: (x, k+\eta)\sim (y,\xi)\right\}\, .
\end{align}
Here, $(y,\xi)\sim (y',\xi')$ means that $(y,\xi)$ and $(y',\xi')$ lie on the same bicharacteristic strip (of $T_s$), and we have implicitly used the embedding of $\sH$ and $\cM$ into $\MK_K$, as well as the corresponding embedding of the cotangent spaces.

Finally, consider the bi-distributions
\begin{align*}
    A_\pm: \Gamma_{\sH}( \cB(-s,0))\times \Gamma_\sH(\cB(s,0))\ni (f,h)\mapsto \int\limits \chi_q f X_\pm(\i\Theta_U)\chi_q h(r_+^2+a^2)\d U\d^2\omega_+\, .
\end{align*}
The distributional kernel of $A_\pm$ in the Kruskal tetrad is given by (see, e.g., \cite{KW, H:PhD, DMP, GHW, H:vol1, Klein})
\begin{align*}
  -\lim\limits_{\epsilon\to 0^+}\frac{1}{2\pi}\frac{\delta_{\ss^2}(\omega_+, \omega_+')\chi_q(U,\omega_+)\chi_q(U',\omega_+')}{(U-U'-\i\epsilon)^2}\, ,  
\end{align*}
so that the wavefront set of the above is given by 
\begin{align*}
    \WF'(A_\pm)=&\left\{ (x,k; x',k')\in T^*(\sH\times\sH): x,\, x'\in \supp(\chi_q),\right.\\  &(\omega_+, k(\partial_{\omega_+}))=(\omega_+', k'(\partial_{\omega_+}))\text{ and } k(\partial_U)=k'(\partial_U)>0 \text{ if } U=U',\\
    &\left.\text{ else } k(\partial_U)=k'(\partial_U)=0\right\}\, .
\end{align*}

Combining these results with the help of \cite[Thm. 8.2.14]{H:vol1}, and taking into account Lemma~\ref{lemma:time ori horizon}, we find that the wavefront set of
\begin{align*}
    (f,h)\in \coinf(\cM)\times\coinf(\cM)\mapsto \IP{\chi_q\fh_s\overline{f}, X_\pm(\i\Theta_U)\chi_q\fh_s h}_\sH
\end{align*}
is contained in $\cN^\pm \times \cN^\mp$, as required.

Altogether, we have shown in this section
\begin{lemma}
    \label{lemma:Had geodesics sH}
    Let $(y,\xi)\in T^*\cM$ lie on a bicharacteristic strip so that the corresponding null geodesic  maximally extended to $\MK_K$ intersects $\sH$. Then $(y,\xi; y,-\xi)\in\WF(W^\pm_{ U})$ only if $(y,\xi)\in \cN^\pm$. 
\end{lemma}


\subsection{Geodesics ending at past null infinity}

To conclude the proof of the Hadamard property, let us now consider a point $q=(y_0,\xi_0)\in \dot{T}^*\cM$ so that the corresponding null geodesic approaches $\sI_-$.
Since we may choose any representative of the bicharacteristic strip, we may assume that $y_0\in \MI$ and that $r(y_0)$ is sufficiently large, so that $\partial_t$ is timelike at $y_0$. 
It then follows immediately form Lemma~\ref{lemma:passivity Had proof} that $(y_0,\xi_0; y_0, -\xi_0)$ can only be in $\WF(W^\pm_\sI)$ if $\pm\xi_0$ is future-directed. It thus remains to examine the wavefront set of $W^\pm_\sH$.

There are two cases that we have to consider. If the geodesic is future-trapped
or if the geodesic crosses the event horizon of the black hole,
then the desired result follows directly from Lemma~\ref{lemma:passivity Had proof},  either because $v_\sH$ becomes timelike at some point along the geodesic, or because of the results in \cite[Prop. 6.6]{HaK}.

However, if the geodesic connects $\sI_-$ and $\sI_+$, it may occur that $v_\sH$ remains spacelike along the geodesic, so that $\pm\xi_0(v_{\sH})>0$ does not imply that $\xi_0$ is future-directed, and the results of Lemma~\ref{lemma:passivity Had proof} are not sufficient.  

To treat this case, let $(y_0,\xi_0)=(\ft_0,x_0,\tau_0,k_0)\in \cN\subset\dot{T}^*\MI$ lie on such a bicharacteristic. Let us note that, as mentioned earlier, we can assume that $y_0$ lies in a region where $\partial_t$ is timelike. This implies that $\tau_0\neq 0$, and by rescaling we may suppose $\tau_0=\pm 1$. Comparing to the semiclassical flow discussed in Section~\ref{subsec:semiflow}, we therefore note that $(x_0,k_0)$ must lie on a semiclassical bicharacteristic connecting $\sR_{in}$ and $\sR_{out}$, since there are no bicharacteristics which leave $\sR_{in}$ or $\sR_{out}$ and then return back to the same set.

Hence, consider the semiclassical flow with initial data $(x_0, k_0)$ and parameter $\tau_0$ connecting ${\mathcal R}_{in}$ and ${\mathcal R}_{out}$. Let $ f^{\lambda}\in \Gamma_c(\cB(s,s))$, $ \lambda>0$, be an oscillatory test section of the form 
\begin{align}
\label{WH1}
f^{\lambda}(t,x)=e^{-\i\lambda(\tau_0,k_0)(\ft,x)}\varphi(\ft,x),
\end{align} 
where $\varphi\in  \Gamma_c(\cB(s,s))$ is supported in  $B_{\epsilon}((\ft_0,x_0))$.  Let $u^{\lambda}$ be a backward solution of 
\begin{align}
\label{WH2}
T_su^{\lambda}=f^{\lambda}. 
\end{align}
We note that by finite speed of propagation, this implies $\supp u^{\lambda}\cap\{U>-\epsilon\}=\emptyset$ for some $\epsilon>0$, where $U$ is the KBL-coordinate defined in Section~\ref{subsec:Kruskal ext}.
\begin{lemma}
\label{lemma10.19}
Let $(y_0,\xi_0)\in T^*\cM$ be such that the bicharacteristic flow with parameter $\tau_0$ passing through $ (x_0,k_0)$ connects $\sR_{in}$ and $\sR_{out}$. Then there exists a small neighbourhood $\cV$ of $(x_0,k_0)$ such that for all $(x_p,\xi_p)\in \cV$ the bicharacteristic flow passing through $(x_p,\xi_p)$ connects $\sR_{in}$ and $\sR_{out}$.  
\end{lemma}
\begin{proof}
Let $r_{max}$ be as defined in \cite[Lemma 5.28]{Millet2}. We write 
\begin{align*}
(y_0,\xi_0)=(t_0,r_0,\theta_0,\varphi_0,\tau_0,\xi_{r0},\xi_{\theta0},\xi_{\varphi0}).
\end{align*}
We can suppose 
\begin{subequations}
\begin{align}
\label{fpm1} 
\exists s_+>0:\, r_0(s_+)>0,\, \dot{r}_0(s_+)>0,\\
\label{fpm2}
\exists s_-<s_+:\, r_0(s_-)>0,\, \dot{r}_0(s_-)<0.
\end{align}
\end{subequations}
Indeed, if this is not the case, Lemmas \cite[Lemma 5.32]{Millet2} and \cite[Lemma 5.33]{Millet2} lead to a contradiction. Now, conditions \eqref{fpm1} and \eqref{fpm2} persist under small perturbations of the initial data. It is then sufficient to apply Lemmas \cite[Lemma 5.30]{Millet2} and \cite[Lemma 5.31]{Millet2} to conclude.
\qeds
\end{proof}
\begin{proposition}
\label{RegHor}
Let $r$ and $\ell\in \rr$ be such that $r+\ell+\frac{1}{2}-2s>0$ and $r>\frac{1}{2}-s$.  Let $B_1\in \Psi^{0,0}_{sc,h}$ with $\WF_h(B_1)$ supported in a small enough neighbourhood of a point $ (x_1,k_1)=(r_+,\omega^*,k_1)\in T^*X$. Suppose that there exists $\eta >0$ such that $-\eta h\le {\rm Im}z\le 0$.  Then there exists  $G\in \Psi^{0,0}_{sc,h}$ with the following property : for all $\chi_1,\, \chi_2$ smooth and compactly supported on some open set of trivializations for $[0,\frac{1}{r_+-\epsilon})\times \cB_s^{\ss^2}$ containing $x_0$, we have  $\chi_1G\chi_2={\rm Op}_h(g)$ with $g$ satisfying $g\equiv 0$ in a neighbourhood of $(x_0,k_0)$. Furthermore, for any $N\in\nn$, there is a constant $C_N>0$ such that we have for all $u\in \bar{H}^{r,\ell}_{(b),h}$
\begin{align}
\label{estimneighhorizon}
\Vert B_1 u\Vert_{\bar{H}^{r,\ell}_{(b),h}}\le C_N(h^{-2}\Vert G\hat{T}_{s,h}(z)u\Vert_{\bar{H}^{r,\ell-1}_{(b),h}}+h^N\Vert u\Vert_{\bar{H}^{-N,\ell}_{(b),h}}). 
\end{align}
\end{proposition}
\begin{proof}
We will make extensive use of Proposition \ref{pp0}. We first consider the case $z_0=-1$ and put $K=K_{1}$. Let $B_{L_{\pm}}\in \Psi^{0,0}_{sc,h}$ be elliptic in a sufficiently small neighbourhood of $L_{\pm}$. By Proposition \ref{pp1}, we have the estimate 
\begin{align}
\label{L+est}
\Vert B_{L_{\pm}}u\Vert_{\bar{H}_{(b),h}^{r,\ell}}\le C_N\left(h^{-1}\Vert G\hat{T}_{s,h}(z)u\Vert_{\bar{H}^{r,\ell-1}_{(b),h}}+h^N\Vert u\Vert_{\bar{H}^{-N,\ell}_{(b),h}}\right).
\end{align}
Using Lemma \ref{lemma10.19}, we see that we can choose $G$ as in Proposition \ref{RegHor}. Let now $B_{\sR_{in}}\in \Psi^{0,0}_{sc,h}$ be elliptic in a sufficiently small neighbourhood of $\sR_{in}$. By Proposition \ref{pp2}, we have the estimate 
\begin{align}
\label{Routest}
\norm{ B_{\sR_{in}}u}_{\bar{H}_{(b),h}^{r,\ell}}\le C_N\left(h^{-1}\norm{ G\hat{T}_{s,h}(z)u}_{\bar{H}^{r,\ell-1}_{(b),h}}+h^N\Vert u\Vert_{\bar{H}^{-N,\ell}_{(b),h}}\right),
\end{align}
where, again, $G$ can be chosen as in Proposition \ref{RegHor}. Let now $B_K,\, B_0$ be as in Proposition \ref{pp4}. By Propositions \ref{pp0} and \ref{pp3}, we see that we have 
\begin{align}
\label{B0est}
\Vert B_0u\Vert_{\bar{H}_{(b),h}^{r,\ell}}\le& C_N\left(h^{-1}\Vert G\hat{T}_{s,h}(z)u\Vert_{\bar{H}_{(b),h}^{r-1,\ell}}+\Vert B_{L_+}u\Vert_{\bar{H}_{(b ),h}^{r,\ell}}\right.\\\nonumber
&\left.+\Vert B_{\sR_{in}}u\Vert_{\bar{H}_{(b ),h}^{r,\ell}}+h^N\Vert u\Vert_{\bar{H}^{-N,\ell}_{(b),h}}\right), 
\end{align}
where $G$ can be chosen as in the proposition. Combining Proposition \ref{pp4} with the estimates \eqref{L+est}, \eqref{Routest}, and \eqref{B0est}, we find 
\begin{align}
\label{Kest}
\Vert B_{K}u\Vert_{\bar{H}_{(b),h}^{r,\ell}}\le C_N\left(h^{-2}\Vert G\hat{T}_{s,h}(z)u\Vert_{\bar{H}^{r,\ell-1}_{(b),h}}+h^N\Vert u\Vert_{\bar{H}^{-N,\ell}_{(b),h}}\right),
\end{align}
where again $G$ can be chosen as in the proposition. Let now $\cV$ as in Lemma \ref{lemma10.19}. Now, every backward (or forward) bicharacteristic staring at a point in $WF(B_1)$ will meet the elliptic sets of $B_{L_{\pm}},\, B_K$ or $B_{\cR_{in}}$ without meeting $\cV$. We then use Propagation of Singularities, Proposition \ref{pp3}, and \eqref{Routest}, \eqref{B0est} and \eqref{Kest} to conclude. 
If $z_0=1$ we replace $\hat{T}_{s,h}(z)$ by $-\hat{T}_{s,h}(-z)$. Then the same argument works modulo inverting the roles of $L_+$ and $L_-$. 
\qeds
\end{proof}

\begin{proposition}
\label{propRegHor}
We have for all $M,\, N>0,$
 \begin{align}
\label{WH3}
\Vert \tho{}^{2s}\hat{u}^{\lambda}(\sigma,(r=r_+,\omega_+))\Vert_{L^2(\rr_{\sigma}\times {\ss}^2_{\omega_+}; \langle \sigma\rangle^M \d\sigma \d\omega_+)}\le C_{M,N}\langle \lambda \rangle^{-N} 
\end{align}
uniformly in $\lambda\ge 1$. 
\end{proposition}

\begin{proof}
Let $\chi\in C_0^{\infty}((r_+-\frac{\epsilon}{3},r_++\epsilon)),\, \chi=1$ in a neighbourhood of $r_+$. Note that  $\{r=r_+\}$ is a smooth hypersurface of $(r_+-\epsilon, \infty)_r\times\ss^2$. We therefore have for all $\ell\in \rr$ 
\begin{align}
&\norm{ \tho{}^{2s}\hat{u}^{\lambda}(\sigma)\vert_{r=r_+}} _{L^2\left(\rr_{\sigma}\times {\ss^2_{\omega_+}},\langle \sigma\rangle^M\d\sigma \d\omega_+\right)}\le C \norm{ \chi \tho{}^{2s}\hat{u}^{\lambda}(\sigma)} _{L^2\left(\rr_{\sigma},\langle \sigma\rangle^M\d\sigma;\bar{H}^{\frac{1}{2}+\tilde{\epsilon},\ell}_{(b)}\right)}\nonumber\\
&\le C\norm{ \widetilde{\chi}\hat{u}^{\lambda}(\sigma)}_{L^2\left(\rr_{\sigma},\langle \sigma\rangle^M\d\sigma;\bar{H}^{\frac{1}{2}+2\vert s\vert+\tilde{\epsilon},\ell}_{(b)}\right)}
\end{align}
for $\tilde{\epsilon}>0$ and some smooth function $\widetilde{\chi}$ with the same support properties as $\chi$. In the last step, we have used that $l^a=(\partial_r)^a$ in the ${}^*K$ coordinate system, and that the Sobolev norm is defined with respect to the connection $\Theta$. 
We now consider the Fourier transformed problem of \eqref{WH2},
\begin{align}
\label{WH4}
\hat{T_s}(\sigma)\hat{u}^{\lambda}(\sigma)=\hat{f}^{\lambda}(\sigma). 
\end{align} 
We start by computing the Fourier transform of $f^{\lambda}$. 
\begin{align}
\label{10.43a}
\hat{f}^{\lambda}(\sigma)&=\frac{1}{(2\pi)^{1/2}}\int\limits_{\rr}e^{\i\sigma \ft}e^{-\i\lambda \tau_0\ft}e^{-\i\lambda k_0x}\varphi(\ft,x)\d\ft\nonumber\\
 &=\cF_{\ft}(\varphi)(\sigma-\lambda \tau_0,x)e^{-\i\lambda k_0x}.
\end{align}
Here, $\cF_{\ft}$ denotes the partial Fourier transform with respect to $\ft$. Note that we only have to show
\begin{align}
\label{WH3m}
\norm{ \hat{u}^{\lambda}(\sigma,r,\omega)}_{L^2\left(\{\vert \sigma\vert\ge 1\}, \langle \sigma\rangle^M \d\sigma;\bar{H}^{r,\ell}_{(b)}\right)}\le C_{M,N}\langle \lambda \rangle^{-N} 
\end{align}
for suitable $r>\frac{1}{2}+2\vert s\vert,\ell\in \rr$. Indeed, let 
\begin{align}
\label{condrl}
r>2\vert s\vert+2\, , \quad  -\frac{3}{2}<\ell<-\frac{1}{2}.
\end{align}
Note that $\hat{u}^{\lambda}$ has the required regularity to apply 
the results of \cite{Millet2}. We have, by
\cite[Proposition 7.6]{Millet2}, for $\widetilde{\epsilon}>0$ and uniformly in $\sigma\in \rr$:
\begin{align*}
\norm{ 1_{\{\vert\sigma\vert\le 1\}}\hat{u}^{\lambda}}_{\bar{H}^{r-1,\ell}_{(b)}}\le C\norm{ 1_{\{\vert\sigma\vert\le 1\}}\hat{T}^{\lambda}_s(\sigma)\hat{u}^{\lambda}}_{\bar{H}^{r,\ell+\widetilde{\epsilon}}_{(b)}}=C\norm{ 1_{\{\vert\sigma\vert\le 1\}}\hat{f}^{\lambda}}_{\bar{H}^{r,\ell+\widetilde{\epsilon}}_{(b)}}\lesssim \langle \lambda\rangle^{-N} 
\end{align*}
for all $N>0$. The last estimate follows from \eqref{10.43a}, because $\varphi$ is a Schwartz function. For $\abs{\sigma}>1$, we will use Proposition \ref{RegHor}, and we have to estimate the right-hand side of \eqref{estimneighhorizon}. Let us first compute $G\hat{f}^{\lambda}$. Recall that we have a finite complete set of trivializations, that $\hat{f}^{\lambda}$ has compact support, and that $G$ is properly supported. It is therefore sufficient to consider $\chi_1G\chi_2$ for a finite number of cutoff functions $\chi_1,\, \chi_2$. We can therefore suppose that $G={\rm Op}_h(g)$. By choosing $\epsilon>0$ sufficiently small, we can suppose that $g\equiv 0$ on $B_{2\epsilon}((x_0,k_0))$. Note, however, that we also need to estimate a possible residual error $h^{N}\Vert E\hat{f}^{\lambda}\Vert_{\bar{H}_{(b),h}^{r,l-1}}$ with $E\in \Psi^{-\infty,-\infty}_{sc,h}$. We start with estimating $h^{-2}\Vert {\rm Op}_h(g)\hat{f}^{\lambda}\Vert_{\bar{H}^{r,l-1}_{(b),h}}$.
For $\delta>0$ small, we first consider the regime $\abs{\sigma}\le (1-\delta)\lambda$ or $ \abs{\sigma}\ge (1+\delta)\lambda$. 

In this regime, we have 
\begin{align*}
{\rm Op}_h(g)\hat{f}^{\lambda}&=h^{-3}\frac{1}{(2\pi)^{3}}\int e^{\i\frac{k}{h}(y-x)}e^{-\i\lambda k_0y}g(x,k)\cF_{\ft} (\varphi)(\sigma-\lambda \tau_0,y)\d^3 y\d^3 k\\
&=\frac{1}{(2\pi)^3}\int e^{ik(y-x)}e^{-i\lambda k_0x}g(x,h(k+\lambda k_0))\cF_{\ft}(\varphi)(\sigma-\lambda \tau_0,y)\d^3y\d^3k\\
&=\frac{1}{(2\pi)^3}\int e^{-i(k+\lambda k_0)x}g(x,h(k+\lambda k_0))\cF(\varphi)(\sigma-\lambda \tau_0,k)\d^3k,
\end{align*}
where $\cF$ stands for the Fourier transform in $(\ft,x)$. As $\varphi$ is Schwartz, $\cF(\varphi)(\sigma-\lambda \tau_0,k)$ is Schwartz as well. This gives the estimate 
\begin{align*}
\abs{\cF(\varphi)(\sigma-\lambda \tau_0,k)}&\le C_{N,\delta}\langle \sigma-\lambda \tau_0\rangle^{-N}\langle k\rangle^{-N}.
\end{align*}
for all $N$. Together with the properties of $g$, this gives 
\begin{align*}
\Vert Op_h(g)\hat{f}^{\lambda}\Vert_{\bar{H}^{r,\ell-1}_{(b),h}}\le C_{N,\delta}\langle \sigma-\lambda \tau_0\rangle^{-N} 
\end{align*}
for all $N>0$. Now we have,
\begin{align*}
\int\limits_{\mathclap{\vert \sigma\vert\le (1-\delta)\lambda}}\langle \sigma-\lambda \tau_0\rangle^{-N}\lesssim \langle \lambda\rangle^{-N+1}\, ,\qquad\quad\int\limits_{\mathclap{\vert \sigma\vert\ge (1+\delta)\lambda}}\langle \sigma-\lambda \tau_0\rangle^{-N}\lesssim \langle \lambda\rangle^{-N+1}.
\end{align*}
This gives 
\begin{align}
\label{10.58a}
\Vert {\rm Op}_h(g)\hat{f}^{\lambda}\Vert_{L^2\left(\{\vert\sigma\vert\ge (1+\delta)\lambda\}\cup \{\vert\sigma\vert\le (1-\delta)\lambda\},\langle \sigma\rangle^M\d\sigma \d\omega_+;\bar{H}^{r,\ell-1}_{(b),h}\right)}\le C_N\langle\lambda\rangle^{-N}
\end{align} 
for all $N>0$. 
We now consider the case 
\begin{align}
\label{insigma}
(1-\delta)\lambda\le \abs{\sigma}\le (1+\delta)\lambda.
\end{align}  
Let $\psi\in C_0^{\infty}((-1,1)),\, \psi(0)=1$. Remember that $\epsilon>0$ is chosen such that $\varphi(\ft,x)=0$ in $\cM\setminus B_{\epsilon}((\ft_0,x_0))$. We then have 
\begin{align*}
{\rm Op}_h(g)\hat{f}^{\lambda}&=h^{-3}\frac{1}{(2\pi)^{3}}\int e^{\i\frac{k}{h}(y-x)}e^{-\i\lambda k_0y}g(x,k)\cF_{\ft}(\sigma-\lambda \tau_0,y)\d^3 y\d^3 k\\
&=\frac{1}{(2\pi)^{3}}\int e^{\i k(y-x)}e^{-\i\lambda k_0x}g(x,h(k+\lambda k_0))\cF_{\ft}(\varphi)(\sigma-\lambda \tau_0,y)\d^3 y\d^3 k\\
&=\frac{1}{(2\pi)^{3}}\int \psi\left(\frac{\vert y-x\vert}{\epsilon}\right)e^{\i k(y-x)}e^{-\i\lambda k_0x}g(x,h(k+\lambda k_0))\cF_{\ft}(\varphi)(\sigma-\lambda \tau_0,y)\d^3 y\d^3 k\\
&+ \frac{1}{(2\pi)^{3}}\int \left(1-\psi\left(\frac{\vert y-x\vert}{\epsilon}\right)\right)e^{\i k(y-x)}e^{-\i\lambda k_0x}g(x,h(k+\lambda k_0))\cF_{\ft}(\varphi)(\sigma-\lambda \tau_0,y)\d^3 y\d^3 k\\
&:=I_1+I_2.
\end{align*}
We first consider $I_1$. If $\vert x-y\vert\le \epsilon$, then we have 
\begin{align*}
\vert x-x_0\vert\le \vert y-x\vert+\vert y-x_0\vert\le 2\epsilon,
\end{align*}
on the support of the integrand,
because $\varphi(\ft,y)=0$ if $\vert y-x_0\vert\ge \epsilon$. 
Note that \eqref{insigma} entails $\abs{h\lambda-1}\le \lambda \delta h$. As $g$ is zero on $B_{2\epsilon}((x_0,k_0))$, this entails that we have on the support of the integrand 
\begin{align*}
&h\abs{k}+\delta\lambda h\abs{k_0}\ge \abs{hk+(h\lambda-1)k_0}\ge 2\epsilon\\
&\Rightarrow h\abs{k}\ge 2\epsilon-\delta\abs{k_0}\lambda h\ge 2\epsilon-\frac{\delta}{1-\delta}\abs{k_0}\, ,
\end{align*}
and thus 
\begin{align}
\label{eqk}
\abs{k}&\ge (1-\delta)\left(2\epsilon-\frac{\delta}{1-\delta}\abs{k_0}\right)\lambda=((1-\delta)2\epsilon-\delta\abs{k_0})\lambda\ge \epsilon \lambda
\end{align}
for $\delta$ sufficiently small. In the following, we will choose $\delta>0$ small enough so that \eqref{eqk} holds. 
Performing the integration in $y$, we find 
\begin{align*}
I_1={\frac{1}{(2\pi)^{3/2}}}\int\limits_{\mathclap{\vert k\vert\ge \epsilon \lambda}}e^{-\i(\lambda k_0+k)x}g(x,h (k+\lambda k_0))\cF_y\left(\cF_{\ft}(\varphi)(\sigma-\lambda \tau_0,y)\psi\left(\frac{\vert y-x\vert}{\epsilon}\right)\right)(k)\d^3 k.
\end{align*}
Thus, 
\begin{align*}
\vert I_1(\sigma,x)\vert\lesssim \langle \sigma-\lambda\tau_0\rangle^{-N}\langle\lambda\rangle^{-N}
\end{align*}
for all $N>0$, because  $k\mapsto \cF_y\left(\cF_{\ft}(\varphi)(\sigma-\lambda \tau_0,y)\psi\left(\frac{\vert y-x\vert}{\epsilon}\right)\right)(k)$ is a Schwartz function. The estimate is uniform in $x$. To see this, put 
\begin{align*}
\cN_p(\varphi)=\sum_{\vert \alpha\vert\le p,\vert \beta\vert\le p}\Vert x^{\alpha}\partial^{\beta}\varphi\Vert_{L^{\infty}},
\end{align*}
and remember that 
\begin{align*}
\cN_p(\cF(\psi))\lesssim \cN_{p+4}(\psi)
\end{align*}
for all $\psi\in \cS(\rr^4)$. Using the properties of $g$, we obtain the same estimates for $\partial^{\alpha}_xI_1(x,\sigma)$ for all $\alpha\in \nn^3$. 
This gives the desired estimate. Indeed, we have for $N$ large enough and $\lambda>1$
\begin{align*}
\int\limits_{\mathclap{(1-\delta)\lambda\le \vert \sigma\vert\le (1+\delta)\lambda}} \langle \sigma\rangle^{M}\langle \sigma-\lambda\tau_0\rangle^{-N}\langle \lambda\rangle^{-N}\d\sigma\lesssim \langle \lambda\rangle^{M+1-N}. 
\end{align*}

To estimate the second integral, we put $\widetilde{\psi}(x,y)=\left(1-\psi\left(\frac{\vert y-x\vert}{\epsilon}\right)\right)$ and use integration by parts,
\begin{align*}
I_2&=\frac{1}{(2\pi)^{3}}\int \frac{-\widetilde{\psi}(x,y)}{\vert y-x\vert^2}\Delta_k(e^{\i k(y-x)})e^{-\i\lambda k_0x}g(x,h(k+\lambda k_0))\cF_{\ft}(\varphi)(\sigma-\lambda \tau_0,y)\d^3 k\d^3 y\\
&=h^2\frac{1}{(2\pi)^{3}}\int \frac{-\widetilde{\psi}(x,y)}{\vert y-x\vert^2}e^{\i k(y-x)}e^{-\i\lambda k_0x}(\Delta_kg)(x,h(k+\lambda k_0))\cF_{\ft}(\varphi)(\sigma-\lambda \tau_0,y)\d^3 k\d^3 y.
\end{align*}
Note that the boundary terms vanish because they are of the form $\lim\limits_{k_1^0\rightarrow \infty}\cB_1(k_1^0,k_2,k_3)$ with
\begin{eqnarray*}
\cB_1(k_1^0,k_2,k_3)
&=&\frac{1}{(2\pi)^{3}}\int\limits_y\int\limits_{k_2,k_3}\widetilde{\psi}(x,y)\frac{(y_1-x_1)}{\i\vert y-x\vert^2}e^{\i k(y-x)}e^{-\i k_0\lambda x}g(x,h((k_1^0,k_2,k_3)+\lambda k_0))\\
&\times&\cF_{\ft}(\varphi)(\sigma-\lambda\tau_0,y)\d k_2\d k_3\d^3 y\\
&=&\frac{1}{(2\pi)^{3/2}}\int\limits_{k_2,k_3}e^{-\i(\lambda k_0+k)x}g(x,h((k_1^0,k_2,k_3)+\lambda k_0))\\
&\times&\cF_y\left(\cF_{\ft}(\sigma-\lambda \tau_0,y)(\varphi)\widetilde{\psi}(x,y)\frac{(y_1-x_1)}{\i\vert y-x\vert^2}\right)(k_1^0,k_2,k_3)\d k_2\d k_3 
\end{eqnarray*}
and this last term goes to zero for each $x$ as $\vert k_1^0\vert$ goes to infinity, because the function $y\mapsto \cF_{\ft}(\varphi)(\sigma-\lambda\tau_0,y)\widetilde{\psi}(x,y)\frac{(y_1-x_1)}{i\vert y-x\vert^2}$ is Schwartz and so is its Fourier transform. We can now write 
\begin{align*}
I_2&=h^2\frac{1}{(2\pi)^{3/2}}\int e^{\i(\lambda k_0-k)x}(\Delta_kg)(x,h(k+\lambda k_0))\cF_y\left(\cF_{\ft}(\varphi)(\sigma-\lambda \tau_0,y)\frac{\widetilde{\psi}(x,y)}{\vert y-x\vert^2}\right)(k)\d^3 k\, ,
\end{align*}
and therefore $\vert I_2(x,\sigma)\vert \lesssim \vert \sigma\vert^{-2}\langle \sigma-\lambda\tau_0\rangle^{-N}$
for all $N>0$. Repeating the argument, we find 
 \begin{align}
\vert I_2(x,\sigma)\vert \lesssim \vert \sigma\vert^{-N'}\langle \sigma-\lambda\tau_0\rangle^{-N}
\end{align}
for all $N',N>0$. Again, the estimate is uniform in $x$. Now we have 
\begin{align*}
\int\limits_{\mathclap{(1-\delta)\le \abs{\sigma}\le (1+\delta)\lambda}}\langle\sigma\rangle^{M-N'}\langle \sigma-\lambda \tau_0\rangle^{-N}\d\sigma\le \langle \lambda\rangle^{M-N'+1}.
\end{align*}

  Using the properties of $g$, we obtain the same estimates for $\partial_x^{\alpha}I_2(x,\sigma)$ for all $\alpha\in \nn^3$. We then obtain for $\ell<-1/2$
\begin{align*}
\Vert I_2((x,\omega),\sigma)\Vert_{L^2(\{(1-\delta)\lambda\le \vert\sigma\vert\le (1+\delta)\lambda\}, \langle \sigma\rangle^M\d\sigma;\bar{H}^{r,\ell-1}_{(b),h})}\lesssim\langle \lambda\rangle^{-N}
\end{align*}
for all $N>0$, and a similar estimate for $I_1$. 
We then find 
\begin{align*}
\Vert {\rm Op}_h(g)\hat{f}^{\lambda}\Vert_{L^2\left(\{(1-\delta)\lambda\le \vert\sigma\vert\le (1+\delta)\lambda\},\langle \sigma\rangle^M\d\sigma \d\omega_+;\bar{H}^{r,\ell-1}_{(b),h}\right)}\le C_N\langle\lambda\rangle^{-N}\, .
\end{align*} 
Adding \eqref{10.58a} we have 
\begin{align}
\label{eqGflambda}
\Vert {\rm Op}_h(g)\hat{f}^{\lambda}\Vert_{L^2\left(\{\vert \sigma\vert\ge 1\},\langle \sigma\rangle^M\d\sigma \d\omega_+;\bar{H}^{r,\ell-1}_{(b),h}\right)}\le C_N\langle\lambda\rangle^{-N} 
\end{align} 
for all $N>0$. We finish the estimate by noting that 
\begin{align*}
h^{-2}\hat{T}_{s,h}(z)\hat{u}^{\lambda}=\hat{T}_s(\sigma)\hat{u}^{\lambda}=\hat{f}^{\lambda}.
\end{align*}  
It remains to estimate the residual error $h^N\Vert E\hat{f}^{\lambda}\Vert_{\bar{H}^{r,l-1}_{(b),h}}\lesssim h^N \Vert \hat{f}^{\lambda}\Vert_{\bar{H}^{r,\ell-1}_{(b),h}}$. We therefore have to estimate $\Vert\hat{f}^{\lambda}\Vert_{\bar{H}^{r,\ell-1}_{(b),h}}$. Weights in $x$ are absorbed by the compact support of $f^{\lambda}$. It is therefore sufficient to estimate 
\begin{align*}
\int\langle h(k+\lambda k_0)\rangle^{2r}\vert \hat{\varphi}(\sigma-\lambda\tau_0,k)\vert^2\d^3 k\lesssim \langle \sigma-\lambda\tau_0\rangle^{-N}\langle\lambda\rangle^{2r} 
\end{align*}
for all $N>0$. We conclude that 
\begin{align*}
h^N\norm{\hat{f}^{\lambda}}_{\bar{H}^{r,\ell-1}_{(b),h}}\lesssim \vert\sigma\vert^{-N}\langle \sigma-\lambda\tau_0\rangle^{-N}\langle\lambda\rangle^{2r}
\end{align*}
for all $N>0$. 
Let now $M,\, N'>0$ be given. We have 
\begin{align}
\label{estim1}
&\int\langle\sigma\rangle^{M-N}\langle \sigma-\lambda \tau_0\rangle^{-N}\d\sigma\le \int\limits_{\mathclap{\frac{1}{2} \lambda\le \vert \sigma\vert\le 2 \lambda}} \langle\sigma\rangle^{M-N}\langle \sigma-\lambda \tau_0\rangle^{-N}\d\sigma\\\nonumber
&+\int\limits_{\mathclap{\vert \sigma\vert\le \frac{1}{2}\lambda}} \langle\sigma\rangle^{M-N}\langle \sigma-\lambda \tau_0\rangle^{-N}\d\sigma+\int\limits_{\mathclap{\vert \sigma\vert\ge 2 \lambda}} \langle\sigma\rangle^{M-N}\langle \sigma-\lambda \tau_0\rangle^{-N}\d\sigma\\\nonumber
&\lesssim\langle \lambda\rangle^{M-N+1}+\langle\lambda\rangle^{-N+1}+\langle \lambda\rangle^{M+1-2N}
\end{align}
for $N$ large enough and $\lambda>1$. Using \eqref{estim1}, we see that we have for $N>0$ sufficiently large
\begin{align}
\label{eqhMulambda}
h^N\norm{E\hat{f}^{\lambda}}_{L^2\left(\{\vert\sigma\vert\ge 1\},\langle \sigma\rangle^M\d\sigma;\bar{H}^{r,\ell-1}_{(b),h}\right)}&\lesssim h^N\norm{\hat{f}^{\lambda}}_{L^2\left(\{\vert\sigma\vert\ge 1\},\langle \sigma\rangle^M\d\sigma ;\bar{H}^{r,\ell-1}_{(b),h}\right)}\\\nonumber
&\le C_N\langle\lambda\rangle^{M+1-N}.
\end{align}  

In the next step, we estimate $h^N\Vert \hat{u}^{\lambda}\Vert_{\bar{H}^{-N,\ell}_{(b),h}}$. We have $\Vert \hat{u}^{\lambda}\Vert_{\bar{H}^{-N,\ell}_{(b),h}}\le \Vert \hat{u}^{\lambda}\Vert_{\bar{H}^{r,\ell}_{(b),h}}$.
We then use \cite[Proposition 6.32]{Millet2} to get 
\begin{align}
\Vert \hat{u}^{\lambda}\Vert_{\bar{H}^{r,\ell}_{(b),h}}\le \Vert \hat{T}^{\lambda}_s(\sigma)\hat{u}^{\lambda}\Vert_{\bar{H}^{r,\ell-1}_{(b),h}}=\Vert\hat{f}^{\lambda}\Vert_{\bar{H}^{r,\ell-1}_{(b),h}}.
\end{align}
By \eqref{eqhMulambda}, we obtain 
\begin{align*}
\norm{ h^{N}\hat{u}^{\lambda}}_{L^2\left(\{\vert\sigma\vert\ge 1\},\langle \sigma\rangle^M\d\omega;\bar{H}^{-N',\ell}_{(b),h}\right)}\le C_N\langle\lambda\rangle^{M+1-N}
\end{align*}   
for all $N$. This finishes the proof.
\qeds
\end{proof}

 With this result, we can prove
\begin{proposition}
\label{prop:Had geo --}
    Let $(y_0,\xi_0)\in \dot{T}^*\MI$ lie on a bicharacteristic connecting $\sI_-$ and $ \sI_+$. Then $(y_0,\xi_0;y_0,-\xi_0)$ is a direction of rapid decrease for $W^\pm_{\sH}$.
\end{proposition}

\begin{proof}
    Since $\cB(s,w)$ are complex line bundles, any non-vanishing smooth section constitutes a local basis. Consider thus some $\varphi(\ft, x)\in \Gamma_c(\cB(s,s))$, and $\psi(\ft,x)\in \Gamma_c(\cB(s-,s))$ and use them to construct $f^\lambda(\ft,x)\in \Gamma_c(\cB(s,s))$ and $h^\lambda(\ft,x)\in\Gamma_c(\cB(s,-s))$ as in \eqref{WH1}. We need to show that
    \begin{equation}
        \abs{W^\pm_\sH(\phi_{\overline{f^\lambda}}, \phi_{h^\lambda})}=\cO(\lambda^{-\infty})\, ,
    \end{equation}
    where we have used the notation \eqref{eq:notation soln of test fct}. If we assume that $u^\lambda$ and $v^\lambda$ are backwards solutions of $\overline{T_{-s}}u^\lambda=\overline{f^\lambda}$ and $\overline{T_{-s}}v^\lambda= h^\lambda$, then one has
    \begin{equation}
        T_\sH\phi_{\overline{f^\lambda}}=({\tho}^{2s}u^\lambda\vert_{r=r_+}, A_su^\lambda\vert_{r=r_+})\, , \quad T_\sH\phi_{h^\lambda}=(({\tho}^{2s}v^\lambda\vert_{r=r_+}, A_sv^\lambda\vert_{r=r_+})\, .
    \end{equation}
    We recall that, due to the support properties of the backwards solutions derived from $f^\lambda$ and $h^\lambda$, we thus have
    \begin{equation}
        \abs{W^\pm_\sH(\phi_{\overline{f^\lambda}}, \phi_{h^\lambda})}\leq \norm{{\tho}'^{s}({\tho}^{2s}u^\lambda)\vert_{r=r_+}}_{L^2(\sH_-)}\norm{X_\pm(\i\Theta_U){\tho}'^s({\tho}^{2s}v^\lambda)\vert_{r=r_+}}_{L^2(\sH_-)}\, .
    \end{equation}
By fixing the normalization of the null vectors $l$ and $n$ of the tetrad to those of the Kruskal tetrad, we can estimate
\begin{equation}
    \abs{W^\pm_\sH(\phi_{\overline{f^\lambda}}, \phi_{h^\lambda})}\leq \norm{\partial_U^{s}\abs{U}^s({\tho}^{2s}u^\lambda)^{K}\vert_{r=r_+}}_{L^2(\sH_-)}\norm{\partial_U^{s+1}\abs{U}^s({\tho}^{2s}v^\lambda)^{K}\vert_{r=r_+}}_{L^2(\sH_-)}\, ,
\end{equation}
where the $K$-superscript denotes quantities computed in the Kinnersley tetrad. Next, we recall that, on $\sH$, one has $U\partial_U\propto \partial_{{}^*t}=\partial_\ft$, and that, consequently, $\partial_U^s\abs{U}^s$ can be written as a polynomial of order $s$ in $\partial_{{}^*t}$, see the proof of Lemma~\ref{lemma:well def wH}.
Then, making use of Parseval's theorem and the identification of $\cB(s,0)\vert_{\sH_-}$ with $\rr_\ft\times\cB_s^{\ss^2}$ using the Kruskal tetrad, we find 
  \begin{align}
      \abs{W^\pm_\sH(\phi_{\overline{f^\lambda}}, \phi_{h^\lambda})}\leq& C \norm{\widehat{({\tho}^{2s}u^\lambda)^{K}}(\sigma, r=r_+, \omega_+)}_{L^2(\rr_{\sigma}\times \ss^2_{\omega_+}; \langle \sigma\rangle^2 \d\sigma \d\omega_+)}\\\nonumber
      &\times \norm{\widehat{({\tho}^{2s}v^\lambda)^{K}}(\sigma, r=r_+, \omega_+)}_{L^2(\rr_{\sigma}\times \ss^2_{\omega_+}; \langle \sigma\rangle^3 \d\sigma \d\omega_+)}
  \end{align}
  The desired result then follows from Proposition~\ref{propRegHor}. 
\qeds
\end{proof}

\paragraph{\bf Proof of Theorem~\ref{thm:Had}}

The Proof of Theorem~\ref{thm:Had} now follows immediately by combining the reduction to the diagonal in Corollary~\ref{cor:Had reduction} with the results of Lemma~\ref{lemma:passivity Had proof}, Lemma~\ref{lemma:Had geodesics sH}, and Proposition~\ref{prop:Had geo --}.

\qeds
    
\appendix

\section{List of Notations}
In this appendix, we will compile a list of notations, ordered by topic, for the convenience of the reader.

Background and bundles of spin-weighted quantities:
\begin{itemize}
    \item $(M,g,\fO, \fT)$ spacetime (with orientation $\fO$ and time orientation $\fT$)
    \item $T_\cc M$ complexified tangent bundle
    \item $V^\pm$ principal null direction (of multiplicity 2)
    \item $\fN$ bundle of oriented future-directed principal null frames
    \item $\fM=(M,g,\fO,\fT, \fV)$ background comprising a spacetime and a choice $\fV$ of outgoing principal null direction
    \item $\fN_0$ connected component of $\fN$ on which $\fV(l)=0$
    \item $\fN_{0,r}$ quotient of $\fN_0$ with respect to the action of $\rr^*$ obtained by fixing the null vectors $l$ and $n$ of the tetrad
    \item $\cB(s,w)$ bundle of spin-weighted scalars with spin weight $s$ and boost weight $w$
    \item $\cB(s)$ bundle of spin-weighted scalars with spin weight $s$ as an associate $U(1)$-bundle to $\fN_{0,r}$
    \item $\cV_s=\cB(s,s)\oplus \cB(s,-s)$ 
    \item $\fV_{\ss^2}$ bundle of oriented orthonormal frames of the sphere $\ss^2$
    \item $\cB_s^{\ss^2}$ bundle of spin-weighted scalar over the sphere (associated to $\fV_{\ss^2}$)
\end{itemize}

GHP-related quantites:
\begin{itemize}
    \item $\Theta_a$ connection on $\cB(s,w)$
    \item $w_a$ connection form of $\Theta_a$
    \item $\tho$, ${\tho}'$, $\edt$, ${\edt}'$ differential operators of GHP formalism
    \item $ \rho$, $\tau$, $\rho'$, $\tau'$ GHP spin coefficients
    \item $T_s$ Teukolsky operator of spin $s$ acting on sections of $\cB(s,s)$
    \item $\cP_s=T_s\oplus \overline{T_{-s}}$ extended Teukolsky operator acting on sections of $\cV_s$
    \item $\Psi_2$ non-vanishing contraction of the Weyl tensor
    \item $\xi$, $\eta$ Killing vector fields (constructed from the Killing spinor $\kappa_{AB}$)
    \item $\zeta$ non-zero component of the Killing spinor, $\sim \Psi_2^{-1/3}$
    \item $\cL_X$ Lie derivative wrt $X$ acting on sections of $\cB(s,w)$
    \item $\cR_s$ and $\cS_s$ symmetry operators of $\varrho^2T_s$
    \item $\Gamma_a=B_a-w_a$ Teukolsky "potential"
    \item $\Phi_{-s,\IRG}$ Hertz-potential in ingoing radiation gauge
    \item $A_s$ ($\widetilde{A}_s$) radial Teukolsky-Starobinsky operators acting on $\cB(s,-s)$ ($\cB(s,s)$)
    \item $\sL_n^\pm$ angular differential operators following the notation of \cite{CTdC}
    \item $B_s$ radial operator connecting solutions of the spin-$s$ Teukolsky equation to solutions of the spin-$-s$ equation
\end{itemize}

Solution spaces on the bulk:
\begin{itemize}
    \item $E^\pm_s$ retarded/advanced Green operators for $T_s$
    \item $\Delta_s^\pm$ retarded/advanced Green operators for $\cP_s$
    \item $TS_s(\cM)=\Gamma_c(\cV_s)/\cP_s\Gamma_c(\cV_s)$ with charged symplectic form $(\cdot,\cdot)_{\Delta_s}$
    \item $\Sol_{s}(\cM)\subset \Gamma_{sc}(\cV_s)$ space of space-compact solutions to $\cP_s \phi=0$ with charged symplectic form $\sigma_s$
    \item $\Sol_{s,p}(\cM)\subset \Sol_{s}(\cM)$ subspace of physical solutions in which the spin-$s$ and spin-$-s$ components derive from common Herz potential
    \item $\cT_{\pm s,sc}(\cM)\subset \Gamma_{sc}(\cB(s,\pm s))$ space of space-compact solutions to $T_s\phi=0$ ($\overline{T_{-s}}\phi=0$)
    \item  $\cS_s(\cM)$ space of functions that show same decay behaviour as those in $\Sol_s(\cM)$ 
\end{itemize}

Sobolev spaces:
\begin{itemize}
    \item $H^m_{[s]}(\ss^2)$ Sobolev space of spin-weighted functions on the sphere
    \item $H^r_{b/sc}(E)$ Sobolev space of sections of $E\to M$ defined using a boundary/scattering volume form and vector fields from the boundary/scattering tangent space.
    \item $\bar{H}^{m,\mu}_{b/sc}(E)=x^\mu \bar{H}^m_{b/sc}(E)$ weighted Sobolev space for sections of $E\to M$, when $M\sim [0,1)_x\times X$ locally; the bar indicates extendibility at $x\to 1$
    \item $H^r_{b/sc,h}$ semiclassical Sobolev spaces, where each derivative in the norm is also scaled by $h$
    \item $H^r_{(b)}$ b-Sobolev space, but using scattering volume form
    \item $\bar{H}^{r, \mu, \nu}_{b}([0,1)_v\times[0,1)_\tau\times\cB_s^{\ss^2})=v^\mu\tau^\nu \bar{H}^{r, }_{b}([0,1)_v\times[0,1)_\tau\times\cB_s^{\ss^2})$ double-weighted Sobolev space of spin-weighted functions
\end{itemize}

Function spaces on the boundary:
\begin{itemize}
    \item $\cS_s(\sI_-)\subset \Gamma_{\sI_-}(\cV_s)$ space of smooth $\cV_s$ sections on $\sI_-$ obeying polynomial decay towards $i^-$ with charged symplectic form $\sigma_{\sI,s}$
    \item $\cS_s(\sH)\subset \Gamma_{\sH}(\cV_s)$ space of smooth $\cV_s$ sections on $\sH$ obeying polynomial decay towards $i^-$ with charged symplectic form $\sigma_{\sH,s}$
    \item $\cS_{s,p}(\sI_-)\subset \cS_s(\sI_-)$ physical boundary space on $\sI_-$ in which spin-$s$ and spin-$-s$ are related and into which the space of physical solutions is mapped by taking a trace
    \item $\cS_{s,p}(\sH)\subset \cS_s(\sH)$ physical boundary space on $\sH$ in which spin-$s$ and spin-$-s$ are related and into which the space of physical solutions is mapped by taking a trace
    \item $\pi_\pm$ projection to the spin-$\pm s$ component of functions in $\cS_{s,p}(\sI_-)$ ($\cS_{s,p}(\sH)$)
\end{itemize}

\bibliographystyle{habbrv}
\bibliography{teukolsky}
\end{document}